\documentclass{biometrika}

\usepackage{amsmath,amsfonts,mathtools}
\usepackage{bm}
\usepackage{soul}
\usepackage{float}
\usepackage{enumitem}
\usepackage{newtxtext}
\usepackage[subscriptcorrection]{newtxmath}
\usepackage{url}
\usepackage{xr}

\graphicspath{{Fig/}}

\newenvironment{reptheorem}[1]
  {\par\medskip\noindent\textbf{Theorem~\ref{#1}.}\ \itshape}
  {\par\medskip}
\newenvironment{replemma}[1]
  {\par\medskip\noindent\textbf{Lemma~\ref{#1}.}\ \itshape}
  {\par\medskip}

\newenvironment{repproposition}[1]
  {\par\medskip\noindent\textbf{Proposition~\ref{#1}.}\ \itshape}
  {\par\medskip}

\makeatletter
\def\ps@plain{%
  \let\@oddhead\@empty
  \let\@evenhead\@empty
  \def\@oddfoot{\hfil\thepage\hfil}
  \let\@evenfoot\@oddfoot
}
\makeatother

\newcommand\ci{\perp\!\!\!\perp}

\begin{document}

\markboth{E. V. Strobl}{COMPACT spectral adjustment scores}

\title{COMPACT: Spectral Adjustment Scores from an Exhaustive and Nonredundant Set of Invariant Relations}

\author{E. V. STROBL}
\affil{Department of Biomedical Informatics, University of Pittsburgh}

\maketitle

\begin{abstract}
Observational datasets frequently contain many baseline variables, yet investigators estimating causal effects may not know which variables to include in the adjustment set. Confounding information may also be distributed weakly across many variables. Propensity scores can simplify adjustment by reducing high-dimensional covariates to a scalar with binary treatment. Although the propensity score is the coarsest balancing score, this balancing property does not necessarily distinguish among causal structures relevant to adjustment. We instead examine all causal graphs among a candidate score, treatment, and outcome while allowing latent variables. Under faithfulness, we identify all unconditional and conditional dependence relations that remain unchanged when the treatment-to-outcome edge is added or removed, leaving treatment-effect estimation to the downstream analysis. Together, these relations define the smallest compatible graph class expressible through the relations considered. We then develop the proposed algorithm, which implements this criterion through a generalized eigenvalue problem whose score space targets the span of a treatment-balancing coordinate and a residual-outcome coordinate. We show that sufficiently informative proxies can recover this span without direct observation of the adjustment variables, characterize the resulting score-recovery and downstream causal-estimation errors, and establish bootstrap validity for the full procedure. Simulations and a real-data application demonstrate superior performance over several alternatives.
\end{abstract}


\section{Introduction}

Confounding remains a central obstacle to causal inference from observational
data. In many applications, treatment assignment and potential outcomes depend
on an underlying baseline state observed only indirectly through many variables
in $\bm X$. Clinical treatment selection, for example, may reflect numerous
individually weak cues involving symptoms, severity, comorbidity, prior
response, tolerance, and perceived benefit \citep{Keepers20,Bauer13}. Adjuvant cancer treatment may depend on a multivariate tumor phenotype reflected
by stage, nodal involvement, grade, receptor status, proliferation markers, and
genomic recurrence scores \citep{Cardoso19}. Consequently, investigators often
adjust for broad collections of baseline variables to capture distributed
confounding information \citep{Schneeweiss09,Zhang22}.

Direct adjustment on $\bm X$ becomes difficult when
adjustment-relevant information is distributed across
many noisy, correlated variables \citep{Stuart10,Lockwood16}.
For binary treatment, propensity scores address this problem by compressing
$\bm X$ into $\mathbb{P}(T=1\mid\bm X)$ \citep{Rosenbaum83}. However, this
requires estimation of the entire treatment law, including treatment-predictive
variation that may be irrelevant to outcome risk. Flexible treatment models can
therefore absorb treatment-only or sample-specific variation, reduce residual
treatment variation, and drive fitted propensities toward zero or one without a
commensurate improvement in confounding control
\citep{Brookhart06,DAmour21}.

We instead ask which score-level dependence relations remain informative about adjustment-relevant structure without requiring knowledge of whether treatment causes the outcome.
Treatment association alone does not distinguish scores useful for adjustment from
scores that predict treatment without carrying outcome-relevant information. For
example, both
$S_{\bm{\beta}}\longrightarrow T$ with 
$Y\ \text{isolated}$
and
$S_{\bm{\beta}}\longrightarrow T\longrightarrow Y$
imply $S_{\bm{\beta}}\not\!\ci T$ under faithfulness, although $S_{\bm{\beta}}$ lies on no
backdoor path between $T$ and $Y$. A treatment-prediction method may nevertheless
devote substantial complexity to such variation. We therefore develop COMPACT
(Confounding-Oriented Multivariate Proxy Adjustment for Causal Targets), which
requires candidate score directions to satisfy
$S_{\bm{\beta}}\not\!\ci T$ and 
$S_{\bm{\beta}}\not\!\ci Y\mid T.$
We show that these relations are compatible with a strictly smaller class of causal
structures than treatment association alone, while remaining unchanged when the
edge $T\rightarrow Y$ is added or removed. The outcome relation therefore prevents treatment
prediction alone from determining the learned representation.

COMPACT also exploits the fact that causal adjustment need not recover a
particular balancing score coordinate. It is sufficient to recover a
low-dimensional space containing the relevant adjustment information.
Specifically, COMPACT uses a generalized eigendecomposition to estimate
$\mathcal U
=
\operatorname{span}(U_T,U_R),$
the joint span of a treatment-balancing coordinate $U_T$ and a residual-outcome coordinate $U_R$. Because the estimand is the span rather than any uniquely
labeled basis, the procedure avoids unnecessary coordinate-identification and
calibration requirements. In the scalar-treatment, scalar-outcome setting,
this space is two-dimensional, making the learned adjustment structure
directly visualizable. The same representation can be used to examine
adjustment-relevant gradients, treatment overlap, and variation in conditional
treatment effects. We provide a more detailed comparison of COMPACT with a broader range of related methods in Supplementary Materials~\ref{supp:related_work}.

Our main contributions are as follows.
\begin{enumerate}[leftmargin=1.25em,label=(\alph*),labelsep=0.4em,itemsep=0.25em]

\item We identify an exhaustive and nonredundant set of score-level dependence
relations whose truth values remain unchanged when the treatment-to-outcome edge
is added or removed (Section~\ref{sec:criterion}).

\item We introduce COMPACT, which operationalizes these restrictions through a
product-based ``AND'' criterion and solves the resulting problem with a single
generalized eigendecomposition (Section~\ref{sec:compact-algorithm:spectral}).

\item We show that the retained directions span regularized linear predictors of
treatment and treatment-residualized outcome, and connect this eigenspace to
latent treatment-balancing and residual-outcome variables
(Section~\ref{sec:compact-algorithm:latent_adj}).

\item We characterize sampling error, proxy-recovery error, and the resulting
downstream treatment-effect estimation error, establishing conditions for consistent estimation of both the average treatment effect (ATE)
and conditional average treatment effect (CATE) after adjustment for the learned scores (Sections~\ref{sec:theory:recovery}--\ref{sec:theory:downstream}).

\item We establish validity of a full-procedure nonparametric bootstrap that
relearns the COMPACT scores and refits the downstream analysis, thereby
propagating uncertainty through the complete estimation procedure
(Section~\ref{sec:theory:bootstrap}).

\item We evaluate COMPACT in fully synthetic and semi-synthetic experiments,
including treatment-effect estimation, latent-coordinate recovery, treatment
overlap, and bootstrap calibration. COMPACT achieves the best or competitive
performance across these criteria while remaining computationally fast
(Section~\ref{sec:empirical}).

\end{enumerate}
Together, the theoretical and empirical results provide a framework
for replacing unstable high-dimensional proxy adjustment with a compact
representation targeted specifically to the requirements of causal-effect
estimation.

\section{Motivating Example} \label{sec:motivating}

We first motivate COMPACT with an application to the Adolescent Brain Cognitive Development (ABCD) Study, a large longitudinal observational cohort of children in the United States followed from late childhood through adolescence \citep{Jernigan18}. Greater reading and print exposure has been associated with higher vocabulary, general knowledge, and other measures of crystallized verbal ability, with longitudinal studies additionally associating reading for pleasure with subsequent vocabulary growth \citep{Cunningham91,Stanovich92,Sullivan15,Sun24}. However, children with greater pre-existing cognitive ability and more favorable developmental environments may also be more likely to read for pleasure, making the causal relationship unclear. We therefore sought to estimate the effect of initiating reading for pleasure on subsequent cognitive performance.

This setting presents a challenging adjustment problem because factors like demographics, baseline cognitive abilities, competing activities, family environment, sleep, screen use, physical characteristics, and mental health may all influence both reading behavior and cognition. We considered six NIH Toolbox Cognition outcomes \citep{Weintraub13}: crystallized cognition, Flanker Inhibitory Control and Attention, Picture Sequence Memory, Picture Vocabulary, Pattern Comparison Processing Speed, and Oral Reading Recognition. We restricted the sample to participants who did not read for pleasure at baseline, defined treatment by initiation of reading for pleasure at the subsequent wave, and evaluated cognition at the following wave. Outcome-specific sample sizes after excluding missing treatment or outcome data ranged from 1,785 to 2,637.

Because the appropriate adjustment set was not known a priori, we constructed a broad adjustment set containing 248 baseline variables: 15 demographic, socioeconomic, cognitive, and study-design variables, including the same six baseline NIH Toolbox measures; 9 Family Environment Scale conflict items; 120 Child Behavior Checklist variables; 14 screen-use variables; 26 sleep variables; 59 sports and activity variables; and 5 anthropometric measures. We fit COMPACT separately for each outcome and estimated both the ATE and COMPACT score-indexed CATE. We used 10,000 family-level cluster bootstrap replicates to account for sibling dependence and controlled the family-wise error rate separately across the six ATE and six CATE-heterogeneity tests using Holm's procedure.

For crystallized cognition, estimated propensity scores were almost uniformly non-extreme across the two-dimensional COMPACT space (Figure~\ref{fig:cognitive_analysis}A), indicating good treatment overlap. Initiating reading for pleasure was associated with higher subsequent crystallized cognition ($\mathrm{ATE}=0.589$, 95\% CI $[0.237,1.044]$, Holm-adjusted $p=0.0072$), whereas the global CATE-heterogeneity test was not significant. Consistent with this result, estimated CATEs varied little across the COMPACT adjustment space on a scale spanning $\pm0.25$ outcome standard deviations around the ATE (Figure~\ref{fig:cognitive_analysis}B).

\begin{figure}
    \centering
    \includegraphics[width=0.95\linewidth]{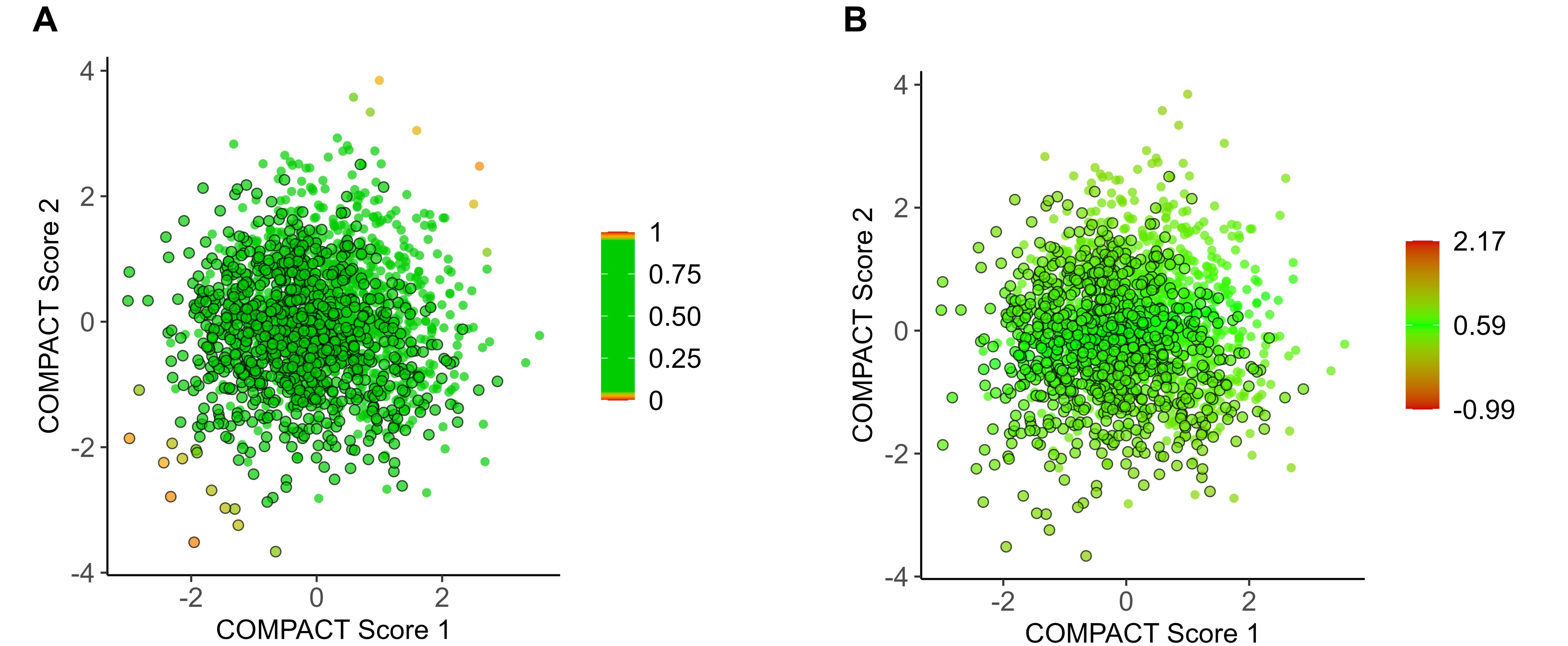}
    \caption{\textbf{COMPACT score plots for crystallized cognition.} COMPACT reduces the 248 baseline adjustment variables to a two-dimensional score space that can be visualized. Black outlines denote participants who initiated reading for pleasure. (A) COMPACT space colored by the estimated propensity score. (B) The same space colored by the estimated CATE; the color scale is centered at the ATE and spans $\pm0.25$ outcome standard deviations.}
    \par
    \raggedright
    \tiny
    \textbf{Alt text:} Two scatter plots of the two-dimensional COMPACT score space for crystallized cognition. Panel A shows participants colored by estimated propensity score, with black outlines marking participants who initiated reading for pleasure. Panel B shows the same score space colored by estimated CATE, with little variation around the ATE.

    \label{fig:cognitive_analysis}
\end{figure}

These results provide evidence consistent with a causal effect of initiating reading for pleasure on subsequent crystallized cognition while illustrating how COMPACT simultaneously provides a low-dimensional adjustment representation, an overlap diagnostic, and a space for examining treatment-effect heterogeneity. Results for the remaining cognitive outcomes are provided in Supplementary Materials~\ref{supp:abcd-additional}.

\section{Setup and notation} \label{sec:setup}
\subsection{Observed data, causal targets, and identification} \label{sec:setup:observed_data}
We observe independent units
$O_i=(Y_i,T_i,\bm X_i)$ for $i=1,\ldots,n$,
where $T_i$ denotes a treatment or exposure, $Y_i$ denotes a post-treatment
outcome, and $\bm X_i\in\mathbb R^p$ denotes a vector of pre-treatment baseline
features. The variables in $\bm X_i$ may be continuous or discrete. Let
\begin{equation} \nonumber
\begin{aligned}
X=(\bm X_1^\top,\ldots,\bm X_n^\top)^\top & \in  \mathbb R^{n\times p},
\qquad
\bm T=(T_1,\ldots,T_n)^\top\in\mathbb R^n,\\
\bm Y&=(Y_1,\ldots,Y_n)^\top\in\mathbb R^n .
\end{aligned}
\end{equation}
Throughout the spectral calculations, $\bm T$, $\bm Y$, and the columns of $X$ are
centered empirically. Treatment levels in the potential-outcome notation remain
on their original scale.

We first state the causal targets. For each treatment level $t
\in\mathcal T$, let $Y_i(t)$ denote the outcome that would be observed
for unit $i$ under treatment $t$. We assume consistency and no interference, so
that $Y_i=Y_i(T_i)$ and one unit's outcome does not depend on another unit's
treatment. For treatment levels $t,t_0\in\mathcal T$, the marginal causal
contrast is
$\tau(t,t_0)
=
\mathbb E\{Y_i(t)-Y_i(t_0)\}.$
In the binary-treatment case, this becomes the average treatment effect
$\tau=\mathbb E\{Y_i(1)-Y_i(0)\}$. Once a low-dimensional score
$\bm S_i=s(\bm X_i)$ has been learned, we can also consider the score-indexed causal
contrast
\begin{equation} \label{eq:CATE_S}
\tau_{\bm S}(t,t_0;\bm s)
=
\mathbb E\{Y_i(t)-Y_i(t_0)\mid \bm S_i=\bm s\}.
\end{equation}
Evaluating this contrast at $\bm S_i$ provides a low-dimensional score-indexed CATE while averaging over distinctions in $\bm X_i$ not retained by the score, thereby avoiding high-dimensional treatment-effect strata that often suffer from weak overlap.

Let $\bm C_i\in\mathbb R^{d_C}$ denote a centered pre-treatment baseline state that
is sufficient for adjustment, so that
\begin{equation} \label{eq:latent_exchangeability}
Y_i(t)\ci T_i \mid \bm C_i
\end{equation}
for all $t\in\mathcal T$. We also assume positivity on the treatment values and
covariate strata relevant to $\tau(t,t_0)$. For a continuous treatment, this
requires
\( 0<f_{T\mid \bm C}(t\mid \bm c)<\infty \)
on the relevant support. In the binary-treatment case, it reduces to
$0<P(T_i=1\mid \bm C_i=\bm c)<1$.

The full covariate-indexed causal contrast is
$\tau_{\bm C}(t,t_0;\bm c)
=
\mathbb E\{Y_i(t)-Y_i(t_0)\mid\bm C_i=\bm c\}.$
We will not recover this entire function. Instead, we seek to recover the
coarser contrast in \eqref{eq:CATE_S}, indexed by the information retained in
the low-dimensional score.

\subsection{Causal model with a single-index treatment law} \label{sec:setup:working_model}

We use the following model to make the population target explicit. Let
\(
\Sigma_{\bm C}=\mathbb E(\bm C_i\bm C_i^\top)
\)
be positive definite, and define the treatment-balancing variable
$U_{T,i}=\bm a^\top\bm C_i$
with
$\bm a^\top\Sigma_{\bm C}\bm a>0.$
We assume that the conditional treatment law depends on $\bm C_i$ only
through $U_{T,i}$:
\begin{equation}
\label{eq:linear-treatment-model}
P(T_i\in A\mid\bm C_i)
=
P(T_i\in A\mid U_{T,i})
\end{equation}
for every measurable set $A$. Let
$m_T(u)=\mathbb E(T_i\mid U_{T,i}=u)$ and 
$\varepsilon_{T,i}=T_i-m_T(U_{T,i}).$
Then $\mathbb E(\varepsilon_{T,i}\mid\bm C_i)=0$.

Equation~\eqref{eq:linear-treatment-model} includes both continuous and binary treatment types. For a continuous treatment, the additive-linear specialization
is
$T_i=U_{T,i}+\varepsilon_{T,i}$ with
$\varepsilon_{T,i}\ci\bm C_i,$
so that $m_T(u)=u$. For a binary treatment,
\[
T_i\mid\bm C_i
\sim
\operatorname{Bernoulli}\{g(U_{T,i})\},
\qquad
0<g(U_{T,i})<1,
\]
so that $m_T(u)=g(u)$. In either case, $U_{T,i}$ is a balancing variable.

When $m_T$ is nonlinear, including the binary-treatment case, we impose the
standard linearity condition
\begin{equation}
\label{eq:inverse-regression-linearity}
\mathbb E(\bm C_i\mid U_{T,i})
=
\frac{\Sigma_{\bm C}\bm a}
{\bm a^\top\Sigma_{\bm C}\bm a}
U_{T,i}.
\end{equation}
This condition holds when $\bm C_i$ follows a centered elliptically symmetric distribution with finite second moments, such as a multivariate Gaussian distribution or a multivariate $t$-distribution with more than two degrees of freedom \citep{Li91}.

Let the potential outcomes satisfy
\begin{equation}
\label{eq:linear-potential-outcome-model}
Y_i(t)
=
\alpha+\bm b^\top\bm C_i
+t\{\gamma+\bm q^\top\bm C_i\}
+\varepsilon_{Y,i}(t),
\end{equation}
where
\( \mathbb E\{\varepsilon_{Y,i}(t)\mid\bm C_i,T_i\}=0.\)
The potential-outcome errors are then conditionally mean independent
of the treatment residual $\varepsilon_{T,i}$ given $\bm C_i$. This condition
applies to both continuous and binary treatments and implies
\eqref{eq:latent_exchangeability} in mean:
\begin{equation}\label{eq:conditional-mean-restriction}
\mathbb E\{Y_i(t)\mid T_i,\bm C_i\}
=
\mathbb E\{Y_i(t)\mid\bm C_i\}.
\end{equation}
The model also implies the linear CATE
\begin{equation}
\label{eq:full-linear-cate}
\tau_{\bm C}(t,t_0;\bm C_i)
=
(t-t_0)\{\gamma+\bm q^\top\bm C_i\}.
\end{equation}

We will see that COMPACT combines the treatment-balancing
coordinate with a residual-outcome coordinate. Define
\begin{equation}
\label{eq:hq-definition}
\bm h_q
=
\Sigma_{\bm C}^{-1}
\mathbb E\left[
\bm C_i m_T(U_{T,i})(\bm q^\top\bm C_i)
\right],
\qquad
U_{R,i}=(\bm b+\bm h_q)^\top\bm C_i.
\end{equation}
The latent adjustment vector and its associated space are
\begin{equation}
\label{eq:latent-adjustment-vector}
\bm U_i=
\begin{pmatrix}
U_{T,i}\\
U_{R,i}
\end{pmatrix},
\qquad
\mathcal U
=
\operatorname{span}_{L_2}
\{U_{T,i},U_{R,i}\}.
\end{equation}
The latent adjustment space thus has dimension at most two.

\subsection{High-dimensional proxy variables} \label{sec:setup:proxies}

We assume that $\bm C_i$ is not observed directly. Instead, the observed
pre-treatment proxy vector $\bm X_i$ satisfies
\begin{equation}
\label{eq:linear-proxy-model}
\bm X_i
\ci
\{T_i,Y_i(t):t\in\mathcal T\}
\mid
\bm C_i,
\qquad
\mathbb E(\bm X_i\mid\bm C_i)
=
\Lambda\bm C_i,
\end{equation}
with $\mathbb E(\|\bm X_i\|^2)<\infty$. The first condition requires the
proxy measurements to be nondifferential: conditional on the latent baseline
state $\bm C_i$, they contain no additional information about treatment or the
potential outcomes. The second requires only a linear conditional mean and
does not impose an additive-error representation, homoskedasticity, or a
continuous proxy distribution. The rows of $\Lambda$ may correspond to many weak or redundant measurements
of the same latent baseline state. Throughout the population spectral calculations, $\bm X_i$ denotes the
centered proxy vector.

\section{Exhaustive and Nonredundant Invariant Relations}
\label{sec:criterion}

At the population level, we seek a coefficient
matrix $B\in\mathbb R^{p\times d}$, with
$d \in \{1,2\}$, whose score vector
$\bm S_{B,i}=B^\top\bm X_i$ defines the
adjustment score space
\[
\mathcal S_B
=
\left\{
\bm X_i^\top\bm\beta:
\bm\beta\in\operatorname{col}(B)
\right\}.
\]
The goal is to reconstruct the low-dimensional
adjustment information carried by the latent baseline
state $\bm C_i$ when the observed proxy system is
sufficiently informative. In particular, when a
nontrivial shared adjustment state is present, the
treatment-balancing score $U_{T,i}$ should lie
arbitrarily close to $\mathcal S_B$ as proxy
information increases.

As discussed in the Introduction, treatment
association alone does not distinguish adjustment
information from information that only predicts
treatment. For example,
\begin{equation}
\label{eq:pure-treatment-chain}
\bm S_{B,i}\longrightarrow T_i\longrightarrow Y_i
\end{equation}
implies $\bm S_{B,i}\not\!\ci T_i$ under
faithfulness even though the representation lies
on no backdoor path from $T_i$ to $Y_i$.
Moreover, adding unconditional outcome relevance
$\bm S_{B,i}\not\!\ci Y_i$ does not resolve this
problem because the same chain also induces
marginal dependence between $\bm S_{B,i}$ and $Y_i$.

We therefore derive the score-space criterion
from the temporal role of treatment. Consider
the six elementary pairwise dependence relations
involving $(\bm S_{B,i},T_i,Y_i)$, including
the three marginal relations and the three
relations conditional on the remaining variable:
\begin{equation}
\label{eq:primitive-dependence-relations}
\begin{split}
\mathcal Q(\bm S_{B,i},T_i,Y_i)=\{&
\bm S_{B,i}\not\!\ci T_i,
\ \bm S_{B,i}\not\!\ci Y_i,
\ T_i\not\!\ci Y_i,\\
&\bm S_{B,i}\not\!\ci T_i\mid Y_i,
\ \bm S_{B,i}\not\!\ci Y_i\mid T_i,
\ T_i\not\!\ci Y_i\mid\bm S_{B,i}
\}.
\end{split}
\end{equation}
We compare these restrictions through the graph
classes they represent. For any
$\mathcal A\subseteq
\mathcal Q(\bm S_{B,i},T_i,Y_i)$,
let $\mathfrak G(\mathcal A)$ denote the class
of directed acyclic graphs (DAGs) that satisfy
every relation in $\mathcal A$ under d-connection
and respect the temporal order
$\bm S_{B,i}\prec T_i\prec Y_i$.
The score vector is treated as a single
vector-valued node in these graphical comparisons.
The graphs may contain additional latent variables,
so we do not assume that
$(\bm S_{B,i},T_i,Y_i)$ are causally sufficient.
We assume that the observational distribution
is faithful to the corresponding graph and that
sampling induces no selection bias. By convention,
$\mathfrak G(\varnothing)$ contains all graphs
under consideration.

The existence of the edge $T_i\rightarrow Y_i$
is unknown when the score space is constructed.
We therefore retain only primitive relations
whose truth values do not depend on whether this
edge is present. For any graph $\mathbb G$ under
consideration that does not contain
$T_i\rightarrow Y_i$, define
$\mathbb G^+=\mathbb G\cup\{T_i\rightarrow Y_i\}$,
leaving every other node and edge unchanged.
We call a primitive relation
$q\in\mathcal Q(\bm S_{B,i},T_i,Y_i)$
treatment-edge invariant if
\begin{equation}
\label{eq:uniform-treatment-effect-invariance}
q\text{ holds in }\mathbb G
\quad\Longleftrightarrow\quad
q\text{ holds in }\mathbb G^+
\end{equation}
for every such $\mathbb G$. Because deleting
$T_i\rightarrow Y_i$ reverses the same
comparison,
\eqref{eq:uniform-treatment-effect-invariance}
covers both addition and deletion of the
treatment-to-outcome edge.

\begin{theorem}
\label{thm:complete-invariant-criterion}
Under the temporal ordering, faithfulness,
and absence of selection bias specified above,
define
\[
\mathcal Q_{\mathrm{inv}}
=
\left\{
\bm S_{B,i}\not\!\ci T_i,\;
\bm S_{B,i}\not\!\ci Y_i\mid T_i
\right\}.
\]
Then $\mathcal Q_{\mathrm{inv}}$ has the
following properties:

\begin{enumerate}
\item[\textnormal{(a)}]
\textbf{Exhaustiveness.}
The set $\mathcal Q_{\mathrm{inv}}$ contains
all treatment-edge-invariant primitive relations
in \eqref{eq:primitive-dependence-relations}:
\[
\mathcal Q_{\mathrm{inv}}
=
\left\{
q\in\mathcal Q(\bm S_{B,i},T_i,Y_i):
q\text{ is treatment-edge invariant}
\right\}.
\]

\item[\textnormal{(b)}]
\textbf{Nonredundancy.}
For every proper subset
$\mathcal A\subsetneq\mathcal Q_{\mathrm{inv}}$,
there exists a faithful ordered DAG
$\mathbb G_{\mathcal A}$ that satisfies every
relation in $\mathcal A$ but does not satisfy
every relation in $\mathcal Q_{\mathrm{inv}}$.
Equivalently, every proper subset represents
a strictly larger graph class:
\(
\mathfrak G(\mathcal Q_{\mathrm{inv}})
\subsetneq\mathfrak G(\mathcal A).
\)
\end{enumerate}
\end{theorem}

\noindent\textbf{Remark.}
Proofs are located in Supplementary
Materials~\ref{supp:proofs}.
Theorem~\ref{thm:complete-invariant-criterion}
formalizes two principles for constructing
adjustment scores: invariance to the unknown
edge $T_i\rightarrow Y_i$, whose existence is
the subject of downstream causal analysis,
and maximal graphical specificity among the
primitive dependence relations considered.
Exhaustiveness identifies all relations that
satisfy the invariance requirement, while
nonredundancy establishes that removing either
relation strictly enlarges the compatible
graph class. Their conjunction therefore
excludes more irrelevant causal structures
than either relation alone.

Importantly, this graphical characterization
does not establish sufficiency for causal
adjustment. For example, in
$\bm S_{B,i}\rightarrow T\leftarrow U\rightarrow Y$,
an instrumental set $\bm S_{B,i}$ satisfies both
$\bm S_{B,i}\not\!\ci T$ and $\bm S_{B,i}\not\!\ci Y\mid T$,
because conditioning on the collider $T$
opens the path between $\bm S_{B,i}$ and $Y$.
Nevertheless, adjusting for $\bm S_{B,i}$ does not
control confounding by $U$. Thus, Theorem~
\ref{thm:complete-invariant-criterion}
provides a maximally specific
treatment-edge-invariant criterion within
the considered class, rather than a
sufficient condition for causal identification.
Adjustment validity instead follows from
the subsequent latent-structure and
proxy-recovery assumptions, under which
the learned score space recovers the
treatment-balancing information required
for causal adjustment.

\vspace{5mm}
We now seek a score space
satisfying both invariant relations.
We operationalize
$\bm S_{B,i}\not\!\ci T_i$
using the squared multiple correlation between
treatment and the score vector.
For residual-outcome relevance, let
$Y_{i,\mathrm{res}}$ and
$\bm S_{B,i,\mathrm{res}}$
denote the population linear regression residuals
of $Y_i$ and $\bm S_{B,i}$, respectively,
after adjustment for treatment.
We measure
$\bm S_{B,i}\not\!\ci Y_i\mid T_i$
through the squared multiple correlation between
these residuals.

This operationalization imposes a
linear-detectability requirement: the underlying
dependencies must induce nonzero marginal and
partial linear associations rather than being
masked by nonlinearity or exact cancellation.
We define the population score-space criterion
\begin{equation}
\label{eq:population-product-objective}
\begin{aligned}
\mathcal J(B)
={}&
\left\{
R^2(T_i\sim\bm S_{B,i})
\right\}^{1/2}
\left\{
R^2(Y_{i,\mathrm{res}}
\sim\bm S_{B,i,\mathrm{res}})
\right\}^{1/2}.
\end{aligned}
\end{equation}
We seek to maximize $\mathcal J(B)$ over
coefficient matrices $B\in\mathbb R^{p\times d}$
with $d \in \{1,2\}$ and
$\operatorname{rank}(B)=d$, for which the
relevant score covariance matrices are nonsingular. Because multiple correlations depend only on
the predictor span, $\mathcal J(B)$ is invariant
under replacing $B$ by $BA$ for any nonsingular
$d\times d$ matrix $A$. Thus, it evaluates the
learned score space rather than the scaling
or orientation of its individual scores.

When the corresponding covariance matrices are
nonsingular, each squared multiple correlation
is positive exactly when its associated
cross-covariance vector is nonzero.
Consequently, under linear detectability,
\[
\mathcal J(B)>0
\quad\Longleftrightarrow\quad
\bm S_{B,i}\not\!\ci T_i
\ \text{and}\
\bm S_{B,i}\not\!\ci Y_i\mid T_i.
\]
The product therefore implements an ``AND''
criterion at the level of the entire score space.
A representation that lacks either treatment
information or residual-outcome information
receives zero population value, even if it
strongly predicts the other variable.
Under faithfulness, the conjunction corresponds
to $\mathfrak G(\mathcal Q_{\mathrm{inv}})$.

\section{COMPACT: A Spectral Approach}
\label{sec:compact-algorithm}

The population objective
\eqref{eq:population-product-objective}
evaluates the joint treatment and residual-outcome
information retained by a score space.
COMPACT maximizes a ridge-regularized version
of this objective, whose optimal score space
can be recovered exactly through a low-rank
generalized eigendecomposition.
The resulting representation contains at most
two nonredundant scores.
We select the ridge regularization parameter
by cross-validation using the original,
unregularized score-space objective.

\subsection{Spectral formulation}
\label{sec:compact-algorithm:spectral}

We construct the score space by maximizing a
ridge-regularized version of the population
objective
\eqref{eq:population-product-objective}.
The regularization stabilizes the spectral calculation
when the proxy covariance matrices are ill-conditioned
or singular, while preserving the joint treatment and
residual-outcome information targeted by the original
criterion.

Let $X$ denote the centered predictor matrix,
and let $\bm T$ and $\bm Y$ denote the centered
treatment and outcome vectors. Define the
treatment residual-maker
$M_T=I_n-\bm T(\bm T^\top\bm T)^{-1}\bm T^\top$.
The treatment-residualized outcome is
$\bm Y_{\mathrm{res}}=M_T\bm Y$. Define the treatment and residual-outcome
cross-covariance vectors
$\widehat{\bm v}_T=n^{-1}X^\top\bm T$
and
$\widehat{\bm v}_R=n^{-1}X^\top M_T\bm Y$,
and the corresponding outcome variances
$\widehat v_T=n^{-1}\bm T^\top\bm T$
and
$\widehat v_R
=n^{-1}\bm Y_{\mathrm{res}}^\top
\bm Y_{\mathrm{res}}$. Also define
$\widehat\Sigma_{\bm X}=n^{-1}X^\top X$
and
$\widehat\Sigma_{\bm X\mid T}
=n^{-1}X^\top M_TX$.
The covariance matrices of the original
and residualized scores are
$B^\top\widehat\Sigma_{\bm X}B$
and
$B^\top\widehat\Sigma_{\bm X\mid T}B$,
respectively.

For $\eta>0$, introduce the ridge-stabilized
covariance matrices
\[
\widehat K_\eta
=
\widehat\Sigma_{\bm X}+\eta I_p,
\qquad
\widehat K_{R,\eta}
=
\widehat\Sigma_{\bm X\mid T}+\eta I_p.
\]
Both are positive definite. We define the
regularized score-space objective
\begin{equation}
\label{eq:regularized-span-objective}
\begin{aligned}
\widehat{\mathcal J}_\eta(B)
={}&
\left\{
\frac{
\widehat{\bm v}_T^\top B
(B^\top\widehat K_\eta B)^{-1}
B^\top\widehat{\bm v}_T
}{\widehat v_T}
\right\}^{1/2}
\left\{
\frac{
\widehat{\bm v}_R^\top B
(B^\top\widehat K_{R,\eta}B)^{-1}
B^\top\widehat{\bm v}_R
}{\widehat v_R}
\right\}^{1/2}.
\end{aligned}
\end{equation}
The criterion is defined for full-column-rank
matrices $B$ and is invariant under replacing
$B$ by $BA$ for any nonsingular $d\times d$
matrix $A$. When $\eta=0$ and the relevant
covariance matrices are nonsingular, $\widehat{\mathcal J}_\eta(B)$ reduces
to the empirical counterpart of
\eqref{eq:population-product-objective}.

Remarkably, the two factors can be maximized
simultaneously using at most two score directions.
To see this, consider any positive-definite
matrix $K$ and vector $\bm v$. Let
$C=K^{1/2}B$ and $\bm w=K^{-1/2}\bm v$.
Then
\[
\begin{aligned}
&\bm v^\top B(B^\top K B)^{-1}B^\top\bm v=
\bm w^\top C(C^\top C)^{-1}C^\top\bm w\leq
\bm w^\top\bm w
=
\bm v^\top K^{-1}\bm v.
\end{aligned}
\]
The inequality follows because
$C(C^\top C)^{-1}C^\top$ is an orthogonal
projection matrix. Equality holds exactly
when $\bm w\in\operatorname{col}(C)$,
or equivalently,
$K^{-1}\bm v\in\operatorname{col}(B)$. Applying this result to
\eqref{eq:regularized-span-objective},
the treatment factor attains its maximum
when
$\widehat K_\eta^{-1}\widehat{\bm v}_T
\in\operatorname{col}(B)$.
Similarly, the residual-outcome factor
attains its maximum when
$\widehat K_{R,\eta}^{-1}\widehat{\bm v}_R
\in\operatorname{col}(B)$.

Recall that
$\widehat K_{R,\eta}
=\widehat K_\eta-
\widehat{\bm v}_T\widehat{\bm v}_T^\top/
\widehat v_T$.
By the Sherman--Morrison identity,
\[
\begin{aligned}
\widehat K_{R,\eta}^{-1}\widehat{\bm v}_R
={}&
\widehat K_\eta^{-1}\widehat{\bm v}_R+
\frac{
\widehat{\bm v}_T^\top
\widehat K_\eta^{-1}\widehat{\bm v}_R
}{
\widehat v_T-
\widehat{\bm v}_T^\top
\widehat K_\eta^{-1}\widehat{\bm v}_T
}
\widehat K_\eta^{-1}\widehat{\bm v}_T.
\end{aligned}
\]
Thus, the optimal residual-outcome direction
is a linear combination of
$\widehat K_\eta^{-1}\widehat{\bm v}_R$
and
$\widehat K_\eta^{-1}\widehat{\bm v}_T$.
Since the coefficient of the former is one,
the two optimal ridge prediction directions
span
\begin{equation}
\label{eq:regularized-optimal-span}
\widehat{\mathcal W}_\eta
=
\operatorname{col}
\left\{
\widehat K_\eta^{-1}
(\widehat{\bm v}_T,\widehat{\bm v}_R)
\right\}.
\end{equation}
Thus, when both cross-covariance vectors
are nonzero, $\widehat{\mathcal W}_\eta$
simultaneously maximizes both factors of
\eqref{eq:regularized-span-objective}.

Although $\widehat{\mathcal W}_\eta$ is spanned by
two optimal ridge prediction directions, these
directions need not be linearly independent.
Retaining both could therefore introduce
redundant scores and cause singularity in
downstream regression. We thus instead use a
low-rank generalized eigendecomposition to
obtain a basis for $\widehat{\mathcal W}_\eta$
with the minimum number of scores needed
to span both optimal prediction directions.

Define
\[
\widehat H
=
\frac{
\widehat{\bm v}_T\widehat{\bm v}_R^\top+
\widehat{\bm v}_R\widehat{\bm v}_T^\top
}{2},
\qquad
\widehat G_\eta
=
\widehat K_\eta
-
\frac{
\widehat{\bm v}_T\widehat{\bm v}_T^\top
}{2\widehat v_T}.
\]
Since $\widehat H$ is constructed from two
cross-covariance vectors, it has rank
$r\leq2$. Consider the generalized
eigenvalue problem
\(
\widehat H\bm\beta
=
\widehat\rho(\eta)\widehat G_\eta\bm\beta.
\)
Because $\widehat G_\eta$ is positive
definite, the symmetric matrix
$\widehat G_\eta^{-1/2}\widehat H
\widehat G_\eta^{-1/2}$ has exactly $r$
nonzero eigenvalues. Its nonzero eigenvectors
span its column space. Transforming these
eigenvectors back to the original coordinates,
the complete nonzero generalized eigenspace is
$\widehat{\mathcal E}_\eta
=\operatorname{col}(\widehat G_\eta^{-1}\widehat H)$.

When both cross-covariance vectors are nonzero,
$\operatorname{col}(\widehat H)
=\operatorname{col}
(\widehat{\bm v}_T,\widehat{\bm v}_R)$,
regardless of whether they are linearly
independent or collinear. Consequently,
\(
\widehat{\mathcal E}_\eta
=
\operatorname{col}
\left\{
\widehat G_\eta^{-1}
(\widehat{\bm v}_T,\widehat{\bm v}_R)
\right\}.
\) The Sherman--Morrison identity also gives
\[
\begin{aligned}
\widehat G_\eta^{-1}\widehat{\bm v}_T
&=
\frac{
2\widehat v_T
}{
2\widehat v_T-
\widehat{\bm v}_T^\top
\widehat K_\eta^{-1}\widehat{\bm v}_T
}
\widehat K_\eta^{-1}\widehat{\bm v}_T,\\
\widehat G_\eta^{-1}\widehat{\bm v}_R
&=
\widehat K_\eta^{-1}\widehat{\bm v}_R+
\frac{
\widehat{\bm v}_T^\top
\widehat K_\eta^{-1}\widehat{\bm v}_R
}{
2\widehat v_T-
\widehat{\bm v}_T^\top
\widehat K_\eta^{-1}\widehat{\bm v}_T
}
\widehat K_\eta^{-1}\widehat{\bm v}_T.
\end{aligned}
\]
The first expression is a nonzero
multiple of
$\widehat K_\eta^{-1}\widehat{\bm v}_T$,
while the second adds a multiple of
that direction to
$\widehat K_\eta^{-1}\widehat{\bm v}_R$.
These transformations preserve their
joint column space. Therefore,
\[
\begin{aligned}
\widehat{\mathcal E}_\eta
&=
\operatorname{col}
\left\{
\widehat G_\eta^{-1}
(\widehat{\bm v}_T,\widehat{\bm v}_R)
\right\}=
\operatorname{col}
\left\{
\widehat K_\eta^{-1}
(\widehat{\bm v}_T,\widehat{\bm v}_R)
\right\}=
\widehat{\mathcal W}_\eta.
\end{aligned}
\]
Hence, retaining the complete nonzero
generalized eigenspace gives exactly the
score space that maximizes
\eqref{eq:regularized-span-objective}
when both cross-covariance vectors
are nonzero. Its dimension is at most
two and is determined by the rank of
$\widehat H$.

We standardize the scores for visualization. Let
$r=\operatorname{rank}(\widehat H)\leq2$
and let $\widehat B_\eta$ contain a
$\widehat G_\eta$-orthonormal basis of the
nonzero generalized eigenspace.
Define $\widehat S_\eta^{\mathrm{raw}}
=X\widehat B_\eta$ and
$\widehat\Gamma_{S,\eta}
=n^{-1}(\widehat S_\eta^{\mathrm{raw}})^\top
\widehat S_\eta^{\mathrm{raw}}$.
We whiten the raw scores to obtain
\begin{equation}
\label{eq:score-whitening}
\widehat S_\eta
=
\widehat S_\eta^{\mathrm{raw}}
\widehat\Gamma_{S,\eta}^{-1/2}.
\end{equation}
When $\widehat\Gamma_{S,\eta}$ is nonsingular,
$n^{-1}\widehat S_\eta^\top
\widehat S_\eta=I_r$.
Whitening preserves the score space and therefore leaves subsequent linear adjustment unchanged.

We also show in Supplementary Materials~\ref{sec:computational-complexity} that COMPACT has overall time complexity \(O(np^2+p^3)\) and memory complexity \(O(np+p^2)\). Importantly, the \(O(p^3)\) term arises solely from a Cholesky factorization, a highly optimized operation in modern numerical linear algebra libraries, allowing COMPACT to remain computationally practical even when \(p\) is large. Consistent with this analysis, COMPACT was among the fastest methods evaluated in our experiments, including settings with hundreds of candidate adjustment variables.

\subsection{Cross-validation for the ridge penalty}
\label{sec:theory:joint-regularization}

In the nondegenerate rank-two case, for any
fixed $\eta$, the COMPACT eigenspace is
equivalent to the span of two ridge regressions---one predicting treatment and the other predicting
the treatment-residualized outcome---as established
in \eqref{eq:regularized-optimal-span}.
However, COMPACT differs from separately tuned
ridge regressions in its regularization objective. Since the goal of COMPACT
is to maximize the joint score-space objective
\eqref{eq:population-product-objective},
we select $\eta$ by cross-validation using
the held-out counterpart of this original,
unregularized objective, rather than minimizing
individual prediction errors. This approach
selects regularization according to the
predictive information retained by the entire
score space, rather than the accuracy of
individual ridge predictions.

For each candidate $\eta$, we fit COMPACT on
the training fold to obtain $\widehat B_\eta$
and the corresponding training-fold score
covariance
\(
\widehat\Gamma_{S,\eta}
=
\frac{1}{n_{\mathrm{tr}}}
\widehat B_\eta^\top
X_{\mathrm{tr}}^\top X_{\mathrm{tr}}
\widehat B_\eta.
\)
The validation scores use the whitening
transformation learned on the training fold:
\(
S_{\eta,\mathrm{val}}
=
X_{\mathrm{val}}\widehat B_\eta
\widehat\Gamma_{S,\eta}^{-1/2}.
\)
We evaluate the unregularized multiple
correlations on the validation fold.
Let
$\bm Y_{\mathrm{res,val}}
=M_{T_{\mathrm{val}}}\bm Y_{\mathrm{val}}$
and
$S_{\mathrm{res,val}}
=M_{T_{\mathrm{val}}}S_{\eta,\mathrm{val}}$.
We select $\eta$ by maximizing
\begin{equation}
\label{eq:CV}
\operatorname{CV}_{\mathrm{span}}(\eta)
=
\left\{
R^2(\bm T_{\mathrm{val}}
\sim S_{\eta,\mathrm{val}})
\right\}^{1/2}
\left\{
R^2(\bm Y_{\mathrm{res,val}}
\sim S_{\mathrm{res,val}})
\right\}^{1/2}.
\end{equation}
This is the held-out sample counterpart of
\eqref{eq:population-product-objective}.

Selecting $\eta$ to optimize \eqref{eq:CV} rather than individual prediction errors is important because it removes an unnecessary calibration requirement. For a fixed $\eta>0$, let $\widehat{\mathcal S}_\eta$ denote the
two-dimensional COMPACT score space, and let $P_\eta$ denote its
$L_2$-orthogonal projection. For $W=T$ or $W=R$, define the
ridge predictor
\(
\widehat W_\eta
=
\bm X_i^\top
(\widehat\Sigma_{\bm X}+\eta I_p)^{-1}
\widehat{\bm v}_W,
\)
and let $g_W$ denote the best population linear predictor of $W_i$
from $\bm X_i$.

\begin{lemma}
\label{lem:ridge-span-decomposition}

Conditional on the training sample, the prediction error of ridge
regression decomposes into the error attainable using the complete
COMPACT score space and an additional error from calibrating the
ridge predictor within that space:
\begin{equation}
\label{eq:ridge-span-error-decomposition}
\begin{aligned}
\underbrace{
\mathbb E(W_i-\widehat W_\eta)^2
}_{\text{Ridge prediction error}}
\quad=
\underbrace{
\mathbb E(W_i-P_\eta W_i)^2
}_{\text{COMPACT span prediction error}}
\quad+
\underbrace{
\mathbb E(P_\eta g_W-\widehat W_\eta)^2.
}_{\text{Additional ridge calibration error}}
\end{aligned}
\end{equation}
\end{lemma}
The decomposition separates loss of predictive information from
error in the calibration of coefficients within the retained space.
Separate ridge prediction penalizes both components. In contrast,
the COMPACT validation criterion refits the linear associations
with the complete score space and is therefore unaffected by
changes of basis or scale within that space. The same projection
argument applies to the treatment-residualized score space used
in the second component of COMPACT's validation criterion.

By avoiding unnecessary calibration requirements,
COMPACT can recover the latent adjustment space
more accurately than separately tuned ridge
regressions. We formalize this possibility by
comparing two candidate regularization parameters
$0<\eta_-<\eta_+$. The analysis distinguishes
the collective information $s_p$ available
about the least-recoverable latent adjustment
direction from the excess score-extraction error
$r_{p,n}$ introduced by estimating the score
space.

Assumption~\ref{assump:joint-regularization},
detailed in Supplementary
Materials~\ref{supp:joint-regularization},
requires that both latent adjustment directions
are recoverable from the observed proxies,
that stronger regularization reduces excess
score-extraction error, and that its increased
calibration error causes separately tuned ridge
regressions to prefer the smaller penalty.
The assumption also requires cross-validation
to distinguish the competing penalties
reliably. Under these conditions, COMPACT's
span-oriented criterion favors stronger
regularization, improving recovery of the
latent adjustment space. Here, we measure the distance between two equal-dimensional subspaces by
$\operatorname{dist}(\mathcal A,\mathcal B)
=
\|P_{\mathcal A}-P_{\mathcal B}\|_{\mathrm{op}},$
where $P_{\mathcal A}$ and $P_{\mathcal B}$ are their $L_2$-orthogonal projection operators.
\begin{theorem}
\label{thm:joint-regularization-recovery}
Suppose Assumption~\ref{assump:joint-regularization}
holds. Let $\widehat\eta_C$ maximize
\eqref{eq:CV} over the two candidate penalties,
and let $\widehat\eta_T,\widehat\eta_R$
minimize the corresponding individual ridge
prediction errors. Define
\[
\widehat{\mathcal S}_{\mathrm{sep}}
=
\operatorname{span}_{L_2}
(\widehat T_{\widehat\eta_T},
\widehat R_{\widehat\eta_R}).
\]
Then there exists a constant $c>0$,
independent of $p$ and $n$, such that
\begin{equation}
\label{eq:joint-selection-recovery-comparison}
\Pr\Bigg[
\operatorname{dist}^2
(\mathcal U,\widehat{\mathcal S}_{\mathrm{sep}})
-
\operatorname{dist}^2
(\mathcal U,\widehat{\mathcal S}_{\widehat\eta_C}) \geq
c\frac{r_{p,n}}
{(1+s_p^{-1}+r_{p,n})^2}
\Bigg]\longrightarrow1.
\end{equation}
\end{theorem}
Thus the joint span-oriented regularization can
recover the latent adjustment space more accurately
than independently tuned ridge regressions.
Specifically, when stronger regularization reduces
excess extraction error but increases the calibration
error of individual predictions, COMPACT selects the
larger penalty while both separate regressions select
the smaller penalty. This yields a positive lower
bound on COMPACT's improvement in adjustment-space
recovery, determined by the avoidable extraction
error $r_{p,n}$ and collective proxy information
$s_p$. Together with
Lemma~\ref{lem:ridge-span-decomposition}, this
establishes why optimizing individual prediction
error need not optimize recovery of the underlying
adjustment space: separate ridge regression penalizes
both loss of predictive information and within-span
calibration error, whereas COMPACT directly evaluates
the predictive information retained by the joint
score space.

\section{Theoretical Results}
\label{sec:theory}

We now establish three theoretical properties of COMPACT. First, we characterize recovery of the latent adjustment space by separating error due to imperfect proxy information from error due to finite-sample estimation of the spectral target. Second, we show how this recovery error propagates to downstream marginal and score-indexed causal estimation. Third, we establish validity of the nonparametric bootstrap when both COMPACT and the downstream regression are recomputed in each bootstrap sample. We provide an overview of the theoretical results here, while Supplementary Materials~\ref{supp:theoretical-details} contain the complete assumptions, intermediate results, and additional discussion.

\subsection{Recovery of the Latent Adjustment Space} \label{sec:theory:recovery}

Let
$\widehat{\mathcal S}_\eta
=
\operatorname{span}_{L_2}
\left\{
\widehat S_{\eta,i,1},
\widehat S_{\eta,i,2}
\right\}$
denote the learned COMPACT score space for a fixed ridge parameter $\eta>0$. Two quantities govern recovery. The effective proxy strength $p_{\mathrm{eff}}$ measures how well the population COMPACT scores linearly reconstruct the worst-recovered direction in $\mathcal U$, with larger values corresponding to more informative proxies. The regularized effective dimension
$d_{\mathrm{eff}}(\eta)
=
\operatorname{tr}
\left\{
\Sigma_{\bm X}
(\Sigma_{\bm X}+\eta I_p)^{-1}
\right\}$
controls the finite-sample complexity of estimating the spectral target, while $\Delta$ denotes the separation of the retained rank-two eigenspace from the zero eigenspace. The formal definitions and regularity conditions are given in Supplementary Materials~\ref{supp:theorem2-details1} and \ref{supp:theorem2-details2}, where Lemmas~\ref{lem:spectral-subspace-error} and~\ref{lem:proxy-recovery} separately establish the sampling and proxy-recovery components of the following result.

\begin{theorem}
\label{thm:proxy-sample-tradeoff}
Suppose the conditions of Lemmas~\ref{lem:spectral-subspace-error} and
\ref{lem:proxy-recovery} hold. Then
\[
\operatorname{dist}^2
\left(
\mathcal U,
\widehat{\mathcal S}_\eta
\right)
\in
O_p\left(
\frac{1}{p_{\mathrm{eff}}}
+
\frac{d_{\mathrm{eff}}(\eta)}{n\Delta^2}
\right).
\]
\end{theorem}

The theorem separates two sources of error. The term $1/p_{\mathrm{eff}}$ quantifies proxy-reconstruction error, decreasing as the observed proxy system becomes more informative about the latent adjustment variables. In contrast, $d_{\mathrm{eff}}(\eta)/(n\Delta^2)$ reflects estimation of the COMPACT eigenspace from a finite sample. Thus the learned score space converges to $\mathcal U$ when proxy information becomes sufficiently strong and the regularized spectral problem remains estimable relative to sample size.

\subsection{Causal consistency of downstream score adjustment} \label{sec:theory:downstream}

We next consider a downstream causal estimator that adjusts for the learned COMPACT scores. For $t\in\mathcal T_0$, let
\(
\theta(t)=\mathbb E\{Y_i(t)\}
\)
denote the marginal causal response. Conditional on the fitted COMPACT map, let $\widehat\theta_{\widehat{\bm S}_\eta}^{\mathrm{marg}}(t)$ denote a downstream estimator of the corresponding score-adjusted marginal response and let $\widehat\theta_{\widehat{\bm S}_\eta}^{\mathrm{cond}}(t;\bm s)$ denote the corresponding score-indexed response estimator. The target score-indexed causal response is
$\theta_{\widehat{\bm S}_\eta}^{\mathrm{cond}}(t;\bm s)
=
\mathbb E\{Y_i(t)\mid \widehat{\bm S}_{\eta,i}=\bm s\}.$

Assumption~\ref{assump:causal}, stated in full in Supplementary Materials~\ref{supp:theorem3-details}, requires latent mean exchangeability, COMPACT reconstruction at the rate established above, smoothness of the latent response surface, bounded residual imbalance within score strata, and consistency of the downstream score-adjusted learner. These conditions translate approximation of the latent adjustment space into control of causal estimation error.

\begin{theorem}
\label{thm:causal-consistency-compact-plugin}
Suppose Assumption~\ref{assump:causal} holds. Then
\[
\sup_{t\in\mathcal T_0}
\left|
\widehat\theta_{\widehat{\bm S}_\eta}^{\mathrm{marg}}(t)-\theta(t)
\right|
\in
O_p\left[
\left\{
\frac{1}{p_{\mathrm{eff}}}
+
\frac{d_{\mathrm{eff}}(\eta)}{n\Delta^2}
\right\}^{1/2}
\right]
+
o_p(1),
\]
and
\[
\sup_{t\in\mathcal T_0,\;\bm s\in\mathcal S_0}
\left|
\widehat\theta_{\widehat{\bm S}_\eta}^{\mathrm{cond}}(t;\bm s)
-
\theta_{\widehat{\bm S}_\eta}^{\mathrm{cond}}(t;\bm s)
\right|
\in
O_p\left[
\left\{
\frac{1}{p_{\mathrm{eff}}}
+
\frac{d_{\mathrm{eff}}(\eta)}{n\Delta^2}
\right\}^{1/2}
\right]
+
o_p(1).
\]
\end{theorem}

Thus the same proxy-recovery and sampling errors governing estimation of the adjustment space also control downstream treatment-effect estimation error. In particular, sufficiently accurate reconstruction of $\mathcal U$ allows adjustment for the learned COMPACT scores to recover both marginal and score-indexed causal responses, provided that the downstream learner is itself consistent. Treatment-contrast bounds follow directly by applying the result at two treatment values.

\subsection{Bootstrap inference} \label{sec:theory:bootstrap}

Finally, we consider inference for finite-dimensional functionals of a downstream ordinary least squares (OLS) regression fitted using the learned COMPACT scores. Let
\(
\bm{\psi}=(\psi_1,\ldots,\psi_q)^\top
\)
denote a fixed vector of population functionals that depends on the learned representation only through its score space. Let $\widehat{\bm{\psi}}$ denote the estimator from the observed sample and $\widehat{\bm{\psi}}^*$ the estimator obtained after recomputing both COMPACT and the downstream regression in a nonparametric bootstrap sample.

Assumption~\ref{assump:L-asymptotic-linearity}, detailed in Supplementary Materials~\ref{supp:theorem4-details}, requires local regularity of the COMPACT spectral map and suitable moment conditions. Together with the fixed-dimensional OLS conditions stated below, these assumptions ensure that perturbations of the empirical moments propagate smoothly through score learning and downstream estimation.

\begin{theorem}
\label{thm:compact-downstream-bootstrap}
Suppose \(O_1,\ldots,O_n\) are independent and identically distributed and
Assumption~\ref{assump:L-asymptotic-linearity} holds. Assume the downstream estimator applies ordinary least squares to the fixed-dimensional regression model in \eqref{eq:linear-potential-outcome-model}:
\begin{equation}\label{eq:regression_model}
Y_i
=
\alpha_0
+
\alpha_T T_i
+
\bm{\alpha}_S^\top \bm S_{\eta,i}
+
T_i\bm{\alpha}_{TS}^\top \bm S_{\eta,i}
+
\varepsilon_i .
\end{equation}
Assume the corresponding population Gram matrix over \(\left(
1,\,
T_i,\,
{\bm S_{\eta,i}}^{\top},\,
T_i{\bm S_{\eta,i}}^{\top}
\right)^\top \) is nonsingular, the regressors have finite fourth moments,
and \(\mathbb E(Y_i^4)<\infty\). Suppose also that \(\bm{\psi}\) is a continuously
differentiable, basis-invariant function of the fitted downstream regression.
Then, for every bounded Lipschitz function
\(h:\mathbb R^q\to\mathbb R\),
\[
\left|
\mathbb E^*
h
\left\{
\sqrt n(\widehat{\bm{\psi}}^*-\widehat{\bm{\psi}})
\right\}
-
\mathbb E
h
\left\{
\sqrt n(\widehat{\bm{\psi}}-\bm{\psi})
\right\}
\right|
\overset{p}{\longrightarrow}
0.
\]
\end{theorem}

The theorem therefore justifies recomputing the complete procedure in each bootstrap sample: the resulting bootstrap distribution accounts jointly for uncertainty in the learned COMPACT representation and in the downstream OLS fit. When the downstream regression is correctly specified and the learned scores provide sufficient causal adjustment, the corresponding OLS functionals inherit their causal interpretation.

\section{Empirical Results} \label{sec:empirical}

\subsection{Comparators}

We selected COMPACT's ridge parameter $\eta$ by five-fold cross-validation over 12 logarithmically spaced values from $10^{-5}$ to $10^{1}$. We compared COMPACT with several propensity-score, outcome-score, and dimension-reduction baselines using each method's default hyperparameter settings:
\begin{enumerate}[leftmargin=1.25em,label=(\alph*),labelsep=0.4em,itemsep=0.25em]

\item \textbf{SuperLearner Propensity Score (SL-PS)} \citep{vanDerLaan07}: estimates the propensity score using an ensemble of penalized logistic regression and random forest learners, with the logit propensity score used for adjustment. This tests whether flexible treatment prediction alone suffices for accurate ATE and CATE recovery.

\item \textbf{Covariate Balancing Propensity Score (CB-PS)} \citep{Imai14}: estimates a propensity score while directly targeting covariate balance between treated and untreated units. This provides a balancing-score comparator focused on treatment-assignment structure rather than outcome-relevant confounding structure.

\item \textbf{Outcome-Adaptive Lasso Propensity Score (OAL-PS)} \citep{Shortreed17}: first identifies outcome-relevant covariates and then penalizes them less strongly in a propensity-score model. This tests whether outcome-guided variable selection within a treatment model can recover an adequate adjustment score.

\item \textbf{Double-score adjustment (DOUBLE)} \citep{Hansen08,Leacy14}: combines a treatment score with a prognostic score. We used a SuperLearner logit propensity score and the predicted untreated outcome, testing whether separately estimated treatment and outcome scores can match a jointly learned adjustment representation.

\item \textbf{Multivariate Inverse-Regression Estimation (MIRE)} \citep{Li91,Cook05}: applies inverse-regression sufficient dimension reduction (SDR) to construct low-dimensional outcome-relevant projections of $\bm X$. We slice $Y$, estimate ridge-stabilized inverse-regression directions, and retain the leading SDR scores. This tests whether outcome-oriented dimension reduction alone suffices for treatment-effect recovery.

\end{enumerate}

COMPACT is the only comparator that learns a low-dimensional span from a joint treatment--residual-outcome spectral operator. 
To evaluate Theorem~\ref{thm:joint-regularization-recovery},
we also compared COMPACT with DOUBLE implemented
using two ridge regressions, with regularization
parameters $\eta$ either tuned separately (\textbf{Ridge-Separate}) or shared
(\textbf{Ridge-Shared}) across the treatment and outcome models. Finally, we evaluated two ablations: \textbf{COMPACT-noT}, which removes the treatment-correlation component and learns outcome-relevant scores only, and \textbf{COMPACT-noY}, which removes the residual-outcome-correlation component and learns treatment-relevant scores only.

\subsection{Metrics}

We evaluated accuracy and practical performance using six complementary metrics. First, we measured \textbf{root mean squared error} (RMSE) for the ATE and CATE across Monte Carlo replications. Because the ATE is scalar within each replication, its replication-level RMSE reduces to absolute error, whereas CATE RMSE averages squared errors across individuals and therefore captures unit-level errors that may cancel in the ATE. These metrics directly assess Theorem~\ref{thm:causal-consistency-compact-plugin}, which links recovery of the latent adjustment information to errors in marginal and score-indexed causal-effect estimation.

Second, we measured recovery of the latent treatment-balancing and residual-outcome coordinates using the \textbf{coefficient of determination}
\[
R^2(U\mid\bm S)
=
1-
\frac{
\sum_{i=1}^n
\left\{
U_i-\widehat{\mathbb E}_{\mathrm{lin}}(U_i\mid\bm S_i)
\right\}^2
}{
\sum_{i=1}^n
(U_i-\overline U)^2
}.
\]
Accordingly, $R^2(U_T\mid\bm S)$ measures preservation of the treatment-balancing coordinate and $R^2(U_R\mid\bm S)$ measures preservation of the residual-outcome coordinate. Values approaching one indicate that the corresponding latent coordinate can be reconstructed accurately from the learned scores. These metrics relate directly to Theorem~\ref{thm:proxy-sample-tradeoff}, which bounds the distance between the learned score space and
$\mathcal U=\operatorname{span}(U_T,U_R).$
If the learned score space recovers $\mathcal U$ exactly, then both $U_T$ and $U_R$ are linear combinations of $\bm S$, so both $R^2$ values equal one. Conversely, because both spaces are two-dimensional, perfect reconstruction of two linearly independent latent coordinates implies recovery of their span. Simultaneously large values therefore indicate accurate recovery of the latent adjustment plane and support interpretation of its two-dimensional visualization up to an invertible change of coordinates.

To assess overlap, we fit a logistic regression of treatment on each learned representation and recorded the \textbf{extreme-propensity proportion},
\[
\frac{1}{n}
\sum_{i=1}^n
\mathbf 1
\left\{
\widehat e(\bm S_i)<0.05
\ \text{or}\
\widehat e(\bm S_i)>0.95
\right\},
\]
reported as a percentage. Lower values indicate fewer observations near estimated positivity violations. We interpret this metric jointly with $R^2(U_T\mid\bm S)$, because a representation containing little treatment information may produce nonextreme propensities without adequately controlling confounding. Finally, we recorded \textbf{total runtime} as the elapsed wall-clock time required to construct the adjustment representation and fit the common downstream treatment-effect estimator.

For the separate COMPACT bootstrap experiments, we evaluated \textbf{empirical coverage}, \textbf{type-I error}, and \textbf{power}. Empirical coverage was the proportion of Monte Carlo confidence intervals containing the true ATE, evaluated across a range of nominal coverage levels. Type-I error and power were the empirical rejection proportions under true and false null hypotheses, respectively. We calculated these separately for the ATE and global CATE tests and report exact 95\% binomial confidence intervals.

\subsection{Simulations}

We evaluated COMPACT in Monte Carlo simulations designed to vary both the amount and quality of proxy information available for confounding adjustment. We generated a 20-dimensional latent baseline state and up to 500 proxy variables, with proxy systems nested across $p\in\{100,300,500\}$ so that larger systems contained all variables from the smaller systems. We varied the fraction of each proxy's variance attributable to the latent state over $\upsilon\in\{0.2,0.4,0.6,0.8\}$, yielding 12 proxy-information conditions. Treatment and outcomes depended on the latent state, with settings spanning constant and heterogeneous treatment effects. A complete description of the data-generating process, proxy construction, estimands, and bootstrap procedures is provided in the Supplementary Materials \ref{supp:setup}.

We summarize the ATE RMSE results in
Figure~\ref{fig:synth_error}A--C. COMPACT achieved among the lowest ATE RMSE across all conditions. At \(p=100\) and \(p=300\), COMPACT and CB-PS performed best, with substantial overlap of their \(95\%\) confidence intervals. At \(p=500\), however, COMPACT outperformed all other algorithms by a significant margin as CB-PS deteriorated sharply, with ATE RMSE increasing to approximately \(0.66\)--\(0.67\). This instability likely arose because CB-PS estimated a high-dimensional propensity model under covariate-balancing constraints without explicit coefficient shrinkage. MIRE also had substantially higher ATE RMSE because its inverse-regression directions explain outcome variation without using treatment assignment to construct a balancing representation. Moreover, COMPACT achieved lower ATE RMSE
than both Ridge-Separate and Ridge-Shared, with
its advantage increasing as $p$ grew (Supplementary Figure \ref{fig:supp_sim_ATE_CATE_RMSE_ridge}A--C). This pattern
is consistent with the mechanism underlying
Theorem~\ref{thm:joint-regularization-recovery}:
as the effective number of nuisance directions
increases, excess score-extraction error
$r_{p,n}\asymp d_p/n$ may grow, making
span-oriented regularization increasingly beneficial. Overall, COMPACT recovered the ATE well across proxy dimensions and levels of $\upsilon$.

\begin{figure}[t]
    \centering
    \includegraphics[width=1\linewidth]{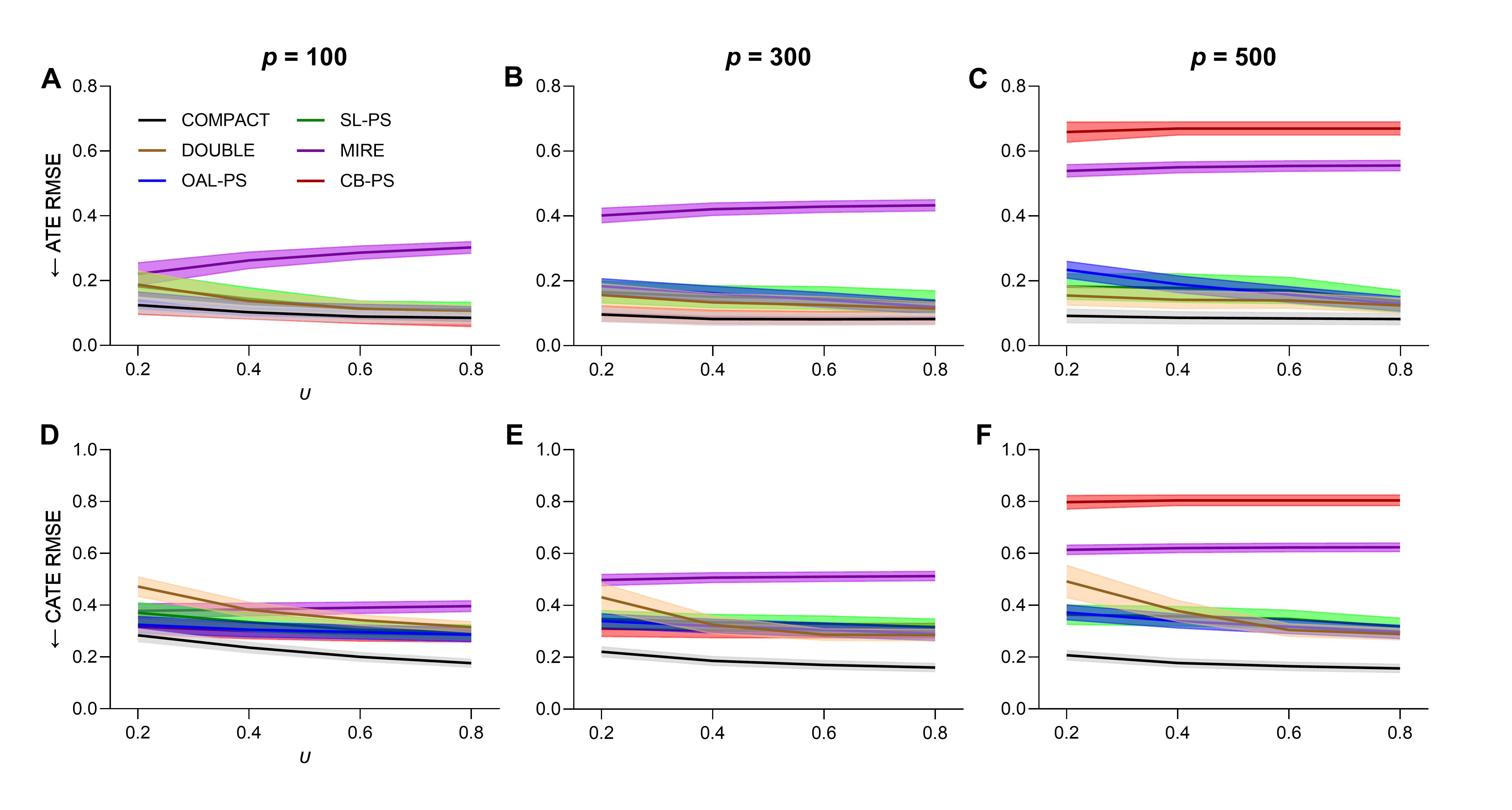}
    \caption{\textbf{Treatment-effect estimation in the fully synthetic
    experiments.} ATE RMSE (A--C) and CATE RMSE (D--F) across proxy
    reliability levels for \(p=100\), \(300\), and \(500\). Lines denote means
    over \(50\) paired Monte Carlo replications, and shaded bands denote
    \(95\%\) confidence intervals.}
    \par
    \raggedright
    \tiny
    \textbf{Alt text:} Six line plots of treatment-effect estimation error in the fully synthetic experiments. Panels A--C show ATE RMSE and panels D--F show CATE RMSE across proxy reliability levels for 100, 300, and 500 proxy variables. COMPACT has among the lowest ATE RMSE and consistently low CATE RMSE across conditions.

    \label{fig:synth_error}
\end{figure}

Because the ATE averages over heterogeneous treatment effects, positive and negative CATE errors can cancel and conceal substantial unit-level error. The CATE results therefore separated COMPACT from the other algorithms more clearly (Figure~\ref{fig:synth_error}D--F). COMPACT achieved the lowest CATE RMSE in every condition, with its advantage increasing with the dimension of \(\bm X\). At \(p=100\), the next-best method had CATE RMSE approximately \(12\%\) to \(62\%\) larger than COMPACT across the four reliability levels; at \(p=500\), this difference increased to approximately \(76\%\) to \(91\%\). Thus, marginal ATE performance understated COMPACT's advantage for recovering treatment-effect heterogeneity. Ablation experiments reinforced this distinction: full COMPACT and both ablated variants achieved comparable ATE RMSE, whereas full COMPACT achieved the lowest CATE RMSE (Supplementary Figure~\ref{fig:supp_sim_ATE_CATE_RMSE}). Furthermore, COMPACT achieved lower CATE RMSE than both ridge-regression variants, demonstrating the benefit of selecting $\eta$ using the joint score-space criterion (Supplementary Figure~\ref{fig:supp_sim_ATE_CATE_RMSE_ridge}D--F).

The latent-recovery diagnostics explained these differences. COMPACT attained the largest \(R^2(U_T\mid\bm S)\) and \(R^2(U_R\mid\bm S)\) in all \(12\) conditions (Figure~\ref{fig:synth_recovery}). Its \(R^2(U_R\mid\bm S)\) increased from approximately \(0.49\) in the least informative condition to \(0.97\) in the most informative condition, while \(R^2(U_T\mid\bm S)\) increased from approximately \(0.44\) to \(0.92\). COMPACT therefore recovered a two-dimensional span that simultaneously explained most variation in both latent adjustment coordinates rather than favoring one at the expense of the other. Both ridge-regression variants recovered the treatment-balancing coordinate $U_T$ less accurately than COMPACT, with the performance gap widening as $p$ increased (Supplementary Figure~\ref{fig:supp_sim_R2_ridge}). Finally, COMPACT-noT preserved only \(U_R\), whereas COMPACT-noY preserved only \(U_T\) (Supplementary Figure~\ref{fig:supp_sim_R2}).

The above results agree with the latent-space recovery analysis in Section~\ref{sec:theory:recovery}. Simultaneously large recovery values for both coordinates imply small worst-direction reconstruction error and therefore large effective proxy strength. Theorem~\ref{thm:proxy-sample-tradeoff} accordingly predicts more accurate recovery of the latent adjustment span as the proxy system becomes more informative, while Theorem~\ref{thm:causal-consistency-compact-plugin} translates this improved reconstruction into smaller marginal and score-indexed causal estimation errors. The observed \(R^2\), ATE, and CATE results follow this predicted relationship.

\begin{figure}[t]
    \centering
    \includegraphics[width=1\linewidth]{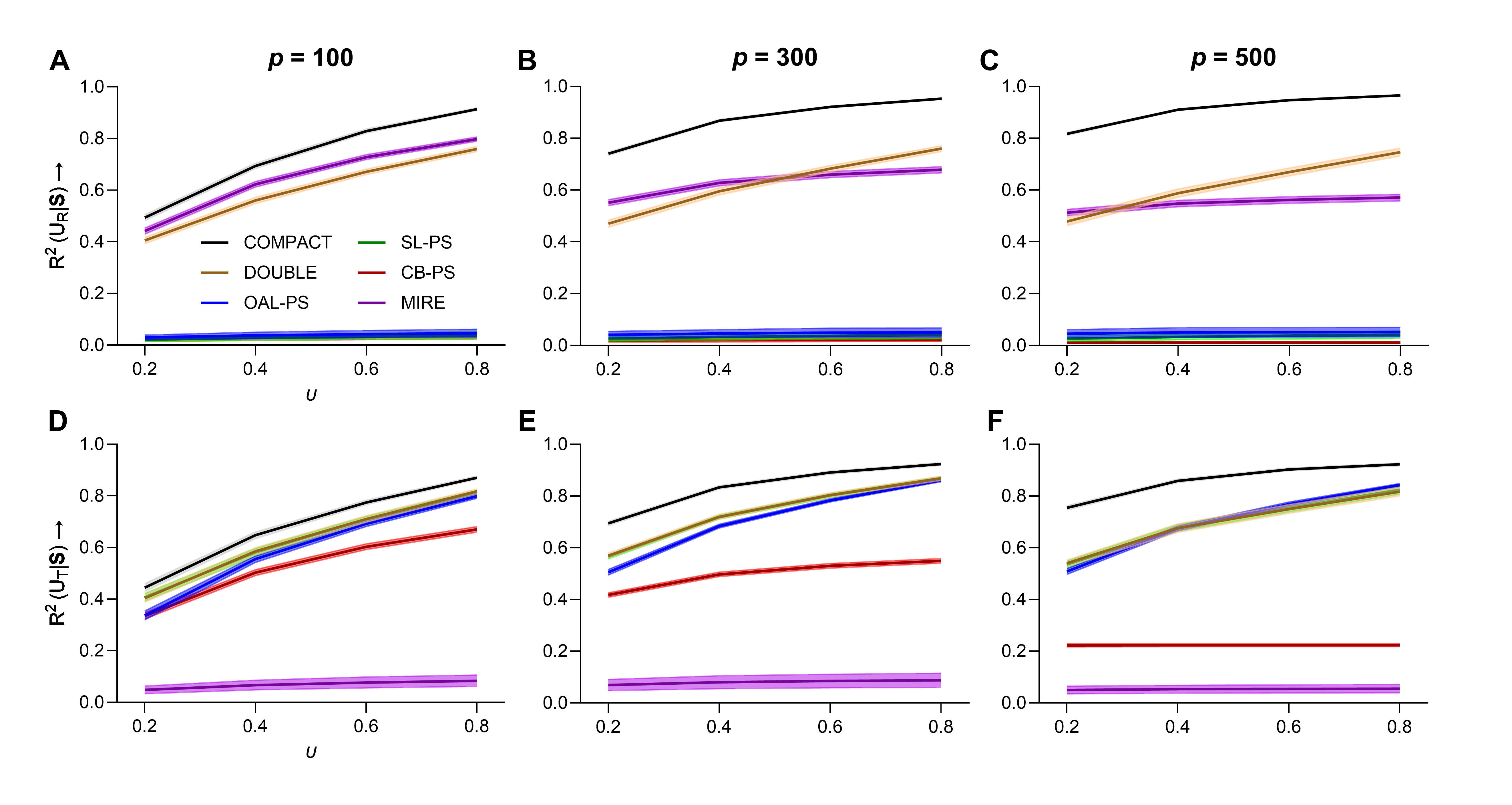}
    \caption{\textbf{Recovery of the latent adjustment coordinates.}
    Linear recovery of \(U_R\) (A--C) and \(U_T\) (D--F) from each learned
    representation across proxy reliability levels and proxy counts. Higher
    values indicate more complete recovery of the corresponding latent
    coordinate.}
    \par
    \raggedright
    \tiny
    \textbf{Alt text:} Six line plots of recovery of the latent adjustment coordinates in the fully synthetic experiments. Panels A--C show linear recovery of the residual-outcome coordinate $U_R$ and panels D--F show linear recovery of the treatment-balancing coordinate $U_T$ across proxy reliability levels and proxy counts.

    \label{fig:synth_recovery}
\end{figure}

COMPACT also maintained stable score-space overlap. The mean percentage of extreme fitted propensities ranged from approximately \(0.4\%\) to \(7.0\%\) across all conditions (Supplementary Figure~\ref{fig:synth_overlap_runtime}A--C). COMPACT-noY achieved similar overlap, indicating that this stability did not depend on the residual-outcome component (Supplementary Figure~\ref{fig:supp_sim_propensity_time}). Instead, COMPACT's low-dimensional, ridge-regularized spectral representation prevented strongly treatment-predictive variation from producing extreme fitted propensities. OAL-PS showed similarly low percentages, ranging from approximately \(0.4\%\) to \(6.7\%\), whereas SL-PS and DOUBLE reached approximately \(57\%\) and \(58\%\), respectively. CB-PS again deteriorated with increasing proxy dimension, producing extreme propensities for approximately \(98\%\)--\(100\%\) of observations at \(p=500\). The MIRE result was less informative because MIRE constructs outcome-oriented directions without using treatment assignment; we nevertheless report it for completeness.

COMPACT remained computationally efficient as proxy dimension increased. Mean runtime remained below \(1.6\) seconds per replication even at \(p=500\) (Supplementary Figure~\ref{fig:synth_overlap_runtime}D--F), making repeated score learning and downstream estimation practical for full-procedure bootstrap inference. Among the competing causal methods, only OAL-PS had comparable runtime. Thus, COMPACT combined more accurate latent-space and CATE recovery with computational efficiency suitable for routine bootstrap inference.

Finally, we evaluated the nonparametric bootstrap by recomputing both the COMPACT score space and downstream OLS regression in each resample. In this experiment, $\tau_0\in\{0,0.5\}$ denotes the population ATE and $\tau_U\in\{0,0.75\}$ controls the magnitude of treatment-effect heterogeneity, with $\tau_U=0$ corresponding to a constant treatment effect. Across 200 Monte Carlo replications, empirical ATE coverage closely followed nominal coverage under both constant and heterogeneous treatment effects (Supplementary Figure~\ref{fig:bootstrap_inference}A), reaching \(95\%\)--\(96\%\) at the nominal \(95\%\) level. Under the ATE null (\(\tau_0=0\)), type-I error was \(4.0\%\) for both \(\tau_U=0\) and \(\tau_U=0.75\). The global CATE test of equality between the score-indexed CATE and ATE had \(3.0\%\) type-I error under \(\tau_U=0\), and the exact \(95\%\) binomial confidence intervals for all three error rates contained the nominal \(5\%\) level (Supplementary Figure~\ref{fig:bootstrap_inference}B). At \(\tau_0=0.5\), the ATE test achieved \(100\%\) power with and without treatment-effect heterogeneity, while the global CATE test achieved \(99.5\%\) power under \(\tau_U=0.75\) (Supplementary Figure~\ref{fig:bootstrap_inference}C). All 500 bootstrap fits succeeded in every replication. These results support accurate bootstrap inference that propagates uncertainty from both score learning and downstream effect estimation.

\subsection{Semi-synthetic data} \label{sec:real_data}
The Clinical Antipsychotic Trials of Intervention Effectiveness (CATIE)
schizophrenia study was a large, multicentre, randomized effectiveness trial
comparing antipsychotic medications among patients with chronic schizophrenia
\citep{Lieberman05}. Its extensive baseline assessments provide a clinically
realistic, high-dimensional covariance structure. We therefore conducted
CATIE-calibrated semi-synthetic experiments in which the CATIE data determined
the proxy-loading geometry, while the latent state, treatment, and outcome were
simulated.

After complete-case filtering, standardization, and removal of constant variables, the CATIE baseline matrix contained \(n=1{,}407\) observations and \(p_{\mathrm{CATIE}}=364\) variables. We extracted the first \(d_C=20\) principal-component loading vectors. We then randomly permuted the variable ordering and selected the first
\(p\in\{121,243,364\}\) variables, producing nested proxy systems containing
approximately one-third, two-thirds, or all CATIE variables. Denote the resulting loading matrix by
\(\Lambda^{(p)}\in\mathbb R^{p\times d_C}\), with each row normalized to unit
Euclidean norm. We otherwise followed the fully synthetic design, crossing four proxy-strength
levels with the three proxy counts to produce \(12\) conditions, each evaluated
over \(50\) Monte Carlo replications. By deriving the proxy-loading directions
from CATIE, this design retained CATIE-calibrated multivariate dependence
geometry while preserving known latent variables and causal ground truth.

\begin{figure}[t]
    \centering
    \includegraphics[width=1\linewidth]{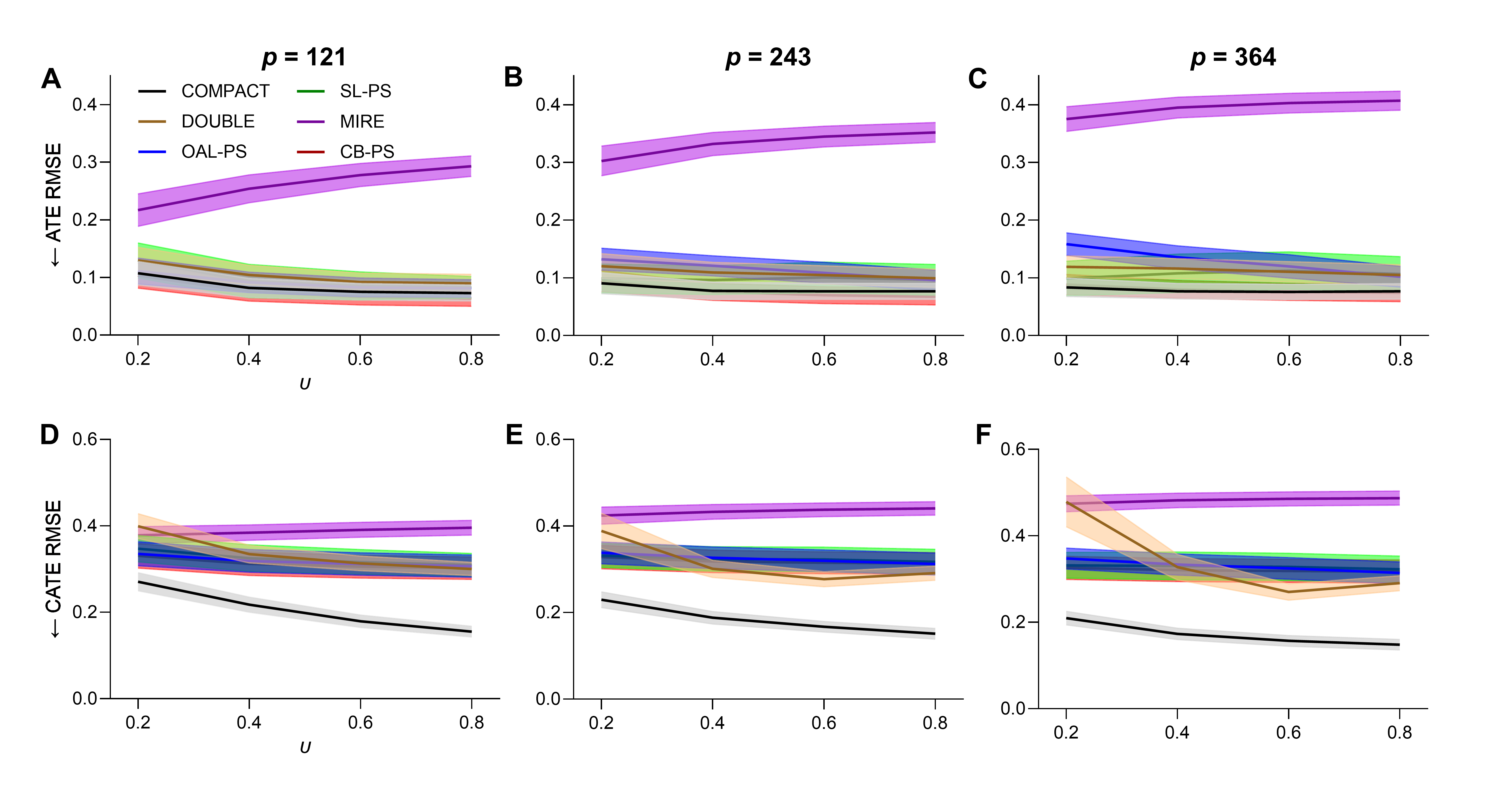}
    \caption{\textbf{Treatment-effect estimation in the CATIE-calibrated
    semi-synthetic experiments.} ATE RMSE (A--C) and CATE RMSE (D--F) across
    proxy-strength levels for \(p=121\), \(243\), and \(364\).}
    \par
    \raggedright
    \tiny
    \textbf{Alt text:} Six line plots of treatment-effect estimation error in the CATIE-calibrated semi-synthetic experiments. Panels A--C show ATE RMSE and panels D--F show CATE RMSE across proxy-strength levels for 121, 243, and 364 proxy variables. COMPACT remains competitive for ATE estimation and has the lowest CATE RMSE across conditions.

    \label{fig:real_ATE_CATE_RMSE}
\end{figure}

Across the semi-synthetic experiments, COMPACT remained competitive for ATE
estimation across all proxy dimensions and proxy-strength levels
(Figure~\ref{fig:real_ATE_CATE_RMSE}A--C; Supplementary Figure \ref{fig:real_Supp_ATE_CATE}A--C). The CATE results differentiated
the methods more clearly: COMPACT achieved the lowest CATE RMSE by a substantial
margin in every condition
(Figure~\ref{fig:real_ATE_CATE_RMSE}D--F; Supplementary Figure \ref{fig:real_Supp_ATE_CATE}D--F). COMPACT also consistently
attained the highest \(R^2(U_T\mid\bm S)\) and \(R^2(U_R\mid\bm S)\), indicating
that its learned score space effectively recovered both the treatment-balancing
and residual-outcome coordinates (Figure~\ref{fig:real_R2}). Achieving optimal performance required both the treatment and residual-outcome components (Supplementary Figure~\ref{fig:real_Supp_R2}). COMPACT further outperformed both Ridge-Separate and Ridge-Shared, achieving lower ATE and CATE RMSE and higher estimation \(R^2(U_T\mid S)\) (Supplementary Figures \ref{fig:catie_ridge_rmse} and \ref{fig:catie_ridge_r2}). Moreover, the algorithm maintained
stable score-space overlap, with the mean percentage of fitted propensities
below \(0.05\) or above \(0.95\) ranging from approximately \(0.5\%\) to
\(3.9\%\) across conditions (Supplementary Figure \ref{fig:real_propensity_time}A--C; Supplementary Figure \ref{fig:real_Supp_Extreme}). DOUBLE and SL-PS again performed poorly on this
diagnostic, while CB-PS deteriorated as the number of proxies increased.
Finally, COMPACT required approximately \(1.1\) seconds per replication on
average and remained below \(1.2\) seconds even when all \(p=364\) proxies were
included (Supplementary Figure~\ref{fig:real_propensity_time}D--F). Overall, the semi-synthetic
experiments closely reproduced the principal findings of the fully synthetic
experiments.

We next evaluated the bootstrap using the CATIE semi-synthetic data. Empirical ATE coverage again closely tracked nominal coverage across confidence levels (Supplementary Figure~\ref{fig:real_bootstrap_inference}A). Moreover, the ATE type-I error rates under the null $\tau_0=0$ remained close to the nominal $5\%$ level, with all corresponding $95\%$ confidence intervals containing $5\%$ (Supplementary Figure~\ref{fig:real_bootstrap_inference}B). Both the ATE and global CATE tests achieved $100\%$ power under their respective alternatives (Supplementary Figure~\ref{fig:real_bootstrap_inference}C). These results indicate that the full-procedure bootstrap provided well-calibrated and powerful inference in the CATIE semi-synthetic setting. Thus, COMPACT maintained accurate inference, downstream ATE and CATE estimation, and reliable subspace recovery in both the synthetic and semi-synthetic settings.

\begin{figure}
    \centering
    \includegraphics[width=1\linewidth]{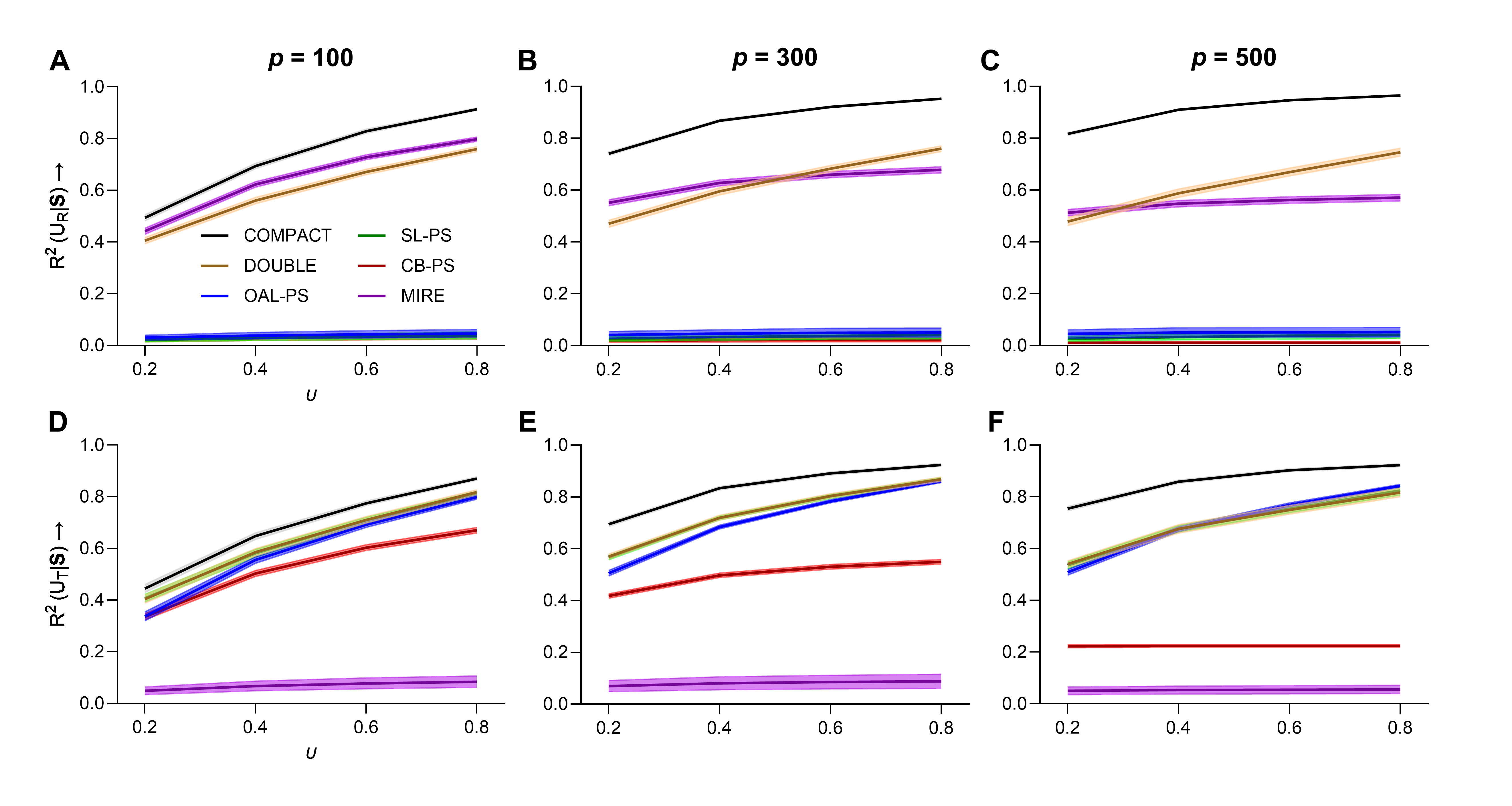}
    \caption{\textbf{Recovery of the latent adjustment coordinates in the
    CATIE-calibrated experiments.} Linear recovery of \(U_R\) (A--C) and
    \(U_T\) (D--F) across proxy-strength levels and proxy counts.}
    \par
    \raggedright
    \tiny
    \textbf{Alt text:} Six line plots of recovery of the latent adjustment coordinates in the CATIE-calibrated semi-synthetic experiments. Panels A--C show linear recovery of the residual-outcome coordinate $U_R$ and panels D--F show linear recovery of the treatment-balancing coordinate $U_T$ across proxy-strength levels and proxy counts.

    \label{fig:real_R2}
\end{figure}

\section{Discussion} \label{sec:discussion}

We introduced COMPACT, a supervised dimension-reduction framework for causal adjustment with high-dimensional proxy variables. Rather than learning a representation from treatment prediction alone, COMPACT jointly targets treatment and treatment-residualized outcome information. At the population level, the resulting two-dimensional eigenspace corresponds to the ridge-regularized linear prediction of a treatment-balancing variable $U_T$ and a residual-outcome variable $U_R$. Our theory characterizes how imperfect proxy recovery and finite-sample estimation of this space propagate into downstream causal estimation.

Computational efficiency was an explicit design goal. COMPACT constructs a low-rank spectral operator from treatment and treatment-residualized outcome cross-covariances, so both adjustment scores are obtained through a single generalized eigendecomposition rather than repeated fitting or iterative optimization of complex prediction models. This smooth spectral formulation also permits full-procedure bootstrap inference that recomputes the learned score space in each resample. The resulting two-dimensional representation is additionally useful for visualization: treatment overlap and estimated treatment-effect heterogeneity can be examined directly in the same supervised space used for adjustment. Our experiments showed favorable computational performance and well-calibrated bootstrap inference.

The theory also highlights a potentially useful role for high-dimensional baseline information. Larger samples reduce estimation error, whereas additional informative proxies can reduce population reconstruction error by providing more information about the latent adjustment variables. Thus high dimensionality need not be purely detrimental when the observed variables are weak, noisy, and partially redundant measurements of an underlying adjustment state. This motivates retaining broad collections of plausibly informative baseline variables when the relevant adjustment variables are uncertain, although arbitrary irrelevant variables need not be beneficial.

Several limitations remain. First, causal validity depends on sufficiently accurate recovery of the treatment-balancing information; residual confounding can persist when the proxy system is weak. Second, the current linear formulation recovers only information expressible through linear combinations of the supplied features, although transformations or kernel extensions could enlarge this function class. Third, COMPACT does not currently distinguish prespecified variables $\bm Z$, such as established moderators, from the proxy system. One possible extension would residualize $\bm X$, $T$, and $Y$ with respect to $\bm Z$, apply COMPACT to the residualized variables, and then reintroduce $\bm Z$ and relevant treatment interactions in the downstream analysis; the formal properties of this procedure remain to be established. Finally, the graph-theoretic criterion is derived for a scalar treatment and outcome. Multiple treatments or outcomes may yield additional dependence restrictions that remain invariant to the treatment-to-outcome edge and potentially more informative adjustment representations.

In summary, COMPACT provides a computationally efficient two-dimensional adjustment representation that jointly uses treatment and residual-outcome information, can benefit from increasingly informative proxy systems, and supports bootstrap inference that accounts for learning the score space. These properties make it a practical approach to causal adjustment when relevant baseline information is distributed across many observed proxy variables.

\section*{Data Availability}
The Adolescent Brain Cognitive Development (ABCD) Study data used in the real-data analysis are controlled-access data available to qualified investigators through the NIMH Data Archive (NDA), subject to NDA approval and applicable data-use requirements. The Clinical Antipsychotic Trials of Intervention Effectiveness (CATIE) data used to calibrate the semi-synthetic experiments are also available through the NIMH Data Archive (DOI: \url{https://doi.org/10.15154/amyn-xj28}).

R code implementing COMPACT and reproducing the numerical results reported
in this article is available at
\url{https://github.com/ericstrobl/COMPACT}. Derived score representations from the CATIE analysis are not publicly released because the derived representations retain information from the restricted data and could potentially permit partial reconstruction of participant-level information.

\bibliographystyle{biometrika}
\bibliography{reference}

\newpage
\section*{Supplementary Materials}
\setcounter{page}{1}

\setcounter{section}{0}
\setcounter{subsection}{0}
\setcounter{secnumdepth}{3}

\renewcommand{\thesection}{S\arabic{section}}
\renewcommand{\thesubsection}{\thesection.\arabic{subsection}}

\setcounter{figure}{0}
\renewcommand{\figurename}{Supplementary Figure}

\section{Related Work} \label{supp:related_work}
COMPACT connects several lines of research that address how baseline information should enter causal adjustment and treatment-effect estimation. Propensity-score methods replace a potentially high-dimensional covariate vector with the conditional probability of treatment, thereby providing a balancing score for observed covariates \citep{Rosenbaum83}. Extensions such as high-dimensional propensity-score adjustment, covariate-balancing propensity scores, and balancing weights select or weight covariates to improve empirical treatment balance \citep{Schneeweiss09,Zhang22,Imai14,Li18}. Because these methods remain centered on treatment assignment, however, they may emphasize treatment predictors that do not improve control of outcome-relevant confounding, thereby increasing variance or reducing finite-sample stability without reducing bias \citep{Brookhart06}.

This limitation has motivated methods that incorporate outcome information into adjustment. Outcome-adaptive propensity-score methods use the outcome to guide covariate selection in the treatment model \citep{Shortreed17}, while prognostic scores summarize outcome risk and provide an outcome-side analog of the propensity score \citep{Hansen08}. Propensity and prognostic scores can then be combined during matching or adjustment \citep{Leacy14}. These scores are nevertheless typically estimated through separate treatment and outcome models and combined only during adjustment. Similarly, double machine-learning methods estimate treatment and outcome nuisance functions separately and combine them through a Neyman-orthogonal estimating equation \citep{Chernozhukov18}. COMPACT differs by using treatment and treatment-residualized outcome information jointly to define the learned adjustment representation itself.

COMPACT also relates to methods that use proxies to address unmeasured confounding. Graphical and proximal causal-inference approaches establish that causal effects can sometimes be identified when confounders are unobserved but suitable proxies or negative controls are available \citep{Kuroki14,Miao18,Tchetgen24}. These approaches typically assign variables specific causal roles and obtain identification through bridge functions, integral equations, or related proxy restrictions. COMPACT considers a different setting in which a high-dimensional baseline feature system $\bm X_i$ contains partial, noisy, and overlapping measurements of a latent baseline state. Rather than assigning each variable a distinct proxy or negative-control role, COMPACT requires the proxy system to contain sufficient linearly recoverable information for the learned score space to preserve the treatment-balancing coordinate needed for adjustment.

Because COMPACT must extract this information from an unpartitioned proxy system, it also connects naturally to sufficient dimension reduction. Classical approaches, including sliced inverse regression and related central-subspace methods, seek low-dimensional projections of covariates that preserve regression information about an outcome \citep{Li91,Cook05}. Their objectives are primarily outcome-oriented, however, and do not simultaneously preserve the treatment-assignment information required for confounding adjustment.

The downstream use of a low dimensional representation also connects COMPACT to methods for treatment-effect heterogeneity. Causal trees, causal forests, Bayesian nonparametric response-surface models, and meta-learners estimate conditional average treatment effects as functions of covariates \citep{Hill11,Athey16,Wager18,Kunzel19}. When these methods use the same broad covariate vector for both confounding control and heterogeneity estimation, they generally do not distinguish variables needed for adjustment from those that index effect modification. COMPACT instead learns a low-dimensional representation of the latent coordinates needed for confounding adjustment and does not seek to preserve effect modifiers unrelated to treatment assignment. Restricting the downstream search to these adjustment-relevant coordinates reduces the dimensionality of the proxy system and, consequently, the opportunity for overfitting. In this way, COMPACT combines the balancing objective of propensity scores, the outcome relevance of prognostic adjustment and sufficient dimension reduction, the latent-state motivation of proxy methods, and the efficiency gains of targeted dimension reduction.

\section{Additional ABCD Cognitive Outcomes}
\label{supp:abcd-additional}

We additionally applied the same analysis to Picture Vocabulary, Flanker Inhibitory Control and Attention, Picture Sequence Memory, Pattern Comparison Processing Speed, and Oral Reading Recognition. Picture Vocabulary showed a pattern similar to crystallized cognition. Estimated propensity scores remained predominantly non-extreme across the learned COMPACT space, and initiation of reading for pleasure was associated with higher subsequent Picture Vocabulary scores ($\mathrm{ATE}=0.575$, 95\% CI $[0.159,1.020]$, Holm-adjusted $p=0.0165$). The global CATE-heterogeneity test was not significant, and estimated CATEs showed little variation across the COMPACT adjustment space. None of the remaining four cognitive outcomes was significant after family-wise error correction (minimum Holm-adjusted $p=0.080$, for Oral Reading).

\begin{figure}[H]
    \centering
    \includegraphics[width=0.95\linewidth]{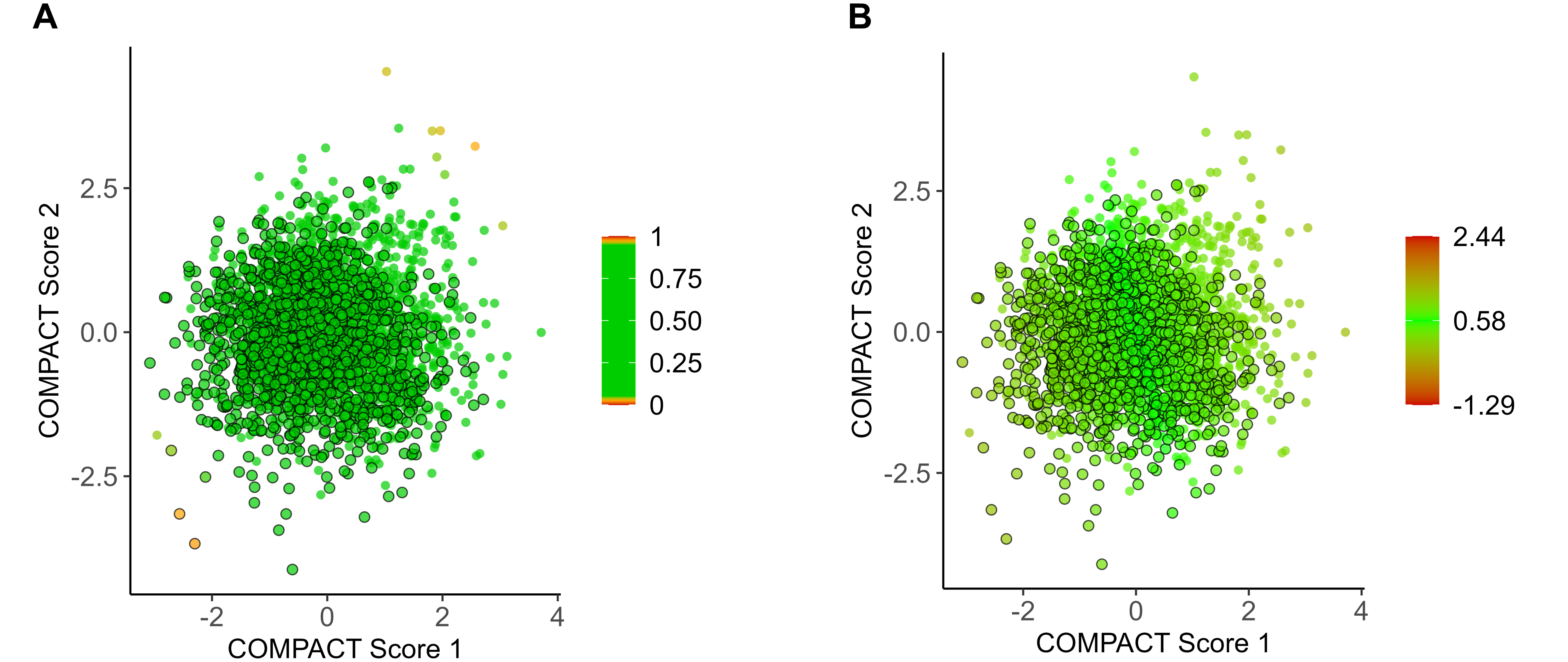}
    \caption{\textbf{COMPACT score plots for Picture Vocabulary.} Black outlines denote participants who initiated reading for pleasure. (A) COMPACT space colored by the estimated propensity score. (B) The same space colored by the estimated CATE using an ATE-centered scale spanning $\pm0.25$ outcome standard deviations.}
    \label{fig:cognitive_analysis_vocabulary}
\end{figure}

\section{Computational Complexity Analysis}
\label{sec:computational-complexity}

Let $n$ denote the number of observations and $p$ the number of candidate
adjustment variables. For a fixed regularization parameter $\eta$, COMPACT
requires $O(np)$ operations to center and residualize the data and
$O(np^2)$ operations to construct the empirical covariance matrices entering
$G_\eta$. The treatment and residual-outcome cross-covariance vectors require
only $O(np)$ additional operations.

The spectral step is particularly simple because $H$ has rank at most two in the scalar-treatment, scalar-outcome setting. COMPACT
therefore does not require a dense $p\times p$ eigendecomposition of $H$.
Instead, after a Cholesky factorization of $G_\eta$, which costs
$O(p^3)$ operations, the generalized eigenproblem can be restricted exactly
to a subspace of dimension at most two. The remaining triangular solves cost
$O(p^2)$, the decomposition of the resulting $p\times2$ matrix costs
$O(p)$, and the final eigendecomposition is at most $2\times2$ and therefore
has constant cost. Computing and whitening the two fitted score coordinates
adds only $O(np+p^2)$ operations. Consequently, a single COMPACT fit has
overall time complexity
$O(np^2+p^3)$
and memory complexity
$O(np+p^2).$

Importantly, although the worst-case time complexity contains an $O(p^3)$ term, this term arises only from a single Cholesky factorization of the positive-definite matrix $G_\eta$. The factorization is highly optimized in modern numerical libraries, so the method remains computationally fast in practice even for large numbers of candidate adjustment variables.

\section{Connection to the Latent Adjustment Representation}
\label{sec:compact-algorithm:latent_adj}

The graphical characterization in
Theorem~\ref{thm:complete-invariant-criterion}
motivates the joint treatment and residual-outcome
criterion, but does not itself establish that the
learned scores are sufficient for causal adjustment.
We now characterize the population representation
targeted by COMPACT under the causal working model
and connect it to the latent treatment-balancing
information required for adjustment.

Define the treatment-residualized outcome
\[
R_i
=
\{Y_i-\mathbb E(Y_i)\}
-
\delta_T\{T_i-\mathbb E(T_i)\},
\qquad
\delta_T
=
\frac{\operatorname{cov}(Y_i,T_i)}
{\operatorname{var}(T_i)},
\]
and let
$\bm{v}_T=\operatorname{cov}(\bm X_i,T_i)$ and
$\bm{v}_R=\operatorname{cov}(\bm X_i,R_i)$.
Define
\[
H
=
\frac{\bm{v}_T\bm{v}_R^\top+
\bm{v}_R\bm{v}_T^\top}{2},
\qquad
K_\eta
=
\Sigma_{\bm X}+\eta I_p,
\qquad
G_\eta
=
K_\eta-
\frac{\bm{v}_T\bm{v}_T^\top}
{2\operatorname{var}(T_i)}.
\]
Let $\mathcal W_\eta$ denote the generalized
eigenspace associated with the nonzero generalized
eigenvalues of $(H,G_\eta)$. Finally, define
\begin{equation}
\label{eq:hq-definition-supp}
\bm h_q
=
\Sigma_{\bm C}^{-1}
\mathbb E\!\left[
\bm C_i m_T(U_{T,i})(\bm q^\top\bm C_i)
\right].
\end{equation}
The following proposition identifies the latent
directions represented by this population eigenspace.

\begin{proposition}
\label{prop:population-compact-span}
Suppose $K_\eta$ and $G_\eta$ are positive definite and $\bm{v}_T$ and $\bm{v}_R$ are
linearly independent. Under
\eqref{eq:linear-treatment-model}--\eqref{eq:linear-proxy-model}, with
\eqref{eq:inverse-regression-linearity} imposed when $m_T$ is nonlinear,
\begin{equation}
\label{eq:population-compact-span}
\mathcal W_\eta
=
\operatorname{span}
\left\{
K_\eta^{-1}\Lambda\Sigma_{\bm C}\bm a,\,
K_\eta^{-1}\Lambda\Sigma_{\bm C}(\bm b+\bm h_q)
\right\}.
\end{equation}
Consequently, the population generalized eigenscores satisfy
\begin{equation}
\label{eq:population-score-latent-span}
\begin{split}
\left\{
\bm X_i^\top \bm w:\bm w\in\mathcal W_\eta
\right\}
=
\operatorname{span}\bigl\{&
\bm a^\top\Sigma_{\bm C}\Lambda^\top
K_\eta^{-1}\bm X_i, (\bm b+\bm h_q)^\top\Sigma_{\bm C}\Lambda^\top
K_\eta^{-1}\bm X_i
\bigr\}.
\end{split}
\end{equation}
\end{proposition}

The proof is provided in Supplementary
Materials~\ref{supp:proofs}.
The proposition has a direct prediction
interpretation. Define
$D=(\bm a,\bm b+\bm h_q)$ and the
latent adjustment vector
\begin{equation}
\label{eq:latent-adjustment-vector-supp}
\bm U_i
=
D^\top\bm C_i
=
\begin{pmatrix}
U_{T,i}\\
U_{R,i}
\end{pmatrix},
\end{equation}
where $U_{T,i}=\bm a^\top\bm C_i$ is
the treatment-balancing variable and
$U_{R,i}=(\bm b+\bm h_q)^\top\bm C_i$
is a residual-outcome variable.
Proposition~\ref{prop:population-compact-span}
shows that the population COMPACT eigenscores span
\begin{equation}
\label{eq:ridge-adjustment-predictions}
\widetilde{\bm U}_{\eta,i}
=
D^\top\Sigma_{\bm C}\Lambda^\top
K_\eta^{-1}\bm X_i.
\end{equation}
When $\eta=0$ and $\Sigma_{\bm X}$ is
nonsingular, this is the best linear predictor
of $\bm U_i$ from the observed proxies
$\bm X_i$, since
$\operatorname{cov}(\bm U_i,\bm X_i)
=D^\top\Sigma_{\bm C}\Lambda^\top$.
For $\eta>0$, it is the corresponding
ridge-regularized linear predictor.
Thus, COMPACT recovers, up to a change of basis,
the ridge-regularized linear reconstruction
of the latent adjustment variables.

The causal relevance of this representation
follows from its inclusion of $U_{T,i}$.
By \eqref{eq:linear-treatment-model},
$T_i\ci\bm C_i\mid U_{T,i}$.
Because $U_{R,i}$ is also a pre-treatment
function of $\bm C_i$, augmenting $U_{T,i}$
with $U_{R,i}$ preserves this balancing property:
\(
T_i\ci\bm C_i\mid\bm U_i.
\)
Together with the conditional mean restriction
in \eqref{eq:conditional-mean-restriction},
this implies
\(
\mathbb E\{Y_i(t)\mid T_i,\bm U_i\}
=
\mathbb E\{Y_i(t)\mid\bm U_i\}.
\)
Consequently, exact recovery of the latent
adjustment space permits adjustment for the
learned scores in place of $\bm U_i$.
With imperfect proxies, COMPACT instead
approximates this space through its regularized
linear reconstruction. The accuracy of this
approximation and its implications for causal
effect estimation are established in
Theorems~\ref{thm:proxy-sample-tradeoff}
and~\ref{thm:causal-consistency-compact-plugin}.

\section{Additional Details for the Theoretical Results}
\label{supp:theoretical-details}

This section provides the detailed notation, assumptions, and intermediate results underlying the theoretical results stated in the main text.

\subsection{Joint Regularization and Score-Space Recovery}
\label{supp:joint-regularization}

We provide the notation and assumptions underlying
Theorem~\ref{thm:joint-regularization-recovery}. Define the standardized latent adjustment vector
\(
\bm U_i^\star
=
\operatorname{var}(\bm U_i)^{-1/2}\bm U_i,
\)
where $\bm U_i$ is defined in
\eqref{eq:latent-adjustment-vector-supp}.
Thus $\operatorname{var}(\bm U_i^\star)=I_2$.
Define
$\Lambda_U=\operatorname{cov}(\bm X_i,\bm U_i^\star)$
and $\bm\xi_i=\bm X_i-\Lambda_U\bm U_i^\star$,
with $\Sigma_\xi=\operatorname{var}(\bm\xi_i)$.
Thus,
$\bm X_i=\Lambda_U\bm U_i^\star+\bm\xi_i$,
where $\operatorname{cov}(\bm\xi_i,\bm U_i^\star)=0$.

For each $\eta>0$, let $\widehat B_\eta$ be
the fitted coefficient basis from
Section~\ref{sec:compact-algorithm:spectral},
and define the raw scores
$\widehat{\bm S}_{\eta,i}^{\mathrm{raw}}
=\widehat B_\eta^\top\bm X_i$.
Substituting the proxy decomposition gives
\[
\begin{aligned}
\widehat{\bm S}_{\eta,i}^{\mathrm{raw}}
&=\widehat B_\eta^\top
(\Lambda_U\bm U_i^\star+\bm\xi_i)\\
&=\widehat B_\eta^\top\Lambda_U\bm U_i^\star
+\widehat B_\eta^\top\bm\xi_i.
\end{aligned}
\]
Define $A_\eta=\widehat B_\eta^\top\Lambda_U$ and
$\Omega_\eta=\widehat B_\eta^\top\Sigma_\xi \widehat B_\eta$.
The first matrix describes how the latent adjustment
variables enter the scores, while the second is the
covariance of their residual component.

When $A_\eta$ is nonsingular, multiplying the scores
by its inverse gives
\[
A_\eta^{-1}\widehat{\bm S}_{\eta,i}^{\mathrm{raw}}
=
\bm U_i^\star+A_\eta^{-1}\widehat B_\eta^\top\bm\xi_i.
\]
This invertible transformation restores the latent
signal to its original scale and basis without
changing the learned score space. The covariance
of the remaining residual component is
\begin{equation}
\label{eq:joint-regularization-relative-noise}
N_\eta=A_\eta^{-1}\Omega_\eta A_\eta^{-\top}.
\end{equation}
Thus, $N_\eta$ measures residual variation relative
to the recoverable latent adjustment information,
independently of score calibration.

We next distinguish unavoidable proxy-reconstruction
error from the additional error introduced by
estimating a score space. Suppose $\Sigma_\xi$ is
positive definite and define the collective proxy
information matrix
\begin{equation}
\label{eq:joint-proxy-information}
\mathcal I_p
=
\Lambda_U^\top\Sigma_\xi^{-1}\Lambda_U,
\qquad
s_p=\lambda_{\min}(\mathcal I_p).
\end{equation}
The quantity $s_p$ measures information about the
worst-recovered latent adjustment direction available
from all $p$ proxies. It can increase as additional
informative proxies are observed, even when each
individual proxy is only weakly associated with
the latent adjustment variables.

When $\mathcal I_p$ is nonsingular, define
$L_p^*=\Sigma_\xi^{-1}\Lambda_U\mathcal I_p^{-1}$.
This is the coefficient matrix of the minimum-noise
linear reconstruction of $\bm U_i^\star$, subject
to preserving both latent coordinates.
For the fitted scores, define
$L_\eta=\widehat B_\eta A_\eta^{-\top}$ and
\begin{equation}
\label{eq:joint-extraction-error}
F_\eta
=
\Sigma_\xi^{1/2}(L_\eta-L_p^*).
\end{equation}
The matrix $F_\eta$ measures the deviation of the
learned reconstruction coefficients from their
minimum-noise population values, weighted by the
residual proxy covariance. Its squared magnitude
therefore measures excess score-extraction error.

Let $\bm\gamma_W=\operatorname{cov}(\bm U_i^\star,W_i)$
and $v_W=\operatorname{var}(W_i)$ for $W=T,R$.
We make the following assumptions:

\begin{assumption}
\label{assump:joint-regularization}
Consider two candidate penalties $0<\eta_-<\eta_+$.
The following conditions hold.

\begin{enumerate}[leftmargin=1.25em,
labelsep=0.4em,itemsep=0.25em]

\item[(a)] \textbf{Recoverable adjustment information.}
The latent-variable and proxy conditions of
Proposition~\ref{prop:population-compact-span} hold.
The matrix $\Sigma_\xi$ is positive definite and
$\Lambda_U$ has rank two. There exist constants
$c_\Gamma,c_T,K_v>0$ such that
\[
\sigma_{\min}
\{(\bm\gamma_T,\bm\gamma_R)\}\geq c_\Gamma,
\qquad
v_T-\|\bm\gamma_T\|^2\geq c_T,
\]
and $0<v_T,v_R\leq K_v$.
The matrices $A_{\eta_-}$ and $A_{\eta_+}$
are nonsingular with probability tending to one.

\item[(b)] \textbf{High-dimensional extraction error
and its attenuation.}
There exist a deterministic sequence $r_{p,n}>0$,
constants $0<c_E\leq C_E<\infty$ and $0\leq q_0<1$,
and deterministic matrices $H_{p,n}$ satisfying
$\|H_{p,n}\|_{\mathrm{op}}\leq q_0$, such that,
with probability tending to one,
\begin{equation}
\label{eq:joint-extraction-scale}
c_Er_{p,n}I_2
\preceq F_{\eta_-}^\top F_{\eta_-}
\preceq C_Er_{p,n}I_2,
\end{equation}
and
\begin{equation}
\label{eq:joint-extraction-contraction}
\|F_{\eta_+}-H_{p,n}F_{\eta_-}\|_{\mathrm{op}}
=o_p(r_{p,n}^{1/2}).
\end{equation}
Thus, $r_{p,n}$ describes the magnitude of
the avoidable extraction error at the smaller
penalty, while stronger regularization contracts
that error in both latent directions.

\item[(c)] \textbf{Calibration-error dominance.}
Increasing regularization from $\eta_-$ to $\eta_+$
increases calibration error by more than it
reduces span prediction error. Specifically,
there exists a deterministic sequence $b_n>0$
such that, with probability tending to one,
\begin{equation}
\label{eq:joint-regularization-calibration-condition}
\begin{aligned}
&
\mathbb E(P_{\eta_+}g_W-\widehat W_{\eta_+})^2
-
\mathbb E(P_{\eta_-}g_W-\widehat W_{\eta_-})^2\\
&\quad\geq
\mathbb E(W_i-P_{\eta_-}W_i)^2
-
\mathbb E(W_i-P_{\eta_+}W_i)^2
+b_n
\end{aligned}
\end{equation}
for both $W=T$ and $W=R$.

\item[(d)] \textbf{Cross-validation concentration.}
Let $\widehat J_C(\eta)$ and
$\widehat L_W(\eta)$ denote the empirical
cross-validation criteria. Set
$K=C_E+1$. Assume
\[
\max_{\eta\in\{\eta_-,\eta_+\}}
|\widehat J_C(\eta)-J_C(\eta)|
=
o_p\left\{
\frac{r_{p,n}}
{(1+s_p^{-1}+Kr_{p,n})^3}
\right\}
\]
and
\[
\max_{\substack{W=T,R\\
\eta\in\{\eta_-,\eta_+\}}}
|\widehat L_W(\eta)-L_W(\eta)|
=o_p(b_n).
\]
These conditions include stability of the
fold-trained score spaces relative to the
final fitted representations.

\end{enumerate}
\end{assumption}

\noindent\textbf{Remark.}
Part (a) requires that the observed proxies
contain sufficient information to distinguish
and recover both latent adjustment directions.
This is plausible when many baseline variables
provide noisy, partially redundant measurements
of the same underlying adjustment state.
Part (b) assumes that stronger regularization
reduces excess score-extraction error by
attenuating estimation noise. In high-dimensional
settings with $d_p$ effective nuisance directions,
this error can scale as $r_{p,n}\asymp d_p/n$.
The condition is plausible when estimation
variance dominates additional shrinkage bias.
Part (c) captures the distinction between
recovering a score space and calibrating
individual predictions. Stronger regularization
may improve the former while increasing
coefficient shrinkage error in the latter,
causing separately tuned ridge regressions
to prefer a smaller penalty. Finally,
Part (d) requires cross-validation fluctuations
to be small relative to the population differences
between candidate penalties, which is plausible
with sufficiently large validation samples and
stable score estimation across folds.

\subsection{Sampling Consistency for the COMPACT Score Space} \label{supp:theorem2-details1}

We fix a ridge parameter \(\eta>0\). Let
$B_\eta\in\mathbb R^{p\times2}$ and
$\widehat B_\eta\in\mathbb R^{p\times2}$ have columns spanning, respectively, the
population and empirical (nonzero) generalized eigenspaces. Define
\[
\mathcal S_\eta
=
\left\{\bm X_i^\top B_\eta{\bm{\xi}}:{\bm{\xi}}\in\mathbb R^2\right\},
\qquad
\widehat{\mathcal S}_\eta
=
\left\{\bm X_i^\top\widehat B_\eta{\bm{\xi}}:{\bm{\xi}}\in\mathbb R^2\right\}.
\]
These spaces do not depend on the particular bases chosen
within the two-dimensional generalized eigenspaces.

Let \(H\) and \(G_\eta\) denote the population matrices
defining the generalized eigenproblem, and define
\[
\mathcal L
=
G_\eta^{-1/2}HG_\eta^{-1/2},
\qquad
\widehat{\mathcal L}
=
\widehat G_\eta^{-1/2}\widehat H\widehat G_\eta^{-1/2}.
\]
For a ridge parameter \(\eta>0\), define the
regularized effective dimension
\[
d_{\mathrm{eff}}(\eta)
=
\operatorname{tr}
\left\{
\Sigma_{\bm X}
(\Sigma_{\bm X}+\eta I_p)^{-1}
\right\}.
\]

We make the following assumptions:
\begin{assumption} \label{assump:regularity}
The following conditions hold:

\begin{enumerate}[leftmargin=1.25em,labelsep=0.4em,itemsep=0.25em]
\item[(a)] \textbf{Regularized operator concentration.}
\(
\|\widehat{\mathcal L}-\mathcal L\|_{\mathrm{op}}
\in
O_p\left[
\left\{
\frac{d_{\mathrm{eff}}(\eta)}{n}
\right\}^{1/2}
\right].
\)

\item[(b)] \textbf{Stable spectral metric and concentration.}
There exist constants $0<c_{\mathbb{G}}<K_{\mathbb{G}}<\infty$ such that
\[
c_{\mathbb{G}}\|{\bm{\zeta}}\|^2
\leq
{\bm{\zeta}}^\top G_\eta{\bm{\zeta}}
\leq
K_{\mathbb{G}}\|{\bm{\zeta}}\|^2
\]
for every ${\bm{\zeta}}\in\mathbb R^p$. Moreover, with probability tending to one, the
empirical spectral metric satisfies
\[
\frac{c_{\mathbb{G}}}{2}\|{\bm{\zeta}}\|^2
\leq
{\bm{\zeta}}^\top \widehat G_\eta{\bm{\zeta}}
\leq
2K_{\mathbb{G}}\|{\bm{\zeta}}\|^2
\]
for every ${\bm{\zeta}}\in\mathbb R^p$, and
\(
\|\widehat G_\eta-G_\eta\|_{\mathrm{op}}
\in
O_p\left[
\left\{
\frac{d_{\mathrm{eff}}(\eta)}{n}
\right\}^{1/2}
\right].
\)

\item[(c)] \textbf{Separated nonzero eigenspace.}
The operator $\mathcal L$ has rank two, and the absolute value of each of its
two nonzero eigenvalues is at least $\Delta$, where
\(
\{d_{\mathrm{eff}}(\eta)/n\}^{1/2}/\Delta\to0
\)
and $\Delta\in O(1)$. Thus the retained
two-dimensional cluster is separated from the zero eigenspace.

\item[(d)] \textbf{Nondegenerate score space.}
The population score Gram matrix
\(
\Gamma_{S,\eta}={B_\eta}^{\top}\Sigma_{\bm X}B_\eta
\)
has smallest eigenvalue
\(
\lambda_{\min}(\Gamma_{S,\eta})\geq c_S>0.
\)

\item[(e)] \textbf{Bounded covariance.}
The covariance matrix satisfies
\(
\|\Sigma_{\bm X}\|_{\mathrm{op}}\leq K_{\bm X} <\infty.
\)
\end{enumerate}
\end{assumption}

\noindent \textbf{Remark.}
Assumption~\ref{assump:regularity} collects regularity conditions for the single
generalized eigendecomposition. Part (a) requires the empirical regularized
spectral operator to concentrate around its population analog.
This is reasonable because the operators are built from sample covariance and
cross-covariance quantities, and the ridge parameter controls the effective
dimension of the problem. Part (b) controls the conditioning of the generalized
eigenproblem. It requires
${\bm{\zeta}}^\top G_\eta{\bm{\zeta}}$
to define a stable squared norm, so that the
spectral solve does not amplify nearly null
coefficient directions. Part (c) separates the retained rank-two eigenspace from the zero
eigenspace. It does not impose an ordering or separation between the two
retained eigenvalues because only their joint span matters. Part (d) requires the population score Gram matrix to be nondegenerate, so
that the extracted scores span a genuine two-dimensional
\(L_2(\mathbb P_X)\) adjustment space. Part (e) imposes bounded variation of \(\bm X_i\). Together,
these conditions ensure that the population COMPACT target is well defined and
that the empirical score space is stable under sampling variation.

We are now ready to state our first result regarding estimation consistency to
the population analog:

\begin{lemma}
\label{lem:spectral-subspace-error}
Consider Assumption~\ref{assump:regularity}. Then
\[
\operatorname{dist}^2
\left(
\widehat{\mathcal S}_\eta,
\mathcal S_\eta
\right)
\in
O_p
\left(
\frac{d_{\mathrm{eff}}(\eta)}{n\Delta^2}
\right).
\]
\end{lemma}

\noindent \textbf{Remark.} The lemma shows that the estimated COMPACT score space converges to
its population analog at the usual root-\(n\) scale. This establishes
sampling consistency of the empirical spectral procedure. However, convergence to the
population COMPACT solution is only an intermediate target. The next question
is whether the span of the population scores accurately reconstruct the span of the latent treatment
and residual-outcome variables $\bm U_i$ through the observed proxy system.

\subsection{Recovery of the Latent Adjustment Space} \label{supp:theorem2-details2}

Recall that
$\bm U_i^\star$ denotes the
standardized latent adjustment vector. Because whitening is
invertible, the raw and standardized vectors span the same latent adjustment space:
\[
\mathcal U
=
\operatorname{span}_{L_2}(U_{T,i},U_{R,i})
=
\left\{
{\bm{\xi}}^\top \bm U_i^\star:{\bm{\xi}}\in\mathbb R^2
\right\}.
\]

Define the population score vector
\(
\bm S_{\eta,i}
=
B_\eta^\top\bm X_i,
\)
so that
\(
\mathcal S_\eta
=
\operatorname{span}_{L_2}(S_{\eta,i,1},S_{\eta,i,2}).
\)
We measure the information about \(\bm U^\star\) not retained by the
population score space through the linear-prediction residual covariance
\(
\Sigma_{\bm U^\star\mid \bm S_\eta}
=
\operatorname{var}(\bm U_i^\star\mid_{\mathrm{lin}}\bm S_{\eta,i}).
\)
Define the effective proxy strength, with the convention
$p_{\mathrm{eff}}=\infty$ under exact linear reconstruction, by
\begin{equation}\label{eq:proxy_strength}
p_{\mathrm{eff}}
=
\frac{1}{\lambda_{\max}(\Sigma_{\bm U^\star\mid \bm S_\eta})} - 1.
\end{equation}
Under the normalization $\operatorname{var}(\bm U_i^\star)=I_2$, every
unit direction in the latent adjustment space has total variance one. Therefore
$\lambda_{\max}(\Sigma_{\bm U^\star\mid \bm S_\eta})$ is the largest residual variance
fraction, over all unit directions in the latent adjustment space, after linearly
projecting $\bm U_i^\star$ onto $\bm S_{\eta,i}$. Thus
$p_{\mathrm{eff}}$ can be interpreted as a
worst-direction explained-to-unexplained variance ratio: it is small when at
least one relevant latent direction is poorly predicted by the scores, and large
only when all relevant latent directions are well recovered.

The following conclusion holds:

\begin{lemma}
\label{lem:proxy-recovery}
If $p_{\mathrm{eff}}>0$, then
\(
\operatorname{dist}^2
\left(
\mathcal U,
\mathcal S_\eta
\right)
\in
O\left(
\frac{1}{p_{\mathrm{eff}}}
\right).
\)
\end{lemma}

\noindent \textbf{Remark.}
The population COMPACT score space
$\mathcal S_\eta$ thus approximates the latent adjustment
space $\mathcal U$ whenever the retained proxy strength
$p_{\mathrm{eff}}$ is large. The lemma does not
assert that the empirical scores have been accurately estimated; that sampling
error is controlled separately by Lemma~\ref{lem:spectral-subspace-error}.
The following theorem combines these two errors: the proxy-recovery error from
the population score space and the estimation error from learning that score
space from a finite sample.

\begin{reptheorem}{thm:proxy-sample-tradeoff}
Suppose the conditions of Lemmas~\ref{lem:spectral-subspace-error} and
\ref{lem:proxy-recovery} hold. Then
\[
\operatorname{dist}^2
\left(
\mathcal U,
\widehat{\mathcal S}_\eta
\right)
\in
O_p\left(
\frac{1}{p_{\mathrm{eff}}}
+
\frac{d_{\mathrm{eff}}(\eta)}{n\Delta^2}
\right).
\]
\end{reptheorem}

\noindent \textbf{Remark.} The theorem separates two requirements for successful recovery. First, the
observed baseline feature system must contain sufficiently strong information
about both latent adjustment variables, as expressed by
\(p_{\mathrm{eff}}\to\infty\). Second, the effective
regularized dimension of the score-learning problem must be small relative to
sample size and the spectral gap, as expressed by
\(d_{\mathrm{eff}}(\eta)/(n\Delta^2)\to0\). Under these two conditions, the
estimated COMPACT score space converges to the latent adjustment space
\(\mathcal U\) in projection distance.

This recovery result also gives a formal interpretation to diagnostic
visualization of the learned score space. Since the COMPACT
directions lie in the treatment--residual-outcome span, plotting units by their
learned COMPACT scores displays the estimated low-dimensional adjustment space. Such plots can be
used as diagnostics for the downstream causal analysis: they show treatment
overlap, clustering, and influential regions in the same space used for
adjustment.

\subsection{Detailed Conditions for Causal Consistency}
\label{supp:theorem3-details}

Let $\mathcal T_0\subseteq\mathcal T$ denote the treatment values over which
uniform statements are made. For $t\in\mathcal T_0$, define the marginal
causal response function
\(
\theta(t)
=
\mathbb E\{Y_i(t)\}.
\)
Let
\(
\mathcal D_n=\{O_1,\ldots,O_n\}
\)
denote the observed sample used to learn the score map and fit the downstream
causal estimator. We allow both stages to use all observations in
\(\mathcal D_n\). Conditional on \(\mathcal D_n\), we treat the fitted map
\(\widehat s_\eta\) as fixed and evaluate all population quantities involving
\(\widehat{\bm S}_\eta=\widehat s_\eta(\bm X)\) using a generic population draw
\(O=(Y,T,\bm X)\) independent of \(\mathcal D_n\). Under this convention, define
\(
m_{\widehat{\bm S}_\eta}(t,\bm s)
=
\mathbb E(Y_i\mid T_i=t,\widehat{\bm S}_{\eta,i}=\bm s).
\)
The corresponding score-adjusted marginal response curve is
\[
\theta_{\widehat{\bm S}_\eta}^{\mathrm{marg}}(t)
=
\mathbb E\left[
m_{\widehat{\bm S}_\eta}(t,\widehat{\bm S}_{\eta,i})
\right],
\]
and the score-indexed causal response curve is
\[
\theta_{\widehat{\bm S}_\eta}^{\mathrm{cond}}(t;\bm s)
=
\mathbb E\{Y_i(t)\mid \widehat{\bm S}_{\eta,i}=\bm s\},
\qquad \bm s\in\mathcal S_0.
\]

We now make the following assumption:

\begin{assumption}
\label{assump:causal}
The following conditions hold:

\begin{enumerate}[leftmargin=1.25em,labelsep=0.4em,itemsep=0.25em]
\item[(a)] \textbf{Consistency and latent mean exchangeability under the evaluation law.}
The observed outcome satisfies $Y_i=Y_i(T_i)$. The treatment law and proxy
conditions satisfy \eqref{eq:linear-treatment-model} and
\eqref{eq:linear-proxy-model}, respectively, and
\eqref{eq:latent_exchangeability} holds in mean. Population expectations
involving $\widehat s_\eta$ follow the evaluation-law convention above and are
therefore taken over a generic population draw conditional on the full-sample
fitted map. Consequently, for each $t\in\mathcal T_0$,
\[
\mathbb E\{Y_i(t)\mid T_i=t,\bm U_i^\star,\widehat{\bm S}_{\eta,i}\}
=
\mathbb E\{Y_i(t)\mid \bm U_i^\star,\widehat{\bm S}_{\eta,i}\}.
\]

\item[(b)] \textbf{COMPACT latent-variable recovery.}
Let \(
\widehat{\mathcal S}_\eta
=
\operatorname{span}_{L_2}
\left\{
\widehat S_{\eta,i,1},
\widehat S_{\eta,i,2}
\right\}
\). The reconstruction map
\[
\mathcal R_n(\widehat{\bm S}_{\eta,i})
=
\begin{pmatrix}
P_{\widehat{\mathcal S}_\eta}U_{i,1}^\star\\
P_{\widehat{\mathcal S}_\eta}U_{i,2}^\star
\end{pmatrix}.
\]
satisfies
\[
\mathbb E\left[
\|\bm U_i^\star-\mathcal R_n(\widehat{\bm S}_{\eta,i})\|^2
\right]
\in
O_p\left(
\frac{1}{p_{\mathrm{eff}}}
+
\frac{d_{\mathrm{eff}}(\eta)}{n\Delta^2}
\right),
\]
and, for score-indexed estimation,
\[
\sup_{\bm s\in\mathcal S_0}
\mathbb E\left[
\|\bm U_i^\star-\mathcal R_n(\widehat{\bm S}_{\eta,i})\|^2
\mid \widehat{\bm S}_{\eta,i}=\bm s
\right]
\in
O_p\left(
\frac{1}{p_{\mathrm{eff}}}
+
\frac{d_{\mathrm{eff}}(\eta)}{n\Delta^2}
\right),
\]
where \(\mathcal S_0\subset\mathbb R^2\) denotes the region of score values over
which the score-indexed causal response is evaluated.

\item[(c)] \textbf{Smooth latent response surface.}
Let
\[
\mu(t,\bm u,\bm s)
=
\mathbb E\{Y_i(t)\mid \bm U_i^\star=\bm u,\widehat{\bm S}_{\eta,i}=\bm s\}.
\]
There exists $L_\mu<\infty$ such that, uniformly over
$t\in\mathcal T_0$ and $\bm s\in\mathcal S_0$,
\[
|\mu(t,\bm u,\bm s)-\mu(t,\bm u',\bm s)|
\leq
L_\mu\|\bm u-\bm u'\|
\]
for all $\bm u,\bm u'$. Also assume
\(
\sup_{t\in\mathcal T_0}\mathbb E\{Y_i(t)^2\}<\infty .
\)

\item[(d)] \textbf{Bounded residual latent-variable imbalance within score strata.}
For every $t\in\mathcal T_0$ and $\bm s\in\mathcal S_0$, the conditional
law of $\bm U_i^\star$ given $(T_i=t,\widehat{\bm S}_{\eta,i}=\bm s)$ is absolutely
continuous with respect to the conditional law of $\bm U_i^\star$ given
$\widehat{\bm S}_{\eta,i}=\bm s$, with uniformly bounded density ratio:
\[
\frac{
dP_{\bm U^\star\mid T_i=t,\widehat{\bm S}_{\eta,i}=\bm s}
}{
dP_{\bm U^\star\mid \widehat{\bm S}_{\eta,i}=\bm s}
}
\leq
K_\omega
\]
almost surely, for some constant $K_\omega<\infty$.

\item[(e)] \textbf{Consistent downstream score-adjusted causal learner.}
Let $\widehat\theta_{\widehat{\bm S}_\eta}^{\mathrm{marg}}(t)$ denote the finite-sample
output of a downstream causal estimator of the marginal response that adjusts
for the unit-level variables $\widehat{\bm S}_{\eta,i}$, and let
$\widehat\theta_{\widehat{\bm S}_\eta}^{\mathrm{cond}}(t;\bm s)$ denote the corresponding
score-indexed response estimator. Assume these estimators are consistent for
their score-adjusted population targets:
\[
\sup_{t\in\mathcal T_0}
\left|
\widehat\theta_{\widehat{\bm S}_\eta}^{\mathrm{marg}}(t)
-
\theta_{\widehat{\bm S}_\eta}^{\mathrm{marg}}(t)
\right|
\in
o_p(1),
\]
and
\[
\sup_{t\in\mathcal T_0,\;\bm s\in\mathcal S_0}
\left|
\widehat\theta_{\widehat{\bm S}_\eta}^{\mathrm{cond}}(t;\bm s)
-
m_{\widehat{\bm S}_\eta}(t,\bm s)
\right|
\in
o_p(1).
\]
\end{enumerate}
\end{assumption}

\noindent \textbf{Remark.}
This assumption collects the causal and downstream regularity
conditions needed to translate COMPACT score recovery into consistency of
score-adjusted causal estimators. Part (a) collects the model conditions under which latent mean exchangeability
\[
\mathbb E\{Y_i(t)\mid
T_i=t,\bm U_i^\star,\widehat{\bm S}_{\eta,i}\}
=
\mathbb E\{Y_i(t)\mid
\bm U_i^\star,\widehat{\bm S}_{\eta,i}\}
\]
holds (Lemma \ref{lem:score-augmented-latent-exchangeability}).
The unconditional reconstruction bound in Part (b) follows from
Theorem~\ref{thm:proxy-sample-tradeoff}; the uniform conditional bound is an
additional requirement that prevents reconstruction error from concentrating within
particular regions of the learned score space. Part (c) is a smoothness condition,
where small errors in reconstructing the latent adjustment variables should lead only to
small changes in the conditional mean of the potential outcome. Part (d) is a
residual overlap condition within score strata. It rules out settings in which,
even after conditioning on $\widehat{\bm S}_{\eta,i}$, treatment assignment creates
arbitrarily large imbalance in the remaining latent adjustment variables. Part (e) separates the quality of the COMPACT adjustment
representation from the quality of the final causal learner: once the learned
scores are treated as the adjustment variables, the downstream estimator is
assumed to consistently estimate both the marginal score-adjusted response and
the score-indexed response surface. Together, these conditions ensure that
approximation of the latent adjustment vector by the learned COMPACT scores is
sufficient to control the bias of both marginal and score-indexed causal
estimation.

\subsection{Detailed Conditions for Bootstrap Inference}
\label{supp:theorem4-details}

For inference, we bootstrap the score-learning and downstream OLS steps to
account for uncertainty from both stages of estimation. We require the following
assumption.

\begin{assumption}
\label{assump:L-asymptotic-linearity}
The following conditions hold.

\begin{enumerate}[leftmargin=1.25em,labelsep=0.4em,itemsep=0.25em]
\item[(a)] \textbf{Local regularity of the COMPACT map.}
The residualization denominators are bounded away from zero, $G_\eta$ is
positive definite, and $H$ and $G_\eta$, viewed as functions of the
finite-dimensional moments entering COMPACT, are continuously differentiable
in a neighborhood of their population values.
The two-dimensional eigenspace
of $\mathcal L$ is separated from the zero eigenspace by a positive gap, and
the population score Gram matrix \( \Gamma_{S,\eta}\) is nonsingular.

\item[(b)] \textbf{Moment condition.}
The empirical moments used to construct $\widehat H$ and
$\widehat G_\eta$ have finite second moments.
\end{enumerate}
\end{assumption}

\noindent \textbf{Remark.}
Assumption~\ref{assump:L-asymptotic-linearity}(a) collects the conditions
needed for the COMPACT score-space projection to vary continuously and
differentiably with the underlying moments
(Lemma~\ref{lem:score-map-asymptotic-linearity}).

\section{Proofs} \label{supp:proofs}

\subsection{Theorem \ref{thm:complete-invariant-criterion}}

\begin{reptheorem}{thm:complete-invariant-criterion}
Let $\bm S_{B,i}$, $T_i$, and $Y_i$ satisfy
the temporal order
$\bm S_{B,i}\prec T_i\prec Y_i$.
Treat $\bm S_{B,i}$ as a single vector-valued
node. Interpret the six primitive relations in
\eqref{eq:primitive-dependence-relations}
under d-connection and faithfulness, and
assume no selection bias. Define
\[
\mathcal Q_{\mathrm{inv}}
=
\left\{
\bm S_{B,i}\not\!\ci T_i,\;
\bm S_{B,i}\not\!\ci Y_i\mid T_i
\right\}.
\]
Then $\mathcal Q_{\mathrm{inv}}$ has the
following properties:

\begin{enumerate}
\item[\textnormal{(a)}]
\textbf{Exhaustiveness.}
The set $\mathcal Q_{\mathrm{inv}}$ contains
all treatment-edge-invariant primitive relations
in \eqref{eq:primitive-dependence-relations}:
\[
\mathcal Q_{\mathrm{inv}}
=
\left\{
q\in\mathcal Q(\bm S_{B,i},T_i,Y_i):
q\text{ is treatment-edge invariant}
\right\}.
\]

\item[\textnormal{(b)}]
\textbf{Nonredundancy.}
For every proper subset
$\mathcal A\subsetneq\mathcal Q_{\mathrm{inv}}$,
there exists a faithful ordered DAG
$\mathbb G_{\mathcal A}$ that satisfies every
relation in $\mathcal A$ but does not satisfy
every relation in $\mathcal Q_{\mathrm{inv}}$.
Equivalently, every proper subset represents
a strictly larger graph class:
\[
\mathfrak G(\mathcal Q_{\mathrm{inv}})
\subsetneq\mathfrak G(\mathcal A).
\]
\end{enumerate}
\end{reptheorem}

\begin{proof}
We treat $\bm S_{B,i}$ as a single
vector-valued pretreatment node and prove
the two statements in order.

\medskip
\noindent\textit{Exhaustiveness.}
First consider $\bm S_{B,i}\not\!\ci T_i$,
so that an active path exists between
$\bm S_{B,i}$ and $T_i$ in the DAG.
Such a path cannot involve
$T_i\rightarrow Y_i$.
Indeed, because
$\bm S_{B,i}\prec T_i\prec Y_i$,
any path that traverses this edge from
$T_i$ toward $Y_i$ and returns to the
pretreatment score node must encounter
an unconditioned collider, which blocks
the path in the absence of selection bias.
Consequently, adding or removing
$T_i\rightarrow Y_i$ cannot change whether
an active path exists between
$\bm S_{B,i}$ and $T_i$.
Thus, $\bm S_{B,i}\not\!\ci T_i$
is treatment-edge invariant under faithfulness.

Next consider
$\bm S_{B,i}\not\!\ci Y_i\mid T_i$.
We follow a similar argument.
Because $T_i$ is conditioned on,
it cannot be a noncollider on an active
path between $\bm S_{B,i}$ and $Y_i$.
Any path involving $T_i\rightarrow Y_i$
has $T_i$ as a noncollider and is therefore
blocked by conditioning on $T_i$.
Consequently, adding or removing this
edge cannot change whether an active path
exists between $\bm S_{B,i}$ and $Y_i$
given $T_i$.
Thus,
$\bm S_{B,i}\not\!\ci Y_i\mid T_i$
is treatment-edge invariant under faithfulness.

We conclude that both relations in
$\mathcal Q_{\mathrm{inv}}$ satisfy
\eqref{eq:uniform-treatment-effect-invariance}.

It remains to show that none of the other
four primitive relations in
\eqref{eq:primitive-dependence-relations}
is treatment-edge invariant.

\begin{enumerate}[label=(\alph*)]
\item
Let $\mathbb G$ contain only
$\bm S_{B,i}\rightarrow T_i$.
Then $\bm S_{B,i}\ci_{\mathbb G}Y_i$.
Adding $T_i\rightarrow Y_i$ creates
the open path
\[
\bm S_{B,i}\rightarrow T_i\rightarrow Y_i,
\]
so
$\bm S_{B,i}\not\!\ci_{\mathbb G^+}Y_i$.
Therefore,
$\bm S_{B,i}\not\!\ci Y_i$
is not treatment-edge invariant.

\item
Let $\mathbb G$ contain only
$\bm S_{B,i}\rightarrow Y_i$,
with $T_i$ isolated.
Then
$\bm S_{B,i}\ci_{\mathbb G}T_i\mid Y_i$.
Adding $T_i\rightarrow Y_i$
creates the collider
\[
\bm S_{B,i}\rightarrow Y_i\leftarrow T_i.
\]
Conditioning on $Y_i$ opens this path,
giving
$\bm S_{B,i}\not\!\ci_{\mathbb G^+}T_i\mid Y_i$.
Therefore,
$\bm S_{B,i}\not\!\ci T_i\mid Y_i$
is not treatment-edge invariant.

\item
Let $\mathbb G$ be the edgeless graph.
Then $T_i\ci_{\mathbb G}Y_i$.
Adding $T_i\rightarrow Y_i$ gives
$T_i\not\!\ci_{\mathbb G^+}Y_i$.
Therefore,
$T_i\not\!\ci Y_i$
is not treatment-edge invariant.

\item
For the same edgeless graph,
$T_i\ci_{\mathbb G}Y_i\mid\bm S_{B,i}$.
After adding $T_i\rightarrow Y_i$,
the direct path remains active conditional
on $\bm S_{B,i}$.
Hence
$T_i\not\!\ci_{\mathbb G^+}Y_i
\mid\bm S_{B,i}$,
so
$T_i\not\!\ci Y_i\mid\bm S_{B,i}$
is not treatment-edge invariant.
\end{enumerate}

These counterexamples exhaust the four
primitive relations outside
$\mathcal Q_{\mathrm{inv}}$.
Therefore, $\mathcal Q_{\mathrm{inv}}$
contains exactly the treatment-edge-invariant
relations, proving (a).

\medskip
\noindent\textit{Nonredundancy.}
Because $\mathcal Q_{\mathrm{inv}}$ contains
two relations, it has three proper subsets.

First, let
$\mathcal A=
\{\bm S_{B,i}\not\!\ci T_i\}$.
The faithful treatment-mediated graph
\[
\bm S_{B,i}\rightarrow T_i\rightarrow Y_i
\]
satisfies the retained relation.
However, conditioning on $T_i$ blocks
its only path from $\bm S_{B,i}$ to $Y_i$,
so $\bm S_{B,i}\ci Y_i\mid T_i$.
The graph therefore satisfies every relation
in $\mathcal A$ but does not satisfy every
relation in $\mathcal Q_{\mathrm{inv}}$.

Second, let
$\mathcal A=
\{\bm S_{B,i}\not\!\ci Y_i\mid T_i\}$.
The faithful outcome-only graph
$\bm S_{B,i}\rightarrow Y_i$,
with $T_i$ isolated, satisfies the
retained relation but has
$\bm S_{B,i}\ci T_i$.
It therefore satisfies every relation
in $\mathcal A$ but not every relation
in $\mathcal Q_{\mathrm{inv}}$.

Finally, let $\mathcal A=\varnothing$.
The edgeless graph satisfies the empty
conjunction but satisfies neither relation
in $\mathcal Q_{\mathrm{inv}}$.

Thus, for every proper subset
$\mathcal A\subsetneq\mathcal Q_{\mathrm{inv}}$,
there exists a faithful ordered graph in
\[
\mathfrak G(\mathcal A)
\setminus
\mathfrak G(\mathcal Q_{\mathrm{inv}}).
\]
Consequently,
\[
\mathfrak G(\mathcal Q_{\mathrm{inv}})
\subsetneq
\mathfrak G(\mathcal A),
\]
which proves nonredundancy.
\end{proof}

\subsection{Proposition \ref{prop:population-compact-span}}

\begin{repproposition}{prop:population-compact-span}
Suppose $K_\eta$ and $G_\eta$ are positive definite and $\bm{v}_T$ and $\bm{v}_R$ are
linearly independent. Under
\eqref{eq:linear-treatment-model}--\eqref{eq:linear-proxy-model}, with
\eqref{eq:inverse-regression-linearity} imposed when $m_T$ is nonlinear,
\begin{equation*}
\mathcal W_\eta
=
\operatorname{span}
\left\{
K_\eta^{-1}\Lambda\Sigma_{\bm C}\bm a,\,
K_\eta^{-1}\Lambda\Sigma_{\bm C}(\bm b+\bm h_q)
\right\}.
\end{equation*}
Consequently, the population generalized eigenscores satisfy
\begin{equation*}
\begin{split}
\left\{
\bm X_i^\top \bm w:\bm w\in\mathcal W_\eta
\right\}
=
\operatorname{span}\bigl\{&
\bm a^\top\Sigma_{\bm C}\Lambda^\top
K_\eta^{-1}\bm X_i,\\
&
(\bm b+\bm h_q)^\top\Sigma_{\bm C}\Lambda^\top
K_\eta^{-1}\bm X_i
\bigr\}.
\end{split}
\end{equation*}
\end{repproposition}

\begin{proof}
Define
\begin{equation}
\label{eq:kappa-treatment}
\kappa_T
=
\frac{
\mathbb E\{U_{T,i}m_T(U_{T,i})\}
}{
\bm a^\top\Sigma_{\bm C}\bm a
}.
\end{equation}
We first express the treatment cross-moment in terms of the latent treatment
direction. In the additive-linear continuous-treatment case,
\[
T_i
=
U_{T,i}+\varepsilon_{T,i}.
\]
Since $U_{T,i}=\bm a^\top\bm C_i$ and
$\mathbb E(\bm C_i\varepsilon_{T,i})=\bm 0$, we have
\[
\mathbb E(\bm C_iT_i)
=
\Sigma_{\bm C}\bm a.
\]
Moreover, $m_T(u)=u$ in this case, so
\eqref{eq:kappa-treatment} gives $\kappa_T=1$.

When $m_T$ is nonlinear, iterated expectation gives
\[
\mathbb E(\bm C_iT_i)
=
\mathbb E\{\bm C_i m_T(U_{T,i})\}.
\]
Conditioning the right-hand side on $U_{T,i}$ gives
\[
\mathbb E(\bm C_iT_i)
=
\mathbb E\!\left[
\mathbb E(\bm C_i\mid U_{T,i})m_T(U_{T,i})
\right].
\]
Substituting \eqref{eq:inverse-regression-linearity} yields
\[
\mathbb E(\bm C_iT_i)
=
\frac{\Sigma_{\bm C}\bm a}
{\bm a^\top\Sigma_{\bm C}\bm a}
\mathbb E\{U_{T,i}m_T(U_{T,i})\}.
\]
Since
\[
\mathbb E(U_{T,i}^2)
=
\bm a^\top\Sigma_{\bm C}\bm a,
\]
the definition of $\kappa_T$ gives
\begin{equation}
\label{eq:latent-treatment-cross-moment}
\mathbb E(\bm C_iT_i)
=
\kappa_T\Sigma_{\bm C}\bm a.
\end{equation}

By consistency, the observed outcome is
\[
Y_i
=
\alpha+\bm b^\top\bm C_i
+
T_i\{\gamma+\bm q^\top\bm C_i\}
+
\varepsilon_{Y,i}(T_i).
\]
The conditional mean restriction in
\eqref{eq:linear-potential-outcome-model} implies
\[
\mathbb E\{\bm C_i\varepsilon_{Y,i}(T_i)\}
=
0.
\]
Multiplying the observed-outcome equation by $\bm C_i$ and taking
expectations therefore gives
\[
\mathbb E(\bm C_iY_i)
=
\Sigma_{\bm C}\bm b
+
\gamma\mathbb E(\bm C_iT_i)
+
\mathbb E\{\bm C_iT_i(\bm q^\top\bm C_i)\}.
\]
Using \eqref{eq:latent-treatment-cross-moment} and
\eqref{eq:hq-definition-supp}, we obtain
\[
\mathbb E(\bm C_iY_i)
=
\Sigma_{\bm C}
\left\{
\bm b+\bm h_q+\gamma\kappa_T\bm a
\right\}.
\]

Because $\bm C_i$ is centered, the definition of $R_i$ gives
\[
\begin{aligned}
\mathbb E(\bm C_iR_i)
&=
\operatorname{cov}(\bm C_i,Y_i)
-
\delta_T\operatorname{cov}(\bm C_i,T_i)\\
&=
\mathbb E(\bm C_iY_i)
-
\delta_T\mathbb E(\bm C_iT_i).
\end{aligned}
\]
Substituting the preceding identities gives
\begin{equation}
\label{eq:latent-residual-cross-moment}
\mathbb E(\bm C_iR_i)
=
\Sigma_{\bm C}
\left\{
\bm b+\bm h_q
+
(\gamma-\delta_T)\kappa_T\bm a
\right\}.
\end{equation}

By consistency, $R_i$ is a measurable function of
$\{T_i,Y_i(t):t\in\mathcal T\}$. The conditional proxy-independence
condition in \eqref{eq:linear-proxy-model} therefore implies
\[
\operatorname{cov}(\bm X_i,T_i\mid\bm C_i)=\bm 0
\qquad\text{and}\qquad
\operatorname{cov}(\bm X_i,R_i\mid\bm C_i)=\bm 0.
\]
The law of total covariance applied twice and
$\mathbb E(\bm X_i\mid\bm C_i)=\Lambda\bm C_i$ consequently give
\[
\begin{aligned}
\bm{v}_T
&=
\operatorname{cov}(\bm X_i,T_i)=
\operatorname{cov}
\left\{
\mathbb E(\bm X_i\mid\bm C_i),
\mathbb E(T_i\mid\bm C_i)
\right\}\\
&=
\Lambda\operatorname{cov}(\bm C_i,T_i) =
\Lambda\mathbb E(\bm C_iT_i).
\end{aligned}
\]
Applying \eqref{eq:latent-treatment-cross-moment} gives
\begin{equation}
\label{eq:observed-treatment-cross-moment}
\bm{v}_T
=
\kappa_T\Lambda\Sigma_{\bm C}\bm a.
\end{equation}

Similarly,
\[
\begin{aligned}
\bm{v}_R
&=
\operatorname{cov}(\bm X_i,R_i)=
\operatorname{cov}
\left\{
\mathbb E(\bm X_i\mid\bm C_i),
\mathbb E(R_i\mid\bm C_i)
\right\}\\
&=
\Lambda\operatorname{cov}(\bm C_i,R_i)=
\Lambda\mathbb E(\bm C_iR_i),
\end{aligned}
\]
where the final equality uses the centering of $\bm C_i$ and $R_i$.
Applying \eqref{eq:latent-residual-cross-moment} gives
\begin{equation}
\label{eq:observed-residual-cross-moment}
\bm{v}_R
=
\Lambda\Sigma_{\bm C}
\left\{
\bm b+\bm h_q
+
(\gamma-\delta_T)\kappa_T\bm a
\right\}.
\end{equation}

We next characterize the generalized eigenspace. Because $\bm{v}_T$ and $\bm{v}_R$ are
linearly independent,
\[
\operatorname{col}(H)
=
\operatorname{span}\{\bm{v}_T,\bm{v}_R\}.
\]
Let ${\bm{\beta}}$ be a generalized eigenvector with nonzero eigenvalue $\rho$.
Then
\(
H{\bm{\beta}}=\rho G_\eta{\bm{\beta}}.
\)
Because $G_\eta$ is nonsingular and $\rho\neq0$,
\[
{\bm{\beta}}
=
\rho^{-1}G_\eta^{-1}H{\bm{\beta}}.
\]
Because
\(
\operatorname{col}(H)=\operatorname{span}\{\bm{v}_T,\bm{v}_R\}
\), there exist scalars $c_T$ and $c_R$ such that
\(
H{\bm{\beta}}=c_T\bm{v}_T+c_R\bm{v}_R.
\)
Substitution gives
\[
{\bm{\beta}}
=
\frac{c_T}{\rho}G_\eta^{-1}\bm{v}_T
+
\frac{c_R}{\rho}G_\eta^{-1}\bm{v}_R.
\]
Therefore,
\[
{\bm{\beta}}
\in
\operatorname{span}
\left\{
G_\eta^{-1}\bm{v}_T,\,
G_\eta^{-1}\bm{v}_R
\right\}.
\]

The matrix $G_\eta$ is a rank-one modification of $K_\eta$. The
Sherman--Morrison identity gives
\[
G_\eta^{-1}
=
K_\eta^{-1}
+
\frac{
K_\eta^{-1}\bm{v}_T\bm{v}_T^\top K_\eta^{-1}
}{
2\operatorname{var}(T_i)-\bm{v}_T^\top K_\eta^{-1}\bm{v}_T
}.
\]
Therefore, $G_\eta^{-1}\bm{v}_T$ is a nonzero scalar multiple of
$K_\eta^{-1}\bm{v}_T$, while $G_\eta^{-1}\bm{v}_R$ equals
$K_\eta^{-1}\bm{v}_R$ plus a scalar multiple of $K_\eta^{-1}\bm{v}_T$. Hence
\[
\operatorname{span}
\left\{
G_\eta^{-1}\bm{v}_T,\,
G_\eta^{-1}\bm{v}_R
\right\}
=
\operatorname{span}
\left\{
K_\eta^{-1}\bm{v}_T,\,
K_\eta^{-1}\bm{v}_R
\right\}.
\]
Because \({\bm{\beta}}
\in
\operatorname{span}
\left\{
G_\eta^{-1}\bm{v}_T,\,
G_\eta^{-1}\bm{v}_R
\right\}\), $H$ has rank two and $G_\eta$ is nonsingular, the two nonzero
generalized eigenvectors span this space. Thus
\begin{equation}
\label{eq:cross-moment-eigenspace}
\mathcal W_\eta
=
\operatorname{span}
\left\{
K_\eta^{-1}\bm{v}_T,\,
K_\eta^{-1}\bm{v}_R
\right\}.
\end{equation}

Substituting \eqref{eq:observed-treatment-cross-moment} and
\eqref{eq:observed-residual-cross-moment} into
\eqref{eq:cross-moment-eigenspace} gives
\[
\mathcal W_\eta
=
\operatorname{span}
\left\{
\kappa_TK_\eta^{-1}\Lambda\Sigma_{\bm C}\bm a,\,
K_\eta^{-1}\Lambda\Sigma_{\bm C}
\left[
\bm b+\bm h_q
+
(\gamma-\delta_T)\kappa_T\bm a
\right]
\right\}.
\]
The assumption $\kappa_T\neq0$ permits rescaling the first vector without
changing its span. The second vector is the sum of
\[
K_\eta^{-1}\Lambda\Sigma_{\bm C}(\bm b+\bm h_q)
\]
and $(\gamma-\delta_T)\kappa_T$ times the rescaled first vector. Recall that $\bm{v}_T$ and $\bm{v}_R$ are linearly independent, so $\bm b+\bm h_q \not = c \bm a$ for some constant $c$ by \eqref{eq:observed-treatment-cross-moment} and \eqref{eq:observed-residual-cross-moment}. We therefore obtain
\[
\mathcal W_\eta
=
\operatorname{span}
\left\{
K_\eta^{-1}\Lambda\Sigma_{\bm C}\bm a,\,
K_\eta^{-1}\Lambda\Sigma_{\bm C}(\bm b+\bm h_q)
\right\},
\]
which proves \eqref{eq:population-compact-span}. Multiplying these coefficient
directions by $\bm X_i^\top$ gives
\eqref{eq:population-score-latent-span}.
\end{proof}

\subsection{Lemma \ref{lem:spectral-subspace-error}}

\begin{replemma}{lem:spectral-subspace-error}
Consider Assumption~\ref{assump:regularity}. Then
\[
\operatorname{dist}^2
\left(
\widehat{\mathcal S}_\eta,
\mathcal S_\eta
\right)
\in
O_p
\left(
\frac{d_{\mathrm{eff}}(\eta)}{n\Delta^2}
\right).
\]
\end{replemma}

\begin{proof}
Let \(\Pi\) and \(\widehat\Pi\) denote the orthogonal projections onto the
two-dimensional nonzero eigenspaces of \(\mathcal L\) and
\(\widehat{\mathcal L}\), respectively.
For brevity, set
\[
a_{n,\eta}
=
\left\{
\frac{d_{\mathrm{eff}}(\eta)}{n}
\right\}^{1/2}.
\]

Assumption~\ref{assump:regularity}(c) separates the two-dimensional eigenspace
from the zero eigenspace by at least $\Delta$. Assumption
\ref{assump:regularity}(a), together with $a_{n,\eta}/\Delta\to0$, implies
\[
\|\widehat{\mathcal L}-\mathcal L\|_{\mathrm{op}}<\Delta/2
\]
with probability tending to one. Lemma~\ref{lem:Davis_Kahan} therefore gives
\begin{equation}
\label{eq:rank-two-DK}
\|\widehat\Pi-\Pi\|_{\mathrm{op}}
\leq
\frac{2\|\widehat{\mathcal L}-\mathcal L\|_{\mathrm{op}}}{\Delta}.
\end{equation}
Using Assumption~\ref{assump:regularity}(a) in
\eqref{eq:rank-two-DK} yields
\begin{equation}
\label{eq:rank-two-projection-rate}
\|\widehat\Pi-\Pi\|_{\mathrm{op}}
\in
O_p\left(\frac{a_{n,\eta}}{\Delta}\right).
\end{equation}

By Assumption~\ref{assump:regularity}(b), the eigenvalues of $G_\eta$ and
$\widehat G_\eta$ are bounded above and away from zero with probability
tending to one. Hence, for fixed constants $0<c<C<\infty$, both matrices
belong with probability tending to one to the set
\[
\mathcal K
=
\left\{
A=A^\top:
cI\preceq A\preceq CI
\right\}.
\]
The set $\mathcal K$ is compact and convex. The finite-dimensional matrix map
$A\mapsto A^{-1/2}$ is continuously differentiable on the positive-definite
cone. Its derivative is therefore bounded on $\mathcal K$. The mean-value
inequality then gives a constant $K_G<\infty$ such that
\[
\|A^{-1/2}-B^{-1/2}\|_{\mathrm{op}}
\leq
K_G\|A-B\|_{\mathrm{op}}
\]
for every $A,B\in\mathcal K$. Applying this inequality with
$A=\widehat G_\eta$ and $B=G_\eta$, and then using the concentration rate in
Assumption~\ref{assump:regularity}(b), gives
\begin{equation}
\label{eq:rank-two-G-root-rate}
\|\widehat G_\eta^{-1/2}-G_\eta^{-1/2}\|_{\mathrm{op}}
\in O_p(a_{n,\eta}).
\end{equation}

Let $Q\in\mathbb R^{p\times2}$ have orthonormal columns spanning
$\operatorname{ran}(\Pi)$. Then
\(
\Pi Q=Q
\)
and
\(
Q^\top Q=I_2
\).
Define
\(
B_\eta=G_\eta^{-1/2}Q
\)
and
\(
\widehat B_{\eta,0}=\widehat G_\eta^{-1/2}\widehat\Pi Q
\).
We first verify that $\widehat B_{\eta,0}$ spans the empirical generalized-eigenvector
coefficient space. Because $\Pi Q=Q$,
\[
Q^\top\widehat\Pi Q-I_2
=
Q^\top(\widehat\Pi-\Pi)Q.
\]
Therefore,
\[
\|Q^\top\widehat\Pi Q-I_2\|_{\mathrm{op}}
\leq
\|Q\|_{\mathrm{op}}^2
\|\widehat\Pi-\Pi\|_{\mathrm{op}}
=
\|\widehat\Pi-\Pi\|_{\mathrm{op}}.
\]
Equation~\eqref{eq:rank-two-projection-rate} and $a_{n,\eta}/\Delta\to0$ in Assumption~\ref{assump:regularity}(c) imply that
the final quantity is smaller than one with probability tending to one. On
this event, $Q^\top\widehat\Pi Q$ is nonsingular. If
\(\widehat\Pi Q\bm a=\bm 0\) for some $\bm a\in\mathbb R^2$, then premultiplication by
$Q^\top$ gives
\(
Q^\top\widehat\Pi Q\bm a=\bm 0
\).
The nonsingularity of $Q^\top\widehat\Pi Q$ then gives $\bm a=\bm 0$. Thus,
$\widehat\Pi Q$ has rank two with probability tending to one. Its columns lie
in $\operatorname{ran}(\widehat\Pi)$, which also has dimension two. Hence,
\[
\operatorname{col}(\widehat\Pi Q)
=
\operatorname{ran}(\widehat\Pi)
\]
with probability tending to one. Since $\widehat G_\eta^{-1/2}$ is
nonsingular, $\widehat B_{\eta,0}$ spans the empirical generalized-eigenvector
coefficient space.

Using $\Pi Q=Q$, we have
\[
\begin{aligned}
\widehat B_{\eta,0}-B_\eta
&=
\widehat G_\eta^{-1/2}\widehat\Pi Q
-
G_\eta^{-1/2}Q \\
&=
\widehat G_\eta^{-1/2}\widehat\Pi Q
-
\widehat G_\eta^{-1/2}\Pi Q
+
\widehat G_\eta^{-1/2}Q
-
G_\eta^{-1/2}Q \\
&=
\widehat G_\eta^{-1/2}
(\widehat\Pi-\Pi)Q
+
(\widehat G_\eta^{-1/2}-G_\eta^{-1/2})Q.
\end{aligned}
\]
Consequently,
\[
\begin{aligned}
\|\widehat B_{\eta,0}-B_\eta\|_{\mathrm{op}}
&\leq
\|\widehat G_\eta^{-1/2}\|_{\mathrm{op}}
\|\widehat\Pi-\Pi\|_{\mathrm{op}}
\|Q\|_{\mathrm{op}} \\
&\quad+
\|\widehat G_\eta^{-1/2}-G_\eta^{-1/2}\|_{\mathrm{op}}
\|Q\|_{\mathrm{op}},
\end{aligned}
\]
where $\|Q\|_{\mathrm{op}}=1$. The matrix bounds of Assumption~\ref{assump:regularity}(b) give
\(\|\widehat G_\eta^{-1/2}\|_{\mathrm{op}}=O_p(1)\). Combining
\eqref{eq:rank-two-projection-rate} and
\eqref{eq:rank-two-G-root-rate}, together with $\Delta\in O(1)$, yields
\begin{equation}
\label{eq:rank-two-B-rate}
\|\widehat B_{\eta,0}-B_\eta\|_{\mathrm{op}}
\in
O_p\left(\frac{a_{n,\eta}}{\Delta}\right).
\end{equation}

We now apply Lemma~\ref{lem:column-space-perturbation}. Set
\[
M=\Sigma_{\bm X}^{1/2}B_\eta,
\qquad
\widehat M=\Sigma_{\bm X}^{1/2}\widehat B_{\eta,0},
\qquad
E=\widehat M-M
=
\Sigma_{\bm X}^{1/2}(\widehat B_{\eta,0}-B_\eta).
\]
Assumption~\ref{assump:regularity}(e) and
\eqref{eq:rank-two-B-rate} give
\[
\begin{aligned}
\|E\|_{\mathrm{op}}
&\leq
\|\Sigma_{\bm X}^{1/2}\|_{\mathrm{op}}
\|\widehat B_{\eta,0}-B_\eta\|_{\mathrm{op}}
\\
&\in
O_p\left(\frac{a_{n,\eta}}{\Delta}\right).
\end{aligned}
\]
Assumption~\ref{assump:regularity}(d) also gives
\[
\begin{aligned}
\sigma_2(M)^2
&=
\lambda_{\min}(M^\top M)
\\
&=
\lambda_{\min}
\left(
{B_\eta}^\top\Sigma_{\bm X}B_\eta
\right)
\geq
c_S.
\end{aligned}
\]
Thus, $M$ has rank two and $\sigma_2(M)\geq\sqrt{c_S}$. Since
$a_{n,\eta}/\Delta\to0$, the preceding bound gives
\(\|E\|_{\mathrm{op}}=o_p(1)\). Therefore,
\(
\|E\|_{\mathrm{op}}\leq\sigma_2(M)/2
\)
with probability tending to one. This verifies the small-perturbation
condition in Lemma~\ref{lem:column-space-perturbation}. The lemma therefore
gives
\[
\left\|
P_{\widehat M}
-
P_M
\right\|_{\mathrm{op}}
\leq
\frac{2\|E\|_{\mathrm{op}}}{\sigma_2(M)}
\in
O_p\left(\frac{a_{n,\eta}}{\Delta}\right),
\]
where the last rate uses \(\sigma_2(M)\geq\sqrt{c_S}\).

The columns of $B_\eta$ span the population generalized-eigenvector
coefficient space, and the columns of $\widehat B_{\eta,0}$ span its empirical
counterpart with probability tending to one. Lemma
\ref{lem:score-space-distance-equivalence} therefore gives
\[
\operatorname{dist}
\left(
\widehat{\mathcal S}_\eta,
\mathcal S_\eta
\right)
=
\left\|
P_{\widehat M}
-
P_M
\right\|_{\mathrm{op}}
\in
O_p\left(\frac{a_{n,\eta}}{\Delta}\right).
\]
Squaring and using
\(a_{n,\eta}^2=d_{\mathrm{eff}}(\eta)/n\) gives
\[
\operatorname{dist}^2
\left(
\widehat{\mathcal S}_\eta,
\mathcal S_\eta
\right)
\in
O_p\left(
\frac{d_{\mathrm{eff}}(\eta)}
{n\Delta^2}
\right),
\]
as claimed.
\end{proof}

\begin{lemma} [Theorem 2 in \cite{Yu15}] \label{lem:Davis_Kahan}
Let $A$ and $\widehat A$ be symmetric matrices. Let $P$ project onto
an invariant subspace of $A$, and suppose the eigenvalues in that subspace
are separated from the remaining spectrum by
\[
\Delta
=
\min_{\rho\in\operatorname{spec}(A|_{\operatorname{ran}(P)}),
\;\rho'\in\operatorname{spec}(A|_{\operatorname{ran}(I-P)})}
|\rho-\rho'|
>
0 .
\]
Let $\widehat P$ project onto the corresponding empirical invariant subspace.
Whenever $\|\widehat A-A\|_{\mathrm{op}}<\Delta/2$,
\[
\|\widehat P-P\|_{\mathrm{op}}
\leq
\frac{2\|\widehat A-A\|_{\mathrm{op}}}{\Delta}.
\]
\end{lemma}

\begin{lemma}
\label{lem:column-space-perturbation}
Let $M,\widehat M\in\mathbb R^{d\times q}$ for a positive integer $q$, and let
\(
E=\widehat M-M.
\)
Suppose that, for some $\sigma_0>0$, we have
\(
\sigma_{\min}(M)\geq \sigma_0,\) and \(
\|E\|_{\mathrm{op}}\leq \frac{\sigma_0}{2}.
\)
Let
\(
\mathcal S=\operatorname{col}(M), \)
and
\(\widehat{\mathcal S}_\eta=\operatorname{col}(\widehat M).
\)
Then $\widehat M$ has full column rank, and
\[
\|P_{\widehat{\mathcal S}_\eta}-P_{\mathcal S}\|_{\mathrm{op}}
\leq
\frac{2\|E\|_{\mathrm{op}}}{\sigma_0}.
\]
\end{lemma}

\begin{proof}
We prove the two conclusions separately. First, we show that
\(\mathcal S\) is close to \(\widehat{\mathcal S}_\eta\) in one direction.

Let \(\bm f\in\mathcal S\) with \(\|\bm f\|_2=1\). Since
\(\mathcal S=\operatorname{col}(M)\), there exists \(\bm a\in\mathbb R^q\) such
that
\begin{equation}\label{eq:colspace-f-Ma}
\bm f=M\bm a .
\end{equation}
Because \(\sigma_{\min}(M)\geq\sigma_0\), \eqref{eq:colspace-f-Ma} gives
\begin{equation*}
1
=
\|\bm f\|_2
=
\|M\bm a\|_2
\geq
\sigma_0\|\bm a\|_2 .
\end{equation*}
Therefore
\begin{equation}\label{eq:colspace-a-rate}
\|\bm a\|_2\leq \sigma_0^{-1}.
\end{equation}
Since \(\widehat M \bm a\in\widehat{\mathcal S}_\eta\), the distance from \(\bm f\) to
\(\widehat{\mathcal S}_\eta\) is no larger than its distance to \(\widehat M \bm a\). Thus
\begin{equation}\label{eq:one-sided-colspace-bound}
\begin{aligned}
\|(I-P_{\widehat{\mathcal S}_\eta})\bm f\|_2
&\leq
\|\bm f-\widehat M \bm a\|_2
\\
&=
\|M\bm a-\widehat M \bm a\|_2
\\
&=
\|(M-\widehat M)\bm a\|_2
\\
&=
\|E\bm a\|_2
\\
&\leq
\|E\|_{\mathrm{op}}\|\bm a\|_2
\\
&\leq
\frac{\|E\|_{\mathrm{op}}}{\sigma_0},
\end{aligned}
\end{equation}
where the last line uses \eqref{eq:colspace-a-rate}. Taking the supremum over
all \(\bm f\in\mathcal S\) with \(\|\bm f\|_2=1\) gives
\begin{equation}\label{eq:one-sided-S-to-Shat}
\operatorname{dist}_{\rightarrow}(\mathcal S,\widehat{\mathcal S}_\eta)
\leq
\frac{\|E\|_{\mathrm{op}}}{\sigma_0}.
\end{equation}

Second, we show that \(\widehat M\) has full column rank and obtain the reverse
one-sided bound. For any \(\bm a\in\mathbb R^q\),
\begin{equation}\label{eq:Mhat-lower-bound}
\begin{aligned}
\|\widehat M \bm a\|_2
&=
\|M\bm a+E\bm a\|_2
\\
&\geq
\|M\bm a\|_2-\|E\bm a\|_2
\\
&\geq
\sigma_0\|\bm a\|_2-\|E\|_{\mathrm{op}}\|\bm a\|_2
\\
&=
(\sigma_0-\|E\|_{\mathrm{op}})\|\bm a\|_2
\\
&\geq
\frac{\sigma_0}{2}\|\bm a\|_2 .
\end{aligned}
\end{equation}
The last inequality uses the assumption
\(\|E\|_{\mathrm{op}}\leq\sigma_0/2\). Equation
\eqref{eq:Mhat-lower-bound} implies that \(\widehat M \bm a=\bm 0\) only if \(\bm a=\bm 0\).
Hence \(\widehat M\) has full column rank and
\begin{equation}\label{eq:Mhat-sigma-lower}
\sigma_{\min}(\widehat M)\geq \sigma_0/2 .
\end{equation}

Now let \(\widehat{\bm f}\in\widehat{\mathcal S}_\eta\) with
\(\|\widehat{\bm f}\|_2=1\). Since
\(\widehat{\mathcal S}_\eta=\operatorname{col}(\widehat M)\), there exists
\(\bm a\in\mathbb R^q\) such that
\begin{equation*}
\widehat{\bm f}=\widehat M \bm a .
\end{equation*}
Using \eqref{eq:Mhat-sigma-lower},
\begin{equation*}
1
=
\|\widehat{\bm f}\|_2
=
\|\widehat M \bm a\|_2
\geq
\frac{\sigma_0}{2}\|\bm a\|_2 ,
\end{equation*}
and hence
\begin{equation*}
\|\bm a\|_2\leq \frac{2}{\sigma_0}.
\end{equation*}
Since \(M\bm a\in\mathcal S\), the same argument as in
\eqref{eq:one-sided-colspace-bound} gives
\begin{equation*}
\begin{aligned}
\|(I-P_{\mathcal S})\widehat{\bm f}\|_2
&\leq
\|\widehat{\bm f}-M\bm a\|_2
\\
&=
\|\widehat M \bm a-M\bm a\|_2
\\
&=
\|E\bm a\|_2
\\
&\leq
\|E\|_{\mathrm{op}}\|\bm a\|_2
\\
&\leq
\frac{2\|E\|_{\mathrm{op}}}{\sigma_0}.
\end{aligned}
\end{equation*}
Taking the supremum over all unit \(\widehat{\bm f}\in\widehat{\mathcal S}_\eta\) gives
\begin{equation}\label{eq:one-sided-Shat-to-S-final}
\operatorname{dist}_{\rightarrow}(\widehat{\mathcal S}_\eta,\mathcal S)
\leq
\frac{2\|E\|_{\mathrm{op}}}{\sigma_0}.
\end{equation}

For subspaces of the same finite dimension, the operator-norm distance between
orthogonal projections equals the larger of the two one-sided distances.
Combining \eqref{eq:one-sided-S-to-Shat} and
\eqref{eq:one-sided-Shat-to-S-final}, we obtain
\[
\|P_{\widehat{\mathcal S}_\eta}-P_{\mathcal S}\|_{\mathrm{op}}
\leq
\frac{2\|E\|_{\mathrm{op}}}{\sigma_0}.
\]
\end{proof}

\begin{lemma}
\label{lem:score-space-distance-equivalence}
Let \(\bm X\) be centered with covariance matrix \(\Sigma_{\bm X}\). For
\(B_1,B_2\in\mathbb R^{p\times2}\), define
\[
\mathcal S(B_j)
=
\left\{
(B_j{\bm{\xi}})^\top\bm X:{\bm{\xi}}\in\mathbb R^2
\right\},
\qquad j\in\{1,2\}.
\]
Let \(P_{\mathcal S(B_j)}\) denote the
\(L_2(P_{\bm X})\)-orthogonal projection onto \(\mathcal S(B_j)\), and let
\(P_{\Sigma_{\bm X}^{1/2}B_j}\) denote the Euclidean orthogonal projection
onto \(\operatorname{col}(\Sigma_{\bm X}^{1/2}B_j)\). Then
\[
\left\|
P_{\mathcal S(B_1)}
-
P_{\mathcal S(B_2)}
\right\|_{\mathrm{op}}
=
\left\|
P_{\Sigma_{\bm X}^{1/2}B_1}
-
P_{\Sigma_{\bm X}^{1/2}B_2}
\right\|_{\mathrm{op}}.
\]
\end{lemma}

\begin{proof}
Let
\(
\mathcal H_{\bm X}
=
\left\{
{\bm{\beta}}^\top\bm X:{\bm{\beta}}\in\mathbb R^p
\right\}
\)
be the space of linear scores, and let
\(
\mathcal C_{\bm X}
=
\operatorname{col}(\Sigma_{\bm X}^{1/2}).
\)
Define
\[
\Phi:\mathcal H_{\bm X}\longrightarrow\mathcal C_{\bm X},
\qquad
\Phi({\bm{\beta}}^\top\bm X)
=
\Sigma_{\bm X}^{1/2}{\bm{\beta}}.
\]
For any \({\bm{\beta}},{\bm{\gamma}}\in\mathbb R^p\),
\[
\begin{aligned}
\left\|
{\bm{\beta}}^\top\bm X-{\bm{\gamma}}^\top\bm X
\right\|_{L_2(P_{\bm X})}^2
&=
\mathbb E\!\left[
\{({\bm{\beta}}-{\bm{\gamma}})^\top\bm X\}^2
\right] \\
&=
({\bm{\beta}}-{\bm{\gamma}})^\top
\Sigma_{\bm X}
({\bm{\beta}}-{\bm{\gamma}}) \\
&=
\left\|
\Sigma_{\bm X}^{1/2}{\bm{\beta}}
-
\Sigma_{\bm X}^{1/2}{\bm{\gamma}}
\right\|_2^2.
\end{aligned}
\]
Thus, \(\Phi\) preserves distances. This calculation also shows that
\(\Phi\) is well defined: if two coefficient vectors produce the same
score in \(L_2(P_{\bm X})\), then they produce the same vector under
\(\Phi\). The map is one-to-one on \(L_2(P_{\bm X})\) equivalence classes
and maps onto \(\mathcal C_{\bm X}\). It is therefore an isometry.

For any \(B\in\mathbb R^{p\times2}\),
\[
\begin{aligned}
\Phi\{\mathcal S(B)\}
&=
\left\{
\Phi\bigl((B{\bm{\xi}})^\top\bm X\bigr):
{\bm{\xi}}\in\mathbb R^2
\right\} \\
&=
\left\{
\Sigma_{\bm X}^{1/2}B{\bm{\xi}}:
{\bm{\xi}}\in\mathbb R^2
\right\} \\
&=
\operatorname{col}(\Sigma_{\bm X}^{1/2}B).
\end{aligned}
\]

Fix \(f\in\mathcal H_{\bm X}\). The orthogonal projection
\(P_{\mathcal S(B)}f\) is the unique element of
\(\mathcal S(B)\) closest to \(f\):
\[
P_{\mathcal S(B)}f
=
\underset{s\in\mathcal S(B)}{\arg\min}\,
\|f-s\|_{L_2(P_{\bm X})}.
\]
Because \(\Phi\) maps \(\mathcal S(B)\) onto
\(\operatorname{col}(\Sigma_{\bm X}^{1/2}B)\) and preserves distances,
\[
\begin{aligned}
\Phi\left\{
P_{\mathcal S(B)}f
\right\}
&=
\Phi\left\{
\underset{s\in\mathcal S(B)}{\arg\min}\,
\|f-s\|_{L_2(P_{\bm X})}
\right\} \\
&=
\underset{s\in\mathcal S(B)}{\arg\min}\,
\|\Phi(f)-\Phi(s)\|_2 \\
&=
\underset{\bm v\in
\operatorname{col}(\Sigma_{\bm X}^{1/2}B)}
{\arg\min}\,
\|\Phi(f)-\bm v\|_2 \\
&=
P_{\Sigma_{\bm X}^{1/2}B}\Phi(f).
\end{aligned}
\]
Applying this identity to \(B_1\) and \(B_2\) gives
\[
\Phi\left\{
\left(
P_{\mathcal S(B_1)}
-
P_{\mathcal S(B_2)}
\right)f
\right\}
=
\left(
P_{\Sigma_{\bm X}^{1/2}B_1}
-
P_{\Sigma_{\bm X}^{1/2}B_2}
\right)\Phi(f).
\]
Because \(\Phi\) preserves norms,
\[
\begin{aligned}
&
\left\|
\left(
P_{\mathcal S(B_1)}
-
P_{\mathcal S(B_2)}
\right)f
\right\|_{L_2(P_{\bm X})} \\
&\qquad=
\left\|
\left(
P_{\Sigma_{\bm X}^{1/2}B_1}
-
P_{\Sigma_{\bm X}^{1/2}B_2}
\right)\Phi(f)
\right\|_2.
\end{aligned}
\]
Moreover, \(\Phi\) maps the unit-norm elements of
\(\mathcal H_{\bm X}\) onto the unit-norm elements of
\(\mathcal C_{\bm X}\). Taking the supremum over these elements therefore
gives equality of the two operator norms.

Finally, the projection difference on the left is zero on
\(\mathcal H_{\bm X}^{\perp}\), while the projection difference on the
right is zero on \(\mathcal C_{\bm X}^{\perp}\). Hence, restricting the
suprema to \(\mathcal H_{\bm X}\) and \(\mathcal C_{\bm X}\) does not
change the operator norms, which proves the result.
\end{proof}

\subsection{Lemma \ref{lem:ridge-span-decomposition}}
\begin{replemma}{lem:ridge-span-decomposition}
Conditional on the training sample, the prediction error of ridge
regression decomposes into the error attainable using the complete
COMPACT score space and an additional error from calibrating the
ridge predictor within that space:
\begin{equation*}
\begin{aligned}
\underbrace{
\mathbb E(W_i-\widehat W_\eta)^2
}_{\text{Ridge prediction error}}
\quad=
\underbrace{
\mathbb E(W_i-P_\eta W_i)^2
}_{\text{COMPACT span prediction error}}
\quad+
\underbrace{
\mathbb E(P_\eta g_W-\widehat W_\eta)^2.
}_{\text{Additional ridge calibration error}}
\end{aligned}
\end{equation*}
\end{replemma}

\begin{proof}
By the empirical counterpart of
Proposition~\ref{prop:population-compact-span}, the complete
COMPACT eigenspace satisfies
\[
\widehat{\mathcal W}_\eta
=
\operatorname{span}
\left(
\widehat K_\eta^{-1}\widehat{\bm v}_T,\,
\widehat K_\eta^{-1}\widehat{\bm v}_R
\right),
\]
where
$\widehat K_\eta=\widehat\Sigma_{\bm X}+\eta I_p$.
Consequently, for either $W=T$ or $W=R$,
the ridge predictor $\widehat W_\eta$ belongs to
$\widehat{\mathcal S}_\eta$. Therefore,
\begin{equation}
\label{eq:ridge-in-compact-space}
P_\eta\widehat W_\eta=\widehat W_\eta.
\end{equation}

We first decompose the prediction residual:
\[
W_i-\widehat W_\eta
=
(W_i-P_\eta W_i)
+
(P_\eta W_i-\widehat W_\eta).
\]
The first term is orthogonal to
$\widehat{\mathcal S}_\eta$ by the definition of $P_\eta$.
The second term belongs to $\widehat{\mathcal S}_\eta$
by \eqref{eq:ridge-in-compact-space}. The two terms are thus orthogonal.
The Pythagorean theorem then gives
\begin{equation}
\label{eq:ridge-span-proof-decomposition}
\begin{aligned}
\mathbb E(W_i-\widehat W_\eta)^2
={}&
\mathbb E(W_i-P_\eta W_i)^2 +
\mathbb E(P_\eta g_W-\widehat W_\eta)^2.
\end{aligned}
\end{equation}

Because $g_W$ is the best population linear predictor
of $W_i$ from $\bm X_i$, the residual $W_i-g_W$
is orthogonal to every linear function of $\bm X_i$.
In particular, it is orthogonal to
$\widehat{\mathcal S}_\eta$, so
\[
P_\eta(W_i-g_W)=0.
\]
By linearity of the projection,
\(
P_\eta W_i=P_\eta g_W.
\)
Substituting this identity into
\eqref{eq:ridge-span-proof-decomposition} gives
\[
\begin{aligned}
\mathbb E(W_i-\widehat W_\eta)^2
={}&
\mathbb E(W_i-P_\eta W_i)^2
+
\mathbb E(P_\eta g_W-\widehat W_\eta)^2,
\end{aligned}
\]
which proves \eqref{eq:ridge-span-error-decomposition}.
\end{proof}

\subsection{Theorem \ref{thm:joint-regularization-recovery}}

\begin{reptheorem}{thm:joint-regularization-recovery}
Suppose Assumption~\ref{assump:joint-regularization}
holds. Let $\widehat\eta_C$ maximize
\eqref{eq:CV} over the two candidate penalties,
and let $\widehat\eta_T,\widehat\eta_R$
minimize the corresponding individual ridge
prediction errors. Define
\[
\widehat{\mathcal S}_{\mathrm{sep}}
=
\operatorname{span}_{L_2}
(\widehat T_{\widehat\eta_T},
\widehat R_{\widehat\eta_R}).
\]
Then there exists a constant $c>0$,
independent of $p$ and $n$, such that
\begin{equation*}
\Pr\Bigg[
\operatorname{dist}^2
(\mathcal U,\widehat{\mathcal S}_{\mathrm{sep}})
-
\operatorname{dist}^2
(\mathcal U,\widehat{\mathcal S}_{\widehat\eta_C}) \geq
c\frac{r_{p,n}}
{(1+s_p^{-1}+r_{p,n})^2}
\Bigg]\longrightarrow1.
\end{equation*}
\end{reptheorem}

\begin{proof}
We first identify the unavoidable proxy-reconstruction
error and separate it from the additional error of
the learned score space. Throughout, fitted
coefficients are treated as fixed when evaluating
population covariances on a new observation.

By the definition of $\Lambda_U$ and
$\operatorname{var}(\bm U_i^\star)=I_2$,
\[
\begin{aligned}
\operatorname{cov}(\bm\xi_i,\bm U_i^\star)
=
\operatorname{cov}
(\bm X_i-\Lambda_U\bm U_i^\star,\bm U_i^\star)=
\Lambda_U-\Lambda_UI_2=0.
\end{aligned}
\]
Consequently,
$\operatorname{var}(\bm X_i)
=\Lambda_U\Lambda_U^\top+\Sigma_\xi$.

Assumption~\ref{assump:joint-regularization}(a)
ensures that $\Sigma_\xi$ and $\mathcal I_p$
are positive definite. By
\eqref{eq:joint-proxy-information},
\[
\begin{aligned}
&\operatorname{var}(\bm X_i)
\Sigma_\xi^{-1}\Lambda_U=
(\Lambda_U\Lambda_U^\top+\Sigma_\xi)
\Sigma_\xi^{-1}\Lambda_U=
\Lambda_U(I_2+\mathcal I_p).
\end{aligned}
\]
Premultiplying by
$\operatorname{var}(\bm X_i)^{-1}$
and postmultiplying by
$(I_2+\mathcal I_p)^{-1}$ gives
\[
\operatorname{var}(\bm X_i)^{-1}\Lambda_U
=
\Sigma_\xi^{-1}\Lambda_U
(I_2+\mathcal I_p)^{-1}.
\]

The residual covariance of the best linear
prediction of $\bm U_i^\star$ from all
observed proxies is therefore
\[
\begin{aligned}
&I_2-\Lambda_U^\top
\operatorname{var}(\bm X_i)^{-1}\Lambda_U\\
&=
I_2-\mathcal I_p(I_2+\mathcal I_p)^{-1}\\
&=(I_2+\mathcal I_p)(I_2+\mathcal I_p)^{-1}
-\mathcal I_p(I_2+\mathcal I_p)^{-1}\\
&=
(I_2+\mathcal I_p-\mathcal I_p)
(I_2+\mathcal I_p)^{-1}\\
&=
(I_2+\mathcal I_p)^{-1}.
\end{aligned}
\]
By the spectral mapping theorem, the
worst-direction reconstruction error is
\begin{equation}
\label{eq:joint-oracle-reconstruction-error}
\lambda_{\max}
\{(I_2+\mathcal I_p)^{-1}\}
=
\frac{1}{1+s_p}.
\end{equation}

We next decompose the relative noise covariance
of the learned scores. By the definitions of
$L_\eta$ and $A_\eta$,
\[
\begin{aligned}
L_\eta^\top\Lambda_U
&=
A_\eta^{-1}\widehat B_\eta^\top\Lambda_U=
A_\eta^{-1}A_\eta=I_2.
\end{aligned}
\]
Similarly, the definition of $L_p^*$ gives
\[
\begin{aligned}
(L_p^*)^\top\Lambda_U
=
\mathcal I_p^{-1}
\Lambda_U^\top\Sigma_\xi^{-1}\Lambda_U=
\mathcal I_p^{-1}\mathcal I_p=I_2.
\end{aligned}
\]
Subtracting these identities yields
\begin{equation}
\label{eq:joint-oracle-orthogonality}
(L_\eta-L_p^*)^\top\Lambda_U=0.
\end{equation}

Since
$\Sigma_\xi L_p^*
=\Lambda_U\mathcal I_p^{-1}$,
equation~\eqref{eq:joint-oracle-orthogonality}
implies
\[
\begin{aligned}
&(L_\eta-L_p^*)^\top\Sigma_\xi L_p^* =
(L_\eta-L_p^*)^\top
\Lambda_U\mathcal I_p^{-1}=0.
\end{aligned}
\]
Thus, the deviation from the oracle coefficients
is orthogonal to the oracle coefficients in
the $\Sigma_\xi$ quadratic form.

The oracle noise covariance satisfies
\[
\begin{aligned}
(L_p^*)^\top\Sigma_\xi L_p^*
=
\mathcal I_p^{-1}
\Lambda_U^\top\Sigma_\xi^{-1}
\Lambda_U\mathcal I_p^{-1}=
\mathcal I_p^{-1}.
\end{aligned}
\]

Let $D_\eta=L_\eta-L_p^*$.
Since $L_\eta=L_p^*+D_\eta$, expanding
$N_\eta=L_\eta^\top\Sigma_\xi L_\eta$
gives
\[
\begin{aligned}
N_\eta
={}
(L_p^*)^\top\Sigma_\xi L_p^*
+
D_\eta^\top\Sigma_\xi L_p^*
+
(L_p^*)^\top\Sigma_\xi D_\eta
+
D_\eta^\top\Sigma_\xi D_\eta.
\end{aligned}
\]
The two cross terms vanish by the
orthogonality established above.
Moreover, the definition
$F_\eta=\Sigma_\xi^{1/2}D_\eta$ gives
$D_\eta^\top\Sigma_\xi D_\eta
=F_\eta^\top F_\eta$.
Consequently,
\begin{equation}
\label{eq:joint-oracle-noise-decomposition}
N_\eta
=
\mathcal I_p^{-1}+F_\eta^\top F_\eta.
\end{equation}
The first term is unavoidable given the observed
proxy system, while the second measures excess
noise from the learned extraction coefficients.

We now show that stronger regularization
reduces this excess error. Write $r=r_{p,n}$
for brevity. Assumption
\ref{assump:joint-regularization}(b) gives
\[
c_ErI_2
\preceq F_{\eta_-}^\top F_{\eta_-}
\preceq C_ErI_2.
\]
It also gives
$\|H_{p,n}\|_{\mathrm{op}}\leq q_0<1$
and
$F_{\eta_+}=H_{p,n}F_{\eta_-}+R_{p,n}$,
where
$\|R_{p,n}\|_{\mathrm{op}}=o_p(r^{1/2})$.

For any $\bm v\in\mathbb R^2$, the lower
matrix bound implies
\[
\begin{aligned}
\|F_{\eta_-}\bm v\|^2
&=
\bm v^\top F_{\eta_-}^\top
F_{\eta_-}\bm v \geq
c_Er\|\bm v\|^2.
\end{aligned}
\]
Taking square roots gives
$\|F_{\eta_-}\bm v\|
\geq(c_Er)^{1/2}\|\bm v\|$.

Using the triangle inequality and the
operator-norm bound on $H_{p,n}$,
\[
\begin{aligned}
\|F_{\eta_+}\bm v\|
&=
\|H_{p,n}F_{\eta_-}\bm v
+R_{p,n}\bm v\|\\
&\leq
q_0\|F_{\eta_-}\bm v\|
+\|R_{p,n}\|_{\mathrm{op}}\|\bm v\|\\
&\leq
q_0\|F_{\eta_-}\bm v\|
+
\frac{\|R_{p,n}\|_{\mathrm{op}}}
{(c_Er)^{1/2}}
\|F_{\eta_-}\bm v\|\\
&\leq
\left(
q_0+
\frac{\|R_{p,n}\|_{\mathrm{op}}}
{(c_Er)^{1/2}}
\right)
\|F_{\eta_-}\bm v\|,
\end{aligned}
\]
Because
$\|R_{p,n}\|_{\mathrm{op}}=o_p(r^{1/2})$,
the second term in parentheses is $o_p(1)$.

Choose the fixed constant
$q=(1+q_0)/2$, so that $q_0<q<1$.
With probability tending to one, the
factor in parentheses is at most $q$.
On this event, simultaneously for every
$\bm v\in\mathbb R^2$,
\[
\|F_{\eta_+}\bm v\|^2
\leq
q^2\|F_{\eta_-}\bm v\|^2.
\]
Equivalently,
\[
\bm v^\top F_{\eta_+}^\top F_{\eta_+}\bm v
\leq
q^2\bm v^\top F_{\eta_-}^\top
F_{\eta_-}\bm v.
\]
Since this holds for every $\bm v$,
the definition of positive-semidefinite
ordering gives
\[
F_{\eta_+}^\top F_{\eta_+}
\preceq q^2F_{\eta_-}^\top F_{\eta_-}.
\]

Subtracting and applying the lower bound
in Assumption~\ref{assump:joint-regularization}(b),
\[
\begin{aligned}
&F_{\eta_-}^\top F_{\eta_-}
-F_{\eta_+}^\top F_{\eta_+}\\
&\succeq
(1-q^2)F_{\eta_-}^\top F_{\eta_-}\\
&\succeq
(1-q^2)c_ErI_2.
\end{aligned}
\]
Let $\kappa=c_E(1-q^2)>0$.
By \eqref{eq:joint-oracle-noise-decomposition},
the common term $\mathcal I_p^{-1}$ cancels
when subtracting the two relative noise matrices.
Therefore,
\begin{equation}
\label{eq:joint-relative-noise-separation}
N_{\eta_-}-N_{\eta_+}
\succeq\kappa rI_2.
\end{equation}

We also need an upper bound on both relative
noise matrices. By Assumption ~\ref{assump:joint-regularization}(b),
$F_{\eta_-}^\top F_{\eta_-}
\preceq C_ErI_2$.
The contraction inequality gives
\[
\begin{aligned}
F_{\eta_+}^\top F_{\eta_+}
&\preceq
q^2F_{\eta_-}^\top F_{\eta_-}\preceq q^2C_ErI_2\preceq C_ErI_2.
\end{aligned}
\]
Thus, setting $K=C_E+1$,
both penalties satisfy
$F_{\eta_\pm}^\top F_{\eta_\pm}
\preceq KrI_2$.

Moreover, the definition of $s_p$ gives
$\mathcal I_p\succeq s_pI_2$.
Because inversion reverses the
positive-semidefinite ordering,
$\mathcal I_p^{-1}\preceq s_p^{-1}I_2$.
Adding these two bounds in
\eqref{eq:joint-oracle-noise-decomposition},
and noting that both terms are
positive semidefinite, yields
\begin{equation}
\label{eq:joint-relative-noise-upper}
0\preceq N_{\eta_\pm}
\preceq(s_p^{-1}+Kr)I_2.
\end{equation}

We next translate the relative noise bounds
into recovery of the latent adjustment space.
Recall that
$L_\eta^\top\Lambda_U=I_2$.
Using
$\bm X_i=\Lambda_U\bm U_i^\star+\bm\xi_i$,
the recalibrated scores satisfy
\[
\begin{aligned}
L_\eta^\top\bm X_i
&=
L_\eta^\top\Lambda_U\bm U_i^\star
+L_\eta^\top\bm\xi_i=
\bm U_i^\star+L_\eta^\top\bm\xi_i.
\end{aligned}
\]
Because $\bm\xi_i$ is uncorrelated with
$\bm U_i^\star$, the residual component
$L_\eta^\top\bm\xi_i$ is also uncorrelated
with $\bm U_i^\star$. Its covariance is
$L_\eta^\top\Sigma_\xi L_\eta=N_\eta$.
Consequently,
\[
\operatorname{cov}
(\bm U_i^\star,L_\eta^\top\bm X_i)=I_2,
\qquad
\operatorname{var}(L_\eta^\top\bm X_i)=I_2+N_\eta.
\]

The residual covariance from the population
linear projection of $\bm U_i^\star$ onto
the learned score space is therefore
\[
\begin{aligned}
&I_2-
\operatorname{cov}
(\bm U_i^\star,L_\eta^\top\bm X_i)
\operatorname{var}(L_\eta^\top\bm X_i)^{-1}
\operatorname{cov}
(L_\eta^\top\bm X_i,\bm U_i^\star)\\
&=
I_2-I_2(I_2+N_\eta)^{-1}I_2\\
&=
I_2-(I_2+N_\eta)^{-1}.
\end{aligned}
\]

By Assumption~\ref{assump:joint-regularization}(a),
$A_\eta=\widehat B_\eta^\top\Lambda_U$ is nonsingular,
so $\widehat B_\eta$ has rank two.
Since $\operatorname{var}(\bm X_i)\succ0$,
the fitted score space has dimension two.
The latent adjustment space also has
dimension two because
$\operatorname{var}(\bm U_i^\star)=I_2$.

The squared largest-principal-angle distance
between these spaces equals the maximum
unexplained variance of a unit-variance
latent direction. By the Rayleigh--Ritz theorem,
this is the largest eigenvalue of the
linear projection residual covariance.
Applying the spectral mapping theorem gives
\begin{equation}
\label{eq:joint-regularization-distance-identity}
\begin{aligned}
&\operatorname{dist}^2
(\mathcal U,\widehat{\mathcal S}_\eta)=
\lambda_{\max}
\{I_2-(I_2+N_\eta)^{-1}\}=
\frac{\lambda_{\max}(N_\eta)}
{1+\lambda_{\max}(N_\eta)}.
\end{aligned}
\end{equation}

Equation~\eqref{eq:joint-relative-noise-separation}
implies
$N_{\eta_-}\succeq
N_{\eta_+}+\kappa rI_2$.
By monotonicity of the largest eigenvalue
under positive-semidefinite ordering,
\[
\lambda_{\max}(N_{\eta_-})
\geq
\lambda_{\max}(N_{\eta_+})+\kappa r.
\]
Thus, writing
$\nu_\pm=\lambda_{\max}(N_{\eta_\pm})$,
we have $\nu_--\nu_+\geq\kappa r$.
Equation~\eqref{eq:joint-relative-noise-upper}
also gives
$0\leq\nu_\pm\leq s_p^{-1}+Kr$.

Using
\eqref{eq:joint-regularization-distance-identity},
the difference in squared distances is
\[
\begin{aligned}
&\operatorname{dist}^2
(\mathcal U,\widehat{\mathcal S}_{\eta_-})
-
\operatorname{dist}^2
(\mathcal U,\widehat{\mathcal S}_{\eta_+})=
\frac{\nu_-}{1+\nu_-}
-\frac{\nu_+}{1+\nu_+}=
\frac{\nu_--\nu_+}
{(1+\nu_-)(1+\nu_+)}.
\end{aligned}
\]
The numerator is at least $\kappa r$,
while each denominator factor is at most
$1+s_p^{-1}+Kr$. Hence,
\begin{equation}
\label{eq:joint-regularization-distance-gap}
\begin{aligned}
&\operatorname{dist}^2
(\mathcal U,\widehat{\mathcal S}_{\eta_-})
-
\operatorname{dist}^2
(\mathcal U,\widehat{\mathcal S}_{\eta_+})\geq
\frac{\kappa r}
{(1+s_p^{-1}+Kr)^2}.
\end{aligned}
\end{equation}

We now show that the joint population
validation criterion favors $\eta_+$.
The population cross-moment identities in
Supplementary Materials~\ref{sec:compact-algorithm:latent_adj}
give
$\operatorname{cov}(\bm X_i,W_i)
=\Lambda_U\bm\gamma_W$ for $W=T,R$.
Therefore,
\[
\begin{aligned}
\operatorname{cov}
(L_\eta^\top\bm X_i,W_i)
&=
L_\eta^\top\operatorname{cov}(\bm X_i,W_i)=
L_\eta^\top\Lambda_U\bm\gamma_W=
\bm\gamma_W.
\end{aligned}
\]

Recall that
$\widehat{\bm S}_{\eta,i}^{\mathrm{raw}}=A_\eta L_\eta^\top\bm X_i$.
Since $A_\eta$ is invertible, the original
and recalibrated scores span the same space.
Let $\widetilde T_{\eta,i}$ denote the
population least-squares prediction of
$T_i$ from this score space. Its regression
coefficient vector in the recalibrated
coordinates is
$\bm\beta_\eta=(I_2+N_\eta)^{-1}\bm\gamma_T$.
Thus,
\[
\widetilde T_{\eta,i}
=
\mathbb E(T_i)+
\bm\beta_\eta^\top
\{L_\eta^\top\bm X_i-
\mathbb E(L_\eta^\top\bm X_i)\}.
\]
Its explained variance is
\[
\begin{aligned}
\operatorname{var}(\widetilde T_{\eta,i})
&=
\bm\beta_\eta^\top
\operatorname{var}(L_\eta^\top\bm X_i)
\bm\beta_\eta\\
&=
\bm\beta_\eta^\top
(I_2+N_\eta)\bm\beta_\eta\\
&=
\bm\gamma_T^\top
(I_2+N_\eta)^{-1}
(I_2+N_\eta)
(I_2+N_\eta)^{-1}\bm\gamma_T\\
&=
\bm\gamma_T^\top
(I_2+N_\eta)^{-1}\bm\gamma_T.
\end{aligned}
\]
Consequently, the squared multiple correlation
between treatment and the learned score space is
\[
t_\eta
=
R^2(T\sim\bm S_\eta)
=
\frac{\bm\gamma_T^\top
(I_2+N_\eta)^{-1}\bm\gamma_T}{v_T}.
\]

Next consider prediction of the
treatment-residualized outcome $R_i$.
Write $\bm V_{\eta,i}=L_\eta^\top\bm X_i$.
The scores after population linear
residualization on treatment are
\[
\bm V_{\eta,i}^{\mathrm{res}}
=
\bm V_{\eta,i}-\mathbb E(\bm V_{\eta,i})
-\frac{\bm\gamma_T}{v_T}
(T_i-\mathbb E T_i).
\]
Using
$\operatorname{var}(\bm V_{\eta,i})=I_2+N_\eta$
and
$\operatorname{cov}(\bm V_{\eta,i},T_i)
=\bm\gamma_T$,
their covariance is
\[
\begin{aligned}
\operatorname{var}
(\bm V_{\eta,i}^{\mathrm{res}})
={}&
I_2+N_\eta-
\frac{\bm\gamma_T\bm\gamma_T^\top}{v_T}.
\end{aligned}
\]
Their covariance with $R_i$ is
\[
\begin{aligned}
&\operatorname{cov}
(\bm V_{\eta,i}^{\mathrm{res}},R_i)=
\operatorname{cov}(\bm V_{\eta,i},R_i)
-\frac{\bm\gamma_T}{v_T}
\operatorname{cov}(T_i,R_i)= L_\eta^\top \Lambda_U\bm\gamma_R
=\bm\gamma_R,
\end{aligned}
\]
because $\operatorname{cov}(T_i,R_i)=0$
by construction.

The residualized covariance is positive
definite under Assumption~
\ref{assump:joint-regularization}(a).
Indeed, Cauchy--Schwarz gives
$\bm\gamma_T\bm\gamma_T^\top
\preceq\|\bm\gamma_T\|^2I_2$.
Since $N_\eta\succeq0$ and
$v_T-\|\bm\gamma_T\|^2\geq c_T>0$,
\[
\begin{aligned}
&I_2+N_\eta-
\frac{\bm\gamma_T\bm\gamma_T^\top}{v_T}\succeq
\left(1-\frac{\|\bm\gamma_T\|^2}{v_T}\right)I_2\succeq
\frac{c_T}{v_T}I_2\succ0.
\end{aligned}
\]

The population squared multiple correlation
for the treatment-residualized outcome is
therefore
\[
\begin{aligned}
u_\eta
&=
R^2(R\sim\bm S_{\eta,\mathrm{res}})\\
&=
\frac{1}{v_R}\bm\gamma_R^\top
\left(
I_2+N_\eta-
\frac{\bm\gamma_T\bm\gamma_T^\top}{v_T}
\right)^{-1}\bm\gamma_R.
\end{aligned}
\]
By the definition of the joint population
criterion, $J_C(\eta)^2=t_\eta u_\eta$.

We next compare these multiple correlations
at the two penalties. For brevity, set
$D_{\max}=1+s_p^{-1}+Kr$ and $\delta=\kappa r$.
Equations~\eqref{eq:joint-relative-noise-separation}
and~\eqref{eq:joint-relative-noise-upper}
give
\[
N_{\eta_-}-N_{\eta_+}\succeq\delta I_2,
\qquad
0\preceq N_{\eta_\pm}\preceq(D_{\max}-1)I_2.
\]

We use an elementary inverse-order inequality.
Suppose $P$ and $Q$ are positive-definite
matrices satisfying
$P\preceq Q\preceq D_{\max}I_2$ and
$Q-P\succeq\delta I_2$, where $0<\delta<D_{\max}$.
Because $Q\preceq D_{\max}I_2$,
\[
\delta I_2\succeq\frac{\delta}{D_{\max}}Q.
\]
Combining this with $Q-P\succeq\delta I_2$
gives
\[
Q-P\succeq\frac{\delta}{D_{\max}}Q.
\]
Rearranging,
\[
P\preceq
\left(1-\frac{\delta}{D_{\max}}\right)Q.
\]
Inversion reverses the positive-semidefinite
ordering, so
\[
P^{-1}\succeq
\frac{D_{\max}}{D_{\max}-\delta}Q^{-1}.
\]
Subtracting $Q^{-1}$ from both sides,
\[
P^{-1}-Q^{-1}
\succeq
\frac{\delta}{D_{\max}-\delta}Q^{-1}.
\]
Finally, $Q\preceq D_{\max}I_2$ implies
$Q^{-1}\succeq D_{\max}^{-1}I_2$.
Since $D_{\max}-\delta\leq D_{\max}$, we obtain
\[
\begin{aligned}
P^{-1}-Q^{-1}
&\succeq
\frac{\delta}{D_{\max}(D_{\max}-\delta)}I_2\succeq
\frac{\delta}{D_{\max}^2}I_2.
\end{aligned}
\]

Apply this result with
$P=I_2+N_{\eta_+}$ and
$Q=I_2+N_{\eta_-}$.
The required bounds follow from the two
relative noise inequalities above.
Therefore,
\[
\begin{aligned}
&(I_2+N_{\eta_+})^{-1}
-(I_2+N_{\eta_-})^{-1}\succeq
\frac{\delta}{D_{\max}^2}I_2.
\end{aligned}
\]
Substituting into the expression for
$t_\eta$ gives
\[
\begin{aligned}
t_{\eta_+}-t_{\eta_-}
&=
\frac{1}{v_T}\bm\gamma_T^\top
\left\{
(I_2+N_{\eta_+})^{-1}
-
(I_2+N_{\eta_-})^{-1}
\right\}
\bm\gamma_T\geq
\frac{\delta\|\bm\gamma_T\|^2}{v_TD_{\max}^2}.
\end{aligned}
\]

For the residualized outcome, define
$P_R=I_2+N_{\eta_+}
-\bm\gamma_T\bm\gamma_T^\top/v_T$
and
$Q_R=I_2+N_{\eta_-}
-\bm\gamma_T\bm\gamma_T^\top/v_T$.
Both matrices are positive definite,
as established above.
Their difference is
\[
Q_R-P_R
=
N_{\eta_-}-N_{\eta_+}
\succeq\delta I_2.
\]
Thus $P_R\preceq Q_R$, and inversion gives
$P_R^{-1}\succeq Q_R^{-1}$.
It follows that
\[
\begin{aligned}
u_{\eta_+}-u_{\eta_-}
&=
\frac{1}{v_R}\bm\gamma_R^\top
(P_R^{-1}-Q_R^{-1})\bm\gamma_R\geq0.
\end{aligned}
\]

Moreover,
$Q_R\preceq I_2+N_{\eta_-}\preceq D_{\max}I_2$.
Hence $Q_R^{-1}\succeq D_{\max}^{-1}I_2$, so
\[
\begin{aligned}
u_{\eta_-}
&=
\frac{1}{v_R}\bm\gamma_R^\top
Q_R^{-1}\bm\gamma_R\geq
\frac{\|\bm\gamma_R\|^2}{v_RD_{\max}}.
\end{aligned}
\]

Combining these inequalities,
\[
\begin{aligned}
&J_C(\eta_+)^2-J_C(\eta_-)^2\\
&=
t_{\eta_+}u_{\eta_+}
-t_{\eta_-}u_{\eta_-}\\
&=
(t_{\eta_+}-t_{\eta_-})u_{\eta_-}
+
t_{\eta_+}(u_{\eta_+}-u_{\eta_-})\\
&\geq
(t_{\eta_+}-t_{\eta_-})u_{\eta_-}\\
&\geq
\frac{\delta
\|\bm\gamma_T\|^2
\|\bm\gamma_R\|^2}
{v_Tv_RD_{\max}^3}.
\end{aligned}
\]

Assumption~\ref{assump:joint-regularization}(a)
gives
$\sigma_{\min}
\{(\bm\gamma_T,\bm\gamma_R)\}\geq c_\Gamma$.
Each column therefore has Euclidean norm
at least $c_\Gamma$.
Together with $v_T,v_R\leq K_v$, this yields
\[
J_C(\eta_+)^2-J_C(\eta_-)^2
\geq
\frac{\kappa c_\Gamma^4r}{K_v^2D_{\max}^3}.
\]
Since $0\leq J_C(\eta)\leq1$,
we have $J_C(\eta_+)+J_C(\eta_-)\leq2$.
Applying the difference-of-squares identity,
\[
\begin{aligned}
&J_C(\eta_+)-J_C(\eta_-)=
\frac{
J_C(\eta_+)^2-J_C(\eta_-)^2
}{
J_C(\eta_+)+J_C(\eta_-)
}\geq
\frac{
J_C(\eta_+)^2-J_C(\eta_-)^2
}{2}.
\end{aligned}
\]
Consequently,
\begin{equation}
\label{eq:joint-regularization-validation-gap}
J_C(\eta_+)-J_C(\eta_-)
\geq
\frac{\kappa c_\Gamma^4r}
{2K_v^2D_{\max}^3}.
\end{equation}
Thus, the joint population criterion favors
the stronger penalty $\eta_+$.

We next show that the individual ridge
prediction criteria favor $\eta_-$.
By Lemma~\ref{lem:ridge-span-decomposition},
the prediction error for $W=T,R$ is
\[
L_W(\eta)
=
\mathbb E(W_i-P_\eta W_i)^2
+
\mathbb E(P_\eta g_W-\widehat W_\eta)^2.
\]
Subtracting the identities at the two
penalties gives
\[
\begin{aligned}
&L_W(\eta_+)-L_W(\eta_-)\\
&=
\mathbb E(W_i-P_{\eta_+}W_i)^2
-
\mathbb E(W_i-P_{\eta_-}W_i)^2+
\mathbb E(P_{\eta_+}g_W-\widehat W_{\eta_+})^2
-
\mathbb E(P_{\eta_-}g_W-\widehat W_{\eta_-})^2.
\end{aligned}
\]
Assumption~\ref{assump:joint-regularization}(c)
states that the second difference exceeds
the negative of the first difference
by at least $b_n$. Therefore,
\begin{equation}
\label{eq:joint-regularization-prediction-gap}
L_W(\eta_+)-L_W(\eta_-)\geq b_n>0,
\qquad W=T,R.
\end{equation}
Thus, separate ridge prediction favors
the weaker penalty for both responses.

Finally, we show that cross-validation
selects the population-preferred penalties.
Let
\[
e_C=
\max_{\eta\in\{\eta_-,\eta_+\}}
|\widehat J_C(\eta)-J_C(\eta)|.
\]
Assumption~\ref{assump:joint-regularization}(d)
gives $e_C=o_p(r/D_{\max}^3)$.
By \eqref{eq:joint-regularization-validation-gap},
the population joint criterion has a
positive gap of order at least $r/D_{\max}^3$.
The empirical difference satisfies

\[
\begin{aligned}
\widehat J_C(\eta_+)-\widehat J_C(\eta_-)&\geq
\{J_C(\eta_+)-e_C\}
-
\{J_C(\eta_-)+e_C\}\\
&=
J_C(\eta_+)-J_C(\eta_-)-2e_C\\
&\geq
\frac{\kappa c_\Gamma^4r}{2K_v^2D_{\max}^3}
-2e_C.
\end{aligned}
\]
Since $e_C=o_p(r/D_{\max}^3)$, the last expression
is positive with probability tending to one.
Therefore, $\Pr(\widehat\eta_C=\eta_+)\to1$.

Similarly, for each $W=T,R$, let
\[
e_W=
\max_{\eta\in\{\eta_-,\eta_+\}}
|\widehat L_W(\eta)-L_W(\eta)|.
\]
Assumption \ref{assump:joint-regularization}(d) gives $e_W=o_p(b_n)$.
Using
\eqref{eq:joint-regularization-prediction-gap},
\[
\begin{aligned}
&\widehat L_W(\eta_+)
-\widehat L_W(\eta_-)\geq
L_W(\eta_+)-L_W(\eta_-)-2e_W\geq
b_n-2e_W.
\end{aligned}
\]
The last expression is positive with
probability tending to one. Hence,
$\Pr(\widehat\eta_W=\eta_-)\to1$
for both $W=T$ and $W=R$.

Combining the selection results,
\[
\Pr\left(
\widehat\eta_C=\eta_+,\,
\widehat\eta_T=\widehat\eta_R=\eta_-
\right)\longrightarrow1.
\]
On this event,
$\widehat{\mathcal S}_{\widehat\eta_C}
=\widehat{\mathcal S}_{\eta_+}$.
By the fixed-penalty ridge-span equivalence,
$\widehat{\mathcal S}_{\mathrm{sep}}
=\widehat{\mathcal S}_{\eta_-}$.

Applying
\eqref{eq:joint-regularization-distance-gap}
therefore gives
\[
\begin{aligned}
&\operatorname{dist}^2
(\mathcal U,\widehat{\mathcal S}_{\mathrm{sep}})
-
\operatorname{dist}^2
(\mathcal U,\widehat{\mathcal S}_{\widehat\eta_C})\geq
\frac{\kappa r}{(1+s_p^{-1}+Kr)^2}.
\end{aligned}
\]
Finally, because $K=C_E+1>1$,
\[
1+s_p^{-1}+Kr
\leq
K(1+s_p^{-1}+r).
\]
Thus,
\[
\frac{\kappa r}{(1+s_p^{-1}+Kr)^2}
\geq
\frac{\kappa}{K^2}
\frac{r}{(1+s_p^{-1}+r)^2}.
\]
The claimed lower bound follows with
$c=\kappa/K^2>0$, completing the proof.
\end{proof}

\subsection{Lemma \ref{lem:proxy-recovery}}

\begin{replemma}{lem:proxy-recovery}
If $p_{\mathrm{eff}}>0$, then
\(
\operatorname{dist}^2
\left(
\mathcal U,
\mathcal S_\eta
\right)
\in
O\left(
\frac{1}{p_{\mathrm{eff}}}
\right).
\)
\end{replemma}

\begin{proof}
Because \(\operatorname{var}(\bm U_i^\star)=I_2\), the latent adjustment
space \(\mathcal U\) has dimension two. Because \(\mathcal S_\eta\) is
spanned by the two components of \(\bm S_{\eta,i}\), it has dimension at
most two.

We first show that \(p_{\mathrm{eff}}>0\) implies that
\(\mathcal S_\eta\) has dimension two. Suppose instead that it had
dimension less than two. The projections of the two components of
\(\bm U_i^\star\) onto \(\mathcal S_\eta\) would then be linearly
dependent. Hence, there would exist \({\bm{\xi}}\in\mathbb R^2\), with
\(\|{\bm{\xi}}\|_2=1\), such that
\[
P_{\mathcal S_\eta}
\left(
{\bm{\xi}}^\top\bm U_i^\star
\right)
=
0.
\]
Because \(\operatorname{var}(\bm U_i^\star)=I_2\), this direction has
unit variance, and therefore
\[
\left\|
(I-P_{\mathcal S_\eta})
({\bm{\xi}}^\top\bm U_i^\star)
\right\|_{L_2}^2
=
1.
\]
This would imply
\(
\lambda_{\max}(\Sigma_{\bm U^\star\mid\bm S_\eta})=1
\)
and hence \(p_{\mathrm{eff}}=0\), a contradiction. Thus,
\(\mathcal S_\eta\) has dimension two.

The spaces \(\mathcal U\) and \(\mathcal S_\eta\) therefore have the same
dimension. For two subspaces of the same finite dimension, we can write
\[
\begin{aligned}
\operatorname{dist}^2
\left(
\mathcal U,
\mathcal S_\eta
\right)
&=
\sup_{\|{\bm{\xi}}\|_2=1}
\left\|
(I-P_{\mathcal S_\eta})
({\bm{\xi}}^\top\bm U_i^\star)
\right\|_{L_2}^2\\
&=
\sup_{\|{\bm{\xi}}\|_2=1}
{\bm{\xi}}^\top
\Sigma_{\bm U^\star\mid\bm S_\eta}
{\bm{\xi}} \\
&=
\lambda_{\max}
\left(
\Sigma_{\bm U^\star\mid\bm S_\eta}
\right).
\end{aligned}
\]
Finally, the definition of \(p_{\mathrm{eff}}\) gives
\[
\lambda_{\max}
\left(
\Sigma_{\bm U^\star\mid\bm S_\eta}
\right)
=
\frac{1}{1+p_{\mathrm{eff}}}
\leq
\frac{1}{p_{\mathrm{eff}}}.
\]
Hence,
\[
\operatorname{dist}^2
\left(
\mathcal U,
\mathcal S_\eta
\right)
\in
O\left(
\frac{1}{p_{\mathrm{eff}}}
\right).
\]
\end{proof}

\subsection{Theorem \ref{thm:proxy-sample-tradeoff}}

\begin{reptheorem}{thm:proxy-sample-tradeoff}
Suppose the conditions of Lemmas~\ref{lem:spectral-subspace-error} and
\ref{lem:proxy-recovery} hold. Then
\[
\operatorname{dist}^2
\left(
\mathcal U,
\widehat{\mathcal S}_\eta
\right)
\in
O_p\left(
\frac{1}{p_{\mathrm{eff}}}
+
\frac{d_{\mathrm{eff}}(\eta)}{n\Delta^2}
\right).
\]
\end{reptheorem}
\begin{proof}
Lemmas~\ref{lem:proxy-recovery} and
\ref{lem:spectral-subspace-error} give
\[
\operatorname{dist}^2
\left(
\mathcal U,
\mathcal S_\eta
\right)
\in
O\left(
\frac{1}{p_{\mathrm{eff}}}
\right)
\]
and
\[
\operatorname{dist}^2
\left(
\mathcal S_\eta,
\widehat{\mathcal S}_\eta
\right)
\in
O_p\left(
\frac{d_{\mathrm{eff}}(\eta)}{n\Delta^2}
\right),
\]
respectively. By the triangle inequality for the operator norm,
\[
\begin{aligned}
\operatorname{dist}
\left(
\mathcal U,
\widehat{\mathcal S}_\eta
\right)
&=
\left\|
P_{\mathcal U}
-
P_{\widehat{\mathcal S}_\eta}
\right\|_{\mathrm{op}} \\
&\leq
\left\|
P_{\mathcal U}
-
P_{\mathcal S_\eta}
\right\|_{\mathrm{op}}
+
\left\|
P_{\mathcal S_\eta}
-
P_{\widehat{\mathcal S}_\eta}
\right\|_{\mathrm{op}} \\
&=
\operatorname{dist}
\left(
\mathcal U,
\mathcal S_\eta
\right)
+
\operatorname{dist}
\left(
\mathcal S_\eta,
\widehat{\mathcal S}_\eta
\right).
\end{aligned}
\]
Using \((a+b)^2\leq2a^2+2b^2\) therefore gives
\[
\operatorname{dist}^2
\left(
\mathcal U,
\widehat{\mathcal S}_\eta
\right)
\leq
2\operatorname{dist}^2
\left(
\mathcal U,
\mathcal S_\eta
\right) +
2\operatorname{dist}^2
\left(
\mathcal S_\eta,
\widehat{\mathcal S}_\eta
\right) \in
O_p\left(
\frac{1}{p_{\mathrm{eff}}}
+
\frac{d_{\mathrm{eff}}(\eta)}{n\Delta^2}
\right).
\]
\end{proof}

The unconditional reconstruction bound used in
Assumption~\ref{assump:causal}(b) follows from
Theorem~\ref{thm:proxy-sample-tradeoff} because
\[
\begin{aligned}
\mathbb E\left[
\left\|
\bm U_i^\star-\mathcal R_n(\widehat{\bm S}_{\eta,i})
\right\|^2
\right]
&=
\sum_{j=1}^2
\left\|
\left(
I-P_{\widehat{\mathcal S}_\eta}
\right)
U_{ij}^\star
\right\|_{L_2}^2
\\
&\leq
\sum_{j=1}^2
\left\|
P_{\mathcal U}
-
P_{\widehat{\mathcal S}_\eta}
\right\|_{\mathrm{op}}^2
\left\|
U_{ij}^\star
\right\|_{L_2}^2
\\
&=
2\operatorname{dist}^2
\left(
\mathcal U,\widehat{\mathcal S}_\eta
\right)
\in
O_p\left(
\frac{1}{p_{\mathrm{eff}}}
+
\frac{d_{\mathrm{eff}}(\eta)}{n\Delta^2}
\right),
\end{aligned}
\]
where
\(
\operatorname{var}(\bm U_i^\star)=I_2
\)
implies
\(
\|U_{ij}^\star\|^2_{L_2}=1
\).

\subsection{Theorem \ref{thm:causal-consistency-compact-plugin}}

\begin{reptheorem}{thm:causal-consistency-compact-plugin}
Suppose Assumption~\ref{assump:causal} holds. Then
\[
\sup_{t\in\mathcal T_0}
\left|
\widehat\theta_{\widehat{\bm S}_\eta}^{\mathrm{marg}}(t)-\theta(t)
\right|
\in
O_p\left[
\left\{
\frac{1}{p_{\mathrm{eff}}}
+
\frac{d_{\mathrm{eff}}(\eta)}{n\Delta^2}
\right\}^{1/2}
\right]
+
o_p(1),
\]
and
\[
\sup_{t\in\mathcal T_0,\;\bm s\in\mathcal S_0}
\left|
\widehat\theta_{\widehat{\bm S}_\eta}^{\mathrm{cond}}(t;\bm s)
-
\theta_{\widehat{\bm S}_\eta}^{\mathrm{cond}}(t;\bm s)
\right|
\in
O_p\left[
\left\{
\frac{1}{p_{\mathrm{eff}}}
+
\frac{d_{\mathrm{eff}}(\eta)}{n\Delta^2}
\right\}^{1/2}
\right]
+
o_p(1).
\]
\end{reptheorem}

\begin{proof}
For brevity, set
\[
\delta_{n,p}^2
=
\frac{1}{p_{\mathrm{eff}}}
+
\frac{d_{\mathrm{eff}}(\eta)}{n\Delta^2}.
\]
All population expectations in this proof are taken under the evaluation law
conditional on $\mathcal D_n$, as defined before
Assumption~\ref{assump:causal}. The empirical score map and the downstream
estimator may nevertheless use the same observations in $\mathcal D_n$.

Let
\(
\widetilde{\bm U}_i^\star
=
\mathcal R_n(\widehat{\bm S}_{\eta,i}).
\)
By Assumption~\ref{assump:causal}(b),
\[
\mathbb E\left[
\|\bm U_i^\star-\widetilde{\bm U}_i^\star\|^2
\right]
\in
O_p(\delta_{n,p}^2).
\]
Therefore, by Cauchy--Schwarz,
\[
\mathbb E\left[
\|\bm U_i^\star-\widetilde{\bm U}_i^\star\|
\right]
\leq
\left\{
\mathbb E\left[
\|\bm U_i^\star-\widetilde{\bm U}_i^\star\|^2
\right]
\right\}^{1/2}
\in
O_p(\delta_{n,p}).
\]
The conditional part of Assumption~\ref{assump:causal}(b) similarly gives
\[
\sup_{\bm s\in\mathcal S_0}
\mathbb E\left[
\|\bm U_i^\star-\widetilde{\bm U}_i^\star\|
\mid \widehat{\bm S}_{\eta,i}=\bm s
\right]
\in
O_p(\delta_{n,p}).
\]

We first bound the population bias in the score-indexed response surface. The
definition of $m_{\widehat{\bm S}_\eta}$ gives
\[
m_{\widehat{\bm S}_\eta}(t,\bm s)
=
\mathbb E(Y_i\mid T_i=t,\widehat{\bm S}_{\eta,i}=\bm s).
\]
Consistency in Assumption~\ref{assump:causal}(a) gives
\[
m_{\widehat{\bm S}_\eta}(t,\bm s)
=
\mathbb E\{Y_i(t)\mid T_i=t,\widehat{\bm S}_{\eta,i}=\bm s\}.
\]
Applying Assumption~\ref{assump:causal}(a) yields
\[
\begin{aligned}
\mathbb E\{\mu(t,\bm U_i^\star,\bm s)
\mid T_i=t,\widehat{\bm S}_{\eta,i}=\bm s\} &= \mathbb E\{ \mathbb{E}(Y_i(t) \mid \bm U_i^\star, \widehat{\bm S}_{\eta,i} = \bm s)
\mid T_i=t,\widehat{\bm S}_{\eta,i}=\bm s\} \\
& = \mathbb E\{ \mathbb{E}(Y_i(t) \mid \bm U_i^\star, T_i = t, \widehat{\bm S}_{\eta,i} = \bm s)
\mid T_i=t,\widehat{\bm S}_{\eta,i}=\bm s\} \\
& = \mathbb E\{ Y_i(t) \mid T_i=t,\widehat{\bm S}_{\eta,i}=\bm s\}\\
& = m_{\widehat{\bm S}_\eta}(t,\bm s),
\end{aligned}
\]
where the second equality follows from Lemma \ref{lem:score-augmented-latent-exchangeability}.
Likewise, the definition of $\theta_{\widehat{\bm S}_\eta}^{\mathrm{cond}}$ and the law of
iterated expectations give
\[
\theta_{\widehat{\bm S}_\eta}^{\mathrm{cond}}(t;\bm s)
=
\mathbb E\{\mu(t,\bm U_i^\star,\bm s)\mid \widehat{\bm S}_{\eta,i}=\bm s\}.
\]
Therefore,
\[
 m_{\widehat{\bm S}_\eta}(t,\bm s)-\theta_{\widehat{\bm S}_\eta}^{\mathrm{cond}}(t;\bm s)
=
C_1(t,\bm s)+C_2(t,\bm s),
\]
where
\[
C_1(t,\bm s)
=
\mathbb E\{\mu(t,\bm U_i^\star,\bm s)-\mu(t,\mathcal R_n(\bm s),\bm s)
\mid T_i=t,\widehat{\bm S}_{\eta,i}=\bm s\},
\]
and
\[
C_2(t,\bm s)
=
\mathbb E\{\mu(t,\mathcal R_n(\bm s),\bm s)-\mu(t,\bm U_i^\star,\bm s)
\mid \widehat{\bm S}_{\eta,i}=\bm s\}.
\]
By the Lipschitz condition in Assumption~\ref{assump:causal}(c),
\[
\begin{aligned}
|C_2(t,\bm s)|
&\leq \mathbb E\{|\mu(t,\mathcal R_n(\bm s),\bm s)-\mu(t,\bm U_i^\star,\bm s) |
\mid \widehat{\bm S}_{\eta,i}=\bm s\}\\ &\leq
L_\mu
\mathbb E\{\|\bm U_i^\star-\mathcal R_n(\bm s)\|\mid \widehat{\bm S}_{\eta,i}=\bm s\}.
\end{aligned}
\]
Thus
\[
\sup_{t\in\mathcal T_0,\;\bm s\in\mathcal S_0}|C_2(t,\bm s)|
\in
O_p(\delta_{n,p}).
\]
To control \(C_1(t,\bm s)\), use the bounded density-ratio condition in
Assumption~\ref{assump:causal}(d). For any nonnegative measurable function
\(\ell(\bm U_i^\star)\),
\[
\mathbb E\{\ell(\bm U_i^\star)\mid T_i=t,\widehat{\bm S}_{\eta,i}=\bm s\}
\leq
K_\omega
\mathbb E\{\ell(\bm U_i^\star)\mid \widehat{\bm S}_{\eta,i}=\bm s\}.
\]
Applying this with
\[
\ell(\bm U_i^\star)=|\mu(t,\bm U_i^\star,\bm s)-\mu(t,\mathcal R_n(\bm s),\bm s)|
\]
and using Lipschitz continuity gives
\[
|C_1(t,\bm s)|
\leq
K_\omega L_\mu
\mathbb E\{\|\bm U_i^\star-\mathcal R_n(\bm s)\|\mid \widehat{\bm S}_{\eta,i}=\bm s\}.
\]
Hence
\[
\sup_{t\in\mathcal T_0,\;\bm s\in\mathcal S_0}|C_1(t,\bm s)|
\in
O_p(\delta_{n,p}).
\]
Combining the bounds for \(C_1(t,\bm s)\) and \(C_2(t,\bm s)\) yields
\[
\sup_{t\in\mathcal T_0,\;\bm s\in\mathcal S_0}
\left|
 m_{\widehat{\bm S}_\eta}(t,\bm s)-\theta_{\widehat{\bm S}_\eta}^{\mathrm{cond}}(t;\bm s)
\right|
\in
O_p(\delta_{n,p}).
\]

The marginal bias follows by averaging the score-indexed bias over the
distribution of \(\widehat{\bm S}_{\eta,i}\). Specifically,
\[
\theta_{\widehat{\bm S}_\eta}^{\mathrm{marg}}(t)-\theta(t)
=
\mathbb E\left[
m_{\widehat{\bm S}_\eta}(t,\widehat{\bm S}_{\eta,i})
-
\theta_{\widehat{\bm S}_\eta}^{\mathrm{cond}}(t;\widehat{\bm S}_{\eta,i})
\right].
\]
Applying the same decomposition used for the score-indexed target, with
\(\bm s=\widehat{\bm S}_{\eta,i}\), gives
\[
\begin{aligned}
&
m_{\widehat{\bm S}_\eta}(t,\widehat{\bm S}_{\eta,i})
-
\theta_{\widehat{\bm S}_\eta}^{\mathrm{cond}}(t;\widehat{\bm S}_{\eta,i})
\\
&=
\mathbb E\left[
\mu(t,\bm U_i^\star,\widehat{\bm S}_{\eta,i})
-
\mu\{t,\mathcal R_n(\widehat{\bm S}_{\eta,i}),\widehat{\bm S}_{\eta,i}\}
\mid T_i=t,\widehat{\bm S}_{\eta,i}
\right]
\\
&\quad -
\mathbb E\left[
\mu(t,\bm U_i^\star,\widehat{\bm S}_{\eta,i})
-
\mu\{t,\mathcal R_n(\widehat{\bm S}_{\eta,i}),\widehat{\bm S}_{\eta,i}\}
\mid \widehat{\bm S}_{\eta,i}
\right].
\end{aligned}
\]
By the Lipschitz condition,
\[
\begin{aligned}
&
\left|
m_{\widehat{\bm S}_\eta}(t,\widehat{\bm S}_{\eta,i})
-
\theta_{\widehat{\bm S}_\eta}^{\mathrm{cond}}(t;\widehat{\bm S}_{\eta,i})
\right|
\\
&\leq
L_\mu
\left[
\mathbb E\left\{
\left\|\bm U_i^\star-\mathcal R_n(\widehat{\bm S}_{\eta,i})\right\|
\mid T_i=t,\widehat{\bm S}_{\eta,i}
\right\}
+
\mathbb E\left\{
\left\|\bm U_i^\star-\mathcal R_n(\widehat{\bm S}_{\eta,i})\right\|
\mid \widehat{\bm S}_{\eta,i}
\right\}
\right].
\end{aligned}
\]
Because
\[
\theta(t)
=
\mathbb E\left[
\theta_{\widehat{\bm S}_\eta}^{\mathrm{cond}}
(t;\widehat{\bm S}_{\eta,i})
\right],
\]
we have
\[
\begin{aligned}
\left|
\theta_{\widehat{\bm S}_\eta}^{\mathrm{marg}}(t)-\theta(t)
\right|
&=
\left|
\mathbb E\left[
m_{\widehat{\bm S}_\eta}(t,\widehat{\bm S}_{\eta,i})
-
\theta_{\widehat{\bm S}_\eta}^{\mathrm{cond}}
(t;\widehat{\bm S}_{\eta,i})
\right]
\right| \\
&\leq
\mathbb E\left[
\left|
m_{\widehat{\bm S}_\eta}(t,\widehat{\bm S}_{\eta,i})
-
\theta_{\widehat{\bm S}_\eta}^{\mathrm{cond}}
(t;\widehat{\bm S}_{\eta,i})
\right|
\right].
\end{aligned}
\]
Applying the preceding pointwise bound gives
\[
\begin{aligned}
\left|
\theta_{\widehat{\bm S}_\eta}^{\mathrm{marg}}(t)-\theta(t)
\right|
&\leq
L_\mu\,
\mathbb E\left[
\mathbb E\left\{
\left\|
\bm U_i^\star-
\mathcal R_n(\widehat{\bm S}_{\eta,i})
\right\|
\mid
T_i=t,\widehat{\bm S}_{\eta,i}
\right\}
\right] \\
&\quad+
L_\mu\,
\mathbb E\left[
\mathbb E\left\{
\left\|
\bm U_i^\star-
\mathcal R_n(\widehat{\bm S}_{\eta,i})
\right\|
\mid
\widehat{\bm S}_{\eta,i}
\right\}
\right].
\end{aligned}
\]
The outer expectations here are taken over the marginal evaluation-law
distribution of \(\widehat{\bm S}_{\eta,i}\). By the bounded density-ratio condition,
the first term is at most
\[
K_\omega L_\mu
\mathbb E\left[
\mathbb E\left\{
\left\|
\bm U_i^\star-
\mathcal R_n(\widehat{\bm S}_{\eta,i})
\right\|
\mid
\widehat{\bm S}_{\eta,i}
\right\}
\right].
\]
The law of iterated expectations then gives
\[
\mathbb E\left[
\mathbb E\left\{
\left\|
\bm U_i^\star-
\mathcal R_n(\widehat{\bm S}_{\eta,i})
\right\|
\mid
\widehat{\bm S}_{\eta,i}
\right\}
\right]
=
\mathbb E\left[
\left\|
\bm U_i^\star-
\mathcal R_n(\widehat{\bm S}_{\eta,i})
\right\|
\right].
\]
Consequently,
\[
\left|
\theta_{\widehat{\bm S}_\eta}^{\mathrm{marg}}(t)-\theta(t)
\right|
\leq
L_\mu(1+K_\omega)
\mathbb E\left[
\left\|
\bm U_i^\star-
\mathcal R_n(\widehat{\bm S}_{\eta,i})
\right\|
\right].
\]
Finally, Jensen's inequality and the average reconstruction bound in
Assumption~\ref{assump:causal}(b) give
\[
\mathbb E
\left[
\left\|\bm U_i^\star-\mathcal R_n(\widehat{\bm S}_{\eta,i})\right\|
\right]
\leq
\left[
\mathbb E
\left\{
\left\|\bm U_i^\star-\mathcal R_n(\widehat{\bm S}_{\eta,i})\right\|^2
\right\}
\right]^{1/2}
\in
O_p(\delta_{n,p}).
\]
Therefore,
\[
\sup_{t\in\mathcal T_0}
\left|
\theta_{\widehat{\bm S}_\eta}^{\mathrm{marg}}(t)-\theta(t)
\right|
\in
O_p(\delta_{n,p}).
\]

Finally, decompose the downstream estimators as
\[
\widehat\theta_{\widehat{\bm S}_\eta}^{\mathrm{marg}}(t)-\theta(t)
=
\left\{
\widehat\theta_{\widehat{\bm S}_\eta}^{\mathrm{marg}}(t)
-
\theta_{\widehat{\bm S}_\eta}^{\mathrm{marg}}(t)
\right\}
+
\left\{
\theta_{\widehat{\bm S}_\eta}^{\mathrm{marg}}(t)-\theta(t)
\right\},
\]
and
\[
\widehat\theta_{\widehat{\bm S}_\eta}^{\mathrm{cond}}(t;\bm s)
-
\theta_{\widehat{\bm S}_\eta}^{\mathrm{cond}}(t;\bm s)
=
\left\{
\widehat\theta_{\widehat{\bm S}_\eta}^{\mathrm{cond}}(t;\bm s)
-
 m_{\widehat{\bm S}_\eta}(t,\bm s)
\right\}
+
\left\{
 m_{\widehat{\bm S}_\eta}(t,\bm s)
-
\theta_{\widehat{\bm S}_\eta}^{\mathrm{cond}}(t;\bm s)
\right\}.
\]
By Assumption~\ref{assump:causal}(e), the first term in each decomposition is
\(o_p(1)\) uniformly over its index set. This assumption concerns the actual
same-sample downstream estimators and therefore accounts for reuse of
$\mathcal D_n$ at both stages. The population bias bounds above give
the second terms as \(O_p(\delta_{n,p})\). Therefore,
\[
\sup_{t\in\mathcal T_0}
\left|
\widehat\theta_{\widehat{\bm S}_\eta}^{\mathrm{marg}}(t)-\theta(t)
\right|
\in
O_p(\delta_{n,p})+o_p(1),
\]
and
\[
\sup_{t\in\mathcal T_0,\;\bm s\in\mathcal S_0}
\left|
\widehat\theta_{\widehat{\bm S}_\eta}^{\mathrm{cond}}(t;\bm s)
-
\theta_{\widehat{\bm S}_\eta}^{\mathrm{cond}}(t;\bm s)
\right|
\in
O_p(\delta_{n,p})+o_p(1).
\]
Substituting the definition of $\delta_{n,p}^2$ gives the stated
response-function bounds.
\end{proof}

\begin{lemma}
\label{lem:score-augmented-latent-exchangeability}
Under \eqref{eq:linear-treatment-model},
\eqref{eq:linear-proxy-model}, and
\eqref{eq:conditional-mean-restriction}, the evaluation-law convention
implies that, for every \(t\in\mathcal T_0\),
\[
\mathbb E\{Y_i(t)\mid
T_i=t,\bm U_i^\star,\widehat{\bm S}_{\eta,i}\}
=
\mathbb E\{Y_i(t)\mid
\bm U_i^\star,\widehat{\bm S}_{\eta,i}\}.
\]
\end{lemma}

\begin{proof}
Conditional on the fitted sample $\mathcal D_n$, the map
$\widehat s_\eta$ is fixed, and
$\widehat{\bm S}_{\eta,i}
=
\widehat s_\eta(\bm X_i)$
is a measurable function of $\bm X_i$. The conditional proxy-independence
condition in \eqref{eq:linear-proxy-model} therefore implies
\[
\widehat{\bm S}_{\eta,i}
\ci
\{T_i,Y_i(t):t\in\mathcal T\}
\mid
\bm C_i.
\]

Because $\bm U_i^\star$ is an invertible transformation of $\bm U_i$,
conditioning on $\bm U_i^\star$ determines $U_{T,i}$. It follows from the
preceding conditional independence and the treatment model that, for every
measurable set $A$,
\[
\begin{aligned}
P(T_i\in A\mid
\bm C_i,\bm U_i^\star,\widehat{\bm S}_{\eta,i})
&=
P(T_i\in A\mid\bm C_i)\\
&=
P(T_i\in A\mid U_{T,i})\\
&=
P(T_i\in A\mid\bm U_i^\star).
\end{aligned}
\]
Consequently,
$T_i
\ci
\bm C_i
\mid
\bm U_i^\star,\widehat{\bm S}_{\eta,i}.$

The conditional proxy independence in
\eqref{eq:linear-proxy-model}, together with
\eqref{eq:conditional-mean-restriction}, also gives
\[
\begin{aligned}
\mathbb E\{Y_i(t)\mid
T_i,\bm C_i,\bm U_i^\star,\widehat{\bm S}_{\eta,i}\}
&=
\mathbb E\{Y_i(t)\mid T_i,\bm C_i\}=
\mathbb E\{Y_i(t)\mid\bm C_i\},
\end{aligned}
\]
and
\[
\mathbb E\{Y_i(t)\mid
\bm C_i,\bm U_i^\star,\widehat{\bm S}_{\eta,i}\}
=
\mathbb E\{Y_i(t)\mid\bm C_i\}.
\]
Therefore, by iterated expectations,
\[
\begin{aligned}
\mathbb E\{Y_i(t)\mid
T_i,\bm U_i^\star,\widehat{\bm S}_{\eta,i}\}
&=
\mathbb E\left[
\mathbb E\{Y_i(t)\mid\bm C_i\}
\mid
T_i,\bm U_i^\star,\widehat{\bm S}_{\eta,i}
\right] \\
&=
\mathbb E\left[
\mathbb E\{Y_i(t)\mid\bm C_i\}
\mid
\bm U_i^\star,\widehat{\bm S}_{\eta,i}
\right] \\
&=
\mathbb E\left[
\mathbb E\{Y_i(t)\mid
\bm C_i,\bm U_i^\star,\widehat{\bm S}_{\eta,i}\}
\mid
\bm U_i^\star,\widehat{\bm S}_{\eta,i}
\right] \\
&=
\mathbb E\{Y_i(t)\mid
\bm U_i^\star,\widehat{\bm S}_{\eta,i}\},
\end{aligned}
\]
where the second equality uses
\(T_i\ci\bm C_i\mid
(\bm U_i^\star,\widehat{\bm S}_{\eta,i})\). Evaluating at \(T_i=t\) proves the
result.
\end{proof}

\subsection{Theorem \ref{thm:compact-downstream-bootstrap}}

\begin{reptheorem}{thm:compact-downstream-bootstrap}
Suppose \(O_1,\ldots,O_n\) are independent and identically distributed and
Assumption~\ref{assump:L-asymptotic-linearity} holds. Assume the downstream estimator applies ordinary least squares to the fixed-dimensional regression model in \eqref{eq:linear-potential-outcome-model}:
\begin{equation}
Y_i
=
\alpha_0
+
\alpha_T T_i
+
\bm{\alpha}_S^\top \bm S_i
+
T_i\bm{\alpha}_{TS}^\top \bm S_i
+
\varepsilon_i .
\end{equation}
Assume the corresponding population Gram matrix over \(\left(
1,\,
T_i,\,
{\bm S_{\eta,i}}^{\top},\,
T_i{\bm S_{\eta,i}}^{\top}
\right)^\top \) is nonsingular, the regressors have finite fourth moments,
and \(\mathbb E(Y_i^4)<\infty\). Suppose also that \(\bm{\psi}\) is a continuously
differentiable, basis-invariant function of the fitted downstream regression.
Then, for every bounded Lipschitz function
\(h:\mathbb R^q\to\mathbb R\),
\[
\left|
\mathbb E^*
h
\left\{
\sqrt n(\widehat{\bm{\psi}}^*-\widehat{\bm{\psi}})
\right\}
-
\mathbb E
h
\left\{
\sqrt n(\widehat{\bm{\psi}}-\bm{\psi})
\right\}
\right|
\overset{p}{\longrightarrow}
0.
\]
\end{reptheorem}
\begin{proof}
Let $\widehat{\mathcal S}_\eta^*$ be the score space recomputed from a
nonparametric bootstrap sample, and let
$P_{\widehat{\mathcal S}_\eta^*}$ denote its $L_2(\mathbb P_X)$-orthogonal
projection.

Let
\[
\bm Z_{\eta,i}
=
\left(
1,\,
T_i,\,
{\bm S_{\eta,i}}^{\top},\,
T_i{\bm S_{\eta,i}}^{\top}
\right)^\top,
\qquad
\Gamma_Z
=
\mathbb E
\left(
\bm Z_{\eta,i}{\bm Z_{\eta,i}}^{\top}
\right),
\]
and define
\[
\bm m_Z
=
\mathbb E
\left(
\bm Z_{\eta,i} Y_i
\right),
\qquad
\bm\alpha_\eta
=
\Gamma_Z^{-1}\bm m_Z,
\qquad
\varepsilon_{\eta,i}
=
Y_i-{\bm Z_{\eta,i}}^{\top}\bm\alpha_\eta.
\]
The corresponding empirical oracle quantities are
\[
\widetilde\Gamma_Z
=
\frac{1}{n}
\sum_{i=1}^n
\bm Z_{\eta,i}{\bm Z_{\eta,i}}^{\top},
\qquad
\widetilde{\bm m}_Z
=
\frac{1}{n}
\sum_{i=1}^n
\bm Z_{\eta,i} Y_i,
\]
\[
\widetilde{\bm\alpha}
=
\widetilde\Gamma_Z^{-1}
\widetilde{\bm m}_Z,
\qquad
\widetilde\varepsilon_i
=
Y_i-{\bm Z_{\eta,i}}^{\top}\widetilde{\bm\alpha}.
\]
Writing \(\bm{\psi}(\bm\alpha)\) for the reported target as a function of the
fitted downstream regression coefficients, the population target and oracle
sample estimator are
\[
\bm{\psi}
=
\bm{\psi}(\bm\alpha_\eta),
\qquad
\widetilde{\bm{\psi}}
=
\bm{\psi}(\widetilde{\bm\alpha}).
\]
Here, \(\dot\psi_{\bm\alpha}\) denotes the derivative of this target map
evaluated at \(\bm\alpha\). Define the population and empirical oracle
influence functions by
\[
{\bm{\varphi}}_\psi(O_i)
=
\dot\psi_{\bm\alpha_\eta}
\left[
\Gamma_Z^{-1}
\bm Z_{\eta,i}\varepsilon_{\eta,i}
\right]
\]
and
\[
\widehat{\bm{\varphi}}_{\psi i}
=
\dot\psi_{\widetilde{\bm\alpha}}
\left[
\widetilde\Gamma_Z^{-1}
\bm Z_{\eta,i}\widetilde\varepsilon_i
\right],
\]
respectively.

For the nonparametric bootstrap, let \(N_i^*\) be the number of times
observation \(i\) appears in the bootstrap sample, and define
\[
\widetilde\Gamma_Z^*
=
\frac{1}{n}
\sum_{i=1}^n
N_i^*
\bm Z_{\eta,i}{\bm Z_{\eta,i}}^{\top},
\qquad
\widetilde{\bm m}_Z^*
=
\frac{1}{n}
\sum_{i=1}^n
N_i^*
\bm Z_{\eta,i} Y_i,
\]
\[
\widetilde{\bm\alpha}^*
=
\left(
\widetilde\Gamma_Z^*
\right)^{-1}
\widetilde{\bm m}_Z^*,
\qquad
\widetilde{\bm{\psi}}^*
=
\bm{\psi}(\widetilde{\bm\alpha}^*).
\]

First, consider the oracle downstream estimator
\(\widetilde{\bm{\psi}}\), which fits the same downstream OLS procedure as \(\widehat{\bm{\psi}}\), using the same
observed sample, but replaces the estimated COMPACT scores with the population
COMPACT scores \(\bm S_{\eta,i}\). The sampling setup gives independent and identically
distributed observations. The theorem assumes that the population downstream
Gram matrix \(\Gamma_Z\) is nonsingular, that
\[
\mathbb E\|\bm Z_{\eta,i}\|^4<\infty,
\qquad
\mathbb E(Y_i^4)<\infty,
\]
and that \(\bm{\psi}\) is continuously differentiable in the fitted downstream
regression. These are the assumptions of
Lemma~\ref{lem:oracle-downstream-sampling-expansion}. Applying that lemma gives
\[
\sqrt n(\widetilde{\bm{\psi}}-\bm{\psi})
=
\frac{1}{\sqrt n}
\sum_{i=1}^n
{\bm{\varphi}}_\psi(O_i)
+
o_p(1),
\]
where
\[
\mathbb E\{{\bm{\varphi}}_\psi(O_i)\}=\bm 0,
\qquad
\mathbb E\|{\bm{\varphi}}_\psi(O_i)\|^2<\infty.
\]

Second, the same assumptions are exactly those required by
Lemma~\ref{lem:oracle-downstream-bootstrap-expansion}. Applying that lemma
gives
\[
\sqrt n(\widetilde{\bm{\psi}}^*-\widetilde{\bm{\psi}})
=
\frac{1}{\sqrt n}
\sum_{i=1}^n
(N_i^*-1)\widehat{\bm{\varphi}}_{\psi i}
+
o_{p^*}(1)
\]
in \(\mathbb P\)-probability. Thus, the second lemma supplies the conditional
bootstrap expansion of the oracle downstream estimator.

Third, the theorem assumes
Assumption~\ref{assump:L-asymptotic-linearity}. We have already verified the
conditions of
Lemmas~\ref{lem:oracle-downstream-sampling-expansion} and
\ref{lem:oracle-downstream-bootstrap-expansion}. Therefore, all assumptions of
Lemma~\ref{lem:joint-oracle-score-bootstrap} are satisfied. Applying that
lemma gives, for every bounded Lipschitz function \(g\),
\[
\begin{aligned}
&
\left|
\mathbb E^*
g
\left(
\sqrt n(\widetilde{\bm{\psi}}^*-\widetilde{\bm{\psi}}),
\operatorname{vec}
\left\{
\sqrt n
\left(
P_{\widehat{\mathcal S}_\eta^*}
-
P_{\widehat{\mathcal S}_\eta}
\right)
\right\}
\right)
\right.
\\
&\qquad\left.
-
\mathbb E
g
\left(
\sqrt n(\widetilde{\bm{\psi}}-\bm{\psi}),
\operatorname{vec}
\left\{
\sqrt n
\left(
P_{\widehat{\mathcal S}_\eta}
-
P_{\mathcal S_\eta}
\right)
\right\}
\right)
\right|
\overset{p}{\longrightarrow}
0.
\end{aligned}
\]
This lemma is needed because the oracle downstream estimator and the estimated
COMPACT score space are learned from the same observations and therefore need
not be asymptotically independent. It gives their joint bootstrap
approximation.

Fourth, the theorem assumes that \(\bm{\psi}\) is basis invariant. Together with
Assumption~\ref{assump:L-asymptotic-linearity} and the conditions already
verified for
Lemma~\ref{lem:oracle-downstream-sampling-expansion}, this satisfies all
assumptions of Lemma~\ref{lem:generated-score-ols-expansion}. Applying that
lemma gives a deterministic bounded linear map \(\mathcal B\) such that
\begin{equation}
\label{eq:theorem-generated-score-sampling}
\sqrt n(\widehat{\bm{\psi}}-\widetilde{\bm{\psi}})
=
\mathcal B
\left[
\sqrt n
\left(
P_{\widehat{\mathcal S}_\eta}
-
P_{\mathcal S_\eta}
\right)
\right]
+
\bm u_n,
\qquad
\bm u_n=o_p(1),
\end{equation}
and
\begin{equation}
\label{eq:theorem-generated-score-bootstrap}
\begin{aligned}
&
\sqrt n
\left\{
(\widehat{\bm{\psi}}^*-\widetilde{\bm{\psi}}^*)
-
(\widehat{\bm{\psi}}-\widetilde{\bm{\psi}})
\right\}
\\
&=
\mathcal B
\left[
\sqrt n
\left(
P_{\widehat{\mathcal S}_\eta^*}
-
P_{\widehat{\mathcal S}_\eta}
\right)
\right]
+
\bm u_n^*,
\qquad
\bm u_n^*=o_{p^*}(1)
\end{aligned}
\end{equation}
in \(\mathbb P\)-probability. These expansions quantify the additional
variation introduced by replacing the population COMPACT score space with its
estimated version.

We now combine the oracle and generated-score terms. The identity
\[
\widehat{\bm{\psi}}-\bm{\psi}
=
(\widetilde{\bm{\psi}}-\bm{\psi})
+
(\widehat{\bm{\psi}}-\widetilde{\bm{\psi}})
\]
and \eqref{eq:theorem-generated-score-sampling} give
\[
\begin{aligned}
\sqrt n(\widehat{\bm{\psi}}-\bm{\psi})
&=
\sqrt n(\widetilde{\bm{\psi}}-\bm{\psi})
+
\mathcal B
\left[
\sqrt n
\left(
P_{\widehat{\mathcal S}_\eta}
-
P_{\mathcal S_\eta}
\right)
\right]
+
\bm u_n.
\end{aligned}
\]
Define the leading sampling term
\[
\bm L_n
=
\sqrt n(\widetilde{\bm{\psi}}-\bm{\psi})
+
\mathcal B
\left[
\sqrt n
\left(
P_{\widehat{\mathcal S}_\eta}
-
P_{\mathcal S_\eta}
\right)
\right].
\]
Then
\begin{equation}
\label{eq:theorem-full-sampling-decomposition}
\sqrt n(\widehat{\bm{\psi}}-\bm{\psi})
=
\bm L_n+\bm u_n,
\qquad
\bm u_n=o_p(1).
\end{equation}

For the bootstrap estimator, the identity
\[
\begin{aligned}
\widehat{\bm{\psi}}^*-\widehat{\bm{\psi}}
&=
(\widetilde{\bm{\psi}}^*-\widetilde{\bm{\psi}})
+
\left\{
(\widehat{\bm{\psi}}^*-\widetilde{\bm{\psi}}^*)
-
(\widehat{\bm{\psi}}-\widetilde{\bm{\psi}})
\right\}
\end{aligned}
\]
and \eqref{eq:theorem-generated-score-bootstrap} give
\[
\begin{aligned}
\sqrt n(\widehat{\bm{\psi}}^*-\widehat{\bm{\psi}})
&=
\sqrt n(\widetilde{\bm{\psi}}^*-\widetilde{\bm{\psi}})
+
\mathcal B
\left[
\sqrt n
\left(
P_{\widehat{\mathcal S}_\eta^*}
-
P_{\widehat{\mathcal S}_\eta}
\right)
\right]
+
\bm u_n^*.
\end{aligned}
\]
Define the leading bootstrap term
\[
\bm L_n^*
=
\sqrt n(\widetilde{\bm{\psi}}^*-\widetilde{\bm{\psi}})
+
\mathcal B
\left[
\sqrt n
\left(
P_{\widehat{\mathcal S}_\eta^*}
-
P_{\widehat{\mathcal S}_\eta}
\right)
\right].
\]
Then
\begin{equation}
\label{eq:theorem-full-bootstrap-decomposition}
\sqrt n(\widehat{\bm{\psi}}^*-\widehat{\bm{\psi}})
=
\bm L_n^*+\bm u_n^*,
\qquad
\bm u_n^*=o_{p^*}(1)
\end{equation}
in \(\mathbb P\)-probability.

We next use
Lemma~\ref{lem:joint-oracle-score-bootstrap} to compare the laws of \(\bm L_n^*\)
and \(\bm L_n\). Let \(\operatorname{mat}(\bm v)\) reverse the vectorization used in
that lemma, and define
\[
\mathcal T(\bm u,\bm v)
=
\bm u+\mathcal B[\operatorname{mat}(\bm v)].
\]
The map \(\operatorname{mat}\) is linear, and \(\mathcal B\) is a deterministic
bounded linear map by
Lemma~\ref{lem:generated-score-ols-expansion}. Therefore, \(\mathcal T\) is
also linear. Moreover, for any \((\bm u,\bm v)\) and \((\bm u',\bm v')\),
\[
\begin{aligned}
&
\left\|
\mathcal T(\bm u,\bm v)-\mathcal T(\bm u',\bm v')
\right\|
\\
&=
\left\|
(\bm u-\bm u')
+
\mathcal B
\left[
\operatorname{mat}(\bm v-\bm v')
\right]
\right\|
\\
&\leq
\|\bm u-\bm u'\|
+
\|\mathcal B\|_{\mathrm{op}}
\left\|
\operatorname{mat}(\bm v-\bm v')
\right\|_{\mathrm F}.
\end{aligned}
\]
Because reversing vectorization preserves the Euclidean and Frobenius norms,
\[
\left\|
\operatorname{mat}(\bm v-\bm v')
\right\|_{\mathrm F}
=
\|\bm v-\bm v'\|.
\]
Hence,
\[
\left\|
\mathcal T(\bm u,\bm v)-\mathcal T(\bm u',\bm v')
\right\|
\leq
\|\bm u-\bm u'\|
+
\|\mathcal B\|_{\mathrm{op}}\|\bm v-\bm v'\|.
\]
Thus, \(\mathcal T\) is bounded and Lipschitz.

Let \(h:\mathbb R^q\to\mathbb R\) be bounded and Lipschitz. Because
\(\mathcal T\) is bounded and Lipschitz, the composition
\(g=h\circ\mathcal T\) is also bounded and Lipschitz. The two arguments of
\(g\) in Lemma~\ref{lem:joint-oracle-score-bootstrap} are mapped by
\(\mathcal T\) to \(\bm L_n^*\) and \(\bm L_n\), respectively. Applying that lemma
with \(g=h\circ\mathcal T\) therefore gives
\begin{equation}
\label{eq:theorem-leading-bootstrap-convergence}
\left|
\mathbb E^*h(\bm L_n^*)
-
\mathbb E h(\bm L_n)
\right|
\overset{p}{\longrightarrow}
0.
\end{equation}

It remains to show that the remainders in
\eqref{eq:theorem-full-sampling-decomposition} and
\eqref{eq:theorem-full-bootstrap-decomposition} do not affect this
bounded-Lipschitz convergence. Let \(M_h<\infty\) satisfy
\[
|h(\bm x)|\leq M_h
\]
for every \(\bm x\), and let \(L_h<\infty\) be a Lipschitz constant for \(h\).
The triangle inequality gives
\[
\begin{aligned}
&
\left|
\mathbb E^*h(\bm L_n^*+\bm u_n^*)
-
\mathbb E h(\bm L_n+\bm u_n)
\right|
\\
&\leq
\mathbb E^*
\left|
h(\bm L_n^*+\bm u_n^*)-h(\bm L_n^*)
\right|
+
\left|
\mathbb E^*h(\bm L_n^*)
-
\mathbb E h(\bm L_n)
\right|
+
\mathbb E
\left|
h(\bm L_n+\bm u_n)-h(\bm L_n)
\right|.
\end{aligned}
\]
The middle term converges to zero in probability by
\eqref{eq:theorem-leading-bootstrap-convergence}.

We next control the first term. Fix \(\varepsilon>0\). Splitting according to
whether \(\|\bm u_n^*\|\leq\varepsilon\) gives
\[
\begin{aligned}
&
\mathbb E^*
\left|
h(\bm L_n^*+\bm u_n^*)-h(\bm L_n^*)
\right|
\\
&=
\mathbb E^*
\left[
\left|
h(\bm L_n^*+\bm u_n^*)-h(\bm L_n^*)
\right|
\mathbf 1\{\|\bm u_n^*\|\leq\varepsilon\}
\right]
+
\mathbb E^*
\left[
\left|
h(\bm L_n^*+\bm u_n^*)-h(\bm L_n^*)
\right|
\mathbf 1\{\|\bm u_n^*\|>\varepsilon\}
\right].
\end{aligned}
\]
On the event \(\{\|\bm u_n^*\|\leq\varepsilon\}\), the Lipschitz property gives
\[
\left|
h(\bm L_n^*+\bm u_n^*)-h(\bm L_n^*)
\right|
\leq
L_h\|\bm u_n^*\|
\leq
L_h\varepsilon.
\]
On the event \(\{\|\bm u_n^*\|>\varepsilon\}\), boundedness gives
\[
\left|
h(\bm L_n^*+\bm u_n^*)-h(\bm L_n^*)
\right|
\leq
|h(\bm L_n^*+\bm u_n^*)|+|h(\bm L_n^*)|
\leq
2M_h.
\]
Therefore,
\begin{equation}
\label{eq:theorem-bootstrap-remainder-bound}
\begin{aligned}
&
\mathbb E^*
\left|
h(\bm L_n^*+\bm u_n^*)-h(\bm L_n^*)
\right|
\leq
L_h\varepsilon
+
2M_h\,
\mathbb P^*(\|\bm u_n^*\|>\varepsilon).
\end{aligned}
\end{equation}
Because \(\bm u_n^*=o_{p^*}(1)\) in \(\mathbb P\)-probability,
\[
\mathbb P^*(\|\bm u_n^*\|>\varepsilon)
\overset{p}{\longrightarrow}
0.
\]
To verify convergence of the left-hand side of
\eqref{eq:theorem-bootstrap-remainder-bound}, fix \(\delta>0\) and choose
\(\varepsilon>0\) sufficiently small that
\(L_h\varepsilon<\delta/2\). Then
\[
\begin{aligned}
&
\mathbb P
\left[
\mathbb E^*
\left|
h(\bm L_n^*+\bm u_n^*)-h(\bm L_n^*)
\right|
>
\delta
\right]
\\
&\leq
\mathbb P
\left[
2M_h
\mathbb P^*(\|\bm u_n^*\|>\varepsilon)
>
\frac{\delta}{2}
\right]
\longrightarrow
0.
\end{aligned}
\]
Hence,
\[
\mathbb E^*
\left|
h(\bm L_n^*+\bm u_n^*)-h(\bm L_n^*)
\right|
\overset{p}{\longrightarrow}
0.
\]

The same argument applies to the original-sample remainder. Specifically,
\[
\mathbb E
\left|
h(\bm L_n+\bm u_n)-h(\bm L_n)
\right|
\leq
L_h\varepsilon
+
2M_h\,
\mathbb P(\|\bm u_n\|>\varepsilon).
\]
Because \(\bm u_n=o_p(1)\),
\[
\mathbb P(\|\bm u_n\|>\varepsilon)
\longrightarrow
0.
\]
Choosing \(\varepsilon\) arbitrarily small therefore gives
\[
\mathbb E
\left|
h(\bm L_n+\bm u_n)-h(\bm L_n)
\right|
\longrightarrow
0.
\]

Combining the convergence of the three terms in the triangle-inequality bound
gives
\[
\left|
\mathbb E^*h(\bm L_n^*+\bm u_n^*)
-
\mathbb E h(\bm L_n+\bm u_n)
\right|
\overset{p}{\longrightarrow}
0.
\]
Finally, using
\eqref{eq:theorem-full-sampling-decomposition} and
\eqref{eq:theorem-full-bootstrap-decomposition}, we obtain
\[
\left|
\mathbb E^*
h
\left\{
\sqrt n(\widehat{\bm{\psi}}^*-\widehat{\bm{\psi}})
\right\}
-
\mathbb E
h
\left\{
\sqrt n(\widehat{\bm{\psi}}-\bm{\psi})
\right\}
\right|
\overset{p}{\longrightarrow}
0.
\]
This proves the theorem.
\end{proof}

\begin{lemma}
\label{lem:conditional-bootstrap-mean-consistency}
Let \(\bm V_1,\ldots,\bm V_n\) be independent and identically distributed
finite-dimensional random vectors satisfying
\(\mathbb E\|\bm V_i\|<\infty\). For multinomial bootstrap counts
\(N_1^*,\ldots,N_n^*\),
\[
\left\|
\frac{1}{n}
\sum_{i=1}^n
(N_i^*-1)\bm V_i
\right\|
=
o_{p^*}(1)
\]
in \(\mathbb P\)-probability.
\end{lemma}

\begin{proof}
We prove the result by decomposing each \(\bm V_i\) into a bounded component
and a tail remainder. For fixed \(K\), conditional Markov's inequality controls
the bounded component, while the integrability of \(\|\bm V_i\|\) controls the
tail remainder. We then combine these two bounds.

Fix \(K>0\), and write
\begin{equation}
\label{eq:conditional-bootstrap-truncation}
\bm V_i^{(K)}
=
\bm V_i\mathbf 1\{\|\bm V_i\|\leq K\},
\qquad
\bm R_i^{(K)}
=
\bm V_i-\bm V_i^{(K)}.
\end{equation}
Conditional on the original sample, the vectors
\(\bm V_i^{(K)}\) and \(\bm R_i^{(K)}\) are fixed. For multinomial bootstrap
counts,
\begin{equation}
\label{eq:conditional-bootstrap-count-covariance}
\mathbb E^*
\left\{
(N_i^*-1)(N_j^*-1)
\right\}
=
\mathbf 1\{i=j\}-\frac{1}{n}.
\end{equation}
Using \eqref{eq:conditional-bootstrap-count-covariance}, we obtain
\[
\begin{aligned}
&
\mathbb E^*
\left\|
\frac{1}{n}
\sum_{i=1}^n
(N_i^*-1)\bm V_i^{(K)}
\right\|^2
\\
&=
\frac{1}{n^2}
\sum_{i=1}^n
\sum_{j=1}^n
\mathbb E^*
\left\{
(N_i^*-1)(N_j^*-1)
\right\}
\left\langle
\bm V_i^{(K)},\bm V_j^{(K)}
\right\rangle
\\
&=
\frac{1}{n^2}
\sum_{i=1}^n
\sum_{j=1}^n
\left\{
\mathbf 1\{i=j\}-\frac{1}{n}
\right\}
\left\langle
\bm V_i^{(K)},\bm V_j^{(K)}
\right\rangle
\\
&=
\frac{1}{n^2}
\sum_{i=1}^n
\left\|
\bm V_i^{(K)}
\right\|^2
-
\frac{1}{n^3}
\left\langle
\sum_{i=1}^n\bm V_i^{(K)},
\sum_{j=1}^n\bm V_j^{(K)}
\right\rangle
\\
&=
\frac{1}{n^2}
\sum_{i=1}^n
\left\|
\bm V_i^{(K)}
\right\|^2
-
\frac{1}{n^3}
\left\|
\sum_{i=1}^n
\bm V_i^{(K)}
\right\|^2
\\
&\leq
\frac{1}{n^2}
\sum_{i=1}^n
\left\|
\bm V_i^{(K)}
\right\|^2.
\end{aligned}
\]
Because \(\|\bm V_i^{(K)}\|\leq K\),
\begin{equation}
\label{eq:conditional-bootstrap-bounded-second-moment}
\mathbb E^*
\left\|
\frac{1}{n}
\sum_{i=1}^n
(N_i^*-1)\bm V_i^{(K)}
\right\|^2
\leq
\frac{1}{n^2}
\sum_{i=1}^n K^2
=
\frac{K^2}{n}.
\end{equation}
The squared norm in
\eqref{eq:conditional-bootstrap-bounded-second-moment} is nonnegative and has
finite conditional expectation. Therefore, conditional Markov's inequality
gives, for every \(a>0\),
\begin{equation}
\label{eq:conditional-bootstrap-bounded-probability}
\begin{aligned}
&
\mathbb P^*
\left\{
\left\|
\frac{1}{n}
\sum_{i=1}^n
(N_i^*-1)\bm V_i^{(K)}
\right\|
>
a
\right\}
\\
&=
\mathbb P^*
\left\{
\left\|
\frac{1}{n}
\sum_{i=1}^n
(N_i^*-1)\bm V_i^{(K)}
\right\|^2
>
a^2
\right\}
\\
&\leq
\frac{1}{a^2}
\mathbb E^*
\left\|
\frac{1}{n}
\sum_{i=1}^n
(N_i^*-1)\bm V_i^{(K)}
\right\|^2
\\
&\leq
\frac{K^2}{na^2}.
\end{aligned}
\end{equation}
Thus, for each fixed \(K\), the conditional probability in
\eqref{eq:conditional-bootstrap-bounded-probability} converges to zero as
\(n\to\infty\).

We next control the remainder. From
\eqref{eq:conditional-bootstrap-truncation},
\[
\bm R_i^{(K)}
=
\bm V_i\mathbf 1\{\|\bm V_i\|>K\},
\qquad
\|\bm R_i^{(K)}\|
=
\|\bm V_i\|\mathbf 1\{\|\bm V_i\|>K\}.
\]
Because \(N_i^*\geq0\),
\(
|N_i^*-1|\leq N_i^*+1.
\)
Moreover, \(\mathbb E^*(N_i^*)=1\), so
\[
\mathbb E^*\!\left(|N_i^*-1|\right)
\leq
\mathbb E^*(N_i^*+1)
=
2.
\]
Conditional on the original sample, the vectors \(\bm R_i^{(K)}\) are fixed.
The triangle inequality therefore gives
\[
\begin{aligned}
&
\mathbb E^*
\left\|
\frac{1}{n}
\sum_{i=1}^n
(N_i^*-1)\bm R_i^{(K)}
\right\|
\\
&\leq
\frac{1}{n}
\sum_{i=1}^n
\mathbb E^*\!\left(
|N_i^*-1|
\right)
\|\bm R_i^{(K)}\|
\\
&\leq
\frac{2}{n}
\sum_{i=1}^n
\|\bm V_i\|
\mathbf 1\{\|\bm V_i\|>K\}.
\end{aligned}
\]
The norm on the left-hand side is nonnegative. Conditional Markov's inequality
therefore gives, for every \(a>0\),
\begin{equation}
\label{eq:conditional-bootstrap-remainder-probability}
\begin{aligned}
&
\mathbb P^*
\left\{
\left\|
\frac{1}{n}
\sum_{i=1}^n
(N_i^*-1)\bm R_i^{(K)}
\right\|
>
a
\right\}
\\
&\leq
\frac{2}{na}
\sum_{i=1}^n
\|\bm V_i\|
\mathbf 1\{\|\bm V_i\|>K\}.
\end{aligned}
\end{equation}

For each fixed \(K\), the variables
\(
\|\bm V_i\|\mathbf 1\{\|\bm V_i\|>K\}
\)
are independent and identically distributed. They also have finite
expectation because
\[
0
\leq
\|\bm V_i\|\mathbf 1\{\|\bm V_i\|>K\}
\leq
\|\bm V_i\|
\]
and \(\mathbb E\|\bm V_i\|<\infty\). Thus, the assumptions of the law of large
numbers are satisfied, and
\begin{equation}
\label{eq:conditional-bootstrap-tail-lln}
\frac{1}{n}
\sum_{i=1}^n
\|\bm V_i\|
\mathbf 1\{\|\bm V_i\|>K\}
\overset{p}{\longrightarrow}
\mathbb E
\left[
\|\bm V_i\|
\mathbf 1\{\|\bm V_i\|>K\}
\right].
\end{equation}
As \(K\to\infty\), the integrand on the right-hand side of
\eqref{eq:conditional-bootstrap-tail-lln} converges pointwise to zero and is
bounded by the integrable variable \(\|\bm V_i\|\). The assumptions of the
dominated convergence theorem are therefore satisfied, so
\begin{equation}
\label{eq:conditional-bootstrap-tail-dct}
\mathbb E
\left[
\|\bm V_i\|
\mathbf 1\{\|\bm V_i\|>K\}
\right]
\longrightarrow
0.
\end{equation}

It remains to combine the bounded and remainder components. By
\eqref{eq:conditional-bootstrap-truncation},
\begin{equation} \label{eq:V_plus_R}
\frac{1}{n}
\sum_{i=1}^n
(N_i^*-1)\bm V_i
=
\frac{1}{n}
\sum_{i=1}^n
(N_i^*-1)\bm V_i^{(K)}
+
\frac{1}{n}
\sum_{i=1}^n
(N_i^*-1)\bm R_i^{(K)}.
\end{equation}
Recall that if both
\(\|\bm A\|\leq\varepsilon/2\) and
\(\|\bm B\|\leq\varepsilon/2\), then
\[
\|\bm A+\bm B\|
\leq
\|\bm A\|+\|\bm B\|
\leq
\frac{\varepsilon}{2}
+
\frac{\varepsilon}{2}
=
\varepsilon.
\]
Therefore, if \(\|\bm A+\bm B\|>\varepsilon\), then \(\|\bm A\|> \varepsilon/2\) or
\(\|\bm B\|> \varepsilon/2\) (or both). For \eqref{eq:V_plus_R}, if the norm of the sum on the left exceeds \(\varepsilon\), then the norm of
at least one component on the right must exceed \(\varepsilon/2\). Consequently,
\[
\begin{aligned}
\left\{
\left\|
\frac{1}{n}
\sum_{i=1}^n
(N_i^*-1)\bm V_i
\right\|
>
\varepsilon
\right\}
&\subseteq
\left\{
\left\|
\frac{1}{n}
\sum_{i=1}^n
(N_i^*-1)\bm V_i^{(K)}
\right\|
>
\frac{\varepsilon}{2}
\right\}
\\
&\qquad\cup
\left\{
\left\|
\frac{1}{n}
\sum_{i=1}^n
(N_i^*-1)\bm R_i^{(K)}
\right\|
>
\frac{\varepsilon}{2}
\right\}.
\end{aligned}
\]
Taking conditional probabilities and applying the union bound gives
\[
\begin{aligned}
&
\mathbb P^*
\left\{
\left\|
\frac{1}{n}
\sum_{i=1}^n
(N_i^*-1)\bm V_i
\right\|
>
\varepsilon
\right\}
\\
&\leq
\mathbb P^*
\left\{
\left\|
\frac{1}{n}
\sum_{i=1}^n
(N_i^*-1)\bm V_i^{(K)}
\right\|
>
\frac{\varepsilon}{2}
\right\}
+
\mathbb P^*
\left\{
\left\|
\frac{1}{n}
\sum_{i=1}^n
(N_i^*-1)\bm R_i^{(K)}
\right\|
>
\frac{\varepsilon}{2}
\right\}.
\end{aligned}
\]
Applying
\eqref{eq:conditional-bootstrap-bounded-probability} with
\(a=\varepsilon/2\) gives
\[
\begin{aligned}
&
\mathbb P^*
\left\{
\left\|
\frac{1}{n}
\sum_{i=1}^n
(N_i^*-1)\bm V_i^{(K)}
\right\|
>
\frac{\varepsilon}{2}
\right\}
\\
&\leq
\frac{K^2}
{n(\varepsilon/2)^2}
=
\frac{4K^2}{n\varepsilon^2}.
\end{aligned}
\]
Similarly, applying
\eqref{eq:conditional-bootstrap-remainder-probability} with
\(a=\varepsilon/2\) gives
\[
\begin{aligned}
&
\mathbb P^*
\left\{
\left\|
\frac{1}{n}
\sum_{i=1}^n
(N_i^*-1)\bm R_i^{(K)}
\right\|
>
\frac{\varepsilon}{2}
\right\}
\\
&\leq
\frac{2}{n(\varepsilon/2)}
\sum_{i=1}^n
\|\bm V_i\|
\mathbf 1\{\|\bm V_i\|>K\}
\\
&=
\frac{4}{n\varepsilon}
\sum_{i=1}^n
\|\bm V_i\|
\mathbf 1\{\|\bm V_i\|>K\}.
\end{aligned}
\]
Substituting these two bounds into the union-bound inequality gives
\begin{equation}
\label{eq:conditional-bootstrap-combined-bound}
\begin{aligned}
&
\mathbb P^*
\left\{
\left\|
\frac{1}{n}
\sum_{i=1}^n
(N_i^*-1)\bm V_i
\right\|
>
\varepsilon
\right\}
\\
&\leq
\frac{4K^2}{n\varepsilon^2}
+
\frac{4}{n\varepsilon}
\sum_{i=1}^n
\|\bm V_i\|
\mathbf 1\{\|\bm V_i\|>K\}.
\end{aligned}
\end{equation}

To verify convergence in \(\mathbb P\)-probability explicitly, fix
\(\varepsilon>0\) and \(\delta>0\). By
\eqref{eq:conditional-bootstrap-tail-dct}, we may choose \(K\) sufficiently
large that
\[
\frac{4}{\varepsilon}
\mathbb E
\left[
\|\bm V_i\|
\mathbf 1\{\|\bm V_i\|>K\}
\right]
<
\frac{\delta}{4}.
\]
For this fixed \(K\), \eqref{eq:conditional-bootstrap-tail-lln} implies
\[
\mathbb P
\left\{
\frac{4}{n\varepsilon}
\sum_{i=1}^n
\|\bm V_i\|
\mathbf 1\{\|\bm V_i\|>K\}
<
\frac{\delta}{2}
\right\}
\longrightarrow
1.
\]
For the same fixed \(K\),
\[
\frac{4K^2}{n\varepsilon^2}
<
\frac{\delta}{2}
\]
for all sufficiently large \(n\). Therefore, with probability tending to one, the sum of the two terms is less than $\delta$. It follows from
\eqref{eq:conditional-bootstrap-combined-bound} that
\[
\mathbb P
\left[
\mathbb P^*
\left\{
\left\|
\frac{1}{n}
\sum_{i=1}^n
(N_i^*-1)\bm V_i
\right\|
>
\varepsilon
\right\}
>
\delta
\right]
\longrightarrow
0.
\]
Because this holds for every \(\varepsilon>0\) and \(\delta>0\),
\[
\left\|
\frac{1}{n}
\sum_{i=1}^n
(N_i^*-1)\bm V_i
\right\|
=
o_{p^*}(1)
\]
in \(\mathbb P\)-probability.
\end{proof}

\begin{lemma}
\label{lem:oracle-downstream-sampling-expansion}
Define the oracle regressor vector
\[
\bm Z_{\eta,i}
=
\left(
1,\,
T_i,\,
{\bm S_{\eta,i}}^{\top},\,
T_i{\bm S_{\eta,i}}^{\top}
\right)^\top
\]
and the population moment
\(
\Gamma_Z
=
\mathbb E
\left(
\bm Z_{\eta,i}{\bm Z_{\eta,i}}^\top
\right).\) Suppose the observed units are independent and identically distributed,
\(\Gamma_Z\) is nonsingular,
\(\mathbb E\|\bm Z_{\eta,i}\|^4<\infty\), and
\(\mathbb E(Y_i^4)<\infty\). If \(\bm{\psi}\) is continuously differentiable in
the fitted regression, then the oracle downstream estimator satisfies
\[
\sqrt n
\left(
\widetilde{\bm\alpha}-\bm\alpha_\eta
\right)
=
\frac{1}{\sqrt n}
\sum_{i=1}^n
\Gamma_Z^{-1}
\bm Z_{\eta,i}\varepsilon_{\eta,i}
+
o_p(1)
\]
and
\[
\sqrt n(\widetilde{\bm{\psi}}-\bm{\psi})
=
\frac{1}{\sqrt n}
\sum_{i=1}^n
{\bm{\varphi}}_\psi(O_i)
+
o_p(1),
\]
where
\(
\mathbb E\{{\bm{\varphi}}_\psi(O_i)\}=\bm 0
\)
and
\(
\mathbb E\|{\bm{\varphi}}_\psi(O_i)\|^2<\infty
\).
\end{lemma}

\begin{proof}
The observed units are independent and identically distributed by the lemma's
sampling condition. Thus, each ordinary law of large numbers and central limit
theorem below applies to independent random vectors with a common
distribution.

We first derive the ordinary least-squares expansion for the oracle estimator.
Let \(\bm m_Z
=
\mathbb E
\left(
\bm Z_{\eta,i} Y_i
\right).
\)
By assumption, \(\Gamma_Z\) is nonsingular. Therefore, the population
least-squares coefficient is
\(
\bm\alpha_\eta=\Gamma_Z^{-1}\bm m_Z
\).
Define the population regression residual by
\(
\varepsilon_{\eta,i}
=
Y_i-{\bm Z_{\eta,i}}^\top\bm\alpha_\eta
\).
The definition of \(\bm\alpha_\eta\) gives
\begin{equation} \label{eq:moment_zero}
\begin{aligned}
\mathbb E
\left(
\bm Z_{\eta,i}\varepsilon_{\eta,i}
\right)
&=
\mathbb E
\left(
\bm Z_{\eta,i} Y_i
\right)
-
\mathbb E
\left(
\bm Z_{\eta,i}{\bm Z_{\eta,i}}^\top
\right)
\bm\alpha_\eta
\\
&=
\bm m_Z
-
\Gamma_Z\Gamma_Z^{-1}\bm m_Z
=
0.
\end{aligned}
\end{equation}

The empirical oracle moments and coefficient are
\[
\widetilde\Gamma_Z
=
\frac{1}{n}
\sum_{i=1}^n
\bm Z_{\eta,i}{\bm Z_{\eta,i}}^\top,
\qquad
\widetilde{\bm m}_Z
=
\frac{1}{n}
\sum_{i=1}^n
\bm Z_{\eta,i} Y_i,
\qquad
\widetilde{\bm\alpha}
=
\widetilde\Gamma_Z^{-1}\widetilde{\bm m}_Z.
\]
The sample normal equations give
\(
\widetilde\Gamma_Z\widetilde{\bm\alpha}
=
\widetilde{\bm m}_Z
\).
Subtracting
\(
\widetilde\Gamma_Z\bm\alpha_\eta
\)
from both sides gives
\[
\begin{aligned}
\widetilde\Gamma_Z
\left(
\widetilde{\bm\alpha}-\bm\alpha_\eta
\right)
&=
\widetilde{\bm m}_Z
-
\widetilde\Gamma_Z\bm\alpha_\eta
\\
&=
\frac{1}{n}
\sum_{i=1}^n
\bm Z_{\eta,i}
\left(
Y_i-{\bm Z_{\eta,i}}^\top\bm\alpha_\eta
\right)
\\
&=
\frac{1}{n}
\sum_{i=1}^n
\bm Z_{\eta,i}\varepsilon_{\eta,i}.
\end{aligned}
\]
Multiplying by \(\sqrt n\widetilde\Gamma_Z^{-1}\) yields the exact identity
\begin{equation}
\label{eq:oracle-ols-exact}
\sqrt n
\left(
\widetilde{\bm\alpha}-\bm\alpha_\eta
\right)
=
\widetilde\Gamma_Z^{-1}
\frac{1}{\sqrt n}
\sum_{i=1}^n
\bm Z_{\eta,i}\varepsilon_{\eta,i}.
\end{equation}

The fourth-moment assumptions in the theorem give
\begin{equation}
\label{eq:oracle-fourth-moments}
\mathbb E\|\bm Z_{\eta,i}\|^4<\infty,
\qquad
\mathbb E(Y_i^4)<\infty.
\end{equation}
Moreover,
\[
\begin{aligned}
\|\bm Z_{\eta,i}\varepsilon_{\eta,i}\|^2
&=
\|\bm Z_{\eta,i}\|^2
\left|
Y_i-{\bm Z_{\eta,i}}^\top\bm\alpha_\eta
\right|^2
\end{aligned}
\]
Applying \(
|a-b|^2
\leq
2a^2+2b^2
\) gives
\[
\left|
Y_i-{\bm Z_{\eta,i}}^\top\bm\alpha_\eta
\right|^2
\leq
2Y_i^2
+
2\left|
{\bm Z_{\eta,i}}^\top\bm\alpha_\eta
\right|^2.
\]
Second, the Cauchy--Schwarz inequality implies
\[
\left|
{\bm Z_{\eta,i}}^\top\bm\alpha_\eta
\right|^2
\leq
\|\bm Z_{\eta,i}\|^2
\|\bm\alpha_\eta\|^2.
\]
Multiplying by \(\|\bm Z_{\eta,i}\|^2\) therefore yields
\[
\begin{aligned}
&
\|\bm Z_{\eta,i}\|^2
\left|
Y_i-{\bm Z_{\eta,i}}^\top\bm\alpha_\eta
\right|^2
\\
&\leq
2\|\bm Z_{\eta,i}\|^2Y_i^2
+
2\|\bm Z_{\eta,i}\|^2
\left|
{\bm Z_{\eta,i}}^\top\bm\alpha_\eta
\right|^2
\\
&\leq
2\|\bm Z_{\eta,i}\|^2Y_i^2
+
2\|\bm Z_{\eta,i}\|^4
\|\bm\alpha_\eta\|^2.
\end{aligned}
\]

The Cauchy--Schwarz inequality gives
\[
\mathbb E
\left(
\|\bm Z_{\eta,i}\|^2Y_i^2
\right)
\leq
\left\{
\mathbb E\|\bm Z_{\eta,i}\|^4
\right\}^{1/2}
\left\{
\mathbb E(Y_i^4)
\right\}^{1/2}
<
\infty.
\]
Consequently,
\begin{equation}
\label{eq:oracle-score-residual-second-moment}
\mathbb E\|\bm Z_{\eta,i}\varepsilon_{\eta,i}\|^2<\infty.
\end{equation}
The vectors \(\bm Z_{\eta,i}\varepsilon_{\eta,i}\) are independent and
identically distributed under the sampling setup, have mean zero by
\eqref{eq:moment_zero}, and have finite covariance by
\eqref{eq:oracle-score-residual-second-moment}. All assumptions of the
multivariate central limit theorem are therefore satisfied, and it gives
\(
n^{-1/2}\sum_{i=1}^n
\bm Z_{\eta,i}\varepsilon_{\eta,i}=O_p(1)
\). In addition, \eqref{eq:oracle-fourth-moments} implies
\[
\mathbb E
\left\|
\bm Z_{\eta,i}{\bm Z_{\eta,i}}^\top
\right\|_{\mathrm F}
=
\mathbb E\|\bm Z_{\eta,i}\|^2
\leq
\left\{
\mathbb E\|\bm Z_{\eta,i}\|^4
\right\}^{1/2}
<
\infty.
\]
Thus, the summands defining \(\widetilde\Gamma_Z\) are independent and
identically distributed and have finite first moments. The law of large
numbers therefore gives
\begin{equation}
\label{eq:oracle-gram-consistency}
\widetilde\Gamma_Z\overset{p}{\longrightarrow}\Gamma_Z.
\end{equation}
Because \(\Gamma_Z\) is nonsingular by the theorem and
\eqref{eq:oracle-gram-consistency} gives convergence to \(\Gamma_Z\), the
assumptions of continuity of matrix inversion are satisfied. Hence,
\(
\widetilde\Gamma_Z^{-1}\overset{p}{\longrightarrow}\Gamma_Z^{-1}
\). Using \eqref{eq:oracle-ols-exact}, add and subtract
\(
\Gamma_Z^{-1}
n^{-1/2}
\sum_{i=1}^n
\bm Z_{\eta,i}\varepsilon_{\eta,i}
\):
\[
\begin{aligned}
&
\widetilde\Gamma_Z^{-1}
\frac{1}{\sqrt n}
\sum_{i=1}^n
\bm Z_{\eta,i}\varepsilon_{\eta,i}
=
\Gamma_Z^{-1}
\frac{1}{\sqrt n}
\sum_{i=1}^n
\bm Z_{\eta,i}\varepsilon_{\eta,i}
+
\left(
\widetilde\Gamma_Z^{-1}
-
\Gamma_Z^{-1}
\right)
\frac{1}{\sqrt n}
\sum_{i=1}^n
\bm Z_{\eta,i}\varepsilon_{\eta,i}.
\end{aligned}
\]
The remainder satisfies
\[
\begin{aligned}
&
\left\|
\left(
\widetilde\Gamma_Z^{-1}
-
\Gamma_Z^{-1}
\right)
\frac{1}{\sqrt n}
\sum_{i=1}^n
\bm Z_{\eta,i}\varepsilon_{\eta,i}
\right\|
\\
&\leq
\left\|
\widetilde\Gamma_Z^{-1}
-
\Gamma_Z^{-1}
\right\|_{\mathrm{op}}
\left\|
\frac{1}{\sqrt n}
\sum_{i=1}^n
\bm Z_{\eta,i}\varepsilon_{\eta,i}
\right\|
\\
&=
o_p(1)O_p(1)
=
o_p(1).
\end{aligned}
\]
Thus,
\begin{equation}
\label{eq:oracle-ols-AL}
\sqrt n
\left(
\widetilde{\bm\alpha}-\bm\alpha_\eta
\right)
=
\frac{1}{\sqrt n}
\sum_{i=1}^n
\Gamma_Z^{-1}
\bm Z_{\eta,i}\varepsilon_{\eta,i}
+
\bm r_{\alpha,n},
\qquad
\|\bm r_{\alpha,n}\|=o_p(1).
\end{equation}

The leading sum in \eqref{eq:oracle-ols-AL} is \(O_p(1)\) by the
multivariate central limit theorem verified above. Therefore,
\begin{equation}
\label{eq:oracle-alpha-rate}
\left\|
\widetilde{\bm\alpha}-\bm\alpha_\eta
\right\|
=
O_p(n^{-1/2}),
\qquad
\left\|
\widetilde{\bm\alpha}-\bm\alpha_\eta
\right\|^2
=
O_p(n^{-1}).
\end{equation}

The lemma assumes that \(\bm{\psi}\) is continuously differentiable in the fitted
regression, and \eqref{eq:oracle-alpha-rate} gives
\(\widetilde{\bm\alpha}\overset{p}{\longrightarrow}\bm\alpha_\eta\).
The assumptions of the first-order Taylor expansion at
\(\bm\alpha_\eta\) are therefore satisfied, so
\[
\widetilde{\bm{\psi}}-\bm{\psi}
=
\dot\psi_{\bm\alpha_\eta}
\left[
\widetilde{\bm\alpha}-\bm\alpha_\eta
\right]
+
\bm r_{\psi,n},
\qquad
\frac{\|\bm r_{\psi,n}\|}
{\|\widetilde{\bm\alpha}-\bm\alpha_\eta\|}
\overset{p}{\longrightarrow}
0.
\]
Equation~\eqref{eq:oracle-alpha-rate} then gives
\(\sqrt n\|\bm r_{\psi,n}\|=o_p(1)\).

Multiplying the Taylor expansion by \(\sqrt n\) therefore gives
\[
\begin{aligned}
\sqrt n(\widetilde{\bm{\psi}}-\bm{\psi})
&=
\sqrt n\,
\dot\psi_{\bm\alpha_\eta}
\left[
\widetilde{\bm\alpha}-\bm\alpha_\eta
\right]
+
\sqrt n\,\bm r_{\psi,n}
\\
&=
\dot\psi_{\bm\alpha_\eta}
\left[
\sqrt n
\left(
\widetilde{\bm\alpha}-\bm\alpha_\eta
\right)
\right]
+
o_p(1),
\end{aligned}
\]
where the second equality uses the linearity of
\(\dot\psi_{\bm\alpha_\eta}\). Substituting \eqref{eq:oracle-ols-AL} gives
\[
\begin{aligned}
\sqrt n(\widetilde{\bm{\psi}}-\bm{\psi})
&=
\dot\psi_{\bm\alpha_\eta}
\left[
\frac{1}{\sqrt n}
\sum_{i=1}^n
\Gamma_Z^{-1}
\bm Z_{\eta,i}\varepsilon_{\eta,i}
+
\bm r_{\alpha,n}
\right]
+
o_p(1)
\\
&=
\dot\psi_{\bm\alpha_\eta}
\left[
\frac{1}{\sqrt n}
\sum_{i=1}^n
\Gamma_Z^{-1}
\bm Z_{\eta,i}\varepsilon_{\eta,i}
\right]
+
\dot\psi_{\bm\alpha_\eta}
\left[
\bm r_{\alpha,n}
\right]
+
o_p(1),
\end{aligned}
\]
where the second equality again uses linearity. Because the derivative is
bounded,
\[
\left\|
\dot\psi_{\bm\alpha_\eta}
\left[
\bm r_{\alpha,n}
\right]
\right\|
\leq
\left\|
\dot\psi_{\bm\alpha_\eta}
\right\|_{\mathrm{op}}
\|\bm r_{\alpha,n}\|
=
o_p(1).
\]
Finally, linearity allows the derivative to pass through the scalar multiple
and finite sum:
\[
\begin{aligned}
\dot\psi_{\bm\alpha_\eta}
\left[
\frac{1}{\sqrt n}
\sum_{i=1}^n
\Gamma_Z^{-1}
\bm Z_{\eta,i}\varepsilon_{\eta,i}
\right]
&=
\frac{1}{\sqrt n}
\sum_{i=1}^n
\dot\psi_{\bm\alpha_\eta}
\left[
\Gamma_Z^{-1}
\bm Z_{\eta,i}\varepsilon_{\eta,i}
\right]
\\
&=
\frac{1}{\sqrt n}
\sum_{i=1}^n
{\bm{\varphi}}_\psi(O_i).
\end{aligned}
\]
Thus,
\begin{equation}
\label{eq:oracle-psi-AL}
\sqrt n(\widetilde{\bm{\psi}}-\bm{\psi})
=
\frac{1}{\sqrt n}
\sum_{i=1}^n
{\bm{\varphi}}_\psi(O_i)
+
o_p(1),
\end{equation}
where
\[
{\bm{\varphi}}_\psi(O_i)
=
\dot\psi_{\bm\alpha_\eta}
\left[
\Gamma_Z^{-1}
\bm Z_{\eta,i}\varepsilon_{\eta,i}
\right].
\]
Because \(\dot\psi_{\bm\alpha_\eta}\) is a deterministic linear map, it can
pass through the expectation:
\[
\begin{aligned}
\mathbb E\{{\bm{\varphi}}_\psi(O_i)\}
&=
\dot\psi_{\bm\alpha_\eta}
\left[
\mathbb E
\left\{
\Gamma_Z^{-1}
\bm Z_{\eta,i}\varepsilon_{\eta,i}
\right\}
\right]
\\
&=
\dot\psi_{\bm\alpha_\eta}
\left[
\Gamma_Z^{-1}
\mathbb E
\left(
\bm Z_{\eta,i}\varepsilon_{\eta,i}
\right)
\right]\\
&= 0,
\end{aligned}
\]
where the last equality follows from \eqref{eq:moment_zero}.

Moreover,
\[
\begin{aligned}
\|{\bm{\varphi}}_\psi(O_i)\|
&=
\left\|
\dot\psi_{\bm\alpha_\eta}
\left[
\Gamma_Z^{-1}
\bm Z_{\eta,i}\varepsilon_{\eta,i}
\right]
\right\| \\
&\leq
\|\dot\psi_{\bm\alpha_\eta}\|_{\mathrm{op}}
\left\|
\Gamma_Z^{-1}
\bm Z_{\eta,i}\varepsilon_{\eta,i}
\right\| \\
&\leq
\|\dot\psi_{\bm\alpha_\eta}\|_{\mathrm{op}}
\|\Gamma_Z^{-1}\|_{\mathrm{op}}
\|\bm Z_{\eta,i}\varepsilon_{\eta,i}\|
\end{aligned}
\]
Squaring and taking expectations gives
\begin{equation}
\label{eq:oracle-influence-second-moment}
\mathbb E\|{\bm{\varphi}}_\psi(O_i)\|^2
\leq
\|\dot\psi_{\bm\alpha_\eta}\|_{\mathrm{op}}^2
\|\Gamma_Z^{-1}\|_{\mathrm{op}}^2
\mathbb E\|\bm Z_{\eta,i}\varepsilon_{\eta,i}\|^2
<
\infty.
\end{equation}
The first two factors are finite because
\(\dot\psi_{\bm\alpha_\eta}\) is bounded and \(\Gamma_Z\) is a
finite-dimensional nonsingular matrix; the final factor is finite by
\eqref{eq:oracle-score-residual-second-moment}.

This proves the lemma.
\end{proof}

\begin{lemma}
\label{lem:oracle-downstream-bootstrap-expansion}
Under the conditions of
Lemma~\ref{lem:oracle-downstream-sampling-expansion}, the nonparametric
bootstrap oracle estimator satisfies
\[
\sqrt n(\widetilde{\bm{\psi}}^*-\widetilde{\bm{\psi}})
=
\frac{1}{\sqrt n}
\sum_{i=1}^n
(N_i^*-1)\widehat{\bm{\varphi}}_{\psi i}
+
o_{p^*}(1)
\]
in \(\mathbb P\)-probability, where \(\widehat{\bm{\varphi}}_{\psi i}
=
\dot\psi_{\widetilde{\bm\alpha}}
\left[
\widetilde\Gamma_Z^{-1}
\bm Z_{\eta,i}\widetilde\varepsilon_i
\right]\).
\end{lemma}

\begin{proof}
The conditions of Lemma~\ref{lem:oracle-downstream-sampling-expansion} hold by
the present lemma's hypothesis. We may therefore use
\eqref{eq:oracle-gram-consistency}, \eqref{eq:oracle-alpha-rate}, and the moment
bounds established in its proof.

For the bootstrap sample, let \(N_i^*\) be the number of times observation
\(i\) appears. Define
\[
\widetilde\Gamma_Z^*
=
\frac{1}{n}
\sum_{i=1}^n
N_i^*\bm Z_{\eta,i}{\bm Z_{\eta,i}}^\top,
\qquad
\widetilde{\bm m}_Z^*
=
\frac{1}{n}
\sum_{i=1}^n
N_i^*\bm Z_{\eta,i} Y_i,
\]
and
\(
\widetilde{\bm\alpha}^*
=
\left(\widetilde\Gamma_Z^*\right)^{-1}\widetilde{\bm m}_Z^*
\).
The bootstrap estimator \(\widetilde{\bm\alpha}^*\) minimizes the bootstrap
least-squares objective
\[
\frac{1}{n}
\sum_{i=1}^n
N_i^*
\left(
Y_i-{\bm Z_{\eta,i}}^{\top}\bm\alpha
\right)^2.
\]
Differentiating with respect to \(\bm\alpha\) and setting the derivative to
zero gives the bootstrap normal equations:
\begin{equation} \label{eq:bootstrap_normal}
\frac{1}{n}
\sum_{i=1}^n
N_i^*\bm Z_{\eta,i}
\left(
Y_i-{\bm Z_{\eta,i}}^{\top}\widetilde{\bm\alpha}^*
\right)
=
0.
\end{equation}
Add and subtract
\({\bm Z_{\eta,i}}^{\top}\widetilde{\bm\alpha}\)
inside the residual:
\[
\begin{aligned}
0
&=
\frac{1}{n}
\sum_{i=1}^n
N_i^*\bm Z_{\eta,i}
\left[
Y_i
-
{\bm Z_{\eta,i}}^{\top}\widetilde{\bm\alpha}
-
{\bm Z_{\eta,i}}^{\top}
\left(
\widetilde{\bm\alpha}^*
-
\widetilde{\bm\alpha}
\right)
\right]
\\
&=
\frac{1}{n}
\sum_{i=1}^n
N_i^*\bm Z_{\eta,i}
\left(
Y_i-{\bm Z_{\eta,i}}^{\top}\widetilde{\bm\alpha}
\right)
\\
&\quad-
\left\{
\frac{1}{n}
\sum_{i=1}^n
N_i^*\bm Z_{\eta,i}{\bm Z_{\eta,i}}^{\top}
\right\}
\left(
\widetilde{\bm\alpha}^*
-
\widetilde{\bm\alpha}
\right).
\end{aligned}
\]
Moving the second term to the other side gives
\[
\widetilde\Gamma_Z^*
\left(
\widetilde{\bm\alpha}^*
-
\widetilde{\bm\alpha}
\right)
=
\frac{1}{n}
\sum_{i=1}^n
N_i^*\bm Z_{\eta,i}
\left(
Y_i-{\bm Z_{\eta,i}}^{\top}\widetilde{\bm\alpha}
\right).
\]
The original-sample normal equation is
\begin{equation}
\label{eq:oracle-sample-normal-equation}
\frac{1}{n}
\sum_{i=1}^n
\bm Z_{\eta,i}
\left(
Y_i-{\bm Z_{\eta,i}}^{\top}\widetilde{\bm\alpha}
\right)
=
0.
\end{equation}
Subtracting the original-sample equation from the bootstrap equation \eqref{eq:bootstrap_normal} gives
\[
\begin{aligned}
0
&=
\frac{1}{n}
\sum_{i=1}^n
\left[
N_i^*\bm Z_{\eta,i}
\left(
Y_i-{\bm Z_{\eta,i}}^{\top}\widetilde{\bm\alpha}^*
\right)
-
\bm Z_{\eta,i}
\left(
Y_i-{\bm Z_{\eta,i}}^{\top}\widetilde{\bm\alpha}
\right)
\right].
\end{aligned}
\]
Let \( \widetilde\varepsilon_i
=
Y_i-{\bm Z_{\eta,i}}^{\top}\widetilde{\bm\alpha} \). Because
\[
Y_i-{\bm Z_{\eta,i}}^{\top}\widetilde{\bm\alpha}^*
=
\widetilde\varepsilon_i
-
{\bm Z_{\eta,i}}^{\top}
\left(
\widetilde{\bm\alpha}^*
-
\widetilde{\bm\alpha}
\right),
\]
substitution gives
\[
\begin{aligned}
0
&=
\frac{1}{n}
\sum_{i=1}^n
\left[
N_i^*\bm Z_{\eta,i}
\left\{
\widetilde\varepsilon_i
-
{\bm Z_{\eta,i}}^{\top}
\left(
\widetilde{\bm\alpha}^*
-
\widetilde{\bm\alpha}
\right)
\right\}
-
\bm Z_{\eta,i}\widetilde\varepsilon_i
\right]
\\
&=
\frac{1}{n}
\sum_{i=1}^n
(N_i^*-1)\bm Z_{\eta,i}\widetilde\varepsilon_i
-
\left\{
\frac{1}{n}
\sum_{i=1}^n
N_i^*\bm Z_{\eta,i}{\bm Z_{\eta,i}}^{\top}
\right\}
\left(
\widetilde{\bm\alpha}^*
-
\widetilde{\bm\alpha}
\right)
\\
&=
\frac{1}{n}
\sum_{i=1}^n
(N_i^*-1)\bm Z_{\eta,i}\widetilde\varepsilon_i
-
\widetilde\Gamma_Z^*
\left(
\widetilde{\bm\alpha}^*
-
\widetilde{\bm\alpha}
\right).
\end{aligned}
\]
Rearranging, multiplying by
\(\left(\widetilde\Gamma_Z^*\right)^{-1}\), and multiplying by \(\sqrt n\)
gives
\begin{equation} \label{eq:bootstrap-oracle-ols-exact}
\sqrt n
\left(
\widetilde{\bm\alpha}^*
-
\widetilde{\bm\alpha}
\right)
=
\left(\widetilde\Gamma_Z^*\right)^{-1}
\frac{1}{\sqrt n}
\sum_{i=1}^n
(N_i^*-1)\bm Z_{\eta,i}\widetilde\varepsilon_i.
\end{equation}

Conditional on the data, the bootstrap Gram matrix satisfies
\[
\begin{aligned}
\widetilde\Gamma_Z^*
-
\widetilde\Gamma_Z
&=
\frac{1}{n}
\sum_{i=1}^n
N_i^*
\bm Z_{\eta,i}{\bm Z_{\eta,i}}^{\top}
-
\frac{1}{n}
\sum_{i=1}^n
\bm Z_{\eta,i}{\bm Z_{\eta,i}}^{\top}
\\
&=
\frac{1}{n}
\sum_{i=1}^n
(N_i^*-1)
\bm Z_{\eta,i}{\bm Z_{\eta,i}}^{\top}.
\end{aligned}
\]
Let
\(
A_i=\bm Z_{\eta,i}{\bm Z_{\eta,i}}^\top
\).
The variables \(\|\bm Z_{\eta,i}\|^4\) are independent and identically
distributed under the conditions inherited from
Lemma~\ref{lem:oracle-downstream-sampling-expansion}, and their expectation is
finite by \eqref{eq:oracle-fourth-moments}. The law of large numbers therefore
gives
\(
n^{-1}\sum_{i=1}^n\|\bm Z_{\eta,i}\|^4=O_p(1)
\).
For multinomial bootstrap counts,
\(
\mathbb E^*\{(N_i^*-1)(N_j^*-1)\}
=
\mathbf 1\{i=j\}-1/n
\).
Using the Frobenius inner product therefore gives
\begin{equation} \label{eq:Z_Op1}
\begin{aligned}
\mathbb E^*
\left\|
\widetilde\Gamma_Z^*-\widetilde\Gamma_Z
\right\|_{\mathrm F}^2
&=
\frac{1}{n^2}
\sum_{i=1}^n
\sum_{j=1}^n
\left\{
\mathbf 1\{i=j\}-\frac{1}{n}
\right\}
\langle A_i,A_j\rangle_{\mathrm F}
\\
&=
\frac{1}{n^2}
\sum_{i=1}^n
\|A_i\|_{\mathrm F}^2
-
\frac{1}{n^3}
\left\langle
\sum_{i=1}^n A_i,\,
\sum_{j=1}^n A_j
\right\rangle_{\mathrm F}\\
&=
\frac{1}{n^2}
\sum_{i=1}^n
\|A_i\|_{\mathrm F}^2
-
\frac{1}{n^3}
\left\|
\sum_{i=1}^n A_i
\right\|_{\mathrm F}^2
\\
&\leq
\frac{1}{n^2}
\sum_{i=1}^n
\|\bm Z_{\eta,i}\|^4
=
O_p(n^{-1}),
\end{aligned}
\end{equation}
where the final equality uses the law-of-large-numbers rate established
immediately before this calculation.

We now apply conditional Markov's inequality with $\mathbb E^*
\left[
n
\left\|
\widetilde\Gamma_Z^*
-
\widetilde\Gamma_Z
\right\|_{\mathrm F}^2 \right]= O_p(1)$ from \eqref{eq:Z_Op1} to conclude that
\[
\sqrt n
\left\|
\widetilde\Gamma_Z^*
-
\widetilde\Gamma_Z
\right\|_{\mathrm F}
=
O_{p^*}(1)
\]
in \(\mathbb P\)-probability, or equivalently,
\[
\left\|
\widetilde\Gamma_Z^*
-
\widetilde\Gamma_Z
\right\|_{\mathrm F}
=
O_{p^*}(n^{-1/2}) =o_{p^*}(1)
\]
in \(\mathbb P\)-probability. Together with
\eqref{eq:oracle-gram-consistency}, the triangle inequality gives
\[
\left\|
\widetilde\Gamma_Z^*-
\Gamma_Z
\right\|_{\mathrm F}
\leq
\left\|
\widetilde\Gamma_Z^*-
\widetilde\Gamma_Z
\right\|_{\mathrm F}
+
\left\|
\widetilde\Gamma_Z-
\Gamma_Z
\right\|_{\mathrm F}
=
o_{p^*}(1)
\]
in \(\mathbb P\)-probability. Thus, both bootstrap and original-sample Gram
matrices converge to the nonsingular matrix \(\Gamma_Z\). The assumptions of
continuity of matrix inversion are therefore satisfied, and
\begin{equation}
\label{eq:bootstrap-oracle-inverse-consistency}
\left(\widetilde\Gamma_Z^*\right)^{-1}
-
\widetilde\Gamma_Z^{-1}
=
o_{p^*}(1)
\end{equation}
in \(\mathbb P\)-probability.

We also need to control the bootstrap sum in
\eqref{eq:bootstrap-oracle-ols-exact}. The same multinomial covariance
calculation gives
\[
\begin{aligned}
&
\mathbb E^*
\left\|
\frac{1}{\sqrt n}
\sum_{i=1}^n
(N_i^*-1)
\bm Z_{\eta,i}\widetilde\varepsilon_i
\right\|^2
\\
&=
\frac{1}{n}
\sum_{i=1}^n
\sum_{j=1}^n
\mathbb E^*
\left\{
(N_i^*-1)(N_j^*-1)
\right\}
\left\langle
\bm Z_{\eta,i}\widetilde\varepsilon_i,\,
\bm Z_{\eta,j}\widetilde\varepsilon_j
\right\rangle
\\
&=
\frac{1}{n}
\sum_{i=1}^n
\sum_{j=1}^n
\left\{
\mathbf 1\{i=j\}-\frac{1}{n}
\right\}
\left\langle
\bm Z_{\eta,i}\widetilde\varepsilon_i,\,
\bm Z_{\eta,j}\widetilde\varepsilon_j
\right\rangle
\\
&=
\frac{1}{n}
\sum_{i=1}^n
\left\|
\bm Z_{\eta,i}\widetilde\varepsilon_i
\right\|^2
-
\frac{1}{n^2}
\left\|
\sum_{i=1}^n
\bm Z_{\eta,i}\widetilde\varepsilon_i
\right\|^2.
\end{aligned}
\]
The original-sample normal equations give
\(
\sum_{i=1}^n
\bm Z_{\eta,i}\widetilde\varepsilon_i
=
0.
\)
Consequently, the second term vanishes, and
\begin{equation}
\label{eq:bootstrap-score-sum-second-moment}
\mathbb E^*
\left\|
\frac{1}{\sqrt n}
\sum_{i=1}^n
(N_i^*-1)
\bm Z_{\eta,i}\widetilde\varepsilon_i
\right\|^2
=
\frac{1}{n}
\sum_{i=1}^n
\left\|
\bm Z_{\eta,i}\widetilde\varepsilon_i
\right\|^2.
\end{equation}
Because
\(
\widetilde\varepsilon_i
=
\varepsilon_{\eta,i}
-
{\bm Z_{\eta,i}}^\top
(\widetilde{\bm\alpha}-\bm\alpha_\eta)
\),
\[
\begin{aligned}
\|\bm Z_{\eta,i}\widetilde\varepsilon_i\|^2
&=
\left\|
\bm Z_{\eta,i}\varepsilon_{\eta,i}
-
\bm Z_{\eta,i}{\bm Z_{\eta,i}}^{\top}
\left(
\widetilde{\bm\alpha}-\bm\alpha_\eta
\right)
\right\|^2
\\
&\leq
2\|\bm Z_{\eta,i}\varepsilon_{\eta,i}\|^2
+
2
\left\|
\bm Z_{\eta,i}{\bm Z_{\eta,i}}^{\top}
\left(
\widetilde{\bm\alpha}-\bm\alpha_\eta
\right)
\right\|^2
\\
&\leq
2\|\bm Z_{\eta,i}\varepsilon_{\eta,i}\|^2
+
2\|\bm Z_{\eta,i}\|^4
\|\widetilde{\bm\alpha}-\bm\alpha_\eta\|^2.
\end{aligned}
\]
Averaging this inequality over \(i\) gives
\begin{equation} \label{eq:two_upper_bounds}
\begin{aligned}
\frac{1}{n}
\sum_{i=1}^n
\|\bm Z_{\eta,i}\widetilde\varepsilon_i\|^2
&\leq
\frac{2}{n}
\sum_{i=1}^n
\|\bm Z_{\eta,i}\varepsilon_{\eta,i}\|^2
\\
&\quad+
2\|\widetilde{\bm\alpha}-\bm\alpha_\eta\|^2
\frac{1}{n}
\sum_{i=1}^n
\|\bm Z_{\eta,i}\|^4.
\end{aligned}
\end{equation}
The summands in the first empirical average are independent and identically
distributed and have finite expectation by
\eqref{eq:oracle-score-residual-second-moment}. The law of large numbers
therefore gives
\[
\frac{1}{n}
\sum_{i=1}^n
\|\bm Z_{\eta,i}\varepsilon_{\eta,i}\|^2
=
O_p(1).
\]
Similarly, \eqref{eq:oracle-fourth-moments} supplies the finite expectation
required for the law of large numbers to give
\[
\frac{1}{n}
\sum_{i=1}^n
\|\bm Z_{\eta,i}\|^4
=
O_p(1).
\]
Finally, \eqref{eq:oracle-alpha-rate} gives
\(
\|\widetilde{\bm\alpha}-\bm\alpha_\eta\|^2=O_p(n^{-1})
\). Hence, the second term in \eqref{eq:two_upper_bounds} is
\(O_p(n^{-1})O_p(1)=O_p(n^{-1})\), whereas the first is \(O_p(1)\).
Consequently,
\begin{equation}
\label{eq:bootstrap-residual-average-rate}
\frac{1}{n}
\sum_{i=1}^n
\|\bm Z_{\eta,i}\widetilde\varepsilon_i\|^2
=
O_p(1).
\end{equation}
The squared norm in \eqref{eq:bootstrap-score-sum-second-moment} is
nonnegative, and its conditional expectation is \(O_p(1)\) by
\eqref{eq:bootstrap-residual-average-rate}. The assumptions of conditional
Markov's inequality are therefore satisfied, and it gives
\[
\begin{aligned}
&
\mathbb P^*
\left\{
\left\|
\frac{1}{\sqrt n}
\sum_{i=1}^n
(N_i^*-1)
\bm Z_{\eta,i}\widetilde\varepsilon_i
\right\|
>
M
\right\}
\\
&=
\mathbb P^*
\left\{
\left\|
\frac{1}{\sqrt n}
\sum_{i=1}^n
(N_i^*-1)
\bm Z_{\eta,i}\widetilde\varepsilon_i
\right\|^2
>
M^2
\right\}
\\
&\leq
\frac{1}{M^2}
\mathbb E^*
\left\|
\frac{1}{\sqrt n}
\sum_{i=1}^n
(N_i^*-1)
\bm Z_{\eta,i}\widetilde\varepsilon_i
\right\|^2
\\
&=
\frac{1}{nM^2}
\sum_{i=1}^n
\left\|
\bm Z_{\eta,i}\widetilde\varepsilon_i
\right\|^2,
\end{aligned}
\]
where the final equality follows from
\eqref{eq:bootstrap-score-sum-second-moment}. By
\eqref{eq:bootstrap-residual-average-rate},
\[
\frac{1}{n}
\sum_{i=1}^n
\left\|
\bm Z_{\eta,i}\widetilde\varepsilon_i
\right\|^2
=
O_p(1).
\]
Therefore,
\begin{equation}
\label{eq:bootstrap-score-sum-rate}
\left\|
\frac{1}{\sqrt n}
\sum_{i=1}^n
(N_i^*-1)
\bm Z_{\eta,i}\widetilde\varepsilon_i
\right\|
=
O_{p^*}(1)
\end{equation}
in \(\mathbb P\)-probability. By
\eqref{eq:bootstrap-oracle-inverse-consistency} and
\eqref{eq:bootstrap-score-sum-rate},
\[
\begin{aligned}
&
\left\|
\left\{
(\widetilde\Gamma_Z^*)^{-1}
-
\widetilde\Gamma_Z^{-1}
\right\}
\frac{1}{\sqrt n}
\sum_{i=1}^n
(N_i^*-1)
\bm Z_{\eta,i}\widetilde\varepsilon_i
\right\|
\\
&\leq
\left\|
(\widetilde\Gamma_Z^*)^{-1}
-
\widetilde\Gamma_Z^{-1}
\right\|_{\mathrm{op}}
\left\|
\frac{1}{\sqrt n}
\sum_{i=1}^n
(N_i^*-1)
\bm Z_{\eta,i}\widetilde\varepsilon_i
\right\|
\\
&=
o_{p^*}(1)O_{p^*}(1)
=
o_{p^*}(1).
\end{aligned}
\]
Thus, replacing
\(
(\widetilde\Gamma_Z^*)^{-1}
\)
by \(\widetilde\Gamma_Z^{-1}\) in
\eqref{eq:bootstrap-oracle-ols-exact} changes its right-hand side by
\(o_{p^*}(1)\). Equation~\eqref{eq:oracle-gram-consistency} and the
nonsingularity of \(\Gamma_Z\) give
\(\|\widetilde\Gamma_Z^{-1}\|_{\mathrm{op}}=O_p(1)\), and
\eqref{eq:bootstrap-oracle-inverse-consistency} consequently gives
\(\|(\widetilde\Gamma_Z^*)^{-1}\|_{\mathrm{op}}=O_{p^*}(1)\) in
\(\mathbb P\)-probability. Combining this rate with
\eqref{eq:bootstrap-score-sum-rate} in
\eqref{eq:bootstrap-oracle-ols-exact} gives
\(
\widetilde{\bm\alpha}^*-\widetilde{\bm\alpha}
=
O_{p^*}(n^{-1/2})
\).

Lemma~\ref{lem:oracle-downstream-sampling-expansion} assumes that \(\bm{\psi}\) is
continuously differentiable. Moreover,
\eqref{eq:oracle-alpha-rate} gives
\(\widetilde{\bm\alpha}\overset{p}{\longrightarrow}\bm\alpha_\eta\),
and the preceding conditional rate gives
\(\widetilde{\bm\alpha}^*-\widetilde{\bm\alpha}=O_{p^*}(n^{-1/2})\).
Thus, with probability tending to one, both coefficient vectors lie in a
neighborhood on which the derivative of \(\bm{\psi}\) is continuous. The
assumptions of the conditional first-order Taylor expansion at
\(\widetilde{\bm\alpha}\) are therefore satisfied, and
\[
\widetilde{\bm{\psi}}^*-\widetilde{\bm{\psi}}
=
\dot\psi_{\widetilde{\bm\alpha}}
\left[
\widetilde{\bm\alpha}^*-
\widetilde{\bm\alpha}
\right]
+
\bm r_{\psi,n}^*,
\qquad
\sqrt n\|\bm r_{\psi,n}^*\|
=
o_{p^*}(1)
\]
in \(\mathbb P\)-probability. Multiplying by \(\sqrt n\) gives
\[
\begin{aligned}
\sqrt n
\left(
\widetilde{\bm{\psi}}^*-\widetilde{\bm{\psi}}
\right)
&=
\sqrt n\,
\dot\psi_{\widetilde{\bm\alpha}}
\left[
\widetilde{\bm\alpha}^*
-
\widetilde{\bm\alpha}
\right]
+
\sqrt n\,\bm r_{\psi,n}^*
\\
&=
\dot\psi_{\widetilde{\bm\alpha}}
\left[
\sqrt n
\left(
\widetilde{\bm\alpha}^*
-
\widetilde{\bm\alpha}
\right)
\right]
+
o_{p^*}(1),
\end{aligned}
\]
where the second equality uses the linearity of
\(\dot\psi_{\widetilde{\bm\alpha}}\).

Substituting \eqref{eq:bootstrap-oracle-ols-exact} into the preceding Taylor expansion gives
\[
\begin{aligned}
\sqrt n
\left(
\widetilde{\bm{\psi}}^*-\widetilde{\bm{\psi}}
\right)
&=
\dot\psi_{\widetilde{\bm\alpha}}
\left[
\left(
\widetilde\Gamma_Z^*
\right)^{-1}
\frac{1}{\sqrt n}
\sum_{i=1}^n
(N_i^*-1)
\bm Z_{\eta,i}\widetilde\varepsilon_i
\right]
+
o_{p^*}(1).
\end{aligned}
\]
The inverse-replacement result proved above states that
\[
\begin{aligned}
\left(
\widetilde\Gamma_Z^*
\right)^{-1}
\frac{1}{\sqrt n}
\sum_{i=1}^n
(N_i^*-1)
\bm Z_{\eta,i}\widetilde\varepsilon_i
=
\widetilde\Gamma_Z^{-1}
\frac{1}{\sqrt n}
\sum_{i=1}^n
(N_i^*-1)
\bm Z_{\eta,i}\widetilde\varepsilon_i
+
\bm r_{\alpha,n}^*,
\end{aligned}
\]
where
\(
\|\bm r_{\alpha,n}^*\|=o_{p^*}(1)
\)
in \(\mathbb P\)-probability. Because
\(\widetilde{\bm\alpha}\overset{p}{\longrightarrow}\bm\alpha_\eta\) and the
derivative of \(\bm{\psi}\) is continuous,
\[
\left\|
\dot\psi_{\widetilde{\bm\alpha}}
\right\|_{\mathrm{op}}
=
O_p(1).
\]
Therefore,
\[
\begin{aligned}
\left\|
\dot\psi_{\widetilde{\bm\alpha}}
\left[
\bm r_{\alpha,n}^*
\right]
\right\|
&\leq
\left\|
\dot\psi_{\widetilde{\bm\alpha}}
\right\|_{\mathrm{op}}
\|\bm r_{\alpha,n}^*\|
\\
&=
O_p(1)o_{p^*}(1)
\\
&=
o_{p^*}(1).
\end{aligned}
\]
Using the linearity of the derivative to separate the remainder consequently
gives
\[
\begin{aligned}
\sqrt n
\left(
\widetilde{\bm{\psi}}^*-\widetilde{\bm{\psi}}
\right)
&=
\dot\psi_{\widetilde{\bm\alpha}}
\left[
\widetilde\Gamma_Z^{-1}
\frac{1}{\sqrt n}
\sum_{i=1}^n
(N_i^*-1)
\bm Z_{\eta,i}\widetilde\varepsilon_i
\right]
+
o_{p^*}(1).
\end{aligned}
\]

The linearity of
\(\dot\psi_{\widetilde{\bm\alpha}}\) then gives
\[
\begin{aligned}
\dot\psi_{\widetilde{\bm\alpha}}
\left[
\widetilde\Gamma_Z^{-1}
\frac{1}{\sqrt n}
\sum_{i=1}^n
(N_i^*-1)
\bm Z_{\eta,i}\widetilde\varepsilon_i
\right]
=
\frac{1}{\sqrt n}
\sum_{i=1}^n
(N_i^*-1)
\dot\psi_{\widetilde{\bm\alpha}}
\left[
\widetilde\Gamma_Z^{-1}
\bm Z_{\eta,i}\widetilde\varepsilon_i
\right].
\end{aligned}
\]
Defining
\begin{equation}
\label{eq:bootstrap-oracle-influence-definition}
\widehat{\bm{\varphi}}_{\psi i}
=
\dot\psi_{\widetilde{\bm\alpha}}
\left[
\widetilde\Gamma_Z^{-1}
\bm Z_{\eta,i}\widetilde\varepsilon_i
\right]
\end{equation}
therefore gives
\begin{equation}
\label{eq:bootstrap-oracle-psi-AL}
\sqrt n
\left(
\widetilde{\bm{\psi}}^*-\widetilde{\bm{\psi}}
\right)
=
\frac{1}{\sqrt n}
\sum_{i=1}^n
(N_i^*-1)\widehat{\bm{\varphi}}_{\psi i}
+
o_{p^*}(1)
\end{equation}
in \(\mathbb P\)-probability.

This proves the lemma.
\end{proof}

Define the moment vector
\[
\bm{\chi}(O_i)
=
\left[
\bm X_i^\top,\,
T_i,\,
Y_i,\,
\operatorname{vec}(\bm X_i\bm X_i^\top)^\top,\,
(T_i\bm X_i)^\top,\,
(Y_i\bm X_i)^\top,\,
T_i^2,\,
T_iY_i
\right]^\top,
\]
and let
\[
\widehat{\bm\vartheta}
=
\frac{1}{n}\sum_{i=1}^n\bm{\chi}(O_i),
\qquad
\bm\vartheta_0
=
\mathbb E\{\bm{\chi}(O_i)\}.
\]

For a moment vector \(\bm\vartheta\) near \(\bm\vartheta_0\), let
\(\mathcal S(\bm\vartheta)\) denote the score space obtained from the
corresponding two-dimensional generalized eigenspace, and define
\[
\mathcal F(\bm\vartheta)
=
P_{\mathcal S(\bm\vartheta)}.
\]
By construction,
\[
\mathcal F(\bm\vartheta_0)
=
P_{\mathcal S_\eta},
\qquad
\mathcal F(\widehat{\bm\vartheta})
=
P_{\widehat{\mathcal S}_\eta}.
\]

\begin{lemma}
\label{lem:score-map-asymptotic-linearity}
Suppose Assumption~\ref{assump:L-asymptotic-linearity} holds. Then the map
\[
\mathcal F:\bm\vartheta\mapsto P_{\mathcal S(\bm\vartheta)}
\]
is continuously differentiable in a neighborhood of $\bm\vartheta_0$. Define
\[
\mathcal I(O)
=
\dot{\mathcal F}_{\bm\vartheta_0}
\left[\bm{\chi}(O)-\bm\vartheta_0\right].
\]
Then
\begin{equation}
\label{eq:one-solve-score-AL}
\sqrt n
\left(
P_{\widehat{\mathcal S}_\eta}-P_{\mathcal S_\eta}
\right)
=
\frac{1}{\sqrt n}
\sum_{i=1}^n
\mathcal I(O_i)
+o_p(1).
\end{equation}
Moreover, if $N_i^*$ denotes the number of times unit $i$ appears in a
nonparametric bootstrap sample, then
\begin{equation}
\label{eq:one-solve-score-bootstrap-AL}
\sqrt n
\left(
P_{\widehat{\mathcal S}_\eta^*}-P_{\widehat{\mathcal S}_\eta}
\right)
=
\frac{1}{\sqrt n}
\sum_{i=1}^n
(N_i^*-1)\widehat{\mathcal I}_i
+o_{p^*}(1)
\end{equation}
in $\mathbb P$-probability, where
\[
\widehat{\mathcal I}_i
=
\dot{\mathcal F}_{\bm\vartheta_0}
\left[\bm{\chi}(O_i)-\widehat{\bm\vartheta}\right].
\]
The influence function satisfies
\[
\mathbb E\{\mathcal I(O)\}=0,
\qquad
\mathbb E\|\mathcal I(O)\|_{\mathrm F}^2<\infty.
\]
In addition,
\[
\sqrt n
\left\|
P_{\widehat{\mathcal S}_\eta}
-
P_{\mathcal S_\eta}
\right\|_{\mathrm{op}}
=
O_p(1),
\]
\[
\sqrt n
\left\|
P_{\widehat{\mathcal S}_\eta^*}
-
P_{\widehat{\mathcal S}_\eta}
\right\|_{\mathrm{op}}
=
O_{p^*}(1),
\]
and
\[
\sqrt n
\left\|
P_{\widehat{\mathcal S}_\eta^*}
-
P_{\mathcal S_\eta}
\right\|_{\mathrm{op}}
=
O_{p^*}(1)
\]
in \(\mathbb P\)-probability for the two bootstrap conclusions.
\end{lemma}

\begin{proof}
By Assumption~\ref{assump:L-asymptotic-linearity}(a), the maps
$\bm\vartheta\mapsto H(\bm\vartheta)$ and
$\bm\vartheta\mapsto G_\eta(\bm\vartheta)$ are continuously differentiable near
$\bm\vartheta_0$. The same assumption makes $G_\eta(\bm\vartheta)$ positive definite throughout a
sufficiently small neighborhood. The inverse-square-root map is continuously
differentiable on the positive-definite cone. Therefore
\begin{equation}
\label{eq:one-solve-L-map}
\mathcal L(\bm\vartheta)
=
G_\eta(\bm\vartheta)^{-1/2}H(\bm\vartheta)G_\eta(\bm\vartheta)^{-1/2}
\end{equation}
is continuously differentiable near $\bm\vartheta_0$.

Let $\Pi$ be the rank-two spectral projection of $\mathcal L(\bm\vartheta_0)$ associated
with its nonzero eigenvalues. The spectral-gap condition allows us to choose a
finite union $\mathcal C$ of disjoint, positively oriented closed contours
that encloses the nonzero spectrum and excludes zero. Continuity of
\eqref{eq:one-solve-L-map} implies that the same contours separate the target
cluster for every $\bm\vartheta$ in a sufficiently small neighborhood of
$\bm\vartheta_0$. The
corresponding spectral projection is
\begin{equation}
\label{eq:rank-two-Riesz-projection}
\Pi(\bm\vartheta)
=
\frac{1}{2\pi\mathrm i}
\oint_{\mathcal C}
\{zI-\mathcal L(\bm\vartheta)\}^{-1}\,dz.
\end{equation}
The matrix inverse is continuously differentiable wherever it exists.
Equation~\eqref{eq:rank-two-Riesz-projection} therefore shows that
$\bm\vartheta\mapsto\Pi(\bm\vartheta)$ is continuously differentiable. This conclusion concerns the
joint rank-two projection and does not require the two enclosed eigenvalues to
be distinct.

The generalized-eigenvector coefficient space is
\[
\mathcal W_\eta(\bm\vartheta)
=
\operatorname{col}\{G_\eta(\bm\vartheta)^{-1/2}\Pi(\bm\vartheta)\}.
\]
Consequently, the corresponding $L_2(\mathbb P_X)$ score space is represented
by the Euclidean column space
\begin{equation}
\label{eq:one-solve-score-column-space}
\operatorname{col}
\left\{
\Sigma_{\bm X}^{1/2}G_\eta(\bm\vartheta)^{-1/2}\Pi(\bm\vartheta)
\right\}.
\end{equation}

Let
\(
A(\bm\vartheta)
=
\Sigma_{\bm X}^{1/2}
G_\eta(\bm\vartheta)^{-1/2}
\Pi(\bm\vartheta).
\) We previously showed that the spectral-gap condition ensures that
\(\Pi(\bm\vartheta)\) remains a rank-two projection when
\(\bm\vartheta\) is near \(\bm\vartheta_0\). Therefore,
\begin{equation} \label{eq:rank_upper_bound}
\operatorname{rank}\{A(\bm\vartheta)\}
\leq
\operatorname{rank}\{\Pi(\bm\vartheta)\}
=
2.
\end{equation}
We now show that the rank of \(A(\bm\vartheta_0)\) cannot fall below two. Let
\(Q\in\mathbb R^{p\times2}\) be an orthonormal basis for
\(\operatorname{ran}\{\Pi(\bm\vartheta_0)\}\), so that
\(
\Pi(\bm\vartheta_0)=QQ^\top,
\)
and define
\(
B_\eta
=
G_\eta(\bm\vartheta_0)^{-1/2}Q
\) so that \( A(\bm\vartheta_0)
=
\Sigma_{\bm X}^{1/2}
B_\eta Q^\top.\) Because \(Q^\top Q=I_2\),
\[
\operatorname{rank}\{A(\bm\vartheta_0)\}
=
\operatorname{rank}
\left(
\Sigma_{\bm X}^{1/2}B_\eta
\right).
\]
We will show that \(\operatorname{rank}
\left(
\Sigma_{\bm X}^{1/2}B_\eta
\right) = 2.\) We have
\[
\begin{aligned}
\left\{
\Sigma_{\bm X}^{1/2}B_\eta
\right\}^{\top}
\left\{
\Sigma_{\bm X}^{1/2}B_\eta
\right\}
=
{B_\eta}^{\top}
\Sigma_{\bm X}
B_\eta
=
\Gamma_{S,\eta}.
\end{aligned}
\]
Because \(\Gamma_{S,\eta}\) is nonsingular by Assumption~\ref{assump:L-asymptotic-linearity}(a) and positive semidefinite, its minimum eigenvalue is
positive. Hence,
\[
\sigma_2
\left(
\Sigma_{\bm X}^{1/2}B_\eta
\right)^2
=
\lambda_{\min}
\left(
\Gamma_{S,\eta}
\right)
>
0.
\]
Thus, \(A(\bm\vartheta_0)\) has rank two.

Finally, \(A(\bm\vartheta)\) is continuous in \(\bm\vartheta\). Because
singular values are continuous,
\[
\sigma_2\{A(\bm\vartheta)\}
\longrightarrow
\sigma_2\{A(\bm\vartheta_0)\}
>
0
\]
as \(\bm\vartheta\to\bm\vartheta_0\). Therefore,
\(\sigma_2\{A(\bm\vartheta)\}>0\) throughout a sufficiently small
neighborhood of \(\bm\vartheta_0\), which implies
\begin{equation} \label{eq:rank_lower_bound}
\operatorname{rank}\{A(\bm\vartheta)\}
\geq
2.
\end{equation}
Combining the \eqref{eq:rank_lower_bound} and \eqref{eq:rank_upper_bound} gives
\(
\operatorname{rank}\{A(\bm\vartheta)\}
=
2
\)
locally.

We have two maps:
\begin{equation}\label{eq:two_maps}
\bm\vartheta
\longmapsto
A(\bm\vartheta)
\longmapsto
P_{\operatorname{col}\{A(\bm\vartheta)\}}.
\end{equation}
We will show that each map is continuously differentiable in a neighborhood of
$\bm\vartheta_0$. We begin with the first map. Because
\(G_\eta(\bm\vartheta)^{-1/2}\) and
\(\Pi(\bm\vartheta)\) are continuously differentiable,
\[
A(\bm\vartheta)
=
\Sigma_{\bm X}^{1/2}
G_\eta(\bm\vartheta)^{-1/2}
\Pi(\bm\vartheta)
\]
is continuously differentiable by the matrix product rule. Hence,
\(
\bm\vartheta\mapsto A(\bm\vartheta)
\)
is continuously differentiable.

We now consider the second map. Let
\[
j_1
=
\min\left\{
j:\|A(\bm\vartheta_0)e_j\|_2>0
\right\},
\]
and let
\[
j_2
=
\min\left\{
j\neq j_1:
A(\bm\vartheta_0)e_j
\notin
\operatorname{span}\{A(\bm\vartheta_0)e_{j_1}\}
\right\}.
\]
Because \(A(\bm\vartheta_0)\) has rank two, both indices are well defined.
Define
\[
C(\bm\vartheta)
=
\begin{pmatrix}
A(\bm\vartheta)e_{j_1}
&
A(\bm\vartheta)e_{j_2}
\end{pmatrix}.
\]
The indices \(j_1\) and \(j_2\) are fixed, although the corresponding columns
vary with \(\bm\vartheta\). By construction, the columns of
\(C(\bm\vartheta_0)\) are linearly independent. Hence, the Gram matrix
\(
C(\bm\vartheta_0)^\top C(\bm\vartheta_0)
\)
is nonsingular. Because \(C(\bm\vartheta)\) consists of fixed indexed columns
of the continuously differentiable matrix \(A(\bm\vartheta)\), it is
continuously differentiable and therefore continuous. Consequently,
\(C(\bm\vartheta)^\top C(\bm\vartheta)\) is continuous in
\(\bm\vartheta\). Its nonsingularity at \(\bm\vartheta_0\) therefore implies
that it remains nonsingular in a sufficiently small neighborhood of
\(\bm\vartheta_0\). Thus, the columns of \(C(\bm\vartheta)\) remain linearly
independent throughout that neighborhood. Because
\(A(\bm\vartheta)\) has rank two,
\(
\operatorname{col}\{A(\bm\vartheta)\}
=
\operatorname{col}\{C(\bm\vartheta)\}.
\)
Therefore,
\[
P_{\mathcal S(\bm\vartheta)}
=
C(\bm\vartheta)
\left\{
C(\bm\vartheta)^\top C(\bm\vartheta)
\right\}^{-1}
C(\bm\vartheta)^\top.
\]
The matrix product rule and the continuous differentiability of matrix
inversion on nonsingular matrices now show directly from this formula that
\(\bm\vartheta\mapsto P_{\mathcal S(\bm\vartheta)}\) is continuously
differentiable.

We have shown that the two maps in \eqref{eq:two_maps} are continuously differentiable in a neighborhood of $\bm\vartheta_0$. We conclude that
$\bm\vartheta\mapsto P_{\mathcal S(\bm\vartheta)}$ is continuously
differentiable in a neighborhood of $\bm\vartheta_0$.

We next establish the sampling and bootstrap expansions stated in the lemma.
By Assumption~\ref{assump:L-asymptotic-linearity}(b), the moment vector has a
finite second moment. Therefore
\begin{equation}
\label{eq:one-solve-moment-rate}
\widehat{\bm\vartheta}-\bm\vartheta_0
=
\frac{1}{n}\sum_{i=1}^n\{\bm{\chi}(O_i)-\bm\vartheta_0\}
=
O_p(n^{-1/2}).
\end{equation}
The continuous differentiability established above means that, for a
deterministic increment $h$ tending to zero,
\[
\mathcal F(\bm\vartheta_0+h)
-
\mathcal F(\bm\vartheta_0)
=
\dot{\mathcal F}_{\bm\vartheta_0}[h]
+
R(h),
\]
where
\[
\frac{\|R(h)\|_{\mathrm F}}{\|h\|}
\longrightarrow
0.
\]
Set
\(
h_n=\widehat{\bm\vartheta}-\bm\vartheta_0
\).
Equation~\eqref{eq:one-solve-moment-rate} gives
\(
\|h_n\|=O_p(n^{-1/2})
\), and the preceding remainder property gives
\(
\|R(h_n)\|_{\mathrm F}/\|h_n\|=o_p(1)
\). Therefore,
\[
\|R(h_n)\|_{\mathrm F}
=
\|h_n\|
\frac{\|R(h_n)\|_{\mathrm F}}{\|h_n\|}
=
O_p(n^{-1/2})o_p(1)
=
o_p(n^{-1/2}).
\]
This proves the first-order expansion
\begin{equation}
\label{eq:one-solve-F-Taylor}
\mathcal F(\widehat{\bm\vartheta})-\mathcal F(\bm\vartheta_0)
=
\dot{\mathcal F}_{\bm\vartheta_0}
[\widehat{\bm\vartheta}-\bm\vartheta_0]
+o_p(n^{-1/2}).
\end{equation}
Substitute \eqref{eq:one-solve-moment-rate} into
\eqref{eq:one-solve-F-Taylor}. Linearity of
$\dot{\mathcal F}_{\bm\vartheta_0}$ then gives \eqref{eq:one-solve-score-AL}.

For the bootstrap sample, the bootstrap moment vector is
\[
\widehat{\bm\vartheta}^*
=
\frac{1}{n}
\sum_{i=1}^n
N_i^*\bm{\chi}(O_i).
\]
Because a bootstrap sample contains $n$ observations,
\(
\sum_{i=1}^nN_i^*=n
\), and hence
\(
\sum_{i=1}^n(N_i^*-1)=0
\).
Consequently,
\begin{equation}
\label{eq:one-solve-bootstrap-moment}
\begin{aligned}
\widehat{\bm\vartheta}^*-\widehat{\bm\vartheta}
&=
\frac{1}{n}
\sum_{i=1}^n
(N_i^*-1)\bm{\chi}(O_i)
\\
&=
\frac{1}{n}
\sum_{i=1}^n
(N_i^*-1)
\{\bm{\chi}(O_i)-\widehat{\bm\vartheta}\}.
\end{aligned}
\end{equation}
Let
\(
\bm V_i
=
\bm{\chi}(O_i)-\widehat{\bm\vartheta}.
\)
By the definition
\(
\widehat{\bm\vartheta}
=
n^{-1}\sum_{i=1}^n\bm{\chi}(O_i)
\),
the vectors \(\bm V_i\) satisfy
\[
\begin{aligned}
\sum_{i=1}^n\bm V_i
&=
\sum_{i=1}^n
\left\{
\bm{\chi}(O_i)-\widehat{\bm\vartheta}
\right\}
\\
&=
\sum_{i=1}^n\bm{\chi}(O_i)
-
n\widehat{\bm\vartheta}
\\
&=
n\widehat{\bm\vartheta}
-
n\widehat{\bm\vartheta}
\\
&=
0.
\end{aligned}
\]
Equation~\eqref{eq:one-solve-bootstrap-moment} may therefore be written as
\[
\widehat{\bm\vartheta}^*
-
\widehat{\bm\vartheta}
=
\frac{1}{n}
\sum_{i=1}^n
(N_i^*-1)\bm V_i.
\]
Conditional on the original data, the vectors \(\bm V_i\) are fixed. Expanding
the squared norm gives
\[
\begin{aligned}
&
\mathbb E^*
\left\|
\widehat{\bm\vartheta}^*
-
\widehat{\bm\vartheta}
\right\|^2
\\
&=
\mathbb E^*
\left\langle
\frac{1}{n}
\sum_{i=1}^n
(N_i^*-1)\bm V_i,\,
\frac{1}{n}
\sum_{j=1}^n
(N_j^*-1)\bm V_j
\right\rangle
\\
&=
\frac{1}{n^2}
\sum_{i=1}^n
\sum_{j=1}^n
\mathbb E^*
\left\{
(N_i^*-1)(N_j^*-1)
\right\}
\langle\bm V_i,\bm V_j\rangle.
\end{aligned}
\]
For the multinomial bootstrap counts,
\[
\mathbb E^*
\left\{
(N_i^*-1)(N_j^*-1)
\right\}
=
\mathbf 1\{i=j\}-\frac{1}{n}.
\]
Substituting this identity gives
\[
\begin{aligned}
&
\mathbb E^*
\left\|
\widehat{\bm\vartheta}^*
-
\widehat{\bm\vartheta}
\right\|^2
\\
&=
\frac{1}{n^2}
\sum_{i=1}^n
\sum_{j=1}^n
\left\{
\mathbf 1\{i=j\}-\frac{1}{n}
\right\}
\langle\bm V_i,\bm V_j\rangle
\\
&=
\frac{1}{n^2}
\sum_{i=1}^n
\|\bm V_i\|^2
-
\frac{1}{n^3}
\sum_{i=1}^n
\sum_{j=1}^n
\langle\bm V_i,\bm V_j\rangle
\\
&=
\frac{1}{n^2}
\sum_{i=1}^n
\|\bm V_i\|^2
-
\frac{1}{n^3}
\left\|
\sum_{i=1}^n
\bm V_i
\right\|^2
\\
&=
\frac{1}{n^2}
\sum_{i=1}^n
\|\bm V_i\|^2.
\end{aligned}
\]
The final equality uses
\(\sum_{i=1}^n\bm V_i=\bm 0\). Substituting
\(\bm V_i=\bm{\chi}(O_i)-\widehat{\bm\vartheta}\) yields
\begin{equation}\label{eq:eq_star_chi}
\mathbb E^*
\left\|
\widehat{\bm\vartheta}^*
-
\widehat{\bm\vartheta}
\right\|^2
=
\frac{1}{n^2}
\sum_{i=1}^n
\left\|
\bm{\chi}(O_i)-\widehat{\bm\vartheta}
\right\|^2.
\end{equation}

We first show that
\[
\frac{1}{n}
\sum_{i=1}^n
\left\|
\bm{\chi}(O_i)-\widehat{\bm\vartheta}
\right\|^2
=
O_p(1).
\]
Write
\[
\bm{\chi}(O_i)-\widehat{\bm\vartheta}
=
\left\{
\bm{\chi}(O_i)-\bm\vartheta_0
\right\}
-
\left\{
\widehat{\bm\vartheta}-\bm\vartheta_0
\right\}.
\]
Expanding the squared norm and summing over \(i\) gives
\[
\begin{aligned}
\sum_{i=1}^n
\left\|
\bm{\chi}(O_i)-\widehat{\bm\vartheta}
\right\|^2
=
&\sum_{i=1}^n
\left\|
\bm{\chi}(O_i)-\bm\vartheta_0
\right\|^2
-
2
\left(
\widehat{\bm\vartheta}-\bm\vartheta_0
\right)^\top
\sum_{i=1}^n
\left\{
\bm{\chi}(O_i)-\bm\vartheta_0
\right\}
\\
&\quad+
n
\left\|
\widehat{\bm\vartheta}-\bm\vartheta_0
\right\|^2.
\end{aligned}
\]
By the definition of \(\widehat{\bm\vartheta}\),
\[
\begin{aligned}
\sum_{i=1}^n
\left\{
\bm{\chi}(O_i)-\bm\vartheta_0
\right\}
&=
n
\left(
\widehat{\bm\vartheta}-\bm\vartheta_0
\right).
\end{aligned}
\]
Substituting this identity into the preceding expansion gives
\[
\begin{aligned}
\sum_{i=1}^n
\left\|
\bm{\chi}(O_i)-\widehat{\bm\vartheta}
\right\|^2
&=
\sum_{i=1}^n
\left\|
\bm{\chi}(O_i)-\bm\vartheta_0
\right\|^2
-
n
\left\|
\widehat{\bm\vartheta}-\bm\vartheta_0
\right\|^2
\leq
\sum_{i=1}^n
\left\|
\bm{\chi}(O_i)-\bm\vartheta_0
\right\|^2.
\end{aligned}
\]
Dividing by \(n\) gives
\[
\frac{1}{n}
\sum_{i=1}^n
\left\|
\bm{\chi}(O_i)-\widehat{\bm\vartheta}
\right\|^2
\leq
\frac{1}{n}
\sum_{i=1}^n
\left\|
\bm{\chi}(O_i)-\bm\vartheta_0
\right\|^2.
\]
Assumption~\ref{assump:L-asymptotic-linearity}(b) gives
\[
\mathbb E
\left\|
\bm{\chi}(O_i)-\bm\vartheta_0
\right\|^2
<
\infty.
\]
The variables
\(
\|\bm{\chi}(O_i)-\bm\vartheta_0\|^2
\)
are independent and identically distributed. Therefore, the law of large
numbers gives
\[
\frac{1}{n}
\sum_{i=1}^n
\left\|
\bm{\chi}(O_i)-\bm\vartheta_0
\right\|^2
\overset{p}{\longrightarrow}
\mathbb E
\left\|
\bm{\chi}(O_i)-\bm\vartheta_0
\right\|^2
<
\infty.
\]
Consequently,
\begin{equation}
\label{eq:bootstrap-moment-empirical-second-moment}
\frac{1}{n}
\sum_{i=1}^n
\left\|
\bm{\chi}(O_i)-\widehat{\bm\vartheta}
\right\|^2
=
O_p(1).
\end{equation}
This second-moment calculation gives
\[
\begin{aligned}
\mathbb E^*
\left\|
\widehat{\bm\vartheta}^*
-
\widehat{\bm\vartheta}
\right\|^2
&=
\frac{1}{n^2}
\sum_{i=1}^n
\left\|
\bm{\chi}(O_i)-\widehat{\bm\vartheta}
\right\|^2
\\
&=
\frac{1}{n}
\left\{
\frac{1}{n}
\sum_{i=1}^n
\left\|
\bm{\chi}(O_i)-\widehat{\bm\vartheta}
\right\|^2
\right\}
\\
&=
O_p(n^{-1}),
\end{aligned}
\]
where the first equality follows from \eqref{eq:eq_star_chi} and the last from
\eqref{eq:bootstrap-moment-empirical-second-moment}. Equivalently,
\[
n\,
\mathbb E^*
\left\|
\widehat{\bm\vartheta}^*
-
\widehat{\bm\vartheta}
\right\|^2
=
O_p(1).
\]
Conditional Markov's inequality gives
\[
\sqrt n
\left\|
\widehat{\bm\vartheta}^*
-
\widehat{\bm\vartheta}
\right\|
=
O_{p^*}(1)
\]
in \(\mathbb P\)-probability. Equivalently,
\(
\widehat{\bm\vartheta}^*
-
\widehat{\bm\vartheta}
=
O_{p^*}(n^{-1/2})
\)
in \(\mathbb P\)-probability.

Because
\(
\widehat{\bm\vartheta}\overset{p}{\longrightarrow}\bm\vartheta_0
\)
and
\(
\widehat{\bm\vartheta}^*-\widehat{\bm\vartheta}
=
O_{p^*}(n^{-1/2})
=
o_{p^*}(1)
\),
the two moment vectors lie in the
neighborhood of \(\bm\vartheta_0\) on which \(\mathcal F\) is continuously
differentiable with probability tending to one. Continuous differentiability
on this event gives
\[
\mathcal F(\widehat{\bm\vartheta}^*)
-
\mathcal F(\widehat{\bm\vartheta})
=
\dot{\mathcal F}_{\widehat{\bm\vartheta}}
\left[
\widehat{\bm\vartheta}^*
-
\widehat{\bm\vartheta}
\right]
+
o_{p^*}(n^{-1/2})
\]
in \(\mathbb P\)-probability.

We replace the derivative at \(\widehat{\bm\vartheta}\) by adding and
subtracting the derivative at \(\bm\vartheta_0\):
\[
\begin{aligned}
\dot{\mathcal F}_{\widehat{\bm\vartheta}}
\left[
\widehat{\bm\vartheta}^*
-
\widehat{\bm\vartheta}
\right]
=
\dot{\mathcal F}_{\bm\vartheta_0}
\left[
\widehat{\bm\vartheta}^*
-
\widehat{\bm\vartheta}
\right]
\quad+
\left(
\dot{\mathcal F}_{\widehat{\bm\vartheta}}
-
\dot{\mathcal F}_{\bm\vartheta_0}
\right)
\left[
\widehat{\bm\vartheta}^*
-
\widehat{\bm\vartheta}
\right].
\end{aligned}
\]
Because the derivative is continuous and
\(\widehat{\bm\vartheta}\overset{p}{\longrightarrow}\bm\vartheta_0\),
\[
\left\|
\dot{\mathcal F}_{\widehat{\bm\vartheta}}
-
\dot{\mathcal F}_{\bm\vartheta_0}
\right\|_{\mathrm{op}}
=
o_p(1).
\]
Consequently, the norm of the additional term is at most
\[
o_p(1)O_{p^*}(n^{-1/2})
=
o_{p^*}(n^{-1/2}).
\]
Thus,
\[
\mathcal F(\widehat{\bm\vartheta}^*)
-
\mathcal F(\widehat{\bm\vartheta})
=
\dot{\mathcal F}_{\bm\vartheta_0}
\left[
\widehat{\bm\vartheta}^*
-
\widehat{\bm\vartheta}
\right]
+
o_{p^*}(n^{-1/2})
\]
in \(\mathbb P\)-probability. We may therefore substitute
\eqref{eq:one-solve-bootstrap-moment} into the first-order expansion with
derivative \(\dot{\mathcal F}_{\bm\vartheta_0}\):
\[
\begin{aligned}
\mathcal F(\widehat{\bm\vartheta}^*)
-
\mathcal F(\widehat{\bm\vartheta})
&=
\dot{\mathcal F}_{\bm\vartheta_0}
\left[
\frac{1}{n}
\sum_{i=1}^n
(N_i^*-1)
\{\bm{\chi}(O_i)-\widehat{\bm\vartheta}\}
\right]
+
o_{p^*}(n^{-1/2})
\\
&=
\frac{1}{n}
\sum_{i=1}^n
(N_i^*-1)
\dot{\mathcal F}_{\bm\vartheta_0}
\left[
\bm{\chi}(O_i)-\widehat{\bm\vartheta}
\right]
+
o_{p^*}(n^{-1/2})
\\
&=
\frac{1}{n}
\sum_{i=1}^n
(N_i^*-1)\widehat{\mathcal I}_i
+
o_{p^*}(n^{-1/2}).
\end{aligned}
\]
Multiplying by \(\sqrt n\) and using
\(
\mathcal F(\widehat{\bm\vartheta}^*)
=
P_{\widehat{\mathcal S}_\eta^*}
\)
and
\(
\mathcal F(\widehat{\bm\vartheta})
=
P_{\widehat{\mathcal S}_\eta}
\)
yields \eqref{eq:one-solve-score-bootstrap-AL}.

It remains to verify the influence-function moments. Linearity gives
\[
\mathbb E\{\mathcal I(O)\}
=
\dot{\mathcal F}_{\bm\vartheta_0}
\left[
\mathbb E\{\bm{\chi}(O)-\bm\vartheta_0\}
\right]
=0.
\]
Because $\dot{\mathcal F}_{\bm\vartheta_0}$ is bounded, there is a constant $K<\infty$
such that
\[
\|\mathcal I(O)\|_{\mathrm F}
\leq
K\|\bm{\chi}(O)-\bm\vartheta_0\|.
\]
The moment condition on $\bm{\chi}(O)$ in Assumption~\ref{assump:L-asymptotic-linearity}(b) therefore implies
$\mathbb E\|\mathcal I(O)\|_{\mathrm F}^2<\infty$.

We finally prove the three projection-tightness conclusions.
By \eqref{eq:one-solve-score-AL},
\[
\sqrt n
\left(
P_{\widehat{\mathcal S}_\eta}
-
P_{\mathcal S_\eta}
\right)
=
\frac{1}{\sqrt n}
\sum_{i=1}^n
\mathcal I(O_i)
+
o_p(1).
\]
The influence-function moment bound established above gives
\(
\mathbb E
\|\mathcal I(O_i)\|_{\mathrm F}^2
<
\infty.
\)
The influence terms are independent and have mean zero. Therefore,
\[
\begin{aligned}
\mathbb E
\left\|
\frac{1}{\sqrt n}
\sum_{i=1}^n
\mathcal I(O_i)
\right\|_{\mathrm F}^2
&=
\frac{1}{n}
\sum_{i=1}^n
\mathbb E
\|\mathcal I(O_i)\|_{\mathrm F}^2
\\
&=
\mathbb E
\|\mathcal I(O_1)\|_{\mathrm F}^2
<
\infty,
\end{aligned}
\]
where the cross terms vanish because the influence terms are independent and
mean zero. Markov's inequality therefore gives
\[
\left\|
\frac{1}{\sqrt n}
\sum_{i=1}^n
\mathcal I(O_i)
\right\|_{\mathrm F}
=
O_p(1).
\]
Since the operator norm is bounded by the Frobenius norm,
\[
\sqrt n
\left\|
P_{\widehat{\mathcal S}_\eta}
-
P_{\mathcal S_\eta}
\right\|_{\mathrm{op}}
=
O_p(1).
\]

For the bootstrap expansion, linearity of
\(\dot{\mathcal F}_{\bm\vartheta_0}\) gives
\begin{equation} \label{eq:linearity_zero}
\sum_{i=1}^n
\widehat{\mathcal I}_i
=
\dot{\mathcal F}_{\bm\vartheta_0}
\left[
\sum_{i=1}^n
\{\bm{\chi}(O_i)-\widehat{\bm\vartheta}\}
\right]
=
0.
\end{equation}
Using the Frobenius inner product,
\[
\begin{aligned}
&
\mathbb E^*
\left\|
\frac{1}{\sqrt n}
\sum_{i=1}^n
(N_i^*-1)\widehat{\mathcal I}_i
\right\|_{\mathrm F}^2
\\
&=
\frac{1}{n}
\sum_{i=1}^n
\sum_{j=1}^n
\mathbb E^*
\left\{
(N_i^*-1)(N_j^*-1)
\right\}
\left\langle
\widehat{\mathcal I}_i,
\widehat{\mathcal I}_j
\right\rangle_{\mathrm F}
\\
&=
\frac{1}{n}
\sum_{i=1}^n
\sum_{j=1}^n
\left\{
\mathbf 1\{i=j\}-\frac{1}{n}
\right\}
\left\langle
\widehat{\mathcal I}_i,
\widehat{\mathcal I}_j
\right\rangle_{\mathrm F}
\\
&=
\frac{1}{n}
\sum_{i=1}^n
\|\widehat{\mathcal I}_i\|_{\mathrm F}^2
-
\frac{1}{n^2}
\sum_{i=1}^n
\sum_{j=1}^n
\left\langle
\widehat{\mathcal I}_i,
\widehat{\mathcal I}_j
\right\rangle_{\mathrm F}.
\end{aligned}
\]

The double sum in the second term satisfies
\[
\begin{aligned}
\sum_{i=1}^n
\sum_{j=1}^n
\left\langle
\widehat{\mathcal I}_i,
\widehat{\mathcal I}_j
\right\rangle_{\mathrm F}
&=
\left\langle
\sum_{i=1}^n
\widehat{\mathcal I}_i,
\sum_{j=1}^n
\widehat{\mathcal I}_j
\right\rangle_{\mathrm F}
=
\left\|
\sum_{i=1}^n
\widehat{\mathcal I}_i
\right\|_{\mathrm F}^2
=
0.
\end{aligned}
\]
The last equality holds because
\(
\sum_{i=1}^n\widehat{\mathcal I}_i=0
\) by \eqref{eq:linearity_zero}.
Therefore,
\[
\mathbb E^*
\left\|
\frac{1}{\sqrt n}
\sum_{i=1}^n
(N_i^*-1)\widehat{\mathcal I}_i
\right\|_{\mathrm F}^2
=
\frac{1}{n}
\sum_{i=1}^n
\|\widehat{\mathcal I}_i\|_{\mathrm F}^2.
\]

Write
\[
\bm{\chi}(O_i)-\widehat{\bm\vartheta}
=
\{\bm{\chi}(O_i)-\bm\vartheta_0\}
-
\{\widehat{\bm\vartheta}-\bm\vartheta_0\}.
\]
Applying the linear map
\(\dot{\mathcal F}_{\bm\vartheta_0}\) gives
\[
\begin{aligned}
\widehat{\mathcal I}_i
&=
\dot{\mathcal F}_{\bm\vartheta_0}
\left[
\bm{\chi}(O_i)-\bm\vartheta_0
\right]
-
\dot{\mathcal F}_{\bm\vartheta_0}
\left[
\widehat{\bm\vartheta}-\bm\vartheta_0
\right]
\\
&=
\mathcal I(O_i)
-
\dot{\mathcal F}_{\bm\vartheta_0}
\left[
\frac{1}{n}
\sum_{j=1}^n
\{\bm{\chi}(O_j)-\bm\vartheta_0\}
\right].
\end{aligned}
\]
Linearity also allows the derivative to pass through the average:
\[
\begin{aligned}
\dot{\mathcal F}_{\bm\vartheta_0}
\left[
\frac{1}{n}
\sum_{j=1}^n
\{\bm{\chi}(O_j)-\bm\vartheta_0\}
\right]
&=
\frac{1}{n}
\sum_{j=1}^n
\dot{\mathcal F}_{\bm\vartheta_0}
\left[
\bm{\chi}(O_j)-\bm\vartheta_0
\right]
\\
&=
\frac{1}{n}
\sum_{j=1}^n
\mathcal I(O_j).
\end{aligned}
\]
Therefore,
\[
\widehat{\mathcal I}_i
=
\mathcal I(O_i)
-
\frac{1}{n}
\sum_{j=1}^n
\mathcal I(O_j).
\]

Centering cannot increase the sum of squared Frobenius norms, so
\[
\frac{1}{n}
\sum_{i=1}^n
\|\widehat{\mathcal I}_i\|_{\mathrm F}^2
\leq
\frac{1}{n}
\sum_{i=1}^n
\|\mathcal I(O_i)\|_{\mathrm F}^2
=
O_p(1)
\]
by the influence-function moment bound $\mathbb E\|\mathcal I(O)\|_{\mathrm F}^2<\infty$ established above and the law of large numbers. We have established that $\mathbb E^*
\left\|
\frac{1}{\sqrt n}
\sum_{i=1}^n
(N_i^*-1)\widehat{\mathcal I}_i
\right\|_{\mathrm F}^2 = O_p(1)$. Conditional Markov's inequality therefore gives
\[
\left\|
\frac{1}{\sqrt n}
\sum_{i=1}^n
(N_i^*-1)\widehat{\mathcal I}_i
\right\|_{\mathrm F}
=
O_{p^*}(1)
\]
in \(\mathbb P\)-probability. Combining this result with
\eqref{eq:one-solve-score-bootstrap-AL} and using
\(\|\cdot\|_{\mathrm{op}}\leq\|\cdot\|_{\mathrm F}\) gives
\[
\sqrt n
\left\|
P_{\widehat{\mathcal S}_\eta^*}
-
P_{\widehat{\mathcal S}_\eta}
\right\|_{\mathrm{op}}
=
O_{p^*}(1)
\]
in \(\mathbb P\)-probability.

Finally, the triangle inequality gives
\[
\begin{aligned}
\sqrt n
\left\|
P_{\widehat{\mathcal S}_\eta^*}
-
P_{\mathcal S_\eta}
\right\|_{\mathrm{op}}
&\leq
\sqrt n
\left\|
P_{\widehat{\mathcal S}_\eta^*}
-
P_{\widehat{\mathcal S}_\eta}
\right\|_{\mathrm{op}} +
\sqrt n
\left\|
P_{\widehat{\mathcal S}_\eta}
-
P_{\mathcal S_\eta}
\right\|_{\mathrm{op}}.
\end{aligned}
\]
The first term is \(O_{p^*}(1)\) in \(\mathbb P\)-probability, and the second is
\(O_p(1)\). Therefore, their sum is \(O_{p^*}(1)\) in \(\mathbb P\)-probability,
which proves the final claim.
\end{proof}

\begin{lemma}
\label{lem:joint-oracle-score-bootstrap}
Suppose Assumption~\ref{assump:L-asymptotic-linearity} and the conditions of
Lemmas~\ref{lem:oracle-downstream-sampling-expansion} and
\ref{lem:oracle-downstream-bootstrap-expansion} hold. Then, for every bounded
Lipschitz function \(g\),
\[
\begin{aligned}
&
\left|
\mathbb E^*
g
\left(
\sqrt n(\widetilde{\bm{\psi}}^*-\widetilde{\bm{\psi}}),
\operatorname{vec}
\left\{
\sqrt n
\left(
P_{\widehat{\mathcal S}_\eta^*}
-
P_{\widehat{\mathcal S}_\eta}
\right)
\right\}
\right)
\right.
\\
&\qquad\left.
-
\mathbb E
g
\left(
\sqrt n(\widetilde{\bm{\psi}}-\bm{\psi}),
\operatorname{vec}
\left\{
\sqrt n
\left(
P_{\widehat{\mathcal S}_\eta}
-
P_{\mathcal S_\eta}
\right)
\right\}
\right)
\right|
\overset{p}{\longrightarrow}
0.
\end{aligned}
\]
\end{lemma}

\begin{proof}
Lemmas~\ref{lem:oracle-downstream-sampling-expansion} and
\ref{lem:oracle-downstream-bootstrap-expansion} supply the oracle sampling and
bootstrap expansions and their required moment bounds. We next verify the
assumptions needed to combine them with the score-space expansion.

The present lemma assumes Assumption~\ref{assump:L-asymptotic-linearity}. We can therefore apply
Lemma~\ref{lem:score-map-asymptotic-linearity} to obtain the sampling and bootstrap
asymptotic-linear expansions of the COMPACT score-space projection.
Let \(\mathcal I(O_i)\) and \(\widehat{\mathcal I}_i\) be the influence
terms defined in that lemma. Define
\[
\bm W_i
=
\begin{pmatrix}
{\bm{\varphi}}_\psi(O_i)\\
\operatorname{vec}\{\mathcal I(O_i)\}
\end{pmatrix},
\qquad
\widehat{\bm W}_i
=
\begin{pmatrix}
\widehat{\bm{\varphi}}_{\psi i}\\
\operatorname{vec}\{\widehat{\mathcal I}_i\}
\end{pmatrix}.
\]
The empirical influence vectors are centered. By
\eqref{eq:oracle-sample-normal-equation}, \eqref{eq:bootstrap-oracle-influence-definition} and the linearity of
\(\dot\psi_{\widetilde{\bm\alpha}}\):
\[
\begin{aligned}
\sum_{i=1}^n
\widehat{\bm{\varphi}}_{\psi i}
&=
\dot\psi_{\widetilde{\bm\alpha}}
\left[
\widetilde\Gamma_Z^{-1}
\sum_{i=1}^n
\bm Z_{\eta,i}\widetilde\varepsilon_i
\right]
=
0.
\end{aligned}
\]
By linearity of the derivative \(\dot{\mathcal F}_{\bm\vartheta_0}\) and the definition of
\(\widehat{\bm\vartheta}\) similarly give
\[
\begin{aligned}
\sum_{i=1}^n
\widehat{\mathcal I}_i
&=
\dot{\mathcal F}_{\bm\vartheta_0}
\left[
\sum_{i=1}^n
\{\bm{\chi}(O_i)-\widehat{\bm\vartheta}\}
\right]
=
\dot{\mathcal F}_{\bm\vartheta_0}[0]
=
0.
\end{aligned}
\]
Therefore,
\begin{equation}
\label{eq:joint-empirical-influence-centering}
\sum_{i=1}^n\widehat{\bm W}_i=\bm 0.
\end{equation}
Also define
\[
\bm A_n
=
\begin{pmatrix}
\sqrt n(\widetilde{\bm{\psi}}-\bm{\psi})\\
\operatorname{vec}\!\left\{
\sqrt n
\left(
P_{\widehat{\mathcal S}_\eta}
-
P_{\mathcal S_\eta}
\right)
\right\}
\end{pmatrix}
\]
and
\[
\bm A_n^*
=
\begin{pmatrix}
\sqrt n(\widetilde{\bm{\psi}}^*-\widetilde{\bm{\psi}})\\
\operatorname{vec}\!\left\{
\sqrt n
\left(
P_{\widehat{\mathcal S}_\eta^*}
-
P_{\widehat{\mathcal S}_\eta}
\right)
\right\}
\end{pmatrix}.
\]
Combining \eqref{eq:oracle-psi-AL} with
\eqref{eq:one-solve-score-AL} gives
\begin{equation}
\label{eq:joint-original-AL}
\bm A_n
=
\frac{1}{\sqrt n}
\sum_{i=1}^n
\bm W_i
+
o_p(1).
\end{equation}
Combining \eqref{eq:bootstrap-oracle-psi-AL} with
\eqref{eq:one-solve-score-bootstrap-AL} gives
\begin{equation}
\label{eq:joint-bootstrap-AL}
\bm A_n^*
=
\frac{1}{\sqrt n}
\sum_{i=1}^n
(N_i^*-1)\widehat{\bm W}_i
+
o_{p^*}(1)
\end{equation}
in \(\mathbb P\)-probability.

We next verify the centering, covariance, and Lindeberg conditions for the
conditional multinomial bootstrap central limit theorem. We first establish
mean-square consistency of the empirical influence terms. Starting from the definitions,
\[
\widehat{\bm{\varphi}}_{\psi i}
=
\dot\psi_{\widetilde{\bm\alpha}}
\left[
\widetilde\Gamma_Z^{-1}
\bm Z_{\eta,i}\widetilde\varepsilon_i
\right]
\]
and
\[
{\bm{\varphi}}_\psi(O_i)
=
\dot\psi_{\bm\alpha_\eta}
\left[
\Gamma_Z^{-1}
\bm Z_{\eta,i}\varepsilon_{\eta,i}
\right].
\]
The two residuals satisfy
\[
\begin{aligned}
\widetilde\varepsilon_i
&=
Y_i
-
{\bm Z_{\eta,i}}^{\top}
\widetilde{\bm\alpha}
\\
&=
Y_i
-
{\bm Z_{\eta,i}}^{\top}
\bm\alpha_\eta
-
{\bm Z_{\eta,i}}^{\top}
\left(
\widetilde{\bm\alpha}
-
\bm\alpha_\eta
\right)
\\
&=
\varepsilon_{\eta,i}
-
{\bm Z_{\eta,i}}^{\top}
\left(
\widetilde{\bm\alpha}
-
\bm\alpha_\eta
\right).
\end{aligned}
\]
Substituting this identity into \(\widehat{\bm{\varphi}}_{\psi i}\) gives
\[
\begin{aligned}
\widehat{\bm{\varphi}}_{\psi i}
&=
\dot\psi_{\widetilde{\bm\alpha}}
\left[
\widetilde\Gamma_Z^{-1}
\bm Z_{\eta,i}
\left\{
\varepsilon_{\eta,i}
-
{\bm Z_{\eta,i}}^{\top}
\left(
\widetilde{\bm\alpha}
-
\bm\alpha_\eta
\right)
\right\}
\right]
\\
&=
\dot\psi_{\widetilde{\bm\alpha}}
\left[
\widetilde\Gamma_Z^{-1}
\bm Z_{\eta,i}\varepsilon_{\eta,i}
\right]
-
\dot\psi_{\widetilde{\bm\alpha}}
\left[
\widetilde\Gamma_Z^{-1}
\bm Z_{\eta,i}{\bm Z_{\eta,i}}^{\top}
\left(
\widetilde{\bm\alpha}
-
\bm\alpha_\eta
\right)
\right],
\end{aligned}
\]
where the second equality uses the linearity of
\(\dot\psi_{\widetilde{\bm\alpha}}\).

Subtracting \({\bm{\varphi}}_\psi(O_i)\) now gives
\[
\begin{aligned}
\widehat{\bm{\varphi}}_{\psi i}
-
{\bm{\varphi}}_\psi(O_i)
&=
\dot\psi_{\widetilde{\bm\alpha}}
\left[
\widetilde\Gamma_Z^{-1}
\bm Z_{\eta,i}\varepsilon_{\eta,i}
\right]
-
\dot\psi_{\bm\alpha_\eta}
\left[
\Gamma_Z^{-1}
\bm Z_{\eta,i}\varepsilon_{\eta,i}
\right]
\\
&\quad-
\dot\psi_{\widetilde{\bm\alpha}}
\left[
\widetilde\Gamma_Z^{-1}
\bm Z_{\eta,i}{\bm Z_{\eta,i}}^{\top}
\left(
\widetilde{\bm\alpha}
-
\bm\alpha_\eta
\right)
\right].
\end{aligned}
\]
Viewing each derivative as a linear operator, the first two terms equal
\[
\left(
\dot\psi_{\widetilde{\bm\alpha}}
\widetilde\Gamma_Z^{-1}
-
\dot\psi_{\bm\alpha_\eta}
\Gamma_Z^{-1}
\right)
\bm Z_{\eta,i}\varepsilon_{\eta,i}.
\]
Therefore,
\begin{equation}\label{eq:psi_decomposition}
\begin{aligned}
\widehat{\bm{\varphi}}_{\psi i}
-
{\bm{\varphi}}_\psi(O_i)
&=
\left(
\dot\psi_{\widetilde{\bm\alpha}}
\widetilde\Gamma_Z^{-1}
-
\dot\psi_{\bm\alpha_\eta}
\Gamma_Z^{-1}
\right)
\bm Z_{\eta,i}\varepsilon_{\eta,i}
\\
&\quad-
\dot\psi_{\widetilde{\bm\alpha}}
\widetilde\Gamma_Z^{-1}
\bm Z_{\eta,i}{\bm Z_{\eta,i}}^{\top}
\left(
\widetilde{\bm\alpha}
-
\bm\alpha_\eta
\right).
\end{aligned}
\end{equation}
Equation~\eqref{eq:oracle-alpha-rate} and continuity of the derivative of
\(\bm{\psi}\) give
\[
\left\|
\dot\psi_{\widetilde{\bm\alpha}}
-
\dot\psi_{\bm\alpha_\eta}
\right\|_{\mathrm{op}}
=
o_p(1).
\]
Equation \eqref{eq:oracle-gram-consistency} and the nonsingularity of $\Gamma_Z$ also give
\[
\begin{aligned}
&
\left\|
\dot\psi_{\widetilde{\bm\alpha}}
\widetilde\Gamma_Z^{-1}
-
\dot\psi_{\bm\alpha_\eta}
\Gamma_Z^{-1}
\right\|_{\mathrm{op}}
\\
&\leq
\left\|
\dot\psi_{\widetilde{\bm\alpha}}
-
\dot\psi_{\bm\alpha_\eta}
\right\|_{\mathrm{op}}
\left\|
\widetilde\Gamma_Z^{-1}
\right\|_{\mathrm{op}}
+
\left\|
\dot\psi_{\bm\alpha_\eta}
\right\|_{\mathrm{op}}
\left\|
\widetilde\Gamma_Z^{-1}
-
\Gamma_Z^{-1}
\right\|_{\mathrm{op}}
\\
&=
o_p(1)O_p(1)
+
O(1)o_p(1)
\\
&=
o_p(1).
\end{aligned}
\]
The variables \(\|\bm Z_{\eta,i}\varepsilon_{\eta,i}\|^2\) are independent
and identically distributed and have finite expectation by
\eqref{eq:oracle-score-residual-second-moment}. The law of large numbers
therefore gives
\[
\frac{1}{n}
\sum_{i=1}^n
\|\bm Z_{\eta,i}\varepsilon_{\eta,i}\|^2
=
O_p(1)
\]
The variables \(\|\bm Z_{\eta,i}\|^4\) are also independent and identically
distributed and have finite expectation by
\eqref{eq:oracle-fourth-moments}. The law of large numbers therefore gives
\[
\frac{1}{n}
\sum_{i=1}^n
\|\bm Z_{\eta,i}{\bm Z_{\eta,i}}^\top\|_{\mathrm{op}}^2
=
\frac{1}{n}
\sum_{i=1}^n
\|\bm Z_{\eta,i}\|^4
=
O_p(1).
\]
Using \(\|a+b\|^2\leq2\|a\|^2+2\|b\|^2\) and \eqref{eq:psi_decomposition} gives
\[
\begin{aligned}
&
\frac{1}{n}
\sum_{i=1}^n
\|\widehat{\bm{\varphi}}_{\psi i}-{\bm{\varphi}}_\psi(O_i)\|^2
\\
&\leq
2
\left\|
\dot\psi_{\widetilde{\bm\alpha}}
\widetilde\Gamma_Z^{-1}
-
\dot\psi_{\bm\alpha_\eta}
\Gamma_Z^{-1}
\right\|_{\mathrm{op}}^2
\frac{1}{n}
\sum_{i=1}^n
\|\bm Z_{\eta,i}\varepsilon_{\eta,i}\|^2
\\
&\quad+
2
\|\dot\psi_{\widetilde{\bm\alpha}}\|_{\mathrm{op}}^2
\|\widetilde\Gamma_Z^{-1}\|_{\mathrm{op}}^2
\|\widetilde{\bm\alpha}-\bm\alpha_\eta\|^2
\frac{1}{n}
\sum_{i=1}^n
\|\bm Z_{\eta,i}{\bm Z_{\eta,i}}^\top\|_{\mathrm{op}}^2.
\end{aligned}
\]
The first product on the right is \(o_p(1)O_p(1)=o_p(1)\). Continuity of
the derivative of \(\bm{\psi}\), \eqref{eq:oracle-gram-consistency}, and the
nonsingularity of \(\Gamma_Z\) make the first two factors in the second
product \(O_p(1)\). Equation~\eqref{eq:oracle-alpha-rate} makes its third
factor \(O_p(n^{-1})\), and the preceding law of large numbers makes its
final factor \(O_p(1)\). The second product is therefore \(O_p(n^{-1})\).
Consequently,
\begin{equation}
\label{eq:oracle-influence-consistency}
\frac{1}{n}
\sum_{i=1}^n
\|\widehat{\bm{\varphi}}_{\psi i}-{\bm{\varphi}}_\psi(O_i)\|^2
\overset{p}{\longrightarrow}
0.
\end{equation}

For the score-space influence terms, linearity gives
\[
\widehat{\mathcal I}_i-\mathcal I(O_i)
=
\dot{\mathcal F}_{\bm\vartheta_0}
\left[
\bm\vartheta_0-\widehat{\bm\vartheta}
\right].
\]
This expression is identical for every \(i\). Because vectorization preserves
the Frobenius norm,
\[
\begin{aligned}
\frac{1}{n}
\sum_{i=1}^n
\left\|
\operatorname{vec}\{\widehat{\mathcal I}_i-\mathcal I(O_i)\}
\right\|^2 &=
\frac{1}{n}
\sum_{i=1}^n
\left\|
\dot{\mathcal F}_{\bm\vartheta_0}
\left[
\bm\vartheta_0-\widehat{\bm\vartheta}
\right]
\right\|_{\mathrm F}^2\\
&=
\left\|
\dot{\mathcal F}_{\bm\vartheta_0}
\left[
\bm\vartheta_0-\widehat{\bm\vartheta}
\right]
\right\|_{\mathrm F}^2
\\
&\leq
\left\|
\dot{\mathcal F}_{\bm\vartheta_0}
\right\|_{\mathrm{op}}^2
\left\|
\widehat{\bm\vartheta}-\bm\vartheta_0
\right\|^2
\overset{p}{\longrightarrow}
0.
\end{aligned}
\]
Here, the derivative is bounded because
Lemma~\ref{lem:score-map-asymptotic-linearity} establishes continuous
differentiability at \(\bm\vartheta_0\), and
\eqref{eq:one-solve-moment-rate} gives
\(\widehat{\bm\vartheta}-\bm\vartheta_0=O_p(n^{-1/2})=o_p(1)\).
Combining \eqref{eq:oracle-influence-consistency} with the preceding
score-space influence bound gives
\begin{equation}
\label{eq:joint-influence-consistency}
\begin{aligned}
\frac{1}{n}
\sum_{i=1}^n
\left\|
\widehat{\bm W}_i-\bm W_i
\right\|^2
&=
\frac{1}{n}
\sum_{i=1}^n
\left\|
\widehat{\bm{\varphi}}_{\psi i}-{\bm{\varphi}}_\psi(O_i)
\right\|^2
\\
&\quad+
\frac{1}{n}
\sum_{i=1}^n
\left\|
\operatorname{vec}
\left\{
\widehat{\mathcal I}_i-\mathcal I(O_i)
\right\}
\right\|^2 \overset{p}{\longrightarrow}0.
\end{aligned}
\end{equation}

Let
\(
\Omega=\mathbb E(\bm W_i\bm W_i^\top)
\).
We established
\(\mathbb E\|{\bm{\varphi}}_\psi(O_i)\|^2<\infty\) in
\eqref{eq:oracle-influence-second-moment}. In addition,
Lemma~\ref{lem:score-map-asymptotic-linearity} gives
\(\mathbb E\|\mathcal I(O_i)\|_{\mathrm F}^2<\infty\) and
\(\mathbb E\{\mathcal I(O_i)\}=0\). Together with the previously established
identity \(\mathbb E\{{\bm{\varphi}}_\psi(O_i)\}=\bm 0\), these results give
\begin{equation}
\label{eq:joint-influence-moments}
\mathbb E(\bm W_i)=\bm 0,
\qquad
\mathbb E\|\bm W_i\|^2<\infty.
\end{equation}
The matrices \(\bm W_i\bm W_i^\top\) are independent and identically
distributed, and
\[
\mathbb E
\|\bm W_i\bm W_i^\top\|_{\mathrm F}
=
\mathbb E\|\bm W_i\|^2
<
\infty
\]
by \eqref{eq:joint-influence-moments}. Thus, all assumptions of the matrix law
of large numbers are satisfied, and it gives
\begin{equation} \label{eq:WW_Omega}
\frac{1}{n}
\sum_{i=1}^n
\bm W_i\bm W_i^\top
\overset{p}{\longrightarrow}
\Omega.
\end{equation}

For each \(i\), add and subtract
\(\bm W_i\widehat{\bm W}_i^\top\):
\[
\begin{aligned}
\widehat{\bm W}_i\widehat{\bm W}_i^\top
-
\bm W_i\bm W_i^\top
&=
\widehat{\bm W}_i\widehat{\bm W}_i^\top
-
\bm W_i\widehat{\bm W}_i^\top
+
\bm W_i\widehat{\bm W}_i^\top
-
\bm W_i\bm W_i^\top
\\
&=
\left(
\widehat{\bm W}_i-\bm W_i
\right)
\widehat{\bm W}_i^\top
+
\bm W_i
\left(
\widehat{\bm W}_i-\bm W_i
\right)^\top.
\end{aligned}
\]
Therefore,
\[
\begin{aligned}
&
\frac{1}{n}
\sum_{i=1}^n
\widehat{\bm W}_i\widehat{\bm W}_i^\top
-
\frac{1}{n}
\sum_{i=1}^n
\bm W_i\bm W_i^\top
\\
&=
\frac{1}{n}
\sum_{i=1}^n
\left(
\widehat{\bm W}_i-\bm W_i
\right)
\widehat{\bm W}_i^\top
+
\frac{1}{n}
\sum_{i=1}^n
\bm W_i
\left(
\widehat{\bm W}_i-\bm W_i
\right)^\top.
\end{aligned}
\]
The triangle inequality gives
\[
\begin{aligned}
&
\left\|
\frac{1}{n}
\sum_{i=1}^n
\widehat{\bm W}_i\widehat{\bm W}_i^\top
-
\frac{1}{n}
\sum_{i=1}^n
\bm W_i\bm W_i^\top
\right\|_{\mathrm{op}}
\\
&\leq
\frac{1}{n}
\sum_{i=1}^n
\left\|
\left(
\widehat{\bm W}_i-\bm W_i
\right)
\widehat{\bm W}_i^\top
\right\|_{\mathrm{op}}
+
\frac{1}{n}
\sum_{i=1}^n
\left\|
\bm W_i
\left(
\widehat{\bm W}_i-\bm W_i
\right)^\top
\right\|_{\mathrm{op}}.
\end{aligned}
\]
For vectors \(\bm a\) and \(\bm b\),
\(
\|\bm a\bm b^\top\|_{\mathrm{op}}
=
\|\bm a\|\|\bm b\|.
\)
Hence,
\[
\begin{aligned}
&
\left\|
\frac{1}{n}
\sum_{i=1}^n
\widehat{\bm W}_i\widehat{\bm W}_i^\top
-
\frac{1}{n}
\sum_{i=1}^n
\bm W_i\bm W_i^\top
\right\|_{\mathrm{op}}
\\
&\leq
\frac{1}{n}
\sum_{i=1}^n
\|\widehat{\bm W}_i-\bm W_i\|
\|\widehat{\bm W}_i\|
+
\frac{1}{n}
\sum_{i=1}^n
\|\widehat{\bm W}_i-\bm W_i\|
\|\bm W_i\|.
\end{aligned}
\]
Applying Cauchy--Schwarz to the first sum gives
\[
\begin{aligned}
&
\frac{1}{n}
\sum_{i=1}^n
\|\widehat{\bm W}_i-\bm W_i\|
\|\widehat{\bm W}_i\|
\leq
\left\{
\frac{1}{n}
\sum_{i=1}^n
\|\widehat{\bm W}_i-\bm W_i\|^2
\right\}^{1/2}
\left\{
\frac{1}{n}
\sum_{i=1}^n
\|\widehat{\bm W}_i\|^2
\right\}^{1/2}.
\end{aligned}
\]
Applying Cauchy--Schwarz to the second sum gives
\[
\begin{aligned}
&
\frac{1}{n}
\sum_{i=1}^n
\|\widehat{\bm W}_i-\bm W_i\|
\|\bm W_i\|
\leq
\left\{
\frac{1}{n}
\sum_{i=1}^n
\|\widehat{\bm W}_i-\bm W_i\|^2
\right\}^{1/2}
\left\{
\frac{1}{n}
\sum_{i=1}^n
\|\bm W_i\|^2
\right\}^{1/2}.
\end{aligned}
\]
Factoring out the common first term yields
\[
\begin{aligned}
&
\left\|
\frac{1}{n}
\sum_{i=1}^n
\widehat{\bm W}_i\widehat{\bm W}_i^\top
-
\frac{1}{n}
\sum_{i=1}^n
\bm W_i\bm W_i^\top
\right\|_{\mathrm{op}}
\\
&\leq
\left\{
\frac{1}{n}
\sum_{i=1}^n
\|\widehat{\bm W}_i-\bm W_i\|^2
\right\}^{1/2}
\left[
\left\{
\frac{1}{n}
\sum_{i=1}^n
\|\widehat{\bm W}_i\|^2
\right\}^{1/2}
+
\left\{
\frac{1}{n}
\sum_{i=1}^n
\|\bm W_i\|^2
\right\}^{1/2}
\right].
\end{aligned}
\]
Equation~\eqref{eq:joint-influence-consistency} makes the first factor
\(o_p(1)\). The law of large numbers makes the final term in brackets
\(O_p(1)\). To control the term involving \(\widehat{\bm W}_i\), write
\[
\widehat{\bm W}_i
=
\left(
\widehat{\bm W}_i-\bm W_i
\right)
+
\bm W_i.
\]
Applying the triangle inequality to the stacked vectors over all \(i\) gives
\[
\begin{aligned}
\left\{
\frac{1}{n}
\sum_{i=1}^n
\|\widehat{\bm W}_i\|^2
\right\}^{1/2}
&\leq
\left\{
\frac{1}{n}
\sum_{i=1}^n
\|\widehat{\bm W}_i-\bm W_i\|^2
\right\}^{1/2}
+
\left\{
\frac{1}{n}
\sum_{i=1}^n
\|\bm W_i\|^2
\right\}^{1/2}.
\end{aligned}
\]
Equation~\eqref{eq:joint-influence-consistency} gives
\[
\left\{
\frac{1}{n}
\sum_{i=1}^n
\|\widehat{\bm W}_i-\bm W_i\|^2
\right\}^{1/2}
=
o_p(1).
\]
Consequently,
\[
\left\{
\frac{1}{n}
\sum_{i=1}^n
\|\widehat{\bm W}_i\|^2
\right\}^{1/2}
\leq
o_p(1)+O_p(1)
=
O_p(1).
\]
Hence, 
\[
\left\|
\frac{1}{n}
\sum_{i=1}^n
\widehat{\bm W}_i\widehat{\bm W}_i^\top
-
\frac{1}{n}
\sum_{i=1}^n
\bm W_i\bm W_i^\top
\right\|_{\mathrm{op}}
\overset{p}{\longrightarrow}
0.
\]
By the above relation and \eqref{eq:WW_Omega}:
\[
\begin{aligned}
&
\left\|
\frac{1}{n}
\sum_{i=1}^n
\widehat{\bm W}_i\widehat{\bm W}_i^\top
-
\Omega
\right\|_{\mathrm{op}}
\\
&\leq
\left\|
\frac{1}{n}
\sum_{i=1}^n
\widehat{\bm W}_i\widehat{\bm W}_i^\top
-
\frac{1}{n}
\sum_{i=1}^n
\bm W_i\bm W_i^\top
\right\|_{\mathrm{op}}
+
\left\|
\frac{1}{n}
\sum_{i=1}^n
\bm W_i\bm W_i^\top
-
\Omega
\right\|_{\mathrm{op}}
\\
&=
o_p(1)+o_p(1)
\\
&=
o_p(1).
\end{aligned}
\]
Thus,
\begin{equation}
\label{eq:joint-bootstrap-covariance}
\frac{1}{n}
\sum_{i=1}^n
\widehat{\bm W}_i\widehat{\bm W}_i^\top
\overset{p}{\longrightarrow}
\Omega.
\end{equation}

We also verify maximal negligibility. The triangle inequality gives
\[
\max_{1\leq i\leq n}
\frac{\|\widehat{\bm W}_i\|}{\sqrt n}
\leq
\max_{1\leq i\leq n}
\frac{\|\bm W_i\|}{\sqrt n}
+
\max_{1\leq i\leq n}
\frac{\|\widehat{\bm W}_i-\bm W_i\|}{\sqrt n}.
\]
For the second term, Equation \eqref{eq:joint-influence-consistency} gives
\[
\max_{1\leq i\leq n}
\frac{\|\widehat{\bm W}_i-\bm W_i\|}{\sqrt n}
\leq
\left\{
\frac{1}{n}
\sum_{i=1}^n
\|\widehat{\bm W}_i-\bm W_i\|^2
\right\}^{1/2}
=
o_p(1).
\]
For the first term, fix
\(\varepsilon>0\). The union bound gives
\begin{equation}\label{eq:max_prob}
\begin{aligned}
\mathbb P
\left(
\max_{1\leq i\leq n}
\frac{\|\bm W_i\|}{\sqrt n}
>
\varepsilon
\right)
&\leq
n\mathbb P
\left(
\|\bm W_1\|>\varepsilon\sqrt n
\right)\\
&=
n\mathbb E
\left[
\mathbf 1
\left\{
\|\bm W_1\|>\varepsilon\sqrt n
\right\}
\right]
\\
&\leq
n\mathbb E
\left[
\frac{\|\bm W_1\|^2}{\varepsilon^2n}
\mathbf 1
\left\{
\|\bm W_1\|>\varepsilon\sqrt n
\right\}
\right]
\\
&=
\varepsilon^{-2}
\mathbb E
\left[
\|\bm W_1\|^2
\mathbf 1
\left\{
\|\bm W_1\|>\varepsilon\sqrt n
\right\}
\right].
\end{aligned}
\end{equation}
The second inequality follows because, on the event
\(\{\|\bm W_1\|>\varepsilon\sqrt n\}\), squaring both sides gives
\(\|\bm W_1\|^2>\varepsilon^2n\), and therefore
\(\|\bm W_1\|^2/(\varepsilon^2n)>1\). Hence, pointwise,
\[
\mathbf 1
\left\{
\|\bm W_1\|>\varepsilon\sqrt n
\right\}
\leq
\frac{\|\bm W_1\|^2}{\varepsilon^2n}
\mathbf 1
\left\{
\|\bm W_1\|>\varepsilon\sqrt n
\right\}.
\]
Now the integrand in the final expectation in \eqref{eq:max_prob} converges pointwise to zero and is
bounded by the integrable random variable \(\|\bm W_1\|^2\). Dominated
convergence therefore makes the final expectation converge to zero. We have
shown that
\begin{equation}
\label{eq:joint-maximal-negligibility}
\max_{1\leq i\leq n}
\frac{\|\widehat{\bm W}_i\|}{\sqrt n}
\overset{p}{\longrightarrow}
0.
\end{equation}

We now verify every condition of the conditional multinomial bootstrap central
limit theorem. The bootstrap weights are multinomial counts by construction,
and \eqref{eq:joint-empirical-influence-centering} gives the required
centering. Equation~\eqref{eq:joint-bootstrap-covariance} gives convergence of
the conditional covariance. Moreover,
\eqref{eq:joint-maximal-negligibility} implies that, for every
\(\varepsilon>0\),
\[
\mathbb E^*
\left[
\|\bm W^*\|^2
\mathbf 1
\left\{
\|\bm W^*\|>\varepsilon\sqrt n
\right\}
\right] =
\frac{1}{n}
\sum_{i=1}^n
\|\widehat{\bm W}_i\|^2
\mathbf 1
\left\{
\|\widehat{\bm W}_i\|>\varepsilon\sqrt n
\right\}
\overset{p}{\longrightarrow}
0,
\]
because the sum equals zero whenever
\(
\max_{1\leq i\leq n}
\|\widehat{\bm W}_i\|/\sqrt n\leq\varepsilon
\).
This is the required Lindeberg condition. The vectors have fixed finite
dimension because both the downstream target and the score-space projection
perturbation space are finite-dimensional. Thus, all assumptions of the
conditional multinomial bootstrap central limit theorem are satisfied, and it
gives
\[
\frac{1}{\sqrt n}
\sum_{i=1}^n
(N_i^*-1)\widehat{\bm W}_i
\rightsquigarrow_{p^*}
N(\bm 0,\Omega)
\]
in \(\mathbb P\)-probability. The vectors \(\bm W_i\) are independent and
identically distributed under the sampling setup, and
\eqref{eq:joint-influence-moments} gives their zero mean and finite covariance.
Thus, all assumptions of the ordinary multivariate central limit theorem are
satisfied. Applying that theorem to the leading sum in
\eqref{eq:joint-original-AL}, and then applying Slutsky's theorem to its
\(o_p(1)\) remainder, gives
\[
\bm A_n
\rightsquigarrow
\bm G,
\qquad
\bm G\sim N(\bm 0,\Omega).
\]
Similarly, \eqref{eq:joint-bootstrap-AL} and the conditional bootstrap central
limit theorem give
\[
\bm A_n^*
\rightsquigarrow_{p^*}
\bm G
\]
in \(\mathbb P\)-probability. Thus, for every bounded Lipschitz function
\(g\), the two convergence statements imply
\[
\left|
\mathbb E g(\bm A_n)
-
\mathbb E g(\bm G)
\right|
\longrightarrow
0
\]
and
\[
\left|
\mathbb E^*g(\bm A_n^*)
-
\mathbb E g(\bm G)
\right|
\overset{p}{\longrightarrow}
0.
\]
The triangle inequality gives
\[
\begin{aligned}
&
\left|
\mathbb E^*g(\bm A_n^*)
-
\mathbb E g(\bm A_n)
\right|
\leq
\left|
\mathbb E^*g(\bm A_n^*)
-
\mathbb E g(\bm G)
\right|
+
\left|
\mathbb E g(\bm G)
-
\mathbb E g(\bm A_n)
\right|.
\end{aligned}
\]
Both terms on the right converge to zero, so
\begin{equation}
\label{eq:joint-bootstrap-BL}
\left|
\mathbb E^*g(\bm A_n^*)
-
\mathbb E g(\bm A_n)
\right|
\overset{p}{\longrightarrow}
0.
\end{equation}

This proves the lemma.
\end{proof}

\begin{lemma}
\label{lem:generated-score-ols-expansion}
Suppose Assumption~\ref{assump:L-asymptotic-linearity} and the conditions of
Lemma~\ref{lem:oracle-downstream-sampling-expansion} hold. Suppose additionally
that \(\bm{\psi}\) is basis invariant. Then there exists a deterministic bounded
linear map \(\mathcal B\)
such that
\[
\sqrt n(\widehat{\bm{\psi}}-\widetilde{\bm{\psi}})
=
\mathcal B
\left[
\sqrt n
\left(
P_{\widehat{\mathcal S}_\eta}
-
P_{\mathcal S_\eta}
\right)
\right]
+
o_p(1)
\]
and
\[
\begin{aligned}
&
\sqrt n
\left\{
(\widehat{\bm{\psi}}^*-\widetilde{\bm{\psi}}^*)
-
(\widehat{\bm{\psi}}-\widetilde{\bm{\psi}})
\right\}
=
\mathcal B
\left[
\sqrt n
\left(
P_{\widehat{\mathcal S}_\eta^*}
-
P_{\widehat{\mathcal S}_\eta}
\right)
\right]
+
o_{p^*}(1)
\end{aligned}
\]
in \(\mathbb P\)-probability.
\end{lemma}

\begin{proof}
The complete downstream construction is the composition
\[
P
\xrightarrow{\textnormal{(a)}}
B(P)
\xrightarrow{\textnormal{(b)}}
\{\bm S_i(P)\}_{i=1}^n
\xrightarrow{\textnormal{(c)}}
\left\{
\widehat\Gamma_D(P),
\widehat{\bm g}_D(P)
\right\}
\xrightarrow{\textnormal{(d)}}
\widehat{\bm\alpha}(P)
\xrightarrow{\textnormal{(e)}}
\bm{\Psi}_n(P).
\]
The precise
definition of each quantity is given as each map is considered below. We prove that each map in this composition is continuously differentiable near
\(P_{\mathcal S_\eta}\).

\medskip
\noindent\textit{Map \(\textnormal{(a)}\).}
The proof of Lemma~\ref{lem:score-map-asymptotic-linearity} constructs a
continuously differentiable local basis \(B(P)\) for score spaces near
\(P_{\mathcal S_\eta}\) and shows that this basis remains linearly independent
throughout a sufficiently small neighborhood. Therefore,
\(P\mapsto B(P)\) is continuously differentiable.

\medskip
\noindent\textit{Map \(\textnormal{(b)}\).}
For each unit,
\(
\bm S_i(P)
=
B(P)^\top\bm X_i.
\)
For fixed \(\bm X_i\), this map is linear in \(B(P)\). Therefore,
\(B(P)\mapsto\{\bm S_i(P)\}_{i=1}^n\) is continuously differentiable.

\medskip
\noindent\textit{Map \(\textnormal{(c)}\).}
Define
\[
\bm D_i(P)
=
\left(
1,\,
T_i,\,
\bm S_i(P)^\top,\,
T_i\bm S_i(P)^\top
\right)^\top.
\]
The empirical Gram matrix and regressor--outcome moment are
\[
\widehat\Gamma_D(P)
=
\frac{1}{n}
\sum_{i=1}^n
\bm D_i(P)\bm D_i(P)^\top
\]
and
\[
\widehat{\bm g}_D(P)
=
\frac{1}{n}
\sum_{i=1}^n
\bm D_i(P)Y_i.
\]
These quantities are finite sums and products of the score coordinates.
Therefore, the map from
\(\{\bm S_i(P)\}_{i=1}^n\) to
\(\{\widehat\Gamma_D(P),\widehat{\bm g}_D(P)\}\) is continuously
differentiable.

\medskip
\noindent\textit{Map \(\textnormal{(d)}\).}
Whenever \(\widehat\Gamma_D(P)\) is nonsingular, ordinary least squares gives
\[
\widehat{\bm\alpha}(P)
=
\widehat\Gamma_D(P)^{-1}
\widehat{\bm g}_D(P).
\]
Matrix inversion is continuously differentiable on the set of nonsingular
matrices, and matrix multiplication is continuously differentiable.
After verifying that \(\widehat\Gamma_D(P)\) remains nonsingular near
\(P_{\mathcal S_\eta}\) with probability tending to one, it follows that map
\(\textnormal{(d)}\) is continuously differentiable on that neighborhood.

\medskip
\noindent\textit{Map \(\textnormal{(e)}\).}
The conditions of
Lemma~\ref{lem:oracle-downstream-sampling-expansion}, which are assumed here,
state that the reported target is a continuously differentiable function of
the fitted downstream regression. Therefore,
\(\widehat{\bm\alpha}(P)\mapsto\bm{\Psi}_n(P)\) is continuously differentiable.

We have shown sequentially that maps
\(\textnormal{(a)}\)--\(\textnormal{(e)}\) are continuously differentiable
near \(P_{\mathcal S_\eta}\) with probability tending to one. The chain rule
therefore implies that the complete map
\[
P\longmapsto\bm{\Psi}_n(P)
\]
is continuously differentiable on that neighborhood with probability tending
to one. 

Let \(\bm{\omega}(O_i)\) collect the distinct scalar entries of the random
quantities required to form the downstream Gram matrix, regressor--outcome
moment, and reported contrast:
\[
1,\quad
T_i,\quad
T_i^2,\quad
\bm X_i,\quad
T_i\bm X_i,\quad
T_i^2\bm X_i,\quad
\bm X_i\bm X_i^\top,\quad
T_i\bm X_i\bm X_i^\top,\quad
T_i^2\bm X_i\bm X_i^\top,
\]
together with
\[
Y_i,\quad
T_iY_i,\quad
Y_i\bm X_i,\quad
T_iY_i\bm X_i.
\]
Because the scores are linear functions of \(\bm X_i\), these raw moments
determine the downstream Gram matrix and regressor--outcome moment for every
local score basis. Their first derivatives with respect to the score-space
coordinates depend on the same raw moments.

By the definition of \(\bm{\chi}(O_i)\),
Assumption~\ref{assump:L-asymptotic-linearity}(b) implies
\begin{equation}
\label{eq:chi-implied-downstream-moments}
\begin{aligned}
&
\mathbb E\|\bm X_i\|^2<\infty,
\qquad
\mathbb E(T_i^2)<\infty,
\qquad
\mathbb E(Y_i^2)<\infty,
\\
&
\mathbb E\|\bm X_i\|^4<\infty,
\qquad
\mathbb E\left\{T_i^2\|\bm X_i\|^2\right\}<\infty,
\qquad
\mathbb E\left\{Y_i^2\|\bm X_i\|^2\right\}<\infty,
\\
&
\mathbb E(T_i^4)<\infty,
\qquad
\mathbb E(T_i^2Y_i^2)<\infty.
\end{aligned}
\end{equation}
Applying the Cauchy--Schwarz inequality to the quantities in
\eqref{eq:chi-implied-downstream-moments} gives
\begin{equation}
\label{eq:downstream-required-first-moments}
\begin{aligned}
\mathbb E\left\{
|T_i|^a\|\bm X_i\|^b
\right\}
&<\infty,
&&a\in\{0,1,2\},\quad b\in\{0,1,2\},
\\
\mathbb E\left\{
|T_i|^a|Y_i|\|\bm X_i\|^b
\right\}
&<\infty,
&&a\in\{0,1\},\quad b\in\{0,1\}.
\end{aligned}
\end{equation}
Every component of \(\bm{\omega}(O_i)\) is covered by
\eqref{eq:downstream-required-first-moments}. Because \(\bm{\omega}(O_i)\) is
finite-dimensional,
\[
\mathbb E\|\bm{\omega}(O_i)\|_2
\leq
\mathbb E\|\bm{\omega}(O_i)\|_1
=
\sum_{\ell}
\mathbb E|\omega_\ell(O_i)|
<
\infty.
\]

Now define
\[
\widehat{\bm\mu}
=
\frac{1}{n}
\sum_{i=1}^n
\bm{\omega}(O_i),
\qquad
\bm\mu_0
=
\mathbb E\{\bm{\omega}(O_i)\}.
\]
The observations are independent and identically distributed by the sampling
condition in Lemma~\ref{lem:oracle-downstream-sampling-expansion}, and the
preceding display establishes the required finite first moment. Thus, all
assumptions of the multivariate law of large numbers are satisfied, and
\begin{equation}\label{eq:multivar_LLN}
\widehat{\bm\mu}
\overset{p}{\longrightarrow}
\bm\mu_0.
\end{equation}

To summarize, the preceding arguments for maps
\(\textnormal{(a)}\)--\(\textnormal{(e)}\) show that there is a deterministic
map \(\mathcal G\), continuously differentiable near
\((P_{\mathcal S_\eta},\bm\mu_0)\), such that
\[
\bm{\Psi}_n(P)
=
\mathcal G(P,\widehat{\bm\mu}).
\]
In particular,
\[
\widehat{\bm{\psi}}
=
\mathcal G
\left(
P_{\widehat{\mathcal S}_\eta},
\widehat{\bm\mu}
\right),
\qquad
\widetilde{\bm{\psi}}
=
\mathcal G
\left(
P_{\mathcal S_\eta},
\widehat{\bm\mu}
\right).
\]
Moreover, \eqref{eq:one-solve-score-AL} gives
\[
P_{\widehat{\mathcal S}_\eta}
\overset{p}{\longrightarrow}
P_{\mathcal S_\eta}.
\]
Combining this convergence with
\(
\widehat{\bm\mu}\overset{p}{\longrightarrow}\bm\mu_0
\)
gives
\[
\left(
P_{\widehat{\mathcal S}_\eta},
\widehat{\bm\mu}
\right)
\overset{p}{\longrightarrow}
\left(
P_{\mathcal S_\eta},
\bm\mu_0
\right).
\]

Thus, the evaluation point converges to
\((P_{\mathcal S_\eta},\bm\mu_0)\), and \(\mathcal G\) is continuously
differentiable on a neighborhood of that point. The assumptions of the
finite-dimensional first-order Taylor expansion are therefore satisfied.
Applying the expansion to the first argument of \(\mathcal G\), while holding
\(\widehat{\bm\mu}\) fixed, gives
\begin{equation}
\label{eq:direct-first-order-Pstar}
\widehat{\bm{\psi}}-\widetilde{\bm{\psi}}
=
\mathcal B_n
\left[
P_{\widehat{\mathcal S}_\eta}
-
P_{\mathcal S_\eta}
\right]
+
r_n,
\end{equation}
where
\[
\mathcal B_n
=
\dot{\mathcal G}_{P}
\left(
P_{\mathcal S_\eta},
\widehat{\bm\mu}
\right)
\]
is the derivative with respect to the first argument and
\[
\frac{\|r_n\|}
{
\left\|
P_{\widehat{\mathcal S}_\eta}
-
P_{\mathcal S_\eta}
\right\|_{\mathrm{op}}
}
\overset{p}{\longrightarrow}
0.
\]
Finally, continuity of the derivative and
\(
\widehat{\bm\mu}\overset{p}{\longrightarrow}\bm\mu_0
\)
give
\[
\mathcal B_n
\overset{p}{\longrightarrow}
\dot{\mathcal G}_{P}
\left(
P_{\mathcal S_\eta},
\bm\mu_0
\right)
=:
\mathcal B.
\]

Because \(\mathcal B\) is the derivative of a finite-dimensional map at the
fixed population point
\((P_{\mathcal S_\eta},\bm\mu_0)\), it is a deterministic bounded linear map.

It remains to convert \eqref{eq:direct-first-order-Pstar} into the stated
root-\(n\) expansion. Equation \eqref{eq:one-solve-score-AL} gives
\begin{equation}
\label{eq:direct-projector-rate}
\sqrt n
\left\|
P_{\widehat{\mathcal S}_\eta}
-
P_{\mathcal S_\eta}
\right\|_{\mathrm{op}}
=
O_p(1).
\end{equation}
Combining this rate with the remainder bound in
\eqref{eq:direct-first-order-Pstar} gives
\[
\begin{aligned}
\sqrt n\|r_n\|
&=
\frac{\|r_n\|}
{
\left\|
P_{\widehat{\mathcal S}_\eta}
-
P_{\mathcal S_\eta}
\right\|_{\mathrm{op}}
}
\sqrt n
\left\|
P_{\widehat{\mathcal S}_\eta}
-
P_{\mathcal S_\eta}
\right\|_{\mathrm{op}}
\\
&=
o_p(1)O_p(1)
=
o_p(1).
\end{aligned}
\]
Multiplying \eqref{eq:direct-first-order-Pstar} by \(\sqrt n\) gives
\[
\sqrt n(\widehat{\bm{\psi}}-\widetilde{\bm{\psi}})
=
\mathcal B_n
\left[
\sqrt n
\left(
P_{\widehat{\mathcal S}_\eta}
-
P_{\mathcal S_\eta}
\right)
\right]
+
o_p(1).
\]
Add and subtract
\(
\mathcal B[
\sqrt n(P_{\widehat{\mathcal S}_\eta}-P_{\mathcal S_\eta})
]
\)
on the right-hand side. This gives
\[
\begin{aligned}
\sqrt n(\widehat{\bm{\psi}}-\widetilde{\bm{\psi}})
&=
\mathcal B
\left[
\sqrt n
\left(
P_{\widehat{\mathcal S}_\eta}
-
P_{\mathcal S_\eta}
\right)
\right]
\\
&\quad+
(\mathcal B_n-\mathcal B)
\left[
\sqrt n
\left(
P_{\widehat{\mathcal S}_\eta}
-
P_{\mathcal S_\eta}
\right)
\right]
+
o_p(1).
\end{aligned}
\]
The second term satisfies
\[
\begin{aligned}
&
\left\|
(\mathcal B_n-\mathcal B)
\left[
\sqrt n
\left(
P_{\widehat{\mathcal S}_\eta}
-
P_{\mathcal S_\eta}
\right)
\right]
\right\|
\\
&\leq
\|\mathcal B_n-\mathcal B\|_{\mathrm{op}}
\sqrt n
\left\|
P_{\widehat{\mathcal S}_\eta}
-
P_{\mathcal S_\eta}
\right\|_{\mathrm{op}}
\\
&=
o_p(1)O_p(1)
=
o_p(1).
\end{aligned}
\]
Therefore,
\begin{equation}
\label{eq:direct-generated-score-original}
\sqrt n(\widehat{\bm{\psi}}-\widetilde{\bm{\psi}})
=
\mathcal B
\left[
\sqrt n
\left(
P_{\widehat{\mathcal S}_\eta}
-
P_{\mathcal S_\eta}
\right)
\right]
+
o_p(1).
\end{equation}

We next establish the corresponding bootstrap expansion. Define the bootstrap
empirical moment vector by
\[
\widehat{\bm\mu}^*
=
\frac{1}{n}
\sum_{i=1}^n
N_i^*\bm{\omega}(O_i).
\]
Using the definition of \(\widehat{\bm\mu}\), we obtain
\[
\begin{aligned}
\widehat{\bm\mu}^*-\widehat{\bm\mu}
&=
\frac{1}{n}
\sum_{i=1}^n
N_i^*\bm{\omega}(O_i)
-
\frac{1}{n}
\sum_{i=1}^n
\bm{\omega}(O_i)
\\
&=
\frac{1}{n}
\sum_{i=1}^n
(N_i^*-1)\bm{\omega}(O_i).
\end{aligned}
\]
The vectors \(\bm{\omega}(O_i)\) are independent and identically distributed by the
sampling condition in
Lemma~\ref{lem:oracle-downstream-sampling-expansion}. They are finite
dimensional, and the preceding moment calculation established that
\(\mathbb E\|\bm{\omega}(O_i)\|<\infty\). Thus, all assumptions of
Lemma~\ref{lem:conditional-bootstrap-mean-consistency} are satisfied.
Applying that lemma gives
\begin{equation}
\label{eq:bootstrap-downstream-moment-consistency}
\left\|
\widehat{\bm\mu}^*
-
\widehat{\bm\mu}
\right\|
=
o_{p^*}(1)
\end{equation}
in \(\mathbb P\)-probability.

The bootstrap estimators can be written as
\[
\widehat{\bm{\psi}}^*
=
\mathcal G
\left(
P_{\widehat{\mathcal S}_\eta^*},
\widehat{\bm\mu}^*
\right),
\qquad
\widetilde{\bm{\psi}}^*
=
\mathcal G
\left(
P_{\mathcal S_\eta},
\widehat{\bm\mu}^*
\right).
\]
Equation \eqref{eq:one-solve-score-bootstrap-AL} gives
\begin{equation}
\label{eq:direct-bootstrap-projector-rate}
\sqrt n
\left\|
P_{\widehat{\mathcal S}_\eta^*}
-
P_{\widehat{\mathcal S}_\eta}
\right\|_{\mathrm{op}}
=
O_{p^*}(1)
\end{equation}
in \(\mathbb P\)-probability. The triangle inequality gives
\[
\begin{aligned}
\sqrt n
\left\|
P_{\widehat{\mathcal S}_\eta^*}
-
P_{\mathcal S_\eta}
\right\|_{\mathrm{op}}
&\leq
\sqrt n
\left\|
P_{\widehat{\mathcal S}_\eta^*}
-
P_{\widehat{\mathcal S}_\eta}
\right\|_{\mathrm{op}}
\\
&\quad+
\sqrt n
\left\|
P_{\widehat{\mathcal S}_\eta}
-
P_{\mathcal S_\eta}
\right\|_{\mathrm{op}}
\\
&=
O_{p^*}(1)+O_p(1)
=
O_{p^*}(1)
\end{aligned}
\]
in \(\mathbb P\)-probability. In particular,
\[
P_{\widehat{\mathcal S}_\eta^*}
-
P_{\mathcal S_\eta}
=
o_{p^*}(1)
\]
in \(\mathbb P\)-probability.

Moreover,
\[
\widehat{\bm\mu}^*-\bm\mu_0
=
\left(
\widehat{\bm\mu}^*
-
\widehat{\bm\mu}
\right)
+
\left(
\widehat{\bm\mu}
-
\bm\mu_0
\right).
\]
The first term is \(o_{p^*}(1)\) by
\eqref{eq:bootstrap-downstream-moment-consistency}, and the second is
\(o_p(1)\) by \eqref{eq:multivar_LLN}.
Therefore,
\[
\widehat{\bm\mu}^*
\overset{p^*}{\longrightarrow}
\bm\mu_0
\]
in \(\mathbb P\)-probability. Hence, the bootstrap evaluation points lie in
the neighborhood on which \(\mathcal G\) is continuously differentiable with
conditional probability tending to one.

Applying the first-order Taylor expansion to the first argument of
\(\mathcal G\), while holding \(\widehat{\bm\mu}^*\) fixed, gives
\begin{equation}
\label{eq:direct-bootstrap-first-order-Pstar}
\widehat{\bm{\psi}}^*-\widetilde{\bm{\psi}}^*
=
\mathcal B_n^*
\left[
P_{\widehat{\mathcal S}_\eta^*}
-
P_{\mathcal S_\eta}
\right]
+
r_n^*,
\end{equation}
where
\[
\mathcal B_n^*
=
\dot{\mathcal G}_{P}
\left(
P_{\mathcal S_\eta},
\widehat{\bm\mu}^*
\right)
\]
and
\[
\frac{\|r_n^*\|}
{
\left\|
P_{\widehat{\mathcal S}_\eta^*}
-
P_{\mathcal S_\eta}
\right\|_{\mathrm{op}}
}
=
o_{p^*}(1)
\]
in \(\mathbb P\)-probability. Combining this remainder bound with the preceding
bootstrap projector rate gives
\[
\begin{aligned}
\sqrt n\|r_n^*\|
&=
\frac{\|r_n^*\|}
{
\left\|
P_{\widehat{\mathcal S}_\eta^*}
-
P_{\mathcal S_\eta}
\right\|_{\mathrm{op}}
}
\sqrt n
\left\|
P_{\widehat{\mathcal S}_\eta^*}
-
P_{\mathcal S_\eta}
\right\|_{\mathrm{op}}
\\
&=
o_{p^*}(1)O_{p^*}(1)
=
o_{p^*}(1)
\end{aligned}
\]
in \(\mathbb P\)-probability.

Continuity of \(\dot{\mathcal G}_{P}\) and
\eqref{eq:bootstrap-downstream-moment-consistency} give
\begin{equation}
\label{eq:direct-bootstrap-derivative-consistency}
\begin{aligned}
\|\mathcal B_n^*-\mathcal B_n\|_{\mathrm{op}}
&=
\left\|
\dot{\mathcal G}_{P}
\left(
P_{\mathcal S_\eta},
\widehat{\bm\mu}^*
\right)
-
\dot{\mathcal G}_{P}
\left(
P_{\mathcal S_\eta},
\widehat{\bm\mu}
\right)
\right\|_{\mathrm{op}}
\\
&=
o_{p^*}(1)
\end{aligned}
\end{equation}
in \(\mathbb P\)-probability.

Subtracting \eqref{eq:direct-first-order-Pstar} from
\eqref{eq:direct-bootstrap-first-order-Pstar} gives
\[
\begin{aligned}
&
(\widehat{\bm{\psi}}^*-\widetilde{\bm{\psi}}^*)
-
(\widehat{\bm{\psi}}-\widetilde{\bm{\psi}})
\\
&=
\mathcal B_n^*
\left[
P_{\widehat{\mathcal S}_\eta^*}
-
P_{\mathcal S_\eta}
\right]
-
\mathcal B_n
\left[
P_{\widehat{\mathcal S}_\eta}
-
P_{\mathcal S_\eta}
\right]
+
(r_n^*-r_n).
\end{aligned}
\]
Add and subtract
\(
\mathcal B_n[
P_{\widehat{\mathcal S}_\eta^*}-P_{\mathcal S_\eta}
]
\)
on the right-hand side. This gives
\[
\begin{aligned}
&
(\widehat{\bm{\psi}}^*-\widetilde{\bm{\psi}}^*)
-
(\widehat{\bm{\psi}}-\widetilde{\bm{\psi}})
\\
&=
(\mathcal B_n^*-\mathcal B_n)
\left[
P_{\widehat{\mathcal S}_\eta^*}
-
P_{\mathcal S_\eta}
\right]
+
\mathcal B_n
\left[
P_{\widehat{\mathcal S}_\eta^*}
-
P_{\mathcal S_\eta}
\right]
-
\mathcal B_n
\left[
P_{\widehat{\mathcal S}_\eta}
-
P_{\mathcal S_\eta}
\right]
+
(r_n^*-r_n).
\end{aligned}
\]
Because \(\mathcal B_n\) is linear, the middle two terms satisfy
\[
\begin{aligned}
&
\mathcal B_n
\left[
P_{\widehat{\mathcal S}_\eta^*}
-
P_{\mathcal S_\eta}
\right]
-
\mathcal B_n
\left[
P_{\widehat{\mathcal S}_\eta}
-
P_{\mathcal S_\eta}
\right]
\\
&=
\mathcal B_n
\left[
\left(
P_{\widehat{\mathcal S}_\eta^*}
-
P_{\mathcal S_\eta}
\right)
-
\left(
P_{\widehat{\mathcal S}_\eta}
-
P_{\mathcal S_\eta}
\right)
\right]
\\
&=
\mathcal B_n
\left[
P_{\widehat{\mathcal S}_\eta^*}
-
P_{\widehat{\mathcal S}_\eta}
\right].
\end{aligned}
\]
Therefore,
\begin{equation}
\label{eq:direct-bootstrap-generated-decomposition}
\begin{aligned}
&
(\widehat{\bm{\psi}}^*-\widetilde{\bm{\psi}}^*)
-
(\widehat{\bm{\psi}}-\widetilde{\bm{\psi}})
\\
&=
\mathcal B_n
\left[
P_{\widehat{\mathcal S}_\eta^*}
-
P_{\widehat{\mathcal S}_\eta}
\right]
+
(\mathcal B_n^*-\mathcal B_n)
\left[
P_{\widehat{\mathcal S}_\eta^*}
-
P_{\mathcal S_\eta}
\right]
+
(r_n^*-r_n).
\end{aligned}
\end{equation}

We now control the final two terms. First,
\[
\begin{aligned}
&
\sqrt n
\left\|
(\mathcal B_n^*-\mathcal B_n)
\left[
P_{\widehat{\mathcal S}_\eta^*}
-
P_{\mathcal S_\eta}
\right]
\right\|
\\
&\leq
\|\mathcal B_n^*-\mathcal B_n\|_{\mathrm{op}}
\sqrt n
\left\|
P_{\widehat{\mathcal S}_\eta^*}
-
P_{\mathcal S_\eta}
\right\|_{\mathrm{op}}
\\
&=
o_{p^*}(1)O_{p^*}(1)
=
o_{p^*}(1)
\end{aligned}
\]
in \(\mathbb P\)-probability. Second,
\[
\begin{aligned}
\sqrt n\|r_n^*-r_n\|
&\leq
\sqrt n\|r_n^*\|
+
\sqrt n\|r_n\|
\\
&=
o_{p^*}(1)+o_p(1)
=
o_{p^*}(1)
\end{aligned}
\]
in \(\mathbb P\)-probability.

It remains to replace \(\mathcal B_n\) in the first term of
\eqref{eq:direct-bootstrap-generated-decomposition} by \(\mathcal B\). Add and
subtract
\(
\mathcal B[
P_{\widehat{\mathcal S}_\eta^*}-P_{\widehat{\mathcal S}_\eta}]
\)
to obtain
\[
\begin{aligned}
&
\mathcal B_n
\left[
\sqrt n
\left(
P_{\widehat{\mathcal S}_\eta^*}
-
P_{\widehat{\mathcal S}_\eta}
\right)
\right]
\\
&=
\mathcal B
\left[
\sqrt n
\left(
P_{\widehat{\mathcal S}_\eta^*}
-
P_{\widehat{\mathcal S}_\eta}
\right)
\right]
+
(\mathcal B_n-\mathcal B)
\left[
\sqrt n
\left(
P_{\widehat{\mathcal S}_\eta^*}
-
P_{\widehat{\mathcal S}_\eta}
\right)
\right].
\end{aligned}
\]
The final term satisfies
\[
\begin{aligned}
&
\left\|
(\mathcal B_n-\mathcal B)
\left[
\sqrt n
\left(
P_{\widehat{\mathcal S}_\eta^*}
-
P_{\widehat{\mathcal S}_\eta}
\right)
\right]
\right\|
\\
&\leq
\|\mathcal B_n-\mathcal B\|_{\mathrm{op}}
\sqrt n
\left\|
P_{\widehat{\mathcal S}_\eta^*}
-
P_{\widehat{\mathcal S}_\eta}
\right\|_{\mathrm{op}}
\\
&=
o_p(1)O_{p^*}(1)
=
o_{p^*}(1)
\end{aligned}
\]
in \(\mathbb P\)-probability.

Multiplying
\eqref{eq:direct-bootstrap-generated-decomposition} by \(\sqrt n\) and
combining the preceding bounds gives
\[
\begin{aligned}
&
\sqrt n
\left\{
(\widehat{\bm{\psi}}^*-\widetilde{\bm{\psi}}^*)
-
(\widehat{\bm{\psi}}-\widetilde{\bm{\psi}})
\right\}
\\
&=
\mathcal B
\left[
\sqrt n
\left(
P_{\widehat{\mathcal S}_\eta^*}
-
P_{\widehat{\mathcal S}_\eta}
\right)
\right]
+
o_{p^*}(1)
\end{aligned}
\]
in \(\mathbb P\)-probability. This proves the lemma.
\end{proof}

\section{Simulation Setup} \label{supp:setup}
We first conducted fully synthetic experiments with \(n=1000\) observations and
\(d_C=20\) latent variables. Within each Monte Carlo replication \(r\), we
independently generated a pool of \(500\) candidate proxy variables and their
corresponding loading vectors. Specifically, for \(j=1,\ldots,500\), we drew
\(
\widetilde{\bm\lambda}_{rj}
\sim
N_{20}(\bm 0,I_{20})
\)
and normalized each vector as
\[
\bm\lambda_{rj}
=
\frac{\widetilde{\bm\lambda}_{rj}}
{\|\widetilde{\bm\lambda}_{rj}\|_2}.
\]
Collecting these loading vectors as rows produced the replication-specific
matrix
\[
\Lambda_r^{(500)}
=
\begin{pmatrix}
\bm\lambda_{r1}^{\top}\\
\vdots\\
\bm\lambda_{r500}^{\top}
\end{pmatrix}
\in\mathbb R^{500\times20}.
\]
We generated a new \(\Lambda_r^{(500)}\) independently in every Monte Carlo
replication but used the same matrix across all conditions within that
replication.

The conditions with \(p=100\), \(300\), and \(500\) used proxy variables
\(1{:}100\), \(1{:}300\), and \(1{:}500\), respectively. For notational
convenience, let
\[
\Lambda_r^{(p)}
=
\Lambda_{r,1:p}^{(500)},
\qquad
p\in\{100,300,500\},
\]
denote the loading vectors associated with the first \(p\) proxy variables.
The observed proxy systems were therefore exactly nested within each
replication: the \(p=300\) condition retained all \(100\) variables from the
\(p=100\) condition and added \(200\) variables, while the \(p=500\) condition
retained all \(300\) variables from the \(p=300\) condition and added another
\(200\).

For each Monte Carlo replication, we generated a dataset containing
\(n=1000\) observations. For each observation \(i=1,\ldots,1000\), we
independently drew a \(20\)-dimensional latent vector
$\bm C_i\sim N_{20}(\bm 0,I_{20}).$
We independently drew
\(\bm z_a,\bm z_b,\bm z_q\sim N_{20}(\bm 0,I_{20})\) and defined
\[
\bm a
=
1.25\frac{\bm z_a}{\|\bm z_a\|_2},
\qquad
\bm b
=
\frac{\bm z_b}{\|\bm z_b\|_2},
\qquad
\bm q
=
0.75\frac{\bm z_q}{\|\bm z_q\|_2}.
\]
Treatment assignment followed
\[
T_i\mid\bm C_i
\sim
\operatorname{Bernoulli}
\left\{
\operatorname{logit}^{-1}
(\bm a^\top\bm C_i)
\right\},
\]
and the outcome followed
\[
Y_i
=
\bm b^\top\bm C_i
+
T_i\{1+\bm q^\top\bm C_i\}
+
\varepsilon_{Y,i},
\qquad
\varepsilon_{Y,i}\sim N(0,1).
\]
The latent CATE was therefore
\(
\tau_{\bm C}(1,0;\bm C_i)
=
1+\bm q^\top\bm C_i,
\)
and the ATE was its sample mean.

The treatment-guided latent coordinate was
\(
U_{T,i}
=
\bm a^\top\bm C_i.
\)
For a Gaussian latent state and a zero-intercept logistic treatment model,
\(\bm h_q=\bm q/2\). We therefore defined
$D
=
\left(
\bm a,\,
\bm b+\frac{\bm q}{2}
\right)$
and
\[
\bm U_i
=
D^\top\bm C_i
=
(U_{T,i},U_{R,i})^\top,
\qquad
U_{R,i}
=
\left(\bm b+\frac{\bm q}{2}\right)^\top\bm C_i.
\]
The exact score-indexed CATE was
\[
\tau_{\bm U}(1,0;\bm U_i)
=
1+
\bm U_i^\top
(D^\top D)^{-1}
D^\top\bm q.
\]

For each condition, we generated the observed proxies according to
\[
\bm X_i^{(p)}
=
\sqrt{\upsilon}\,
\Lambda_r^{(p)}\bm C_i
+
\sqrt{1-\upsilon}\,
\bm\varepsilon_{X,i}^{(p)},
\qquad
\bm\varepsilon_{X,i}^{(p)}
\sim
N_p(\bm 0,I_p),
\]
where the proxy errors were independent of the latent state, treatment, and
outcome errors. Because every loading vector had unit Euclidean norm,
\(\upsilon\) represented the fraction of each proxy variable's variance
attributable to the latent state, and the corresponding signal-to-noise ratio
was \(\upsilon/(1-\upsilon)\). We varied
\[
\upsilon\in\{0.2,0.4,0.6,0.8\}
\quad\text{and}\quad
p\in\{100,300,500\},
\]
producing \(12\) conditions. Because COMPACT and the ridge-based
methods were computationally inexpensive, their comparisons used
\(500\) Monte Carlo replications per condition, whereas the broader
algorithm comparisons used \(50\) replications per condition.

Finally, we evaluated full-procedure bootstrap inference in a separate
experiment with \(n=1000\), \(p=300\), \(d_C=20\), and proxy reliability
\(\upsilon=0.5\). We drew all randomly generated data-generating parameter
vectors, including \(\widetilde{\bm q}\), once at the beginning of the
experiment and held them fixed across the \(200\) Monte Carlo replications.
Within each replication, we regenerated the latent states, proxies, treatments,
and outcome errors as above, but parameterized the treatment effect as
\[
\tau_i
=
\tau_0+\tau_U\widetilde{\bm q}^{\top}\bm C_i,
\qquad
\|\widetilde{\bm q}\|_2=1,
\]
where \(\tau_0\in\{0,0.5\}\) determined the population ATE and
\(\tau_U\in\{0,0.75\}\) determined treatment-effect heterogeneity. For each
original dataset, we applied COMPACT to estimate the two-dimensional score
representation \(\widehat{\bm S}_i\), fit the downstream OLS model in
Equation~\eqref{eq:regression_model}, and calculated
\[
\widehat\tau_i
=
\widehat\alpha_T
+
\widehat{\bm{\alpha}}_{TS}^{\top}\widehat{\bm S}_i,
\qquad
\widehat\tau
=
\frac{1}{n}\sum_{i=1}^n\widehat\tau_i,
\]
where \(\widehat\tau_i\) and \(\widehat\tau\) denote the fitted CATE for
observation \(i\) and the estimated ATE, respectively.

For each dataset, we drew \(500\) nonparametric bootstrap samples. In every
bootstrap sample, we relearned COMPACT, refit the downstream OLS model, and
calculated the bootstrap ATE \(\widehat\tau_b^*\). Because the
bootstrap distribution is centered at \(\widehat\tau\), we defined
$\widetilde\tau_b^*
=
\widehat\tau_b^*-\widehat\tau.$
This recentring produced a bootstrap reference distribution whose null value
was \(\tau=0\). We used this distribution to construct confidence
intervals and to compare the observed null deviation \(\widehat\tau-0\) with
the bootstrap deviations \(\widetilde\tau_b^*\).

We used the analogous construction for the global CATE test. Define the
population score-indexed heterogeneity function as
\[
h_0(\bm s)
=
\tau_{\bm S}(\bm s)-\tau,
\qquad
\tau_{\bm S}(\bm s)
=
\mathbb E\!\left\{
Y_i(1)-Y_i(0)
\mid
\bm S_i=\bm s
\right\}.
\]
The global null was
\[
H_0:
h_0(\bm S_i)=0
\quad\text{for every }i=1,\ldots,n,
\]
and its estimated original-sample analog was
$\widehat h_i
=
\widehat\tau_i-\widehat\tau.$ We measured the observed departure from the null function \(h_0\equiv0\) using
\[
Q_n
=
\sum_{i=1}^n\widehat h_i^2.
\]

For bootstrap sample \(b\), let \(I_{bj}\) denote the original-sample
observation selected at bootstrap position \(j\), and define the bootstrap
estimated heterogeneity as
$\widehat h_{bj}^*
=
\widehat\tau_{bj}^*-\widehat\tau_b^*.$
Because the bootstrap distribution of \(\widehat{\bm h}_b^*\) is
centered at the original estimate \(\widehat{\bm h}\), we recentred it at the null
function \(h_0\equiv0\) by defining
$\widetilde h_{bj}^*
=
\widehat h_{bj}^*
-
\widehat h_{I_{bj}}.$
The corresponding bootstrap statistic was
\[
Q_{n,b}^*
=
\sum_{j=1}^n
\left(\widetilde h_{bj}^*\right)^2
=
\sum_{j=1}^n
\left(
\widehat h_{bj}^*
-
\widehat h_{I_{bj}}
\right)^2.
\]
Thus, \(Q_n\) measured the observed departure from \(h_0\equiv0\), whereas
\(Q_{n,b}^*\) measured the bootstrap departure after recentring the bootstrap
reference distribution at \(h_0\equiv0\). We rejected the global null when
\(Q_n\) exceeded the \(95\)th percentile of the bootstrap statistics. We
assessed ATE coverage under constant and heterogeneous treatment effects, ATE
type-I error under \(\tau_0=0\), ATE power under \(\tau_0=0.5\), global CATE
type-I error under \(\tau_U=0\), and global CATE power under
\(\tau_U=0.75\).

\newpage 
\section{Extra Results for Fully Synthetic Data}

\begin{figure}[H]
    \centering
    \includegraphics[width=1\linewidth]{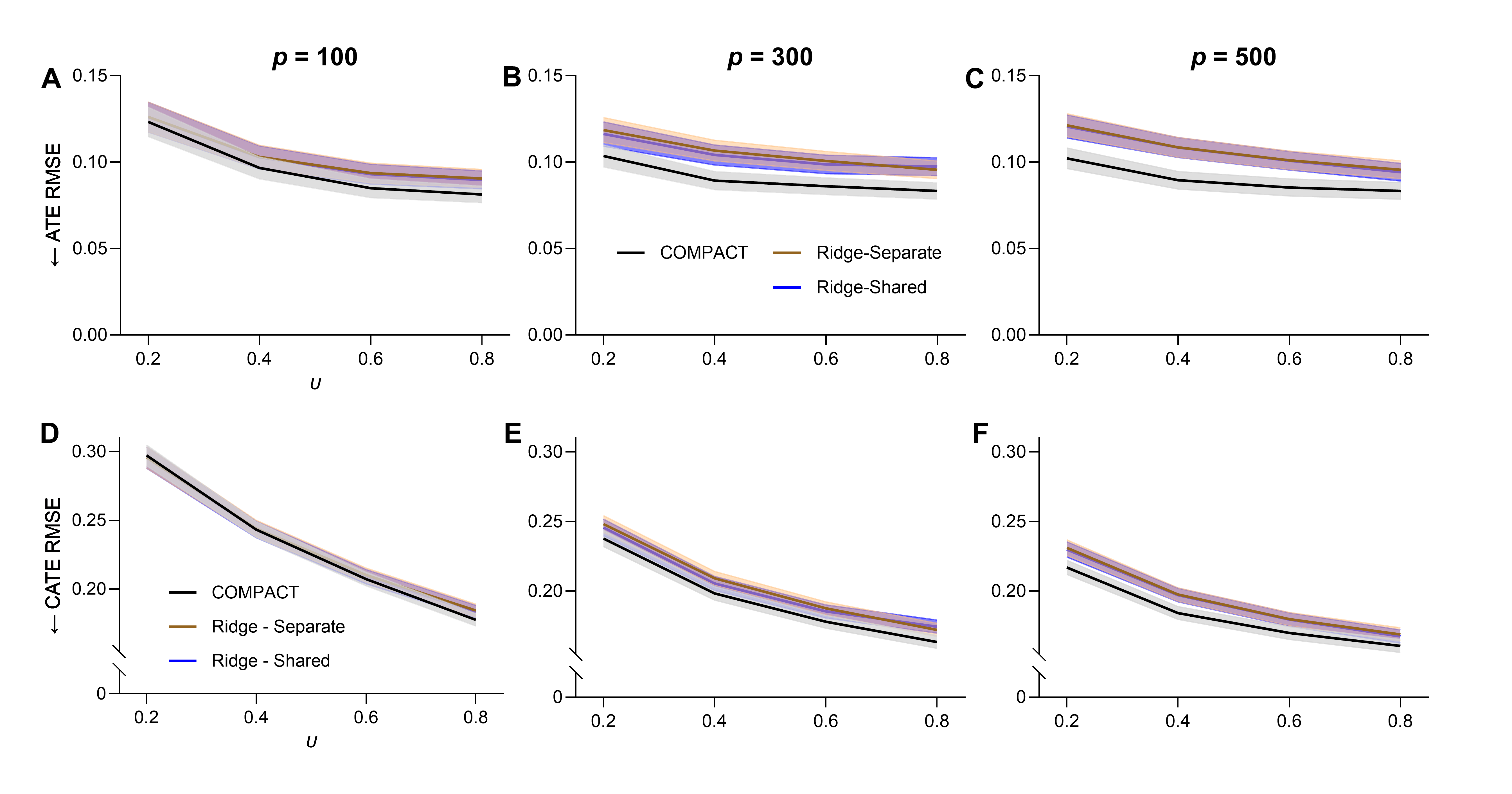}
    \caption{\textbf{Treatment-effect estimation comparing COMPACT with ridge regression.}
    ATE RMSE (A--C) and CATE RMSE (D--F) across proxy dimensions and reliability levels for COMPACT, Ridge-Separate (independently tuned penalties), and Ridge-Shared (a common penalty).}
    \label{fig:supp_sim_ATE_CATE_RMSE_ridge}
\end{figure}

\begin{figure}[H]
    \centering
    \includegraphics[width=1\linewidth]{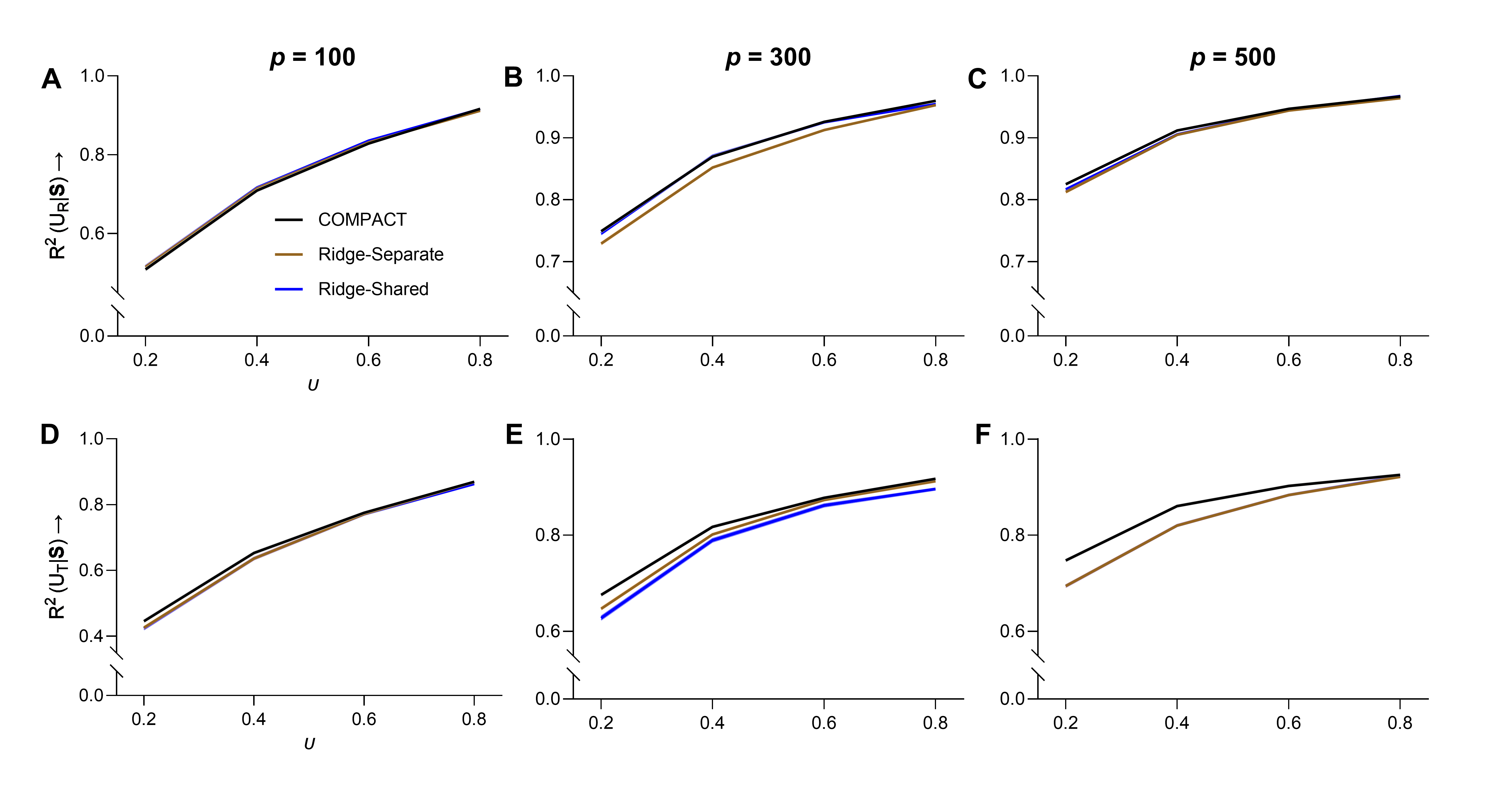}
    \caption{\textbf{Latent adjustment-space recovery comparing COMPACT with ridge regression.}
    Linear recovery of the residual-outcome coordinate $R^2(U_R\mid\bm S)$ (A--C) and treatment-balancing coordinate $R^2(U_T\mid\bm S)$ (D--F) across proxy dimensions and reliability levels for COMPACT, Ridge-Separate (independently tuned penalties), and Ridge-Shared (a common penalty).}
    \label{fig:supp_sim_R2_ridge}
\end{figure}

\begin{figure}[H]
    \centering
    \includegraphics[width=1\linewidth]{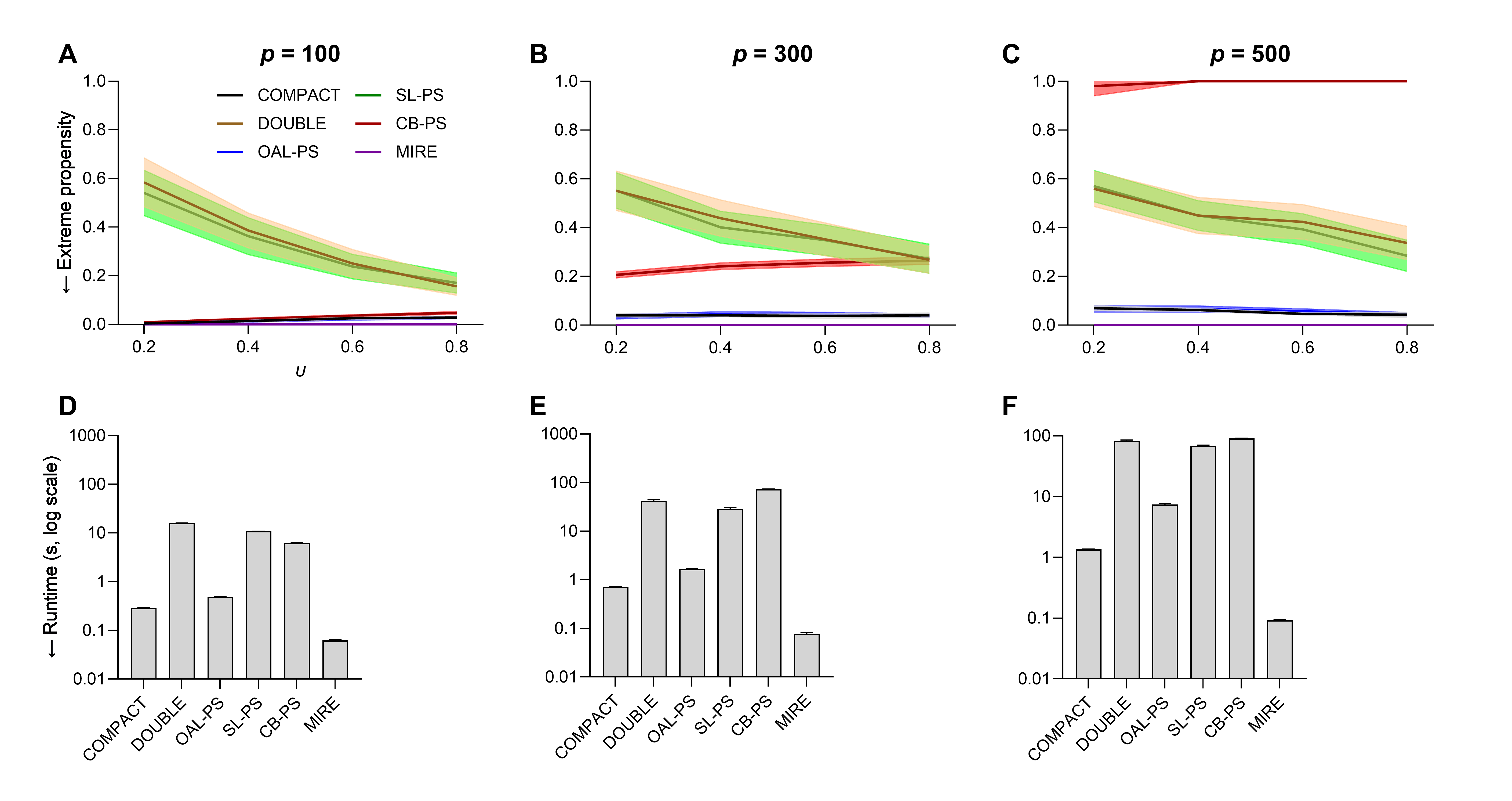}
    \caption{\textbf{Score-space overlap and computation.} Proportion of
    observations with fitted propensities below \(0.05\) or above \(0.95\)
    (A--C) and mean total runtime on a logarithmic scale (D--F) for
    \(p=100\), \(300\), and \(500\). Runtime was first averaged across
    reliability levels within each replication. Shaded bands and error bars
    denote \(95\%\) confidence intervals.}
    \label{fig:synth_overlap_runtime}
\end{figure}

\begin{figure}[H]
\centering
\includegraphics[width=0.8\linewidth]{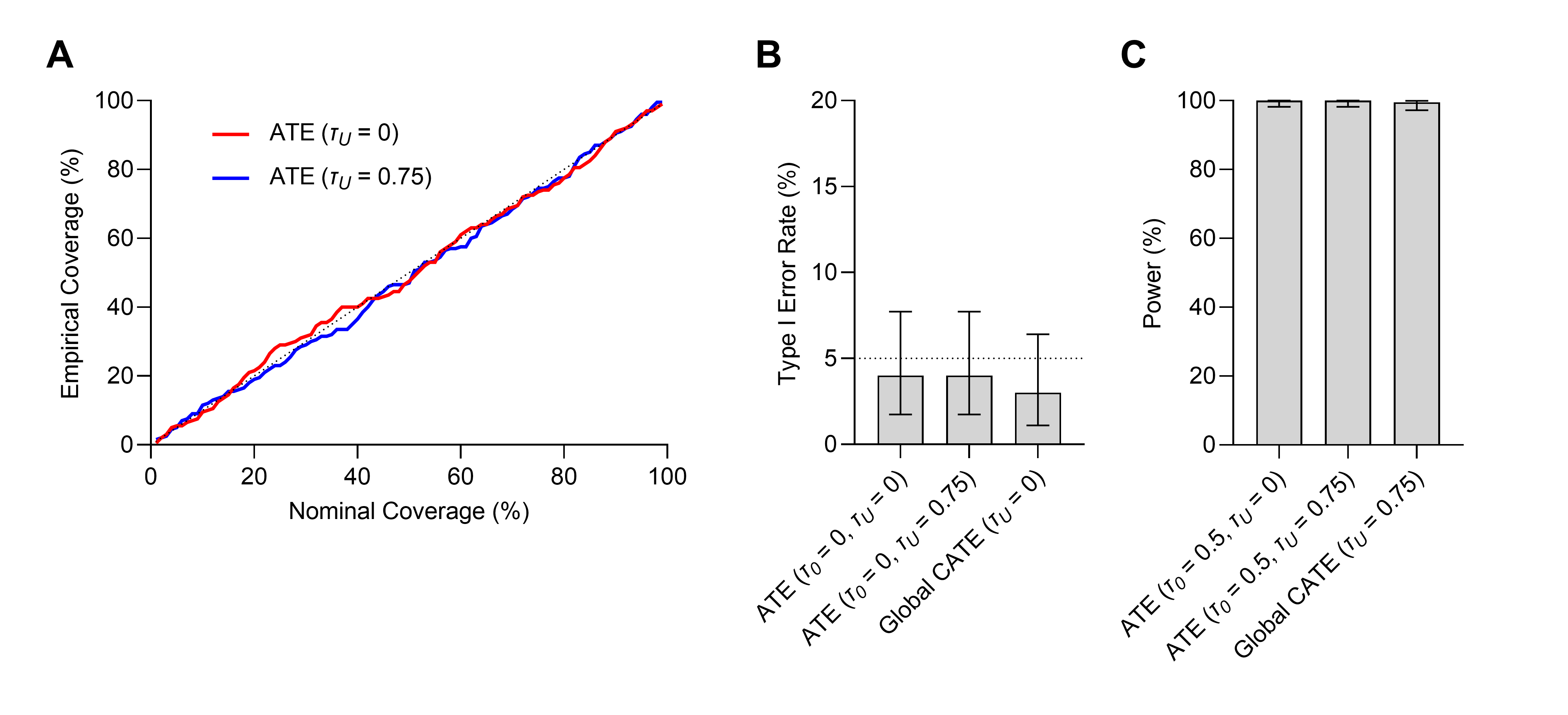}
\caption{\textbf{Full-procedure bootstrap inference.}
(A) Empirical coverage of the ATE confidence interval across nominal
coverage levels under constant treatment effects (\(\tau_U=0\)) and
heterogeneous treatment effects (\(\tau_U=0.75\)). The dotted diagonal
denotes perfect calibration. Coverage is invariant to the constant ATE
shift $\tau_0$, so one curve is shown for each value of $\tau_U$.
(B) Type-I error for the ATE test of ($H_0:\tau_0=0$) under
$\tau_U=0$ and \(\tau_U=0.75\), and for the global CATE test of
\(H_0:\tau_{\bm S}(\bm S_i)=\operatorname{ATE}\) under \(\tau_U=0\).
The dotted horizontal line denotes the nominal 5\% level.
(C) Power of the ATE test at \(\tau_0=0.5\) under both values of
\(\tau_U\), and power of the global CATE test under \(\tau_U=0.75\).
Error bars denote
95\% confidence intervals.}
\label{fig:bootstrap_inference}
\end{figure}

\begin{figure}[H]
    \centering
    \includegraphics[width=1\linewidth]{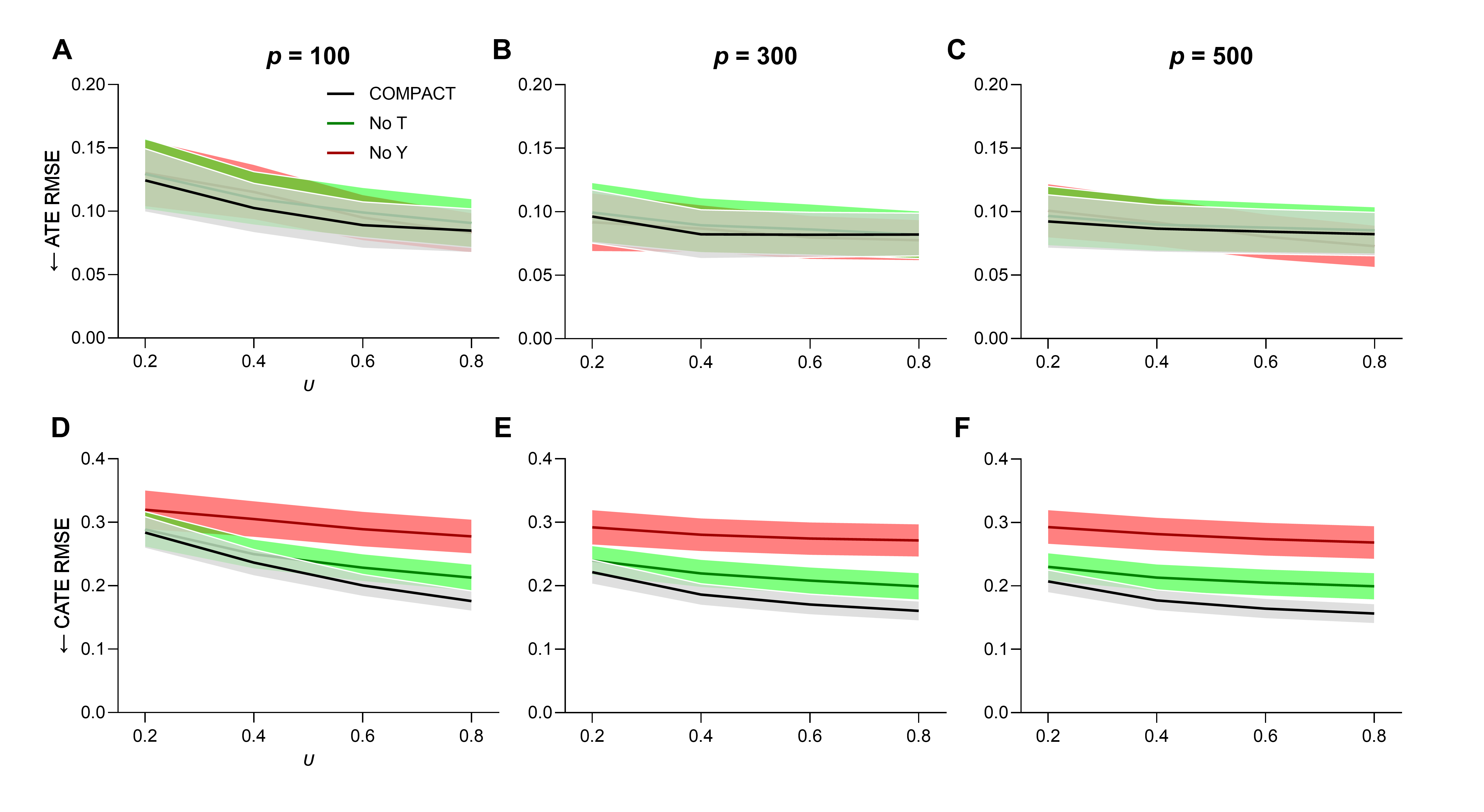}
    \caption{\textbf{Treatment-effect estimation in the fully synthetic ablation experiments.}
    ATE RMSE (A--C) and CATE RMSE (D--F) for full COMPACT and its two ablated
    variants.}
    \label{fig:supp_sim_ATE_CATE_RMSE}
\end{figure}

\begin{figure}[H]
    \centering
    \includegraphics[width=1\linewidth]{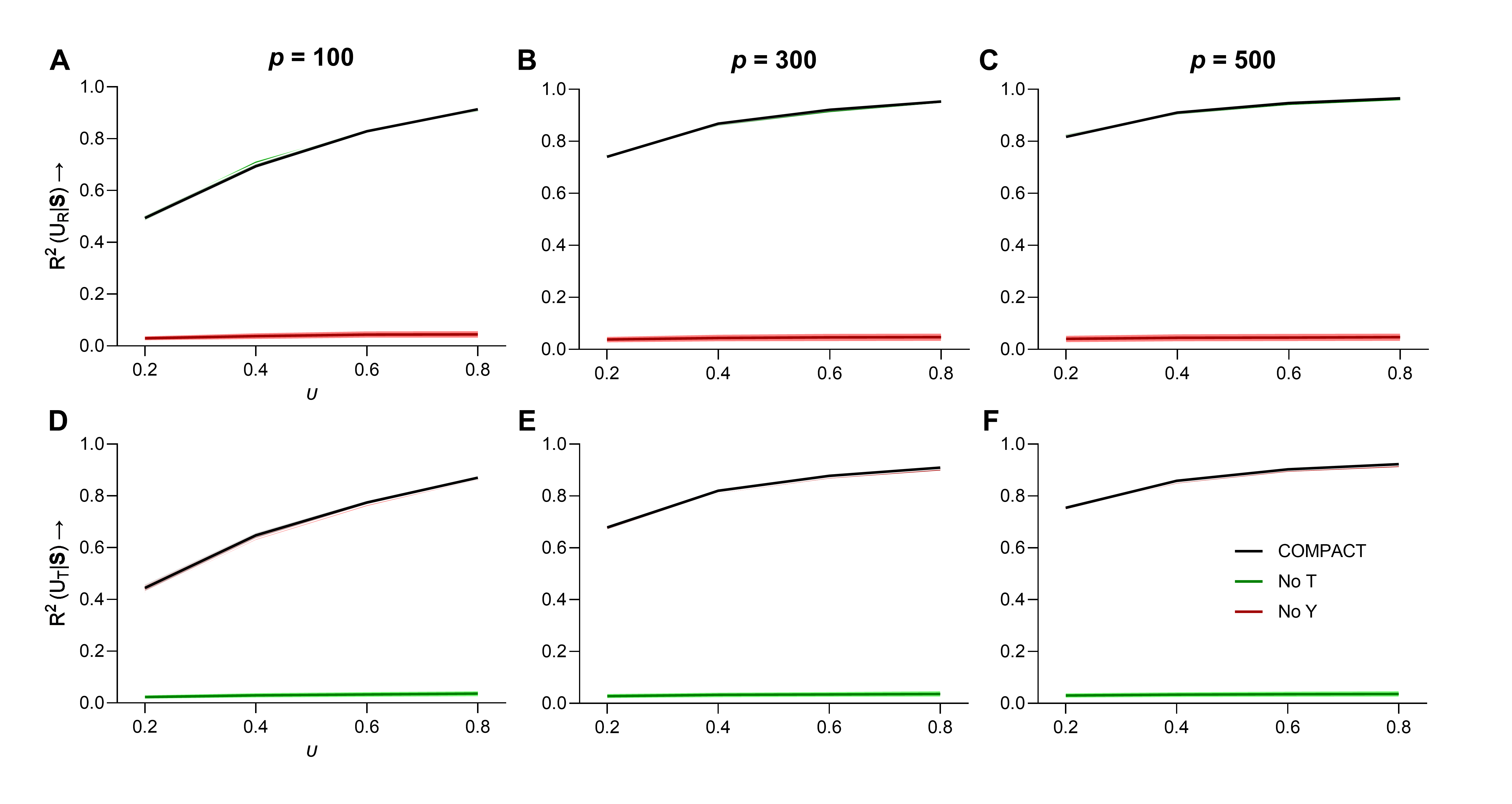}
    \caption{\textbf{Latent-coordinate recovery in the fully synthetic ablation experiments.}
    Linear recovery of \(U_R\) (A--C) and \(U_T\) (D--F) for full COMPACT and
    its two ablated variants.}
    \label{fig:supp_sim_R2}
\end{figure}

\begin{figure}[H]
    \centering
    \includegraphics[width=1\linewidth]{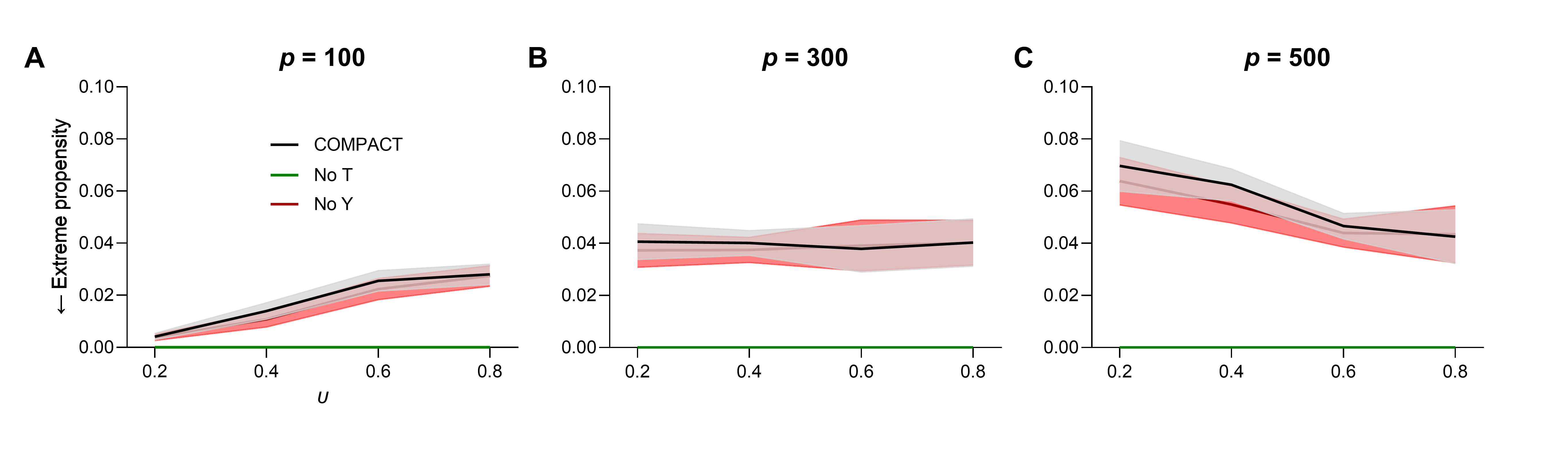}
    \caption{\textbf{Score-space overlap in the fully synthetic ablation experiments.}
    Proportions of extreme fitted propensities for full COMPACT and its two
    ablated variants.}
    \label{fig:supp_sim_propensity_time}
\end{figure}

\section{Extra Results for Semi-Synthetic Data} \label{supp:real_data}

\begin{figure}[H]
    \centering
    \includegraphics[width=1\linewidth]{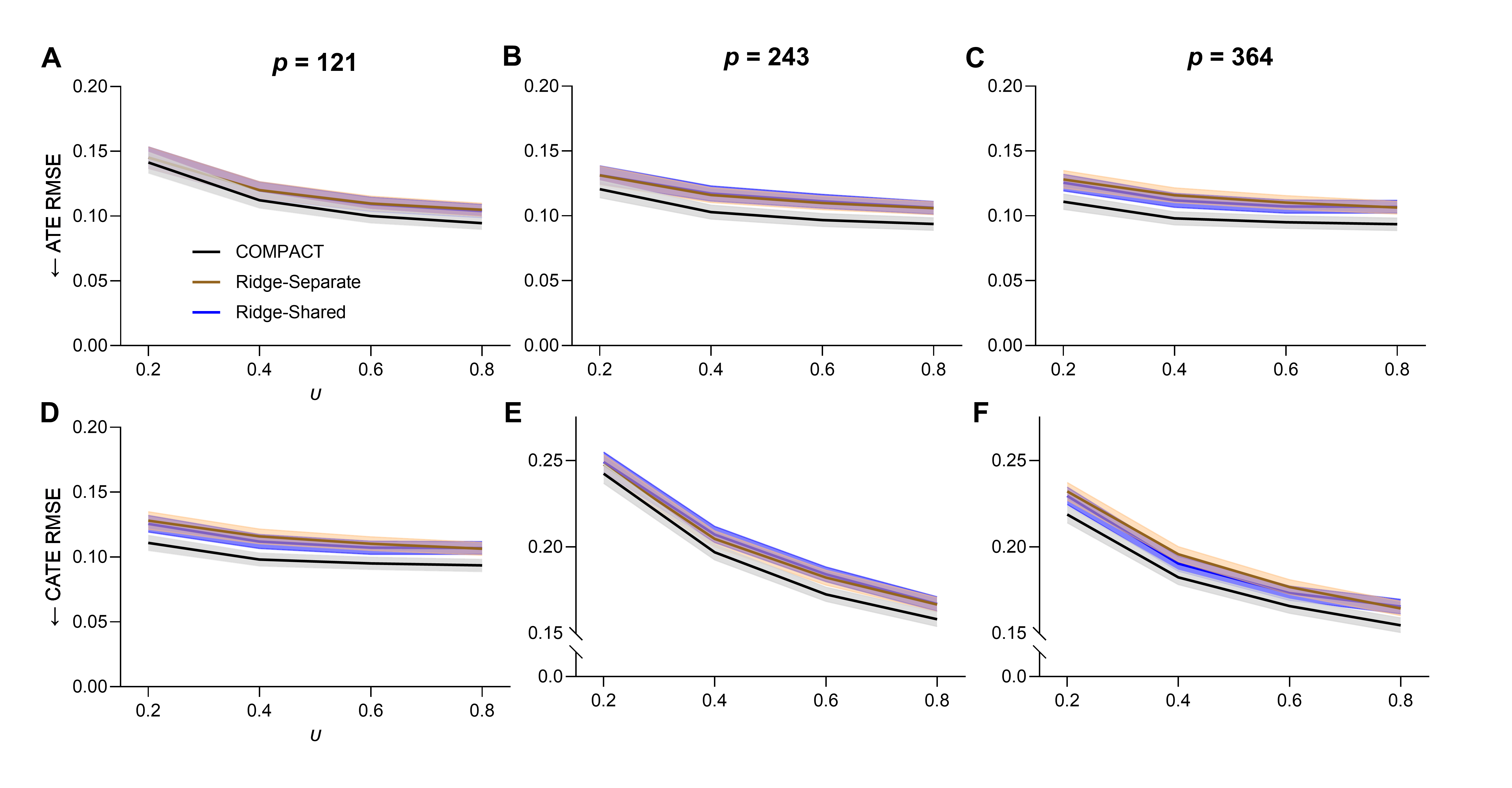}
    \caption{Treatment effect estimation performance in CATIE-calibrated simulations. (A--C) ATE RMSE. (D--F) CATE RMSE.}
    \label{fig:catie_ridge_rmse}
\end{figure}

\begin{figure}[H]
    \centering
    \includegraphics[width=1\linewidth]{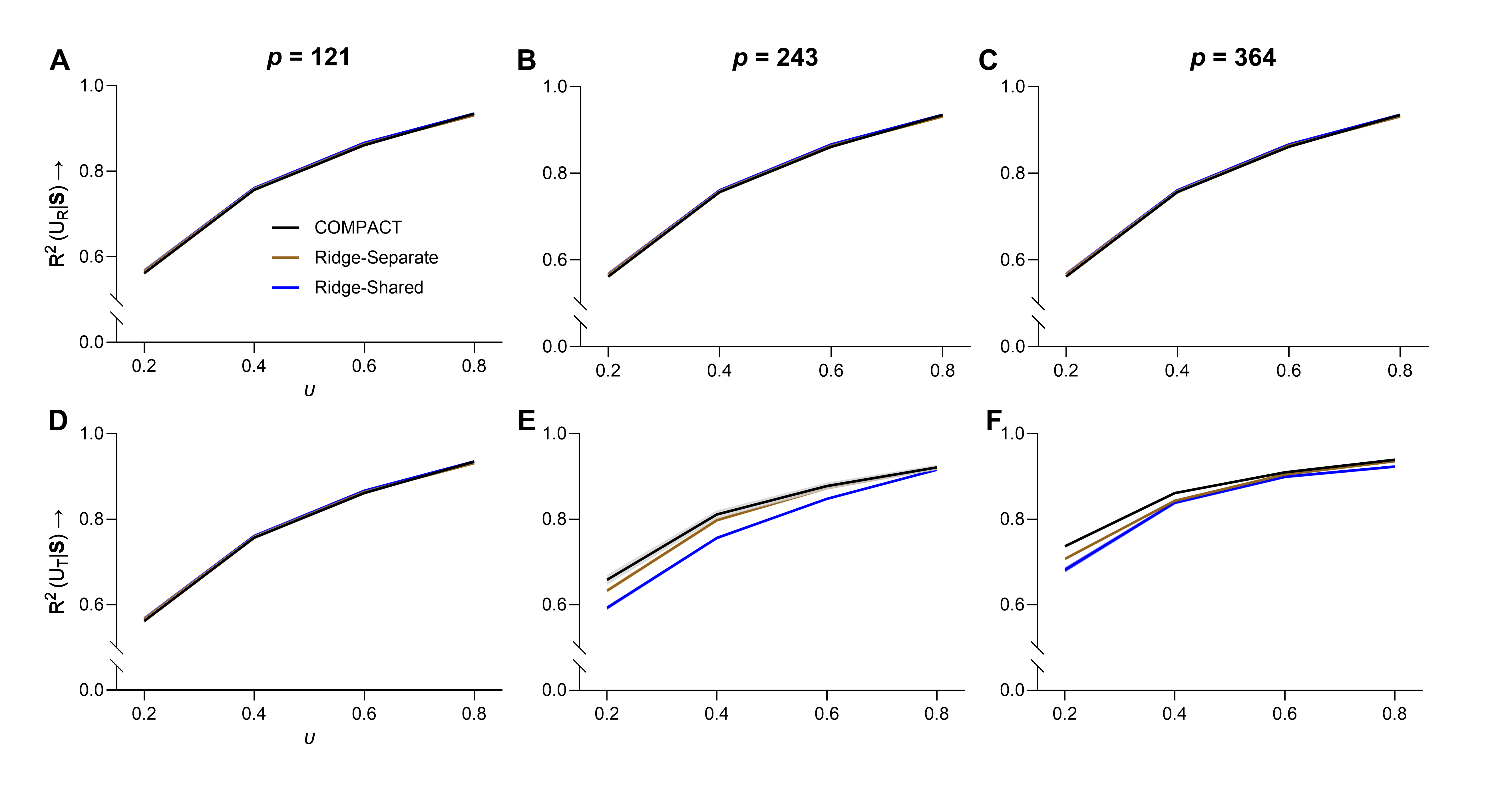}
    \caption{Representation estimation performance in CATIE-calibrated simulations. (A--C) Estimation $R^2(U_R \mid S)$. (D--F) Estimation $R^2(U_T \mid S)$.}
    \label{fig:catie_ridge_r2}
\end{figure}

\begin{figure}[H]
    \centering
    \includegraphics[width=1\linewidth]{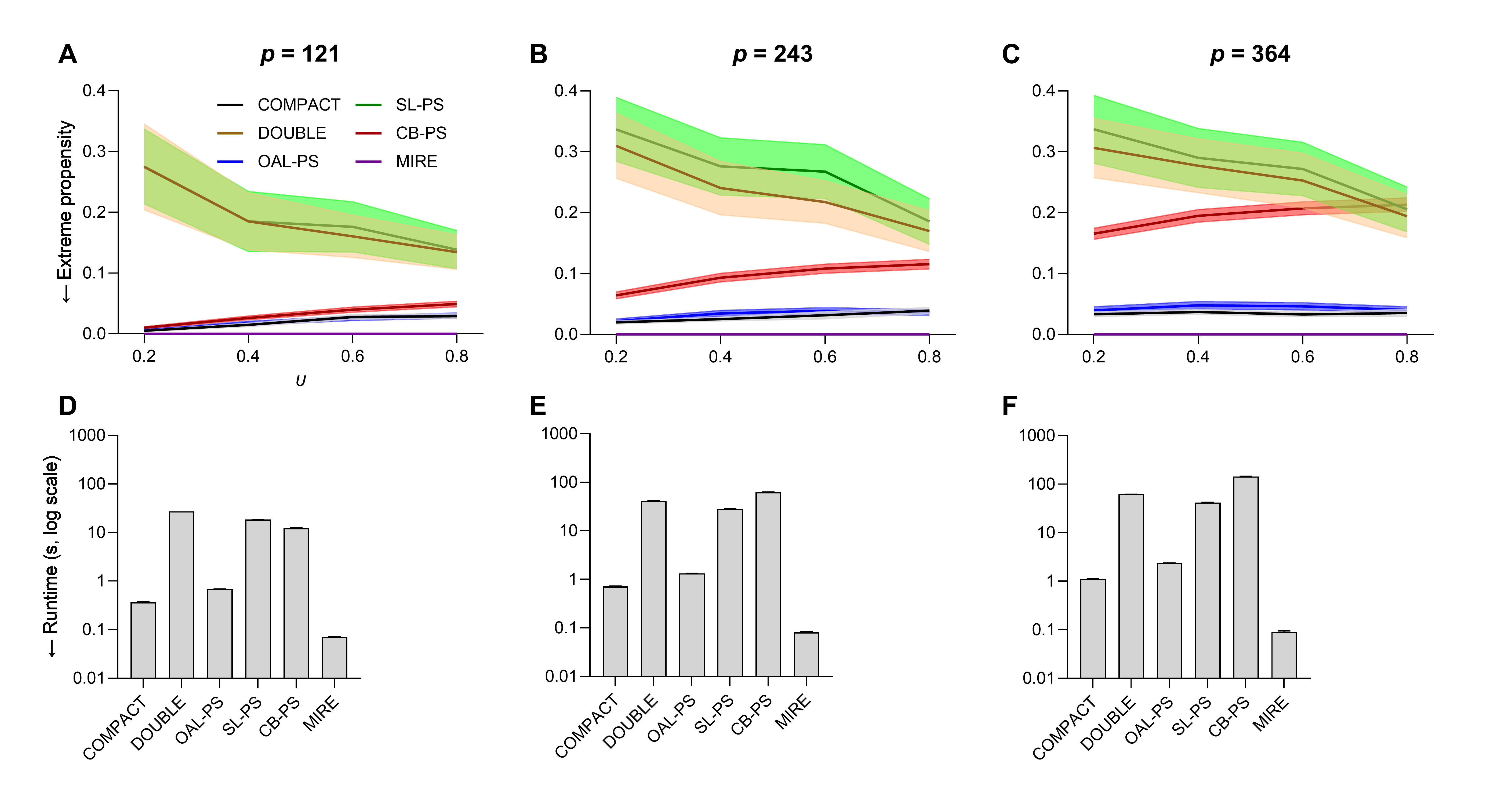}
    \caption{\textbf{Score-space overlap and computation in the
    CATIE-calibrated experiments.} Proportions of extreme fitted propensities
    (A--C) and mean total runtime on a logarithmic scale (D--F) for
    \(p=121\), \(243\), and \(364\).}
    \label{fig:real_propensity_time}
\end{figure}

\begin{figure}[H]
\centering
\includegraphics[width=0.8\linewidth]{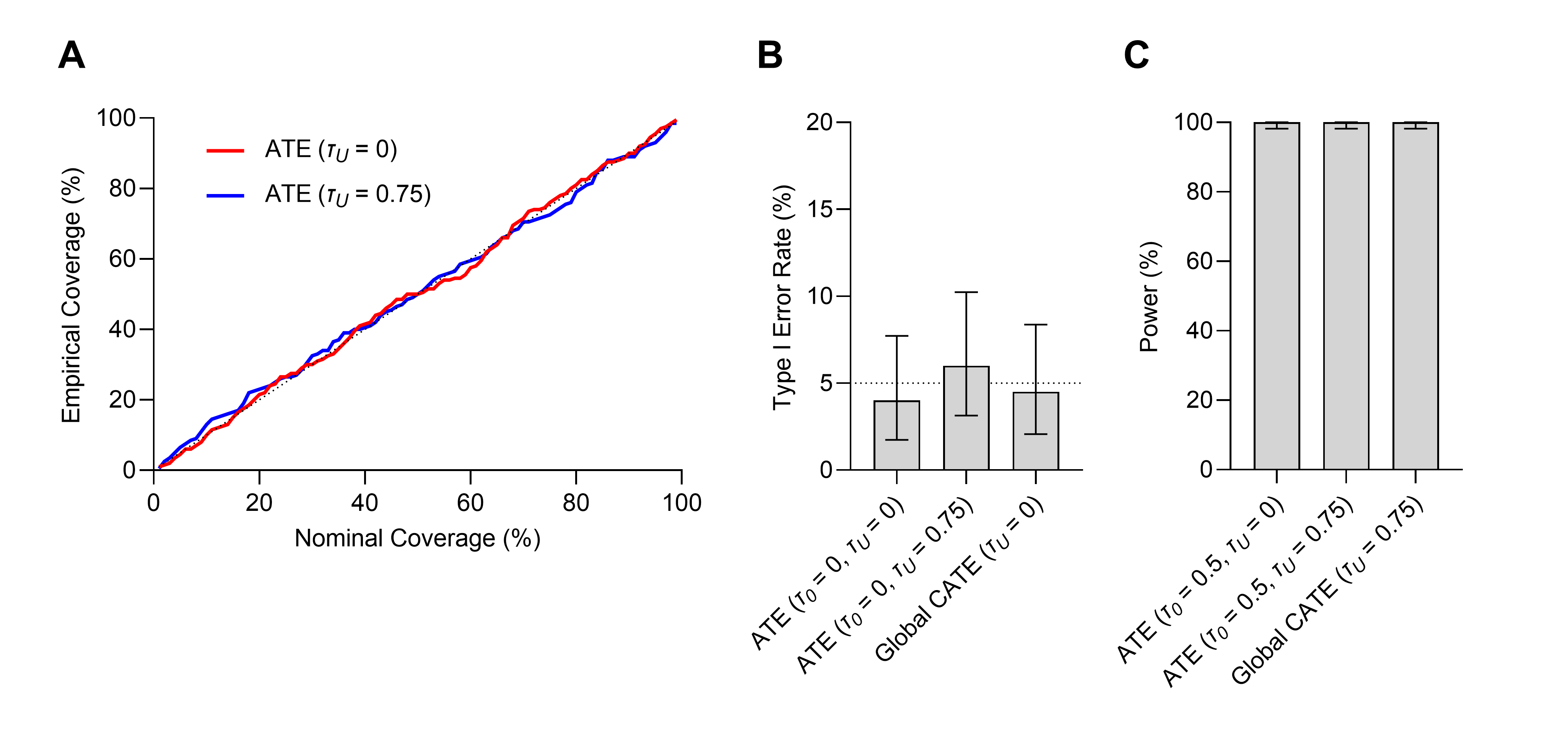}
\caption{\textbf{Bootstrap inference in the
    CATIE-calibrated experiments.}
(A) Empirical coverage of the ATE confidence interval.
(B) Type-I error for the ATE test of ($H_0:\tau_0=0$).
(C) Power of the ATE test at \(\tau_0=0.5\) under both values of
\(\tau_U\), and power of the global CATE test under \(\tau_U=0.75\).}
\label{fig:real_bootstrap_inference}
\end{figure}

\begin{figure}[H]
    \centering
    \includegraphics[width=1\linewidth]{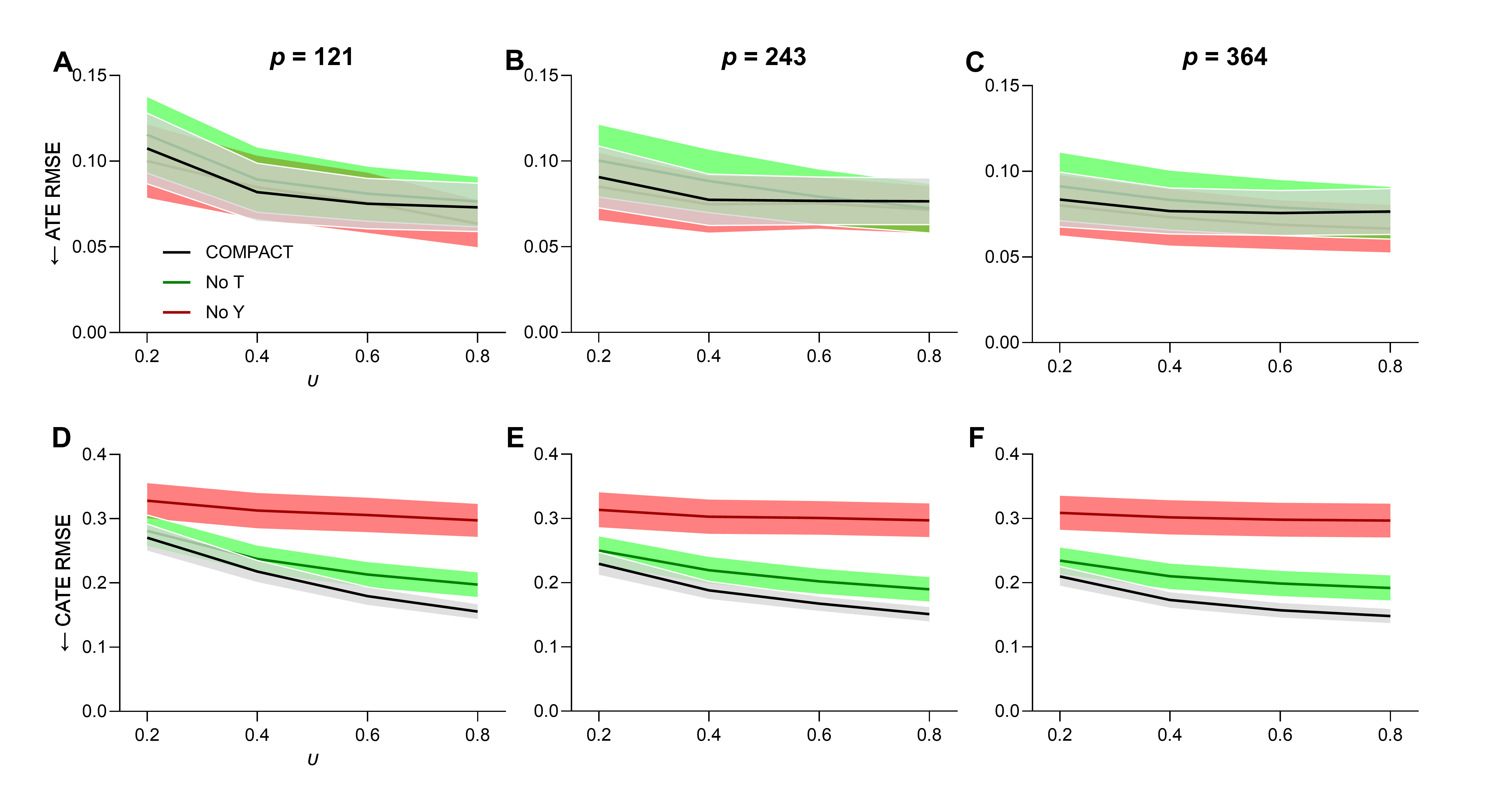}
    \caption{\textbf{Treatment-effect estimation in the semi-synthetic ablation experiments.}
    ATE RMSE (A--C) and CATE RMSE (D--F) for full COMPACT and its two ablated
    variants.}
    \label{fig:real_Supp_ATE_CATE}
\end{figure}

\begin{figure}[H]
    \centering
    \includegraphics[width=1\linewidth]{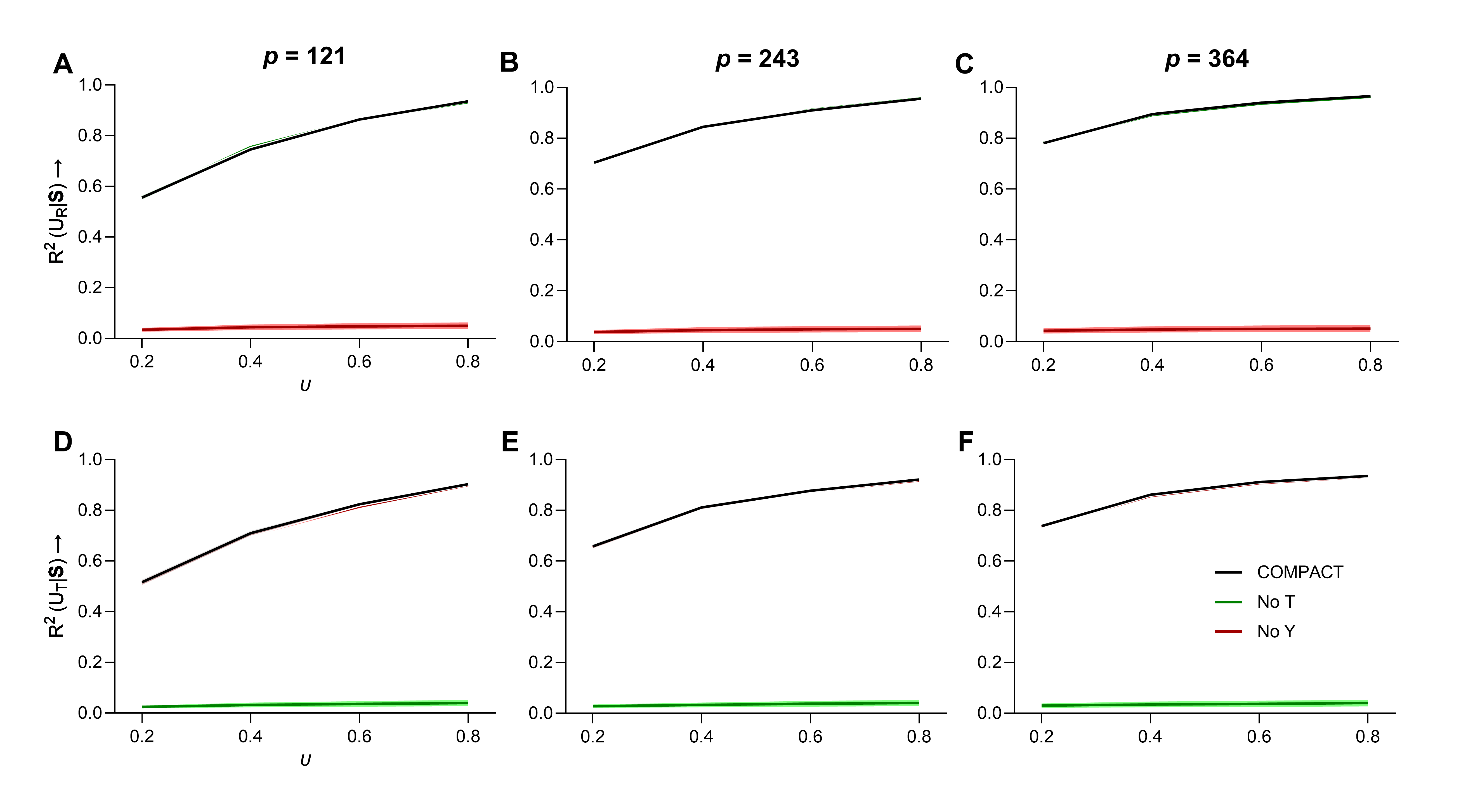}
    \caption{\textbf{Latent-coordinate recovery in the semi-synthetic ablation experiments.}
    Linear recovery of \(U_R\) (A--C) and \(U_T\) (D--F) for full COMPACT and
    its two ablated variants.}
    \label{fig:real_Supp_R2}
\end{figure}

\begin{figure}[H]
    \centering
    \includegraphics[width=1\linewidth]{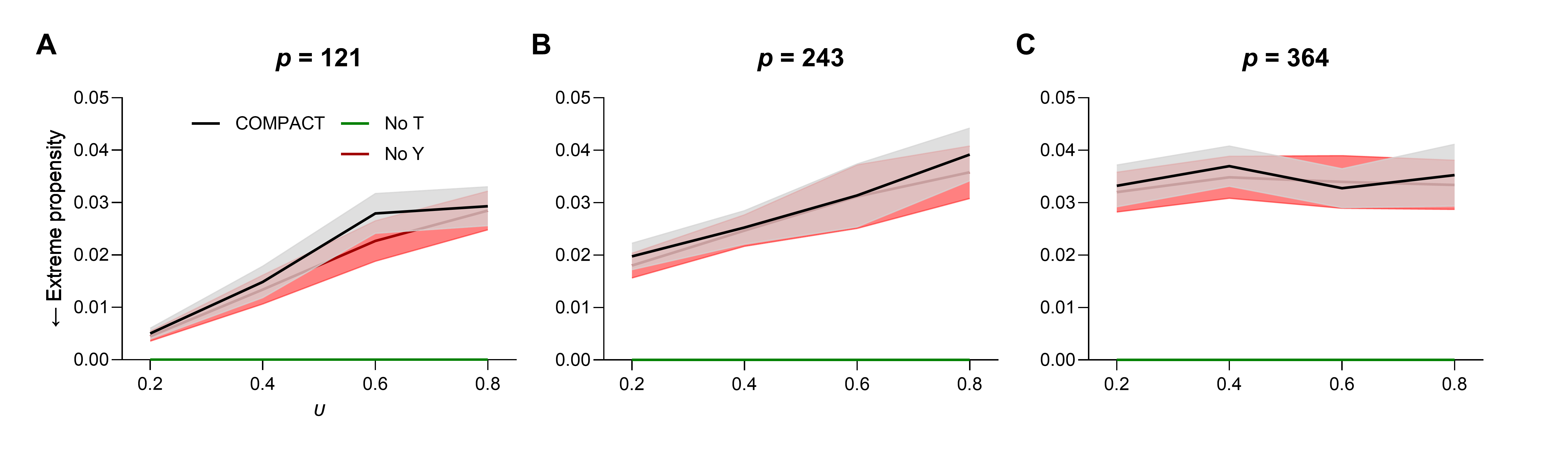}
    \caption{\textbf{Score-space overlap in the semi-synthetic ablation experiments.}
    Proportions of extreme fitted propensities for full COMPACT and its two
    ablated variants.}
    \label{fig:real_Supp_Extreme}
\end{figure}

\end{document}